%% file: main.tex
\documentclass{article}
\input{commands.tex}

\usepackage{cite}

\title{Orthogonal Approximate Message Passing with Optimal Spectral Initializations for Rectangular Spiked Matrix Models}

\author{Haohua Chen\thanks{Academy of Mathematics and Systems Science, Chinese Academy of Sciences. Email: \texttt{chenhaohua25@mails.ucas.ac.cn}}
\and Songbin Liu\thanks{Department of Electrical Engineering, Columbia University. Email: \texttt{sl5878@columbia.edu}}
\and Junjie Ma\thanks{Academy of Mathematics and Systems Science, Chinese Academy of Sciences. Email: \texttt{majunjie@lsec.cc.ac.cn}}}

\graphicspath{{figures/}}

\begin{document}
\maketitle
\begin{abstract}

We propose an orthogonal approximate message passing (OAMP) algorithm for signal estimation in the rectangular spiked matrix model with general rotationally invariant (RI) noise. We establish a rigorous state evolution that precisely characterizes the algorithm’s high-dimensional dynamics and enables the construction of iteration-wise optimal denoisers. Within this framework, we accommodate spectral initializations under minimal assumptions on the empirical noise spectrum. In the rectangular setting, where a single rank-one component typically generates multiple informative outliers, we further propose a procedure for combining these outliers under mild non-Gaussian signal assumptions. For general RI noise models, the predicted performance of the proposed optimal OAMP algorithm agrees with replica-symmetric predictions for the associated Bayes-optimal estimator, and we conjecture that it is statistically optimal within a broad class of iterative estimation methods.
\end{abstract}
\tableofcontents

\input{body/intro}

\input{body/notations}

\input{body/prelims}

\input{body/main_results}

\input{body/example}

\input{body/spectral_combo}

\input{body/spectral_init}

\input{body/simulation}

\bibliographystyle{ieeetr}
\bibliography{bib}

\appendix
\input{appendix/ap-prelims3}

\input{appendix/ap-SE}

\input{appendix/ap-Denoisers2}

\input{appendix/ap-OptimalSE}

\input{appendix/ap-Example}

\input{appendix/ap-spec_combo}

\input{appendix/ap-Global_Sign2}

\input{appendix/ap-spec.init2}
\input{appendix/misc}

\end{document}

%% file: commands.tex
\makeatletter
\renewcommand\paragraph{\@startsection{paragraph}{4}{\z@}%
                                    {1ex \@plus1ex \@minus.2ex}%
                                    {-1em}%
                                    {\normalfont\normalsize\bfseries}}
\makeatother
{\begin{list}               %    with flush left bullets.
    {$\bullet$ \hfill}{
        \setlength{\leftmargin}{\parindent}
        \setlength{\parsep}{0.04\baselineskip}
        \setlength{\itemsep}{0.5\parsep}
        \setlength{\labelwidth}{\leftmargin}
        \setlength{\labelsep}{0em}}
    }
{\end{list}}

\usepackage{amsmath}
\usepackage{amssymb}
\usepackage{amsthm}
\usepackage[margin=1in]{geometry}
\usepackage{enumitem}
\usepackage{hyperref}
 \hypersetup{
     colorlinks=true,
     linkcolor=blue,
     filecolor=blue,
     citecolor = blue,      
     urlcolor=black,
     }
\usepackage{bm}
\usepackage{mathrsfs}
\usepackage[scr=boondox]{mathalpha}
 \usepackage{booktabs} 
\usepackage{extarrows}
\usepackage{graphicx}
\usepackage{subfig}
\usepackage{algorithmic,algorithm}
\usepackage{cmll}
\usepackage{mathtools}

\usepackage[cmyk]{xcolor} 
\usepackage[title]{appendix}

\usepackage{cleveref}
\newcommand{\supp}{\mathrm{supp}}

\newtheorem{theorem}{Theorem}
\newtheorem{proposition}{Proposition}
\newtheorem{lemma}{Lemma}
\newtheorem{corollary}{Corollary}

\theoremstyle{definition}

\newtheorem{definition}{Definition}
\newtheorem{assumption}{Assumption}

\theoremstyle{remark}
\newtheorem{remark}{Remark}
\newtheorem{fact}{Fact}

\providecommand{\sref}[1]{Section~\ref{#1}}
\providecommand{\appref}[1]{Appendix~\ref{#1}}

\providecommand{\thref}[1]{Theorem~\ref{#1}}
\providecommand{\defref}[1]{Definition~\ref{#1}}
\providecommand{\lemref}[1]{Lemma~\ref{#1}}

\providecommand{\assumpref}[1]{Assumption~\ref{#1}}

\providecommand{\factref}[1]{Fact~\ref{#1}}
\providecommand{\propref}[1]{Proposition~\ref{#1}}

\providecommand{\R}{\ensuremath{\mathbb{R}}}
\providecommand{\C}{\ensuremath{\mathbb{C}}}
\providecommand{\N}{\ensuremath{\mathbb{N}}}
\renewcommand{\vec}[1]{\ensuremath{\boldsymbol{#1}}}

\providecommand{\vu}{\vec{u}}

\providecommand{\vv}{\vec{v}}

\DeclareMathOperator{\Tr}{Tr}

\DeclareMathOperator{\op}{op}

\newcommand{\explain}[2]{\overset{\text{\tiny{#1}}}{#2}} %to provide explanations
\newcommand{\bydef}{\stackrel{\mathrm{def}}{=}}

\newcommand{\E}{\mathbb{E}} %% for expectations
\newcommand{\Var}{\mathrm{Var}}
\newcommand{\Cov}{\mathrm{Cov}}
\newcommand{\ac}{\overset{\text{a.s.}}{\longrightarrow}} %a.s. conv
\newcommand{\wc}{\overset{W}{\longrightarrow}}

\renewcommand{\tilde}{\widetilde}
\renewcommand{\hat}{\widehat}

\DeclareFontFamily{U}{mathx}{\hyphenchar\font45}
\DeclareFontShape{U}{mathx}{m}{n}{
      <5> <6> <7> <8> <9> <10>
      <10.95> <12> <14.4> <17.28> <20.74> <24.88>
      mathx10
      }{}
\DeclareSymbolFont{mathx}{U}{mathx}{m}{n}
\DeclareFontSubstitution{U}{mathx}{m}{n}
\DeclareMathAccent{\widecheck}{0}{mathx}{"71}
\DeclareMathAccent{\wideparen}{0}{mathx}{"75}

\newcommand{\serv}[1]{\mathsf{#1}} %for state evolution random variables

\newcommand{\BE}{\begin{equation}}
\newcommand{\EE}{\end{equation}}
\newcommand{\BS}{\begin{subequations}}
\newcommand{\ES}{\end{subequations}}
\newcommand{\UT}{\mathsf{T}}   %% upright T

\providecommand{\muM}{\mu}
\providecommand{\muN}{\widetilde{\mu}}

\providecommand{\muM}{\mu}

\newcommand{\paraB}[1]{\bm{\mathsf{#1}}}

\providecommand{\ConvMode}{\mathrm{a.s.}}
\DeclareMathOperator{\sign}{\mathrm{sign}\,}
\providecommand{\ase}{\stackrel{\mathrm{a.s.}}{=}}

\providecommand{\Smu}{\mathcal{S}_{\mu}}

\providecommand{\Su}{\mathcal{S}_{\nu_1}}

\providecommand{\Sv}{\mathcal{S}_{\nu_2}}

\providecommand{\Suv}{\mathcal{S}_{\nu_3}}

\providecommand{\CT}{\mathcal{C}}
\newcommand{\wpc}{\overset{W_p}{\longrightarrow}}

\providecommand{\card}[1]{\left\lvert #1 \right\rvert}
\newcommand{\spp}{\mathrm{sp}\,}
\newcommand{\indep}{\mathrel{\perp\mspace{-10mu}\perp}}
\usepackage{float}
\newcommand{\argmax}{\mathrm{argmax}\,}

\newcommand{\proofseeapp}[2]{%
  \begin{proof}%
  See Appendix~\ref{#2}.%
  \end{proof}%
}

%% file: body/intro.tex
\section{Introduction}

We study the estimation of asymmetric rank-one signals
$\bm{u}_*\in\mathbb{R}^M$ and $\bm{v}_*\in\mathbb{R}^N$ generated from the \emph{rectangular spiked model}
\begin{equation}
    \bm{Y}
    = \frac{\theta}{\sqrt{MN}}
      \bm{u}_* \bm{v}_*^\UT + \bm{W}
      \in \mathbb{R}^{M\times N},
    \label{eq:rectangular spiked model}
\end{equation}
where $\theta>0$ is a signal-to-noise ratio (SNR) and $\bm{W}$ is a noise
matrix.  This model is widely used to analyze high-dimensional data in which
the number of features and samples are of comparable scale, with applications
ranging from financial data analysis~\cite{plerou2002financial,Bouchaud2009Financial}
to community detection~\cite{Krzakala2013clustering,Abbe2016block}.

In the classical setting where $\bm{W}$ has i.i.d.\ Gaussian entries, the
fundamental behavior of PCA is well understood.  A sharp phase transition
governs when the leading singular vectors correlate with the underlying
signals~\cite{baik2005phase,feral2007largest,paul2007asymptotics,Bai2010SpectralAO}.
These guarantees can be improved by incorporating structural priors such as
sparsity~\cite{zou2006sparse,deshpande2014information,lesieur2015phase,
montanari2015non,lesieur2017constrained} or Bayesian assumptions
\cite{bishop1999bayesian,lawrence2009nonlinear,fletcher2018iterative}.
Approximate Message Passing (AMP) algorithms
\cite{donoho2009message,bayati2011dynamics,rangan2012iterative,krzakala2016mutual,dia2016mutual,ma2017orthogonal,rangan2019vector,maillard2019high,takeuchi2019rigorous,miolane2017fundamental,takeuchi2020convolutional,montanari2021estimation,takeuchi2021bayes,fan2022approximate,liu2022memory,gerbelot2023graph,dudeja2023universality,wang2024universality,lovig2025universality}, play an important role in these settings: in the
high-dimensional limit, their empirical performance is exactly described by a
deterministic state evolution (SE) recursion.  This property has enabled
rigorous optimality guarantees among large classes of first order methods
\cite{montanari2021estimation} and low-degree polynomial estimators
\cite{montanari2025equivalence}, and agreement with replica predictions for the
minimum mean square error (MMSE) in certain regimes~\cite{dia2016mutual,miolane2017fundamental,
el2018estimation,lelarge2019Fundamental}.

Naturally, in many practical high-dimensional settings, the noise structure
deviates from the idealized i.i.d.\ Gaussian assumption.  Rotationally invariant
(RI) noise models, in which the singular vectors of $\bm{W}$ are Haar
distributed and independent of its singular values, admit arbitrary limiting
spectra and therefore form a broad and expressive class for high-dimensional
noise. For such models, the PCA outlier behavior is well understood
\cite{Benaych-Georges2011eigenvalues,Benaych-Georges2012singular,capitaine2018spectrum,noiry2021spectral}, and several
AMP-type extensions have been proposed
\cite{fan2022approximate,zhong2021approximate,barbier2023fundamental}, though their MSE optimality
remains unresolved. For the \textit{symmetric}
counterpart of \eqref{eq:rectangular spiked model},
\cite{dudeja2024optimality} introduced an orthogonal approximate message
passing (OAMP) algorithm \cite{ma2017orthogonal,rangan2019vector} in which each
iteration applies a matrix denoiser tailored to the limiting noise spectrum, and
established a state evolution characterization together with optimality
guarantees. The form of these matrix denoisers is closely related to classical
matrix denoising and covariance shrinkage procedures
\cite{ledoit2012nonlinear,Nada2014Optshrink,Bun2016RIE,pourkamali2023rectangular},
and the performance achieved in \cite{dudeja2024optimality} was shown to match
replica-symmetric predictions for the Bayes-optimal performance in certain regimes \cite{barbier2023fundamental,Barbier2025TAP}.

When turning to rectangular models, an additional difficulty arises from the
structure of spectral information itself.  Unlike the i.i.d.\ Gaussian case,
where a rank-one signal produces a single informative outlier (see
Remark~\ref{rem:mp-no-left}), rectangular RI models can generate multiple outlier singular values
\cite{Benaych-Georges2012singular,belinschi2017outliers}, distributing the
signal energy across them.  Standard PCA is therefore suboptimal, as the
principal components may not be the most informative
\cite{Nadakuditi2013informative}, with the signal energy often spread across
several singular directions \cite{monasson2015estimating,fan2021principal}.
These works point to the need for combining all informative singular vectors,
while the development of a practically executable aggregation scheme remains
open.

A natural and effective initialization for AMP is based on the principal components (PCs) of the data \cite{mondelli2019fundamental,mondelli2021approximate,venkataramanan2022estimation,lu2020phase,chen2021spectral}. However, this approach induces dependence on the noise matrix $\bm{W}$, violating a crucial assumption underlying standard state evolution analysis. In the i.i.d.\ Gaussian setting, this difficulty was resolved using a decoupling technique that separates the PCs from the spectral bulk \cite{montanari2021estimation}. This approach relies critically on the entrywise independence of $\bm{W}$ and does not directly extend to general rotationally invariant (RI) ensembles. For RI noise, subsequent work developed a different
approach based on a two-phase artificial AMP construction
\cite{mondelli2021pca,zhong2021approximate}, in which an auxiliary AMP with a
noise-independent initialization is designed to converge to the empirical PCA
estimator.  Variants of this method have also proved useful in Gaussian
generalized linear models \cite{mondelli2021approximate,venkataramanan2022estimation},
although the analyses in both the matrix and GLM settings rely on additional
technical conditions, such as non-negative free cumulants \cite{mondelli2021pca}
or sufficiently large SNR \cite{zhong2021approximate}. 

\paragraph{Our Contributions.}
This paper develops and analyzes an OAMP algorithm for the rectangular spiked model with rotationally invariant (RI) noise. Our main contributions are as follows.
\begin{itemize}
  \item \textbf{Optimal OAMP for Rectangular RI Models:} We extend the OAMP framework \cite{dudeja2024optimality} to the rectangular setting and establish a rigorous state evolution (SE) that characterizes the joint dynamics of the left and right singular vector estimates. This analysis allows us to derive iteration-wise Bayes-optimal matrix and scalar denoisers. We demonstrate that the algorithm's fixed point aligns with replica symmetric predictions for the minimum mean square error (MMSE) {\cite{Chenqun2025}} in the absence of statistical-computational gaps, and in the specific case of i.i.d. Gaussian noise, it recovers the performance of standard AMP.

\item \textbf{Optimal Spectral Initialization with Multiple Outliers:} A key challenge in rectangular RI models is that a single rank-one signal generally generates multiple informative outlier singular values. Standard PCA (using only the top singular vector) is therefore suboptimal. We characterize the theoretically optimal linear combination of all informative outliers. To implement this in practice, we solve the ``relative sign alignment" problem, where the signs of the outliers are unknown, by proposing two methods: a Maximum Likelihood Estimator (MLE) and a computationally efficient Non-Gaussian Moment Contrast (NGMC) scheme. The NGMC method requires only a mild condition (the existence of an even moment distinct from the Gaussian) to asymptotically match oracle performance.
  \item \textbf{Spectrally-Initialized OAMP:} We integrate the optimal spectral estimator as a principled initialization for OAMP. Based on a resolvent reformulation of the singular equation, we show that the spectral step can be viewed as a single OAMP update, thereby incorporating spectral initialization seamlessly into the OAMP framework and eliminating the need for artificial two-phase constructions or additional assumptions (such as nonnegative free cumulants \cite{mondelli2021pca} or sufficiently large signal-to-noise ratios \cite{zhong2021approximate}) used in earlier analyses of spectrally initialized AMP. This formulation yields a state evolution characterization for spectrally initialized OAMP that explicitly accounts for the intrinsic global-sign ambiguity and applies to general rotationally invariant noise models.
\end{itemize}

%% file: body/notations.tex
%%%%% Organization
We conclude this section by introducing the notations used in this paper.\\
\emph{Algebra.} Let $\mathbb{N}$, $\mathbb{R}$, $\mathbb{R_+}$ and $\mathbb{C}$ denote the sets of positive integers, real numbers, non-negative real numbers and complex numbers, respectively. For any $N \in \mathbb{N}$, we define the set $[N] := \{1, 2, 3, \dots, N\}$, and let $\mathbb{O}(N)$ denote the set of $N \times N$ orthogonal matrices. We use bold-face font for vectors and matrices whose dimensions diverge, such as signal vectors $\bm{u}_* \in \mathbb{R}^M, \bm{v}_* \in \mathbb{R}^N$ or a noise matrix $\bm{W} \in \mathbb{R}^{M \times N}$. For objects in a fixed finite dimension $k$, we use regular font. Specifically, for vectors $x, y \in \mathbb{R}^k$, let $\|x\|_2$ denote the $\ell_2$ norm, $\langle x, y \rangle = \sum_{i=1}^k x_i y_i$ be the standard inner product, and $\operatorname{diag}(x)$ be the $k \times k$ diagonal matrix formed by the entries of $x$. For any vector $s\in\mathbb{R}^k$, we write $[s]_j$ for its $j$-th coordinate. For a matrix $M \in \mathbb{R}^{k \times k}$, we write $\operatorname{Tr}(M)$, $\|M\|_{\text{op}}$, and $\|M\|_F$ for its trace, operator (spectral) norm, and Frobenius norm, respectively. If $M$ is symmetric, its eigenvalues are ordered $\lambda_1(M) \ge \dots \ge \lambda_k(M)$, with corresponding eigenvectors $u_1(M), \dots, u_k(M)$. The spectrum of $\bm{M}$ is denoted $\spp (\bm{M})$. $\card{\mathcal{A}}$ denotes the cardinality of a finite set $\mathcal{A}$. $\Im(z)$ denotes the imaginary part of a complex number $z\in\mathbb{C}$.\\
\noindent \emph{Probability and Analysis.} We denote the Gaussian distribution on $\mathbb{R}^k$ with mean vector $\mu \in \mathbb{R}^k$ and covariance matrix $\Sigma \in \mathbb{R}^{k \times k}$ by $\mathcal{N}(\mu, \Sigma)$. For a finite set $A$, $\operatorname{Unif}(A)$ represents the uniform distribution on $A$, and for any $x \in \mathbb{R}$, the measure $\delta_{\{x\}}$ denotes the point mass (or Dirac measure) at $x$. By extension, $\operatorname{Unif}(\mathbb{O}(N))$ denotes the Haar measure on the orthogonal group $\mathbb{O}(N)$. Furthermore, for sequences of random variables, convergence almost surely and in distribution are denoted by $\xrightarrow{a.s.}$ and $\xrightarrow{d}$, respectively. For random variables $X$ and $Y$, we write $X \indep Y$ to denote that $X$ and $Y$ are independent. For a finite measure $\chi$ of bounded variation on a space $\mathcal{E}$ and any bounded, Borel-measurable function $f: \mathcal{E} \to \mathbb{R}$, we denote its integral by $\langle f(x) \rangle_{\chi} \bydef \int_\mathcal{E} f(x)\, d\chi(x)$. The support of a measure $\chi$ is denoted as $\supp(\chi)$. We use $\sign(\cdot)$ to denote the signum function, which returns $1$, $-1$, or $0$ if its argument is positive, negative, or zero, respectively.

%% file: body/prelims.tex
\section{Preliminaries}
\label{sec:preliminaries}

We collect here the main probabilistic and spectral tools used in our analysis of the rectangular spiked
model. We first formalize the high-dimensional asymptotic regime and convergence notions, and then review
the spectral transform machinery associated with the noise spectrum. We introduce signal-projected spectral
measures that encode how the eigenspace of the observation aligns with the true signal directions, and
summarize their limiting behavior and the resulting outlier eigenvalues. These results will be used in the
spectral estimators and as initialization for the OAMP state evolution in
Sections~\ref{sec: spec method} and~\ref{sec: spec OAMP}.

\subsection{Rectangular spiked model and assumptions}

We recall the rectangular spiked model introduced in \eqref{eq:rectangular spiked model}:
\[
\bm{Y}
= \frac{\theta}{\sqrt{MN}}\,\bm{u}_*\bm{v}_*^\UT + \bm{W} \in \mathbb{R}^{M\times N},
\]
and detail the asymptotic regime and structural assumptions on the signal and the noise.

\begin{assumption}
\label{assump:main}
We make the following assumptions on the signal and noise in the model \eqref{eq:rectangular spiked model}.
\begin{enumerate}
\item[(a)] We consider the asymptotic regime where $M,N\to\infty$ such that the aspect ratio converges,
\(
M/N \to \delta \in (0,1].
\)

\item[(b)] The signal and side information, represented by random vector pairs, converge in Wasserstein
distance:
\[
(\bm{u}_*,\bm{a}) \wc (\serv{U}_*,\serv{A}),
\qquad
(\bm{v}_*,\bm{b}) \wc (\serv{V}_*,\serv{B}),
\]
where $(\serv{U}_*,\serv{A},\serv{V}_*,\serv{B})$ have finite moments of all orders. Without loss of generality, we assume
\[
\mathbb{E}[\serv{U}_*^2] = \mathbb{E}[\serv{V}_*^2] = 1.
\]
\item[(c)] The noise matrix $\bm{W}\in\mathbb{R}^{M\times N}$ is independent of $(\bm{u}_\ast,\bm{a},\bm{v}_\ast,\bm{b})$ and is orthogonally invariant. Specifically, its singular value
decomposition $\bm{W}=\bm{U_W}\mathrm{diag}(\bm{\sigma})\bm{V_W}^\UT$. The matrices
$\bm{U_W}\in\mathbb{O}(M)$ and $\bm{V_W}\in\mathbb{O}(N)$ are independent and Haar distributed on
their respective orthogonal groups, and $\mathrm{diag}(\bm{\sigma})$ is deterministic. We assume
$\|\bm{W}\|_{\mathrm{op}}\le C$ for some constant $C$ independent of $(M,N)$ and that the empirical
spectral distribution of $\bm{W}\bm{W}^\UT$ converges weakly to a deterministic probability measure
$\mu$. We define the limiting spectral measure of $\bm{W}^\UT\bm{W}$ as
\[
\widetilde{\mu} \;\triangleq\; \delta\mu + (1-\delta)\delta_{\{0\}}.
\]
We assume $\mu$ is absolutely continuous with a H\"older continuous density and has compact support
$\supp(\mu)\subset\mathbb{R}_+$.

\item[(d)] Let $\spp(\bm{W}\bm{W}^\UT)$ denote the sets of empirical eigenvalues of
$\bm{W}\bm{W}^\UT$. We assume that
the empirical spectrum of $\bm{W}\bm{W}^\UT$ is asymptotically contained in any small neighborhood
of the limiting support:
\begin{equation*}
\lim_{M\to\infty} \sup_{\lambda\in\spp(\bm{W}\bm{W}^\UT)}
   \mathrm{d}(\lambda,\supp(\mu)) \;\xrightarrow{\text{a.s.}}\; 0,
\end{equation*}
where $\mathrm{d}(\lambda,S)\bydef\inf_{x\in S}|\lambda-x|$ denotes the distance from a point $\lambda$ to a
set $S$.
\end{enumerate}
\end{assumption}

\subsection{High-dimensional asymptotics and notation}

We next specify the mode of convergence used throughout to describe limits of empirical distributions
of vector entries and asymptotic equivalence of high-dimensional random vectors. More information can
be found in \cite{bayati2011dynamics,fan2022approximate,feng2022unifying}.

\begin{definition}[Wasserstein convergence]
\label{def:W2}
Let $(\vv_1,\dots,\vv_\ell)$ be a collection of random vectors in $\mathbb{R}^d$. We say that the
empirical distribution of the entries of $(\vv_1,\dots,\vv_\ell)$ converges to random variables
$(\mathsf{V}_1,\dots,\mathsf{V}_\ell)$ in the Wasserstein space of order $p$ if for any test function
$h:\mathbb{R}^\ell\to\mathbb{R}$ satisfying
\begin{equation}
\label{eq:PL2-testfunc}
|h(v)-h(v')|
 \le L\,(1+\|v\|^{p-1}+\|v'\|^{p-1})\|v-v'\|,
 \quad \forall v,v'\in\mathbb{R}^\ell,
\end{equation}
for some $L<\infty$, we have
\[
\frac{1}{d}\sum_{i=1}^d h(v_1[i],\dots,v_\ell[i])
  \xrightarrow{\text{a.s.}} \mathbb{E}[h(\mathsf{V}_1,\dots,\mathsf{V}_\ell)],
\quad \text{as } d\to\infty.
\]
We denote convergence in this sense by
$(\vv_1,\dots,\vv_\ell)\wpc(\mathsf{V}_1,\dots,\mathsf{V}_\ell)$. If this convergence holds for all
$p\ge 1$, we write $(\vv_1,\dots,\vv_\ell)\wc(\mathsf{V}_1,\dots,\mathsf{V}_\ell)$.
\end{definition}

Following \cite{dudeja2024optimality}, we also introduce a notion of asymptotic equivalence for
high-dimensional vectors.

\begin{definition}[Asymptotic equivalence]
\label{def:eq}
Two $d$-dimensional random vectors $\vu$ and $\vv$ are asymptotically equivalent if
\[
\frac{\|\vu-\vv\|^2}{d}\xrightarrow{\text{a.s.}} 0
\quad \text{as } d\to\infty.
\]
We denote this by $\vu \explain{$d\to\infty$}{\simeq} \vv$.
\end{definition}

\subsection{Stieltjes Transform, Hilbert Transform, and C-transform}

The analysis of the singular spectrum of the rectangular spiked model is most naturally expressed in
terms of transforms of spectral measures. We recall the Stieltjes transform and associated Hilbert
transform for finite signed measures, and define the C-transform that will appear in the master
equation governing outlier eigenvalues.

A \emph{signed measure} $\chi$ on the Borel subsets of $\mathbb{R}$ generalizes a measure by allowing
both positive and negative values. There is a well-known bijection between finite signed measures on
$\mathbb{R}$ and right continuous functions of bounded variation
\cite[Proposition~4.4.3]{Cohn2013MeasureTheory}.

\begin{definition}[Stieltjes transform of a finite signed measure]
\label{def: Stieltjes Transform for BV}
Let $\chi$ be a finite signed measure on $\mathbb{R}$, and let
$F_\chi(x)\bydef\chi((-\infty,x])$ be its right continuous function of bounded variation. The
Stieltjes transform $\mathcal{S}_\chi$ is defined for $z\in\mathbb{C}\setminus\supp(\chi)$ by
\begin{equation}
\label{eq:def of Stieltjes Transform}
\mathcal{S}_\chi(z)
  \;\bydef\; \int_{\mathbb{R}}\frac{1}{z-\lambda}\,dF_\chi(\lambda),
\qquad z\in\mathbb{C}\setminus\supp(\chi).
\end{equation}
\end{definition}

This transform uniquely determines the measure $\chi$ (and hence $F_\chi$); see, e.g.,
\cite[Theorem~B.8]{Bai2010SpectralAO}. The Hilbert transform of $\chi$, denoted by
$\mathcal{H}_\chi$, is defined by the Cauchy principal value integral
\begin{equation}
\label{eq: Hilbert Transform Def}
\mathcal{H}_\chi(x)
  \;\bydef\; \frac{1}{\pi}\,\mathrm{P.V.}\int_{\mathbb{R}}\frac{1}{x-\lambda}\,dF_\chi(\lambda),
\qquad x\in\mathbb{R}.
\end{equation}
When $\chi$ is absolutely continuous with respect to Lebesgue measure with a H\"older continuous
density, the integral \eqref{eq: Hilbert Transform Def} exists and $\mathcal{H}_\chi$ is itself
H\"older continuous \cite[Section~2.1]{pastur2011eigenvalue}.

In what follows, we write $\Smu$ and $\mathcal{H}$ for the Stieltjes and Hilbert transforms
associated with the spectral measure $\mu$ in Assumption~\ref{assump:main}. The Stieltjes and
Hilbert transforms are linked on the real axis, and this relationship will allow us to express the
densities of certain limiting measures in closed form.

The C-transform plays a central role in the master equation governing the emergence of outlier
eigenvalues. Its structure is closely related to the $D$-transform appearing in the analysis of
deformed random matrix models; see, e.g., \cite[Section~2.3]{Benaych-Georges2012singular}.

\begin{definition}[C-transform]
\label{Def:C_transform}
Let $\mu$ be the limiting spectral measure of the noise matrix in Assumption~\ref{assump:main}. The
C-transform associated with $\mu$ is defined for $z\in\mathbb{C}\setminus\supp(\mu)$ by
\begin{equation}
\label{eq:Ctrans}
\mathcal{C}(z)
  \;\bydef\; z\,\Smu(z)\Bigl[\delta\,\Smu(z) + \frac{1-\delta}{z}\Bigr],
\end{equation}
where $\Smu$ denotes the Stieltjes transform of $\mu$.
\end{definition}

\subsection{Signal–eigenspace Spectral Measures}

We now introduce spectral measures that project the eigenspace of the observed matrix onto the spans
of the true signals $\bm{u}_*$ and $\bm{v}_*$. These measures encode how signal energy is distributed
across the empirical spectrum and will be central in our state evolution analysis of OAMP.

\begin{definition}[Signal–eigenspace spectral measures]
\label{def:spectral_measures}
Let $(\lambda_i(\cdot),\bm{u}_i(\cdot))$ be the eigenvalue/eigenvector pairs of a symmetric matrix.
\begin{enumerate}
\item[(a)] \textbf{Parallel spectral measures.}  
To analyze quadratic forms in the signal directions, we define the parallel spectral measures
$\nu_{M,1}$ and $\nu_{N,2}$ as the weighted empirical measures
\begin{align*}
\nu_{M,1}
 &\;\bydef\; \frac{1}{M}\sum_{i=1}^M \langle \bm{u}_i(\bm{Y}\bm{Y}^\UT),\bm{u}_* \rangle^2
    \,\delta_{\lambda_i(\bm{Y}\bm{Y}^\UT)}, \\
\nu_{N,2}
 &\;\bydef\; \frac{1}{N}\sum_{i=1}^N \langle \bm{u}_i(\bm{Y}^\UT\bm{Y}),\bm{v}_* \rangle^2
    \,\delta_{\lambda_i(\bm{Y}^\UT\bm{Y})}.
\end{align*}

\item[(b)] \textbf{Cross spectral measure.}  
To analyze bilinear forms coupling the two signal directions, we construct the symmetric dilation
\[
\widehat{\bm{Y}}
  \;\bydef\; \begin{bmatrix} 0 & \bm{Y}\\ \bm{Y}^\UT & 0 \end{bmatrix} \in\mathbb{R}^{L\times L},
\quad L=M+N,
\]
and define the cross spectral measure $\nu_{L,3}$ by
\[
\nu_{L,3}
  \;\bydef\; \frac{1}{L}\sum_{i=1}^L
    \langle \bm{u}_i(\widehat{\bm{Y}}),\widehat{\bm{u}}_* \rangle
    \langle \bm{u}_i(\widehat{\bm{Y}}),\widehat{\bm{v}}_* \rangle
    \,\delta_{\lambda_i(\widehat{\bm{Y}})},
\]
where the zero-padded vectors are
$\widehat{\bm{u}}_* \bydef[\bm{u}_*^\UT,\bm{0}^\UT]^\UT$ and
$\widehat{\bm{v}}_* \bydef [\bm{0}^\UT,\bm{v}_*^\UT]^\UT$.
\end{enumerate}
\end{definition}

\paragraph{Shrinkage functions.}

In order to state the limiting characterization of the measures $\nu_i$, we introduce shrinkage
functions $\varphi_i:\mathbb{R}\to\mathbb{R}$ that describe their absolutely continuous components.
Let $\mathcal{H}(\lambda)$ be the Hilbert transform of $\mu$ at $\lambda$. Using the
Sokhotski–Plemelj formula~\cite{Blanchard2003MathematicalPhysics}, one can compute boundary values
of $1-\theta^2\mathcal{C}(\lambda-i\epsilon)$ as $\epsilon\to 0^+$:
\BS
\begin{align}
&\lim_{\epsilon \to 0^+} \big|1-\theta^2\CT(\lambda - \mathrm{i} \epsilon)\big|^2 \\
&\qquad = \Big\{1 - \delta\theta^2\pi^2\lambda \mathcal{H}^2(\lambda)
                  + \delta\theta^2\pi^2\lambda \mu^2(\lambda)
                  - (1-\delta)\theta^2\pi \mathcal{H}(\lambda)\Big\}^2 \\
&\qquad\quad + \Big\{\pi\theta^2\mu(\lambda)\big[(1-\delta) + 2\delta\pi\lambda\mathcal{H}(\lambda)\big]\Big\}^2.
\end{align}
\label{eq: norm of the master equation}
\ES
We then define the shrinkage functions $\varphi_1,\varphi_2,\varphi_3:\mathbb{R}\to\mathbb{R}$ by
\BS
\label{eq:shrikage-function}
\begin{align}
    \varphi_1(\lambda)
     &\bydef\frac{1 + \delta\theta^2\pi^2\lambda\left(\mathcal{H}(\lambda)^2 + \mu(\lambda)^2\right)}
                 {\lim_{\epsilon \to 0^+} \left|1-\theta^2\CT(\lambda-\mathrm{i}\epsilon)\right|^2},
       \label{eq:phiu} \\
    \varphi_3(\lambda)
     &\bydef\frac{\theta\left(1-\delta + 2\delta\pi\lambda\mathcal{H}(\lambda)\right)}
                 {\lim_{\epsilon \to 0^+} \left|1-\theta^2\CT(\lambda-\mathrm{i}\epsilon)\right|^2}
                 \cdot \mathbf{1}_{\{\lambda \neq 0\}}, \label{eq:phiuv} \\
    \varphi_2(\lambda)
     &\bydef
    \begin{cases}
      \delta \varphi_1(\lambda) + \dfrac{\theta(1-\delta)}{\lambda}\varphi_3(\lambda),
        & \lambda > 0, \\[0.4em]
      \dfrac{\delta}{1 - \theta^2(1-\delta)\pi\mathcal{H}(0)},
        & \lambda = 0.
    \end{cases}
    \label{eq:phiv}
\end{align}
\ES

\subsection{The Master Equation and Outlier Location}
The emergence of outlier singular-values in the rectangular spiked model is governed by the 
\textit{master equation}
\begin{equation}
\label{Eqn:master0}
\Gamma(z) \;\bydef\; 1 - \theta^{2}\,\mathcal{C}(z)=0,
\qquad z \in \mathbb{C}\setminus\supp(\mu),
\end{equation}
where $\mu$ is the limiting spectral measure of $\bm{W}\bm{W}^{\UT}$ and 
$\mathcal{C}$ is the C-transform defined in Definition~\ref{Def:C_transform}.
Real solutions of \eqref{Eqn:master0} outside the support of $\mu$ correspond to isolated spectral components, while the analytic properties of $\Gamma$ determine where such solutions may occur and ensure that any such solution is isolated and simple. We therefore begin by establishing the basic analytic properties of $\Gamma$.

\begin{lemma}[Analytic structure and zeros of the master equation]
\label{lem:Gamma-analytic}
Under Assumption~\ref{assump:main}(c), $\Gamma$ has the following properties:
\begin{enumerate}
\item $\Gamma$ is holomorphic on $\mathbb{C}\setminus\supp(\mu)$ and not identically zero.
      In particular, $\Gamma(z)\to 1$ as $|z|\to\infty$.
\item Every real zero of $\Gamma$ lying in $\mathbb{R}\setminus\supp(\mu)$ is isolated and simple;
      in particular, if $\Gamma(\lambda^{\*})=0$ for some
      $\lambda^{\*}\in\mathbb{R}\setminus\supp(\mu)$, then
      $\Gamma'(\lambda^{\*})\neq 0$. 
\item $\Gamma(\lambda)\neq 0$ for all $\lambda$ in the interior of $\supp(\mu)$, and moreover $\Gamma(0)\neq 0$. Consequently, any real solution of the
master equation that produces an isolated spectral component must lie in $\mathbb{R} \setminus (\operatorname{supp}(\mu) \cup \{0\})$.
\item Let $\lambda\in\R\setminus\supp(\mu)$ such that $\Gamma(\lambda)=0$. Then,
$$\operatorname{sign}\!\big(\Gamma'(\lambda)\big)=\operatorname{sign}\!\big(\Smu(\lambda)\big)$$
Equivalently, $\operatorname{sign}\!\big(\CT'(\lambda)\big)=-\operatorname{sign}\!\big(\Smu(\lambda)\big)$.
In particular, $\Gamma'(\lambda)\neq0$.      
\end{enumerate}
\end{lemma}

\proofseeapp{\lemref{lem:Gamma-analytic}}{App:master_eqn_analytic}

As an immediate consequence, any real solution of the master equation
\eqref{Eqn:master0} that generates an isolated spectral component must lie in
$\mathbb{R}\setminus(\supp(\mu)\cup{0})$ and corresponds to a simple zero of $\Gamma$.
These properties will be used repeatedly in the characterization of the singular
parts of the limiting spectral measures.

\subsection{Limiting Spectral Measures and Outlier Behavior}

We are now ready to state the limiting behavior of the signal–eigenspace spectral measures, and to
connect them with the isolated outlier eigenvalues and singular vectors of the rectangular spiked model.

\begin{lemma}
\label{lem:spectral_measures_properties}
Under Assumption~\ref{assump:main}, in the rectangular spiked model 
\eqref{eq:rectangular spiked model}, the following hold.
\begin{enumerate}
\item \textbf{Weak convergence.}  
The measures $\nu_{M,1},\nu_{N,2}$ and $\nu_{L,3}$ from Definition~\ref{def:spectral_measures}
converge weakly almost surely to deterministic, compactly supported measures $\nu_1,\nu_2,\nu_3$,
respectively. The limiting measures $\nu_1$ and $\nu_2$ are probability measures on $\mathbb{R}_+$, and
$\nu_3$ is a finite signed measure on $\mathbb{R}$.

\item \textbf{Stieltjes transforms.}  
For $z\in\mathbb{C}\setminus\mathbb{R}$, their Stieltjes transforms are
\begin{equation}
\label{eq:nu-Stieltjes}
\Su(z) = \frac{\Smu(z)}{1-\theta^2\CT(z)}, \quad
\Sv(z) = \frac{\delta \Smu(z) +\frac{1-\delta}{z} }{1-\theta^2 \CT(z)}, \quad 
\Suv(z) = \frac{\sqrt{\delta}}{1+\delta}\cdot\frac{\theta \CT(z^2)}{1- \theta^2 \CT(z^2)}.
\end{equation}

\item \textbf{Absolutely continuous parts.}  
Let the Lebesgue decomposition of each measure be $\nu_i = \nu_i^{\parallel} + \nu_i^{\perp}$ for 
$i\in\{1,2,3\}$, where $\nu_i^{\parallel}$ is the absolutely continuous component and
$\nu_i^{\perp}$ is the singular component. Let $\varphi_1,\varphi_2,\varphi_3$ be the shrinkage
functions in \eqref{eq:shrikage-function}. Then
\begin{equation}
\label{eq:nu-ac}
\frac{d\nu_{1}^{\parallel}}{d\lambda} = \mu(\lambda)\varphi_1(\lambda), \quad
\frac{d\nu_{2}^{\parallel}}{d\lambda} = \mu(\lambda)\varphi_2(\lambda), \quad
\frac{d\nu_{3}^{\parallel}}{d\sigma}
    = \frac{\sqrt{\delta}}{1+\delta}\,\mathrm{sign}(\sigma)\,\mu(\sigma^2)\varphi_{3}(\sigma^2).
\end{equation}

\item \textbf{Singular parts.}
Let
\[
\mathcal{K}^* \bydef \{\lambda\in\R\setminus\supp(\mu):\Gamma(\lambda)=0\},
\]
and assume $\mathcal{K}^*$ is finite.
Then the singular components are purely atomic and admit the representations
\BS
\begin{align}
\nu_1^{\perp}
  &= \sum_{\lambda_*\in\mathcal{K}^*}\nu_1(\{\lambda_*\})\,\delta_{\lambda_*},\\
\nu_2^{\perp}
  &= \sum_{\lambda_*\in\mathcal{K}^*}\nu_2(\{\lambda_*\})\,\delta_{\lambda_*}
    \;+\;\mathbf{1}_{\{\delta<1\}}\,\nu_2(\{0\})\,\delta_0,\\
\nu_3^{\perp}
 & = \sum_{\lambda_*\in\mathcal{K}^*}
    \bigl(\nu_3(\{\sigma_*\})\,\delta_{\sigma_*}
          + \nu_3(\{-\sigma_*\})\,\delta_{-\sigma_*}\bigr),\quad \sigma_*\bydef\sqrt{\lambda_*}.
\end{align}
\ES
Moreover, for $\lambda_*\in\mathcal{K}^*$,
\[
\nu_1(\{\lambda_*\})
  = -\frac{\Smu(\lambda_*)}{\theta^2 \mathcal{C}'(\lambda_*)}, \qquad
\nu_2(\{\lambda_*\})
  = -\frac{\delta \Smu(\lambda_*)+(1-\delta)/\lambda_*}
          {\theta^2 \mathcal{C}'(\lambda_*)},
\]
and
\[
\nu_3(\{\pm\sigma_*\})
  = \mp\,\frac{\sqrt{\delta}}{1+\delta}\,
         \frac{1}{2\theta^3\sigma_*\,\mathcal{C}'(\sigma_*^2)}.
\]
If $\delta<1$, $\nu_2$ has an additional atom at $0$ with mass
\[
\nu_2(\{0\})
  = \frac{1-\delta}{1-\theta^2(1-\delta)\pi\mathcal{H}(0)}.
\]
\end{enumerate}
\end{lemma}

\proofseeapp{\lemref{lem:spectral_measures_properties}}{subsec:proof_lemma2}

To connect the spectral measures in Definition~\ref{def:spectral_measures} with the eigen-structure of
the rectangular spiked model \eqref{eq:rectangular spiked model}, we next characterize the limiting
outlier eigenvalues and their associated overlaps. This extends the results of
\cite[Theorems~2.8--2.9]{Benaych-Georges2012singular} to the multi-outlier setting.

\begin{proposition}[Outlier characterization]
\label{prop:outlier_characterization}
Under Assumption~\ref{assump:main}, let $\mathcal{K}^*$ be the real zero set of the master equation
$\Gamma(\lambda)=0$ as in Lemma~\ref{lem:spectral_measures_properties}. Assume we are in a supercritical
regime in which $\mathcal{K}^*$ is finite. Then the following hold.
\begin{enumerate}
\item \textbf{Isolation of population outliers.}
There exists $\varepsilon>0$ such that the intervals
$I_k\bydef(\lambda_k-\varepsilon,\lambda_k+\varepsilon)$, $\lambda_k\in\mathcal{K}^*$, are pairwise disjoint and satisfy
\[
I_k\cap\supp(\mu)=\varnothing,
\qquad
I_k\cap\mathcal{K}^*=\{\lambda_k\},
\qquad \forall\,\lambda_k\in\mathcal{K}^*.
\]

\item \textbf{Exact spectral separation.}
Let $\varepsilon>0$ be as in item~(1), and define
\[
\mathcal{K}^*_{\varepsilon} \bydef \bigcup_{\lambda_k \in \mathcal{K}^*} (\lambda_k - \varepsilon, \lambda_k + \varepsilon),
\qquad
\supp_{\varepsilon}(\mu)\bydef \{\lambda\in\R:\ d(\lambda,\supp(\mu))<\varepsilon\}.
\]
Almost surely, there exists $M_0<\infty$ such that for all $M\ge M_0$:
\begin{align}
&\spp(\bm{WW}^\UT)\cap \mathcal{K}^*_{\varepsilon}=\varnothing, \label{eq:sep-W}\\
&\card{\left(\spp(\bm{YY}^\UT)\cap (\lambda_k - \varepsilon, \lambda_k + \varepsilon)\right)}=1,
\qquad \forall\,\lambda_k\in\mathcal{K}^*, \label{eq:uniq-outlier}\\
&\spp(\bm{YY}^\UT)\subseteq \supp_{\varepsilon}(\mu)\,\cup\, \mathcal{K}^*_{\varepsilon}. \label{eq:no-extra}
\end{align}
In words, for all sufficiently large $M$, all eigenvalues of $\bm{YY}^\UT$ are within
$\supp_{\varepsilon}(\mu)\cup\mathcal{K}^*_{\varepsilon}$, and each outlier window
$(\lambda_k-\varepsilon,\lambda_k+\varepsilon)$ contains exactly one eigenvalue.

\item \textbf{Convergence of empirical outliers.}
For each $\lambda_k\in\mathcal{K}^*$, let $\lambda_{k,M}$ denote the unique eigenvalue in
$\spp(\bm{YY}^\UT)\cap(\lambda_k-\varepsilon,\lambda_k+\varepsilon)$ (well-defined for all $M\ge M_0$ by
\eqref{eq:uniq-outlier}). Then
\[
\lambda_{k,M}\ac \lambda_k
\qquad\text{as } M\to\infty.
\]

\item \textbf{Limiting overlaps.}
Let $(\sigma_{k,M},\bm{u}_k(\bm{Y}),\bm{v}_k(\bm{Y}))$ be a singular value--vector triplet of
$\bm{Y}\in\mathbb{R}^{M\times N}$ such that $\sigma_{k,M}^2\to\lambda_k$, and write $\sigma_k\bydef\sqrt{\lambda_k}$. Then
\begin{align}
\frac{1}{M}\,\langle \bm{u}_*,\bm{u}_k(\bm{Y})\rangle^2
&\ac \nu_1(\{\lambda_k\}),\\
\frac{1}{N}\,\langle \bm{v}_*,\bm{v}_k(\bm{Y})\rangle^2
&\ac \nu_2(\{\lambda_k\}), \\
\frac{1}{\sqrt{MN}}\,
\langle \bm{u}_*,\bm{u}_k(\bm{Y})\rangle
\langle \bm{v}_*,\bm{v}_k(\bm{Y})\rangle
&\ac
\, 2\frac{1+\delta}{\sqrt{\delta}}\,\nu_3(\{\sigma_k\}),
\end{align}
where $\nu_1,\nu_2,\nu_3$ are the limiting spectral measures in
Lemma~\ref{lem:spectral_measures_properties}.
\end{enumerate}
\end{proposition}

\proofseeapp{\propref{prop:outlier_characterization}}{sec:spectral_measures_properties}

\begin{figure}[htbp]
  \centering

  %----------- Left Image -----------%
  \begin{minipage}[t]{0.33\textwidth}
    \centering
    \includegraphics[width=\linewidth]{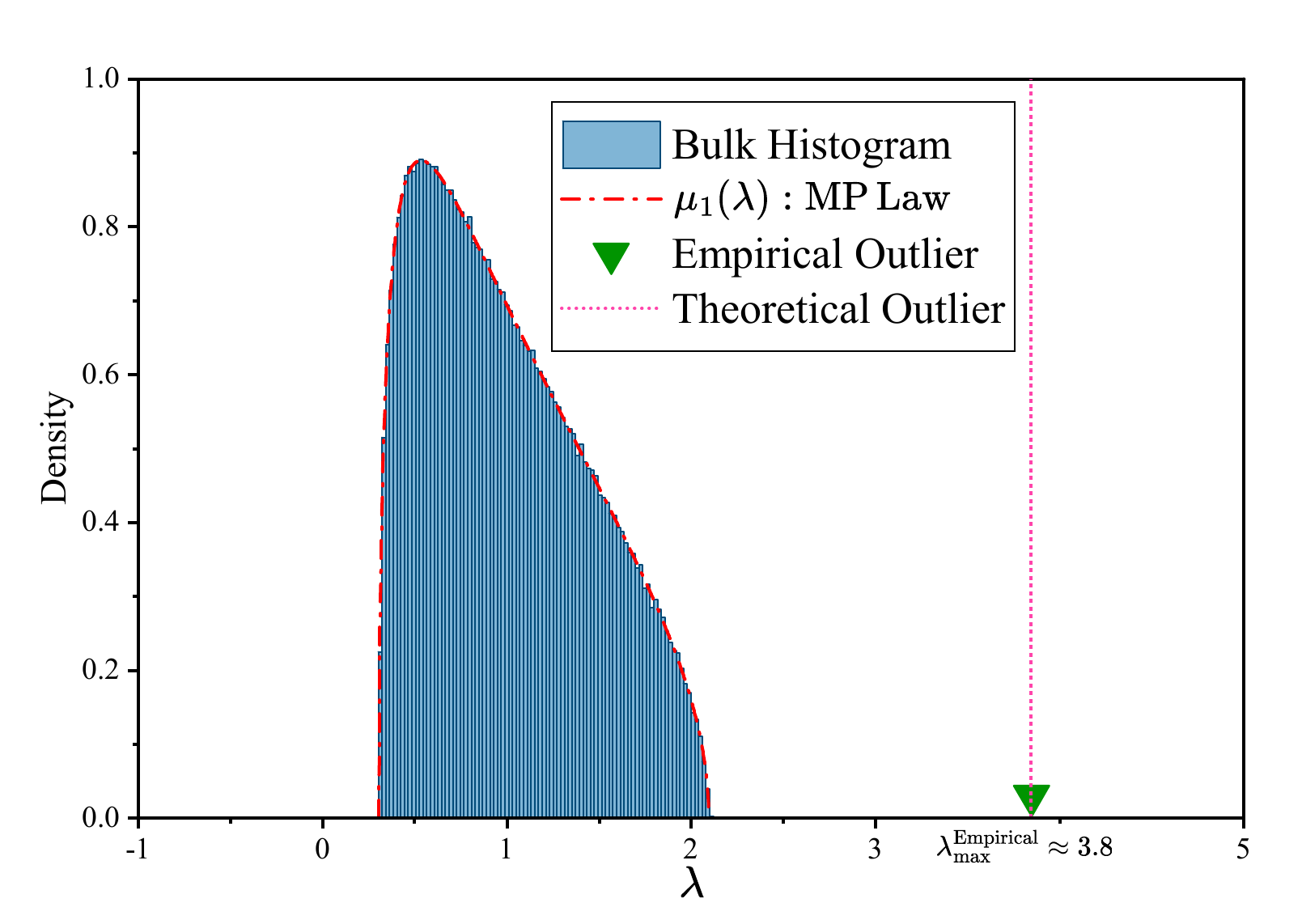}
  \end{minipage}\hfill
  %----------- Middle Image -----------%
  \begin{minipage}[t]{0.33\textwidth}
    \centering
    \includegraphics[width=\linewidth]{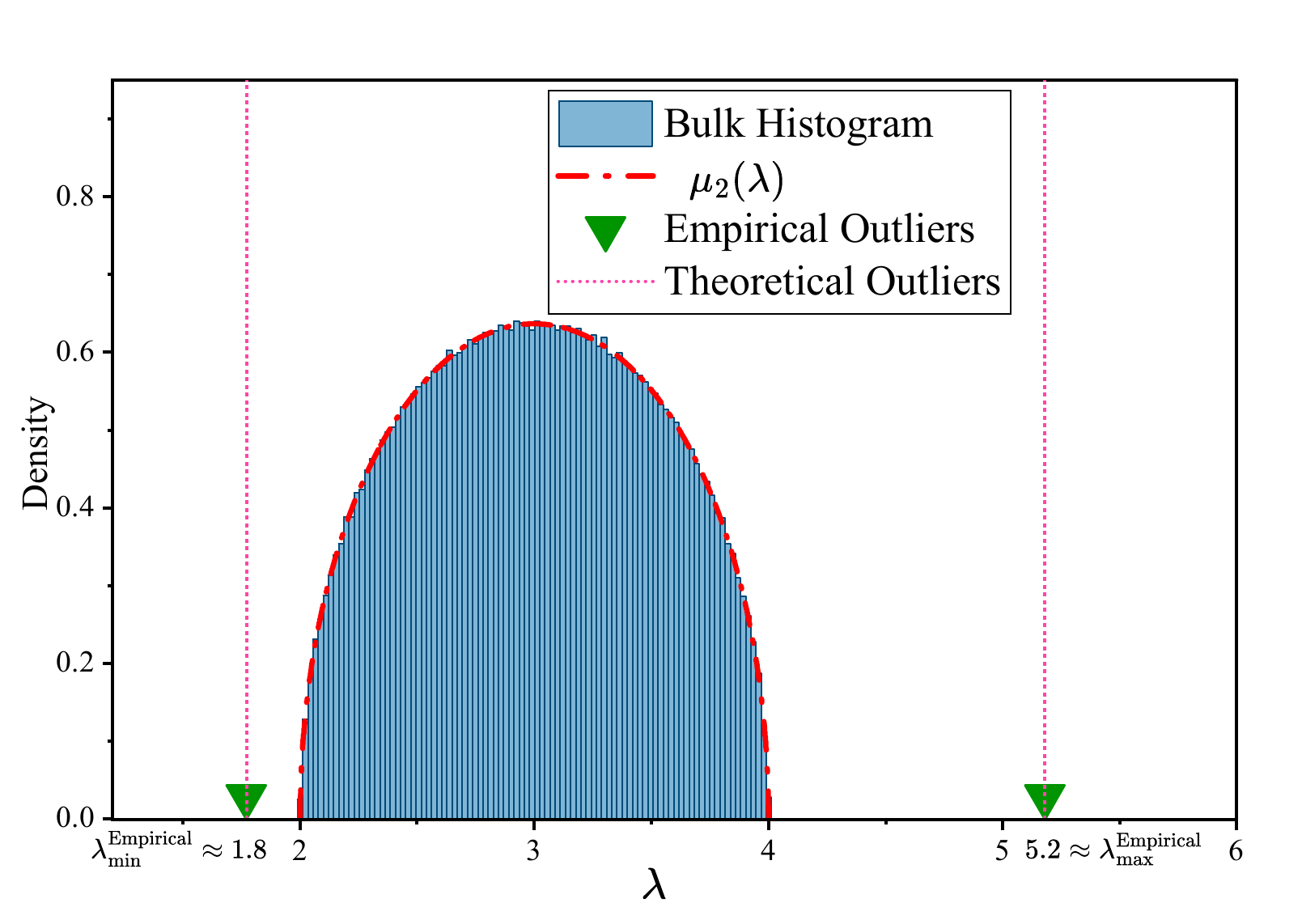}
  \end{minipage}\hfill % <--- ERROR FIXED HERE: Closed the second minipage and added hfill
  %----------- Right Image -----------%
  \begin{minipage}[t]{0.33\textwidth} % <--- Started the third minipage explicitly
    \centering
    \includegraphics[width=\linewidth]{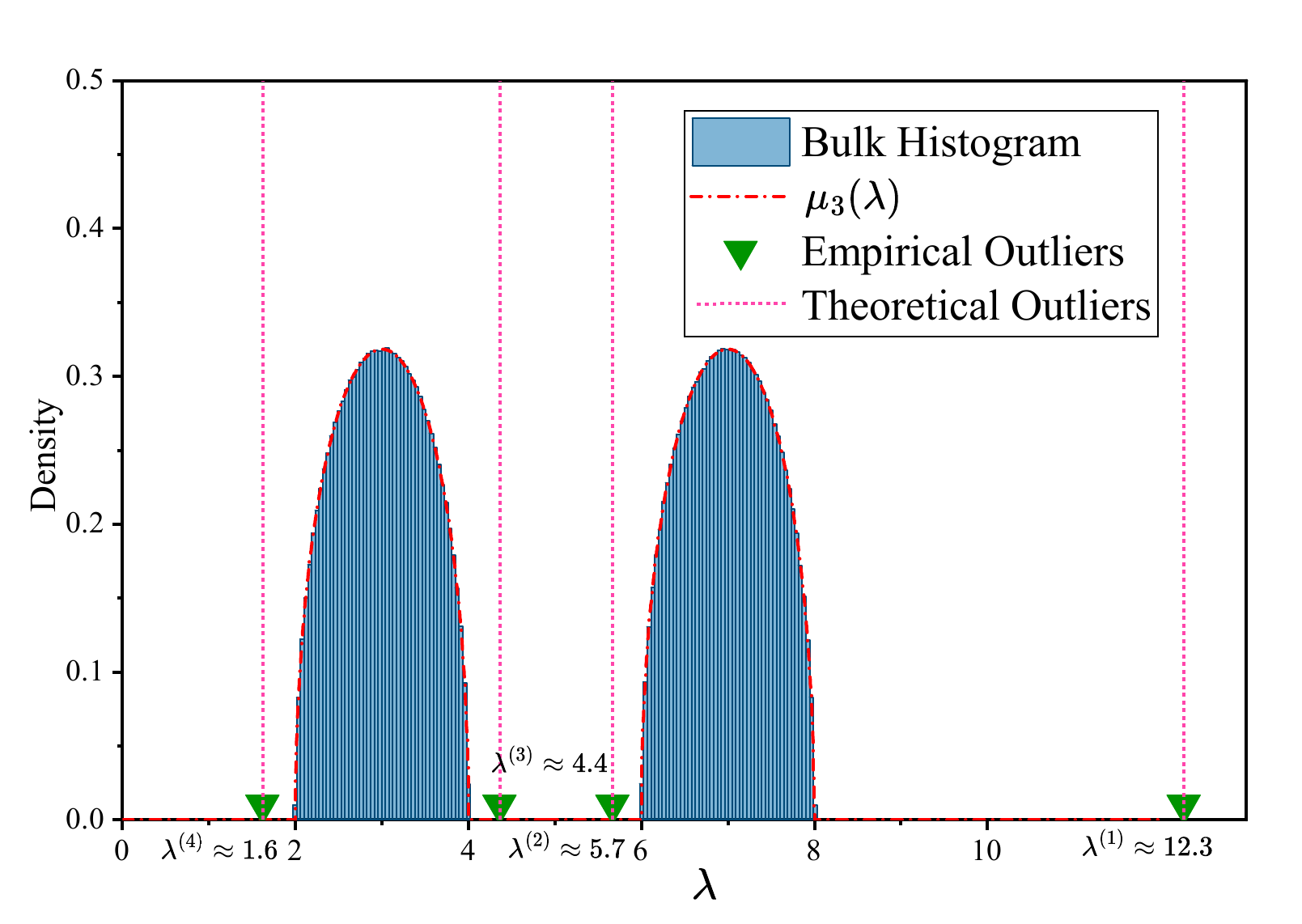}
  \end{minipage}

  \caption{Spectral behavior of the rectangular spiked model in super critical $\theta$-regime under different noise distributions. 
Left: Gaussian noise; the bulk follows the Mar\v{c}enko--Pastur density $\mu_1(\lambda)$ and exhibits a single outlier. Center: Non-Gaussian noise with bulk $\mu_2(\lambda)=\frac{2}{\pi}\sqrt{(\lambda-2)(4-\lambda)}\,\mathbf{1}_{[2,4]}(\lambda)$, producing two outliers. Right: Non-Gaussian noise with bulk $\mu_3(\lambda)=\frac{1}{\pi}\sqrt{(\lambda-2)(4-\lambda)}\,\mathbf{1}_{[2,4]}(\lambda)+\frac{1}{\pi}\sqrt{(\lambda-6)(8-\lambda)}\,\mathbf{1}_{[6,8]}(\lambda)$, producing multiple outliers. Dashed vertical lines indicate the real roots of the master equation in Lemma~\ref{lem:Gamma-analytic}.}
  \label{fig:combined-outliers}
\end{figure}

\begin{remark}[Multiplicity of outliers and spectral behavior]
\label{rem:multiple_outliers}
In the rectangular rotationally invariant (RI) model, a single rank-one spike may give rise to
\emph{multiple} outlier singular values.
This behavior can be understood through the analytic structure of the master equation
$\mathcal{C}(\lambda)=1/\theta^2$, where the associated C-transform
$\mathcal{C}(\lambda)$ need not be monotone on $\R\setminus\supp(\mu)$.
The importance of possible non-monotonicity of such transforms has already been noted in the
general theory of rectangular low-rank perturbations
(e.g., in the work of Benaych-Georges and Nadakuditi~\cite{Benaych-Georges2012singular}),
and it naturally allows the outlier equation to admit multiple real solutions, each
corresponding to a distinct outlier.

This phenomenon appears in two qualitatively different regimes, illustrated in
Figure~\ref{fig:combined-outliers}:
\begin{itemize}

    \item \textit{Single-interval bulk.}
    Even when the noise spectrum $\supp(\mu)$ consists of a single connected interval,
    the effective rank-two structure of the perturbation $\bm Y\bm Y^\UT$ can generate
    \emph{two} distinct outliers, typically appearing on opposite sides of the bulk;
    see the center panel of Fig.~\ref{fig:combined-outliers}.

    \item \textit{Multi-interval bulk.}
    When $\supp(\mu)$ consists of finitely many disjoint intervals, outliers typically emerge
    within the spectral gaps separating these intervals, as illustrated in the right panel of
    Fig.~\ref{fig:combined-outliers}.
\end{itemize}
\end{remark}

\begin{remark}[MP law and absence of sub-bulk solutions]
\label{rem:mp-no-left}
As noted in Remark~\ref{rem:multiple_outliers}, even when the noise spectrum consists of a
single interval, a rank-one rectangular spike may generate \emph{two} outliers, reflecting the
effective rank-two structure of the perturbation $\bm Y\bm Y^\UT$.
In the classical spiked Mar\v{c}enko--Pastur (MP) setting, however, the potential outlier below the
lower edge of the MP law does not occur: when the spike is supercritical, the unique outlier
lies strictly above the upper edge.
See the left panel of Fig.~\ref{fig:combined-outliers} for illustration.
While this conclusion follows immediately from the Mar\v{c}enko--Pastur specialization of the
general outlier equation, we are not aware of a reference where it is stated explicitly; for
this reason, we record it as Lemma~\ref{lem:MP-right-only} (in Appendix \ref{sec:example}).
\end{remark}

%% file: body/main_results.tex
\section{Orthogonal Approximate Message Passing Algorithms}\label{Sec:OAMP_main}

This section introduces a family of Orthogonal Approximate Message Passing (OAMP) algorithms for rank-one rectangular matrix estimation. The construction relies on a set of spectral denoisers, iterate denoisers, and side information, together with \textit{trace-free} and \textit{divergence-free} conditions that ensure a closed-form state evolution description.

\subsection{Orthogonal AMP for Rectangular Spiked Models}

\begin{definition}[OAMP algorithm]\label{def:OAMP}
Given the observation matrix $\bm{Y}\in\R^{M\times N}$ and side information vectors $\bm{a}\in\R^{M\times k}$ and $\bm{b}\in\R^{N\times k}$, an OAMP algorithm generates iterates $(\bm{u}_t)_{t\ge 1}$ and $(\bm{v}_t)_{t\ge 1}$ through the updates
\begin{align}
\bm{u}_t &=
F_t(\bm{Y}\bm{Y}^\top)\,f_t(\bm{u}_{1},\dots,\bm{u}_{t-1};\bm{a})
\;+\;
\widetilde{F}_t(\bm{Y}\bm{Y}^\top)\,\bm{Y}\,g_t(\bm{v}_{1},\dots,\bm{v}_{t-1};\bm{b}), 
\label{eq:OAMP Algo u}\\
\bm{v}_t &=
G_t(\bm{Y}^\top\bm{Y})\,g_t(\bm{v}_{1},\dots,\bm{v}_{t-1};\bm{b})
\;+\;
\widetilde{G}_t(\bm{Y}^\top\bm{Y})\,\bm{Y}^\top f_t(\bm{u}_{1},\dots,\bm{u}_{t-1};\bm{a}),
\label{eq:OAMP Algo v}
\end{align}
for $t\ge 1$. Here $F_t,\widetilde{F}_t,G_t,\widetilde{G}_t$ are \emph{spectral denoisers}: if
$\bm{Y}\bm{Y}^\top=\bm{U}\mathrm{diag}(\lambda_i)\bm{U}^\top$, then
\[
F_t(\bm{Y}\bm{Y}^\top)
= \bm{U}\mathrm{diag}(F_t(\lambda_i))\bm{U}^\top,
\]
and similarly for the other spectral denoisers. The functions $f_t,g_t$ are \emph{iterate denoisers} applied entrywise to vector inputs. At iteration $t$, the estimates of the signals $\bm{u}_*$ and $\bm{v}_*$ are produced by postprocessing maps $\phi_{u,t}$ and $\phi_{v,t}$,
\begin{align}
\hat{\bm{u}}_t &= \phi_{u,t}(\bm{u}_1,\dots,\bm{u}_t;\bm{a}), 
\label{eq:Esitimate u}\\
\hat{\bm{v}}_t &= \phi_{v,t}(\bm{v}_1,\dots,\bm{v}_t;\bm{b}).
\label{eq: Esitimate v}
\end{align}
\end{definition}

\paragraph{Regularity and orthogonality constraints.}
We require the spectral denoisers be dimension-independent and continuous on $\mathrm{supp}(\mu)$ and $F_t,G_t$ (but not $\tilde{F}_t$ and $\tilde{G}_t$) satisfy the \emph{trace-free constraint}
\begin{align}\label{eq:trace free}
\frac{1}{M}\mathrm{Tr}\,F_t(\bm{Y}\bm{Y}^\top)
&\xrightarrow{\mathrm{a.s.}} \langle F_t(\lambda)\rangle_\mu = 0, \\
\frac{1}{N}\mathrm{Tr}\,G_t(\bm{Y}^\top\bm{Y})
&\xrightarrow{\mathrm{a.s.}} \langle G_t(\lambda)\rangle_{\widetilde{\mu}} = 0.
\end{align}
The iterate denoisers $f_t,g_t$ and the postprocessing functions must be dimension-independent, Lipschitz, and continuously differentiable. Furthermore, the sequence $(f_t)_{t\ge1}$ and $(g_t)_{t\ge1}$ satisfy the \emph{divergence-free condition}
\begin{align}\label{eq: divergence free constraint}
\mathbb{E}[\partial_s f_t(\mathsf{U}_1,\dots,\mathsf{U}_{t-1};\mathsf{A})]=0,\qquad
\mathbb{E}[\partial_s g_t(\mathsf{V}_1,\dots,\mathsf{V}_{t-1};\mathsf{B})]=0,
\end{align}
for every $s<t$, where the expectations are taken under the limiting joint laws $(\mathsf{U}_*,\mathsf{A})\sim\pi_u$ and $(\mathsf{V}_*,\mathsf{B})\sim\pi_v$.

The trace-free and divergence-free conditions in OAMP (and vector AMP) algorithms \cite{ma2017orthogonal,rangan2019vector} ensure that the effective noise in each update is asymptotically orthogonal to all past iterates, thereby removing the Onsager correction and enabling a valid state evolution characterization.

\paragraph{State Evolution Random Variables.} Each OAMP algorithm is associated with a collection of state evolution random variables. It describes the joint asymptotic behavior of the signal, the iterates and the side information \( \bm{a}, \bm{b} \). Let $\mathsf{Z}_{u,t},\mathsf{Z}_{v,t}$ be Gaussian random variables, the distributions are given by
\begin{subequations} \label{eq: SE R.V.-original} 
\begin{align}
(\mathsf{U}_{*}, \mathsf{A}) &\sim \pi_u, \quad \mathsf{U}_t = \mu_{u,t} \mathsf{U}_{*} + \mathsf{Z}_{u,t} \quad \forall t \in \mathbb{N},  \label{eq:SEuchannel} \\  
(\mathsf{V}_{*}, \mathsf{B}) &\sim \pi_v, \quad \mathsf{V}_t = \mu_{v,t}  \mathsf{V}_{*} + \mathsf{Z}_{v,t} \quad \forall t \in \mathbb{N}.  \label{eq:SEvchannel}  
\end{align}
\end{subequations}
The random variables of the iterative denoisers with side information are formally defined as
\BS
\begin{align}\label{eq:longmemory-F-G}
\mathsf{F}_t &= f_t(\mathsf{U}_1,\ldots,\mathsf{U}_{t-1};\mathsf{A}),  \quad \mathsf{G}_t = g_t(\mathsf{V}_1,\ldots,\mathsf{V}_{t-1};\mathsf{B}).
\end{align}
The alignment metrics between estimates and ground truth are characterized by
\begin{align}\label{eq:alignment-alphabeta}
\alpha_t &= \mathbb{E}[\mathsf{U}_*\mathsf{F}_t], \quad \beta_t = \mathbb{E}[\mathsf{V}_*\mathsf{G}_t].
\end{align}
Consequently, the residual covariances are defined as
\begin{align}\label{eq:SE-sigma-fg}
\sigma^2_{f,st} &= \mathbb{E}[\mathsf{F}_s \mathsf{F}_t] - \alpha_s\alpha_t, \quad \sigma^2_{g,st} = \mathbb{E}[\mathsf{G}_s\mathsf{G}_t] - \beta_s\beta_t. 
\end{align}
\ES
The coefficients in \eqref{eq:SEuchannel} and \eqref{eq:SEvchannel} are defined via the recursion
\BS \label{eq: mean value of SE1}
\begin{align}
 \mu_{u,t}
 &\bydef \alpha_t \big\langle F_t(\lambda) \big\rangle_{\nu_1}
   + \beta_t (1+\delta^{-1}) \big\langle \sigma\,\widetilde{F}_t(\sigma^2) \big\rangle_{\nu_3},
 \\
 \mu_{v,t}
 &\bydef \beta_t \big\langle G_t(\lambda) \big\rangle_{\nu_2}
   + \alpha_t (1+\delta) \big\langle \sigma\,\widetilde{G}_t(\sigma^2) \big\rangle_{\nu_3}.
\end{align}
\ES
The variables \( (\mathsf{Z}_{u,t}, \mathsf{Z}_{v,t})_{t\in \mathbb{N}} \) are zero-mean jointly Gaussian random variables, sampled independently of the true signals. Their covariance matrix entries are given by the following recursions for any $s, t \in \mathbb{N}$
\BS \label{eq: Cov of SE1}
\begin{align}
\Sigma_{u,st} \bydef \mathbb{E}[\mathsf{Z}_{u,s}\mathsf{Z}_{u,t}] 
&= \alpha_s\alpha_t \big\langle F_s(\lambda) F_t(\lambda) \big\rangle_{\nu_1}
  + \beta_s\beta_t \delta^{-1} \big\langle \lambda \,\tilde{F}_s(\lambda)\tilde{F}_t(\lambda) \big\rangle_{\nu_2}
  - \mu_{u,s}\mu_{u,t} \nonumber\\
&\quad+ (1+\delta^{-1})\Big(
    \alpha_s\beta_t \big\langle \sigma\,F_s(\sigma^2)\tilde{F}_t(\sigma^2) \big\rangle_{\nu_3}
  + \alpha_t\beta_s \big\langle \sigma\,F_t(\sigma^2)\tilde{F}_s(\sigma^2) \big\rangle_{\nu_3}
  \Big) \nonumber\\
&\quad+ \sigma_{f,st}^2 \big\langle F_s(\lambda) F_t(\lambda) \big\rangle_{\mu}
  + \delta^{-1} \sigma_{g,st}^2 \big\langle \lambda\,\tilde{F}_s(\lambda)\tilde{F}_t(\lambda) \big\rangle_{\widetilde{\mu}},
  \label{eq:cov_Z_u} \\
\Sigma_{v,st} \bydef \mathbb{E}[\mathsf{Z}_{v,s}\mathsf{Z}_{v,t}] 
&= \beta_s\beta_t \big\langle G_s(\lambda) G_t(\lambda) \big\rangle_{\nu_2}
  + \alpha_s\alpha_t \delta \big\langle \lambda\,\tilde{G}_s(\lambda)\tilde{G}_t(\lambda) \big\rangle_{\nu_1}
  - \mu_{v,s}\mu_{v,t} \nonumber\\
&\quad+ (1+\delta)\Big(
    \beta_s\alpha_t \big\langle \sigma\,G_s(\sigma^2)\tilde{G}_t(\sigma^2) \big\rangle_{\nu_3}
  + \beta_t\alpha_s \big\langle \sigma\,G_t(\sigma^2)\tilde{G}_s(\sigma^2) \big\rangle_{\nu_3}
  \Big) \nonumber\\
&\quad+ \sigma_{g,st}^2 \big\langle G_s(\lambda) G_t(\lambda) \big\rangle_{\widetilde{\mu}}
  + \delta \sigma_{f,st}^2 \big\langle \lambda\,\tilde{G}_s(\lambda)\tilde{G}_t(\lambda) \big\rangle_{\mu}.
  \label{eq:cov_Z_v}
\end{align}
\ES

Our first main result is the following theorem on the state evolution of the proposed OAMP algorithm for spiked matrix models.
\begin{theorem}[State evolution]\label{Thm: State Evolution}
Consider the OAMP algorithm in Definition~\ref{def:OAMP}, and let the state evolution random variables be defined as in~\eqref{eq: SE R.V.-original}. Then for each fixed $t\in\mathbb{N}$,
\begin{align}
(\bm{u}_*,\bm{u}_1,\dots,\bm{u}_t;\bm{a})
&\xrightarrow{W_2}
(\mathsf{U}_*,\mathsf{U}_1,\dots,\mathsf{U}_t;\mathsf{A}),\\
(\bm{v}_*,\bm{v}_1,\dots,\bm{v}_t;\bm{b})
&\xrightarrow{W_2}
(\mathsf{V}_*,\mathsf{V}_1,\dots,\mathsf{V}_t;\mathsf{B}),
\end{align}
where the convergence is in the Wasserstein sense of Definition~\ref{def:W2}.
\end{theorem}

\begin{proof}
See Appendix~\ref{Sec:OAMP_SE_proof}.%
\end{proof}%

%========================================================================================
%========================================================================================
%                                      revised version
%========================================================================================
%========================================================================================

\subsection{The Optimal OAMP Algorithm}\label{sec:optimal OAMP}

We now specialize the OAMP framework to derive an algorithm that achieves the Bayes-optimal performance predicted by state evolution. The resulting procedure uses MMSE-based scalar denoisers, optimal spectral shrinkage functions derived from the limiting spectral measures, and cosine similarity parameters that track alignment with the true signals.

\paragraph{Final algorithm.}
The optimal OAMP iterates $(\bm{u}_t^*)_{t\ge1}$ and $(\bm{v}_t^*)_{t\ge1}$ are defined using the squared cosine similarities 
$w_{1,t}, w_{2,t}\in[0,1)$ and take the form
\begin{align}
\bm{u}_t^*
&= \frac{1}{\sqrt{w_{1,t}}} \Big[
F_t^*(\bm{Y}\bm{Y}^\top)\,
\bar{\phi}\big(\bm{u}_{t-1}^*;\bm{a}\,|\,w_{1,t-1}\big)
+
\widetilde{F}_t^*(\bm{Y}\bm{Y}^\top)\,
\bm{Y}\,\bar{\phi}\big(\bm{v}_{t-1}^*;\bm{b}\,|\,w_{2,t-1}\big)
\Big], \label{eq:u-opt}\\[0.3em]
\bm{v}_t^*
&= \frac{1}{\sqrt{w_{2,t}}} \Big[
G_t^*(\bm{Y}^\top\bm{Y})\,
\bar{\phi}\big(\bm{v}_{t-1}^*;\bm{b}\,|\,w_{2,t-1}\big)
+
\widetilde{G}_t^*(\bm{Y}^\top\bm{Y})\,
\bm{Y}^\top\,\bar{\phi}\big(\bm{u}_{t-1}^*;\bm{a}\,|\,w_{1,t-1}\big)
\Big].\label{eq:v-opt}
\end{align}
The estimates at iteration $t$ are
\[
\hat{\bm{u}}_t^* \,=\, \phi(\bm{u}_t^*;\bm{a}\,|\,w_{1,t}),
\qquad
\hat{\bm{v}}_t^* \,=\, \phi(\bm{v}_t^*;\bm{b}\,|\,w_{2,t}).
\]

\begin{remark}
The prefactors $1/\sqrt{w_{1,t}}$ and $1/\sqrt{w_{2,t}}$ normalize the iterates so that the corresponding state-evolution variables satisfy $\mathbb{E}[(\mathsf{U}_t^*)^2]=\mathbb{E}[(\mathsf{V}_t^*)^2]=1$. This is a convention.
\end{remark}

\subsubsection*{1. Scalar MMSE and DMMSE denoisers}

The scalar MMSE function $\phi$ and the divergence-free MMSE (DMMSE) function $\bar\phi$
follow \cite[Definition~3]{dudeja2024optimality}.  
For a scalar Gaussian channel
\[
\mathsf{X} \mid (\mathsf{X}_*,\mathsf{C}) \sim 
\mathcal{N}\!\left(\sqrt{\omega}\,\mathsf{X}_*,\,1-\omega\right),
\qquad \omega\in[0,1),
\]
the denoisers are
\begin{align}\label{eq:mmse-scalar}
\phi(x;c\,|\,\omega)
&\bydef \mathbb{E}[\mathsf{X}_* \mid \mathsf{X} = x,\,\mathsf{C}=c],\\
\bar{\phi}(x;c\,|\,\omega)
&\bydef
\frac{\phi(x;c\,|\,\omega) - \frac{1}{\sqrt{1-\omega}}\mathbb{E}[\,\mathsf{Z}\,\phi(\mathsf{X};\mathsf{C}\,|\,\omega)\,]\,x}
{1 - \frac{\sqrt{\omega}}{\sqrt{1-\omega}}\mathbb{E}[\,\mathsf{Z}\,\phi(\mathsf{X};\mathsf{C}\,|\,\omega)\,]},
\quad \mathsf{Z}\sim\mathcal{N}(0,1).
\end{align}
The DMMSE denoiser enforces the divergence-free condition required by OAMP.

\begin{assumption}\label{assump:lip-MMSE}
For every $\omega\in[0,1)$, the MMSE estimator 
$\phi(\cdot |\omega)$ is continuously differentiable and Lipschitz.
\end{assumption}

\subsubsection*{2. Optimal spectral denoisers}

Let the shrinkage functions $\varphi_1,\varphi_2,\varphi_3$ be given by
\eqref{eq:phiu}–\eqref{eq:phiv}.  
For SE parameters $\rho_{1,t},\rho_{2,t}>0$, define
\begin{align}
P_t^*(\lambda)
&\bydef
\frac{\lambda(\rho_{2,t}\varphi_2(\lambda)+\delta)}
{(\rho_{1,t}\varphi_1(\lambda)+1)\,(\rho_{2,t}\varphi_2(\lambda)+\delta)\,\lambda
 - \rho_{1,t}\rho_{2,t}\varphi_3(\lambda)^2},\\[0.2em]
\widetilde{P}_t^*(\lambda)
&\bydef
\frac{\sqrt{\delta}\,\rho_{2,t}\varphi_3(\lambda)}
{(\rho_{1,t}\varphi_1(\lambda)+1)\,(\rho_{2,t}\varphi_2(\lambda)+\delta)\,\lambda
 - \rho_{1,t}\rho_{2,t}\varphi_3(\lambda)^2},\\[0.2em]
Q_t^*(\lambda)
&\bydef
\frac{\delta\lambda(\rho_{1,t}\varphi_1(\lambda)+1)}
{(\rho_{1,t}\varphi_1(\lambda)+1)\,(\rho_{2,t}\varphi_2(\lambda)+\delta)\,\lambda
 - \rho_{1,t}\rho_{2,t}\varphi_3(\lambda)^2},\\[0.2em]
\widetilde{Q}_t^*(\lambda)
&\bydef
\frac{\sqrt{\delta}\,\rho_{1,t}\varphi_3(\lambda)}
{(\rho_{1,t}\varphi_1(\lambda)+1)\,(\rho_{2,t}\varphi_2(\lambda)+\delta)\,\lambda
 - \rho_{1,t}\rho_{2,t}\varphi_3(\lambda)^2}.
\end{align}

The trace-free optimal matrix denoisers are then
\begin{align}
F_t^*(\lambda)
&\bydef
\left(1+\frac{1}{\rho_{1,t}}\right)\left(1 - \frac{P_t^*(\lambda)}{\langle P_t^*\rangle_{\mu}}\right),\qquad
\widetilde{F}_t^*(\lambda)
\bydef
\left(1+\frac{1}{\rho_{2,t}}\right)\frac{\widetilde{P}_t^*(\lambda)}{\langle P_t^*\rangle_{\mu}},\label{eq:Fstar}\\
G_t^*(\lambda)
&\bydef
\left(1+\frac{1}{\rho_{2,t}}\right)\left( 1 - \frac{Q_t^*(\lambda)}{\langle Q_t^*\rangle_{\widetilde{\mu}}}\right),\qquad
\widetilde{G}_t^*(\lambda)
\bydef
\left(1+\frac{1}{\rho_{1,t}}\right)\frac{\widetilde{Q}_t^*(\lambda)}{\langle Q_t^*\rangle_{\widetilde{\mu}}}.\label{eq:tGstar}
\end{align}

\subsubsection*{3. Recursion for $(w_{i,t},\rho_{i,t})$}

Let $\mathrm{mmse}_{\mathsf{X}}(w)
\bydef
\mathbb{E}\!\left[(\mathsf{X}_* - \mathbb{E}[\mathsf{X}_*|\mathsf{X}])^2\right]$, where $\mathsf{X}= \sqrt{w}\mathsf{X}_*+\sqrt{1-w}\mathsf{Z}$. Then
\begin{align}\label{eq:OptimalRecursion}
\rho_{1,t}
&= \frac{1}{\mathrm{mmse}_{\mathsf{U}}(w_{1,t-1})}
  - \frac{1}{1-w_{1,t-1}},&
\rho_{2,t}
&= \frac{1}{\mathrm{mmse}_{\mathsf{V}}(w_{2,t-1})}
  - \frac{1}{1-w_{2,t-1}},\\
w_{1,t}
&= 1 -
\frac{1-\langle P_t^*\rangle_{\mu}}
{\langle P_t^*\rangle_{\mu}}\,\frac{1}{\rho_{1,t}},&
w_{2,t}
&= 1 -
\frac{1-\langle Q_t^*\rangle_{\widetilde{\mu}}}
{\langle Q_t^*\rangle_{\widetilde{\mu}}}\,\frac{1}{\rho_{2,t}}.
\end{align}
The recursion is initialized via $w_{1,0},w_{2,0}\in(0,1)$.

\begin{proposition}[State evolution: optimal OAMP]\label{prop: optimal SE}
Let $(\mathsf{U}_*,\mathsf{U}_t^*;\mathsf{A})$ and $(\mathsf{V}_*,\mathsf{V}_t^*;\mathsf{B})$ denote the state-evolution variables associated with the optimal OAMP iterates. Then:

\begin{enumerate}
\item For $i\in\{1,2\}$, we have $w_{i,t}\in(0,1)$ and $\rho_{i,t}>0$, and
$(\mathsf{U}_*,\mathsf{U}_t^*)$ and $(\mathsf{V}_*,\mathsf{V}_t^*)$
form scalar Gaussian channels with similarities $w_{1,t}$ and $w_{2,t}$, respectively. Moreover,
\[
\lim_{M\to\infty}\frac{\|\hat{\bm{u}}_t^*-\bm{u}_*\|^2}{M}
\overset{\mathrm{a.s.}}{=}
\mathrm{mmse}_{\mathsf{U}}(w_{1,t}),\qquad
\lim_{N\to\infty}\frac{\|\hat{\bm{v}}_t^*-\bm{v}_*\|^2}{N}
\overset{\mathrm{a.s.}}{=}
\mathrm{mmse}_{\mathsf{V}}(w_{2,t}).
\]

\item The sequence $(\rho_{1,t},\rho_{2,t},w_{1,t},w_{2,t})$ is monotone and converges to $(\rho_1^*,\rho_2^*,w_1^*,w_2^*)\in(0,\infty)^2\times[0,1)^2$ satisfying
\BS\label{eq: fixed point equation of OAMP-general}
\begin{align}
\rho_{1}
&= \frac{1}{\mathrm{mmse}_{\mathsf{U}}(w_{1})} - \frac{1}{1-w_{1}}, &
\mathrm{mmse}_{\mathsf{U}}(w_{1})
&= \frac{1}{\rho_{1}}\left(1-\langle P^*\rangle_{\mu}\right), \\[0.3em]
\rho_{2}
&= \frac{1}{\mathrm{mmse}_{\mathsf{V}}(w_{2})} - \frac{1}{1-w_{2}}, &
\mathrm{mmse}_{\mathsf{V}}(w_{2})
&= \frac{1}{\rho_{2}}\left(1-\langle Q^*\rangle_{\widetilde{\mu}}\right).
\end{align}
\ES
Consequently,
\[
\lim_{t\to\infty}\lim_{M\to\infty}
\frac{\|\hat{\bm{u}}_t^*-\bm{u}_*\|^2}{M}
\overset{\mathrm{a.s.}}{=}
\mathrm{mmse}_{\mathsf{U}}(w_1^*),\qquad
\lim_{t\to\infty}\lim_{N\to\infty}
\frac{\|\hat{\bm{v}}_t^*-\bm{v}_*\|^2}{N}
\overset{\mathrm{a.s.}}{=}
\mathrm{mmse}_{\mathsf{V}}(w_2^*).
\]
\end{enumerate}
\end{proposition}

\proofseeapp{\propref{prop: optimal SE}}{sec:optimal_SE}

\begin{remark}[Connection with replica-symmetric Bayes-risk predictions]
The fixed-point equations in \eqref{eq: fixed point equation of OAMP-general} match the replica-symmetric characterization of the Bayes risk for the rectangular spiked rotationally-invariant model, whenever \eqref{eq: fixed point equation of OAMP-general} has a unique solution. Further details will appear in a forthcoming paper~\cite{Chenqun2025}.
\end{remark}

%% file: body/example.tex
\subsection{Example: I.I.D Gaussian Noise}\label{sec:IIDGaussian}

%\paragraph{Spectral Analysis in I.I.D. Gaussian Noise.}
We now specialize our results to the noise matrix $\bm{W}$ with i.i.d.\ $\mathcal{N}(0, 1/N)$ entries. In this canonical setting, the limiting spectral measure $\mu$ of $\bm{W}\bm{W}^\UT$ is the Mar\v{c}henko--Pastur law~\cite{Bai2010SpectralAO} with aspect ratio $\delta\in(0,1)$, whose density is
\[
\mu_{\mathrm{MP}}(\lambda)
= \frac{\sqrt{(b_+-\lambda)(\lambda-a_-)}}{2\pi\delta\,\lambda}\,
  \mathbf{1}_{[a_-,b_+]}(\lambda),
\qquad
a_- \bydef (1-\sqrt{\delta})^2,\quad
b_+ \bydef (1+\sqrt{\delta})^2.
\]
\noindent As an application of \propref{prop:outlier_characterization}, a detailed spectral analysis in such I.I.D.\ Gaussian noise model, which derives the phase transition and the outlier location, is provided in Appendix~\ref{app:pf-MP-right-only}.

%\paragraph{Equivalence in Fixed Point Equation with AMP.} 
We demonstrate that for this model, the fixed-point equations \eqref{eq: fixed point equation of OAMP-general} governing our optimal OAMP algorithm coincides with that of the standard AMP \cite{rangan2012iterative,montanari2021estimation} up to a  re-parameterization.
\begin{proposition}[]\label{thm:OAMP_Wigner_FP}
For the rectangular spiked model \eqref{eq:rectangular spiked model} with i.i.d. Gaussian noise matrix, the fixed point equations \eqref{eq: fixed point equation of OAMP-general} can be simplified to
\begin{align}
\frac{w_1}{1-w_1} & = \frac{\theta^2}{\delta} \left(1 -\mathrm{mmse}(w_2)\right), \\
\frac{w_2}{1-w_2}  & = \theta^2 \left(1 -\mathrm{mmse}(w_1)\right).
\end{align}
\end{proposition}

\proofseeapp{\propref{thm:OAMP_Wigner_FP}}{sec:example}

%% file: body/spectral_combo.tex
\section{Optimal Spectral Estimation Under Multiple Outliers}\label{sec: spec method}

In the absence of a nonzero mean or side information, a random initialization fails for the OAMP algorithm: its state evolution converges to a trivial fixed point, as observed previously in phase retrieval \cite{ma2019optimization,mondelli2021approximate} and spiked models \cite{montanari2021estimation,mondelli2021pca,zhong2021approximate}. A spectral initialization is therefore required to produce a nontrivial estimate.

In this section, we study spectral estimation for the rectangular spiked model. As detailed in Remark~\ref{rem:multiple_outliers} and Figure~\ref{fig:combined-outliers}, a single rank-one signal in this setting typically generates \emph{multiple} informative outlier singular values. In such regimes, relying solely on the leading singular vector (standard PCA) is suboptimal because it discards the signal energy carried by secondary outliers. Prior work \cite{Nadakuditi2013informative} notes this phenomenon but does not provide an optimal method for combining the outlier components. Here, we develop a data-driven estimator that aggregates the informative outliers optimally under mild non-Gaussian assumptions on the signal.

\subsection{Optimal Oracle Spectral Estimators}

As established in Proposition~\ref{prop:outlier_characterization}, each outlying
singular vector of $\bm Y$ retains a nonvanishing asymptotic correlation with the
true signal directions. In the multiple-outlier regime, it is therefore natural
to consider linear combinations of all informative components rather than relying
on a single leading singular vector. This subsection characterizes the optimal
such combination.

Let the singular-value decomposition of $\bm Y$ be
\[
\bm Y=\sum_{i=1}^{M}\sigma_{i,M}(\bm Y)\,\bm u_i(\bm Y)\,\bm v_i(\bm Y)^\UT,
\qquad 
\sigma_{1,M}\ge\cdots\ge\sigma_{M,M},
\]
and define
$\lambda_{i,M}\triangleq\sigma_{i,M}(\bm Y)^2$.  
Under Proposition~\ref{prop:outlier_characterization} (supercritical $\theta$), 
let $\mu$ denote the limiting spectral distribution of $\bm W\bm W^\UT$, and let 
$\mathcal{K}^*$ be the finite set of population outliers.

Choose $\varepsilon>0$ small enough so that the $\varepsilon$-neighborhood of 
points in $\mathcal{K}^*$ are disjoint and lie outside $\supp(\mu)$; denote this 
union by $\mathcal{K}^*_\varepsilon$.  
We then define the \emph{empirical outlier index set}
\begin{equation}\label{eq:emp-outlier-index}
\mathcal{I}_M 
    \,\bydef\,\{\,i:\,\lambda_{i,M}\in
        \spp(\bm Y\bm Y^\UT)\cap\mathcal{K}^*_\varepsilon\,\}.
\end{equation}

As ensured by Proposition~\ref{prop:outlier_characterization}, Claim~(2), for
all sufficiently large $M$, the empirical outliers in $\mathcal{I}_M$ correspond
one-to-one with the population outliers in $\mathcal{K}^*$. With these
informative components reliably identified, we consider linear spectral
estimators supported on $\mathcal{I}_M$:
\begin{align}\label{eq:linear_spectral_estimators}
\bm u_{\mathrm{PCA}}(\bm c_u)
    &\bydef \sqrt{M}\sum_{i\in\mathcal{I}_M} c_{u,i}\,\bm u_i(\bm Y), &
\bm v_{\mathrm{PCA}}(\bm c_v)
    &\bydef \sqrt{N}\sum_{i\in\mathcal{I}_M} c_{v,i}\,\bm v_i(\bm Y),
\end{align}
where $\bm c_u,\bm c_v\in\mathbb{R}^{\card{\mathcal{I}_M}}$ denote the combination
coefficients. Since each outlying singular vector carries nonvanishing alignment
with the true signal, an appropriate linear combination may improve the overall
directional accuracy compared to using any single component.

The next proposition characterizes the \textit{oracle} asymptotic squared cosine
similarity achievable by this class of estimators, which equals the projection
of the true signal onto the outlier eigenspace. The optimal coefficients
attaining this limit depend on the unknown signal and are therefore not
implementable in practice, but the result serves as the fundamental
performance benchmark for all linear spectral methods.

\begin{proposition}\label{prop:optimal_linear_estimator_cos}
Consider the class of estimators in \eqref{eq:linear_spectral_estimators}.
For any $\bm c_u,\bm c_v\in\mathbb{R}^{\card{\mathcal{I}_M}}$, almost surely,
\begin{align}\label{eq: max overlap}
\lim_{M\to\infty}
\frac{\langle \bm u_{\mathrm{PCA}}(\bm c_u),\bm u_* \rangle^2}
     {\|\bm u_{\mathrm{PCA}}(\bm c_u)\|^2\,\|\bm u_*\|^2}
    \le \sum_{\lambda_i\in\mathcal{K}^*}\nu_1(\{\lambda_i\}),
\qquad
\lim_{N\to\infty}
\frac{\langle \bm v_{\mathrm{PCA}}(\bm c_v),\bm v_* \rangle^2}
     {\|\bm v_{\mathrm{PCA}}(\bm c_v)\|^2\,\|\bm v_*\|^2}
    \le \sum_{\lambda_i\in\mathcal{K}^*}\nu_2(\{\lambda_i\}),
\end{align}
where $\nu_1(\{\lambda_k\})$ and $\nu_2(\{\lambda_k\}$ are defined in
Lemma~\ref{lem:spectral_measures_properties}, Claim~(4).
Moreover, these upper bounds are asymptotically attained by the oracle
combinations
\begin{align}\label{eq:oracle_optimal_spectral_est.}
\bm u_{\mathrm{ora}}^*
    &\bydef \sqrt{M}\sum_{i\in\mathcal{I}_M}
        \langle \bm u_*,\bm u_i(\bm Y)\rangle\,\bm u_i(\bm Y), \\
\bm v_{\mathrm{ora}}^*
    &\bydef \sqrt{N}\sum_{i\in\mathcal{I}_M}
        \langle \bm v_*,\bm v_i(\bm Y)\rangle\,\bm v_i(\bm Y).
\end{align}
\end{proposition}

\proofseeapp{\propref{prop:optimal_linear_estimator_cos}}{app:pf-opt-linear-cos}

\begin{remark}[Connection to RIE estimators \cite{Bun2016RIE}]
\label{rem:RIE}
Our construction of optimal spectral estimators is structurally related to the 
\emph{rotationally invariant estimator} (RIE) framework developed for 
extensive–rank matrix denoising in \cite{Bun2016RIE} and for rectangular models
in \cite{pourkamali2023rectangular}. In both settings, one first characterizes
an oracle estimator and then constructs a data-driven procedure that
asymptotically attains the oracle performance.

There are, however, important differences. RIE operates in the extensive rank
regime, where the signal information is distributed across the whole spectrum
and the optimal estimator applies an eigenvalue-dependent shrinkage to all
singular values. In contrast, our model is rank one, and all informative content
is concentrated in a finite number of outlier singular components; optimal
estimation thus requires combining only these outliers. A second distinction concerns the estimation objective. Whereas RIE aims to
reconstruct the underlying low-rank matrix, our goal is to recover the rank-one
signal vectors. In this setting, the optimal linear combination of the outlier
components involves relative signs that cannot be inferred from random matrix
theory alone. As a result, accurate aggregation of multiple outliers requires an
explicit sign-resolution procedure, addressed in
Section~\ref{sec:pratical-optimal}.
\end{remark}
\vspace{5pt}

\subsection{Data-Driven Optimal Linear Spectral Estimators}\label{sec:pratical-optimal}

The oracle estimator in \eqref{eq:oracle_optimal_spectral_est.} achieves the
linear-spectral performance bound of Proposition~\ref{prop:optimal_linear_estimator_cos},
but its coefficients depend on the unknown signal through the outlier--signal
overlaps. Thus Proposition~\ref{prop:optimal_linear_estimator_cos} provides an
\emph{oracle} benchmark for what any linear spectral method based solely on
$\bm Y$ can achieve. In this subsection, we construct a data-driven estimator
that asymptotically attains this benchmark. As a first step, we derive a
signal--plus--noise limit law for each outlier direction.

\begin{proposition}
\label{prop:signal_plus_noise_sv}
Under the assumptions of Proposition~\ref{prop:outlier_characterization}, let
$\mathcal{K}^\ast$ denote the finite set of outlier eigenvalues. Let
$\mathcal{K}^\ast=\{\lambda_1,\ldots,\lambda_K\}$ and
$\sigma_k=\sqrt{\lambda_k}$ for $1\le k\le K$. For each $k\in\{1,\ldots,K\}$ and
all sufficiently large $M$, let $\lambda_{k,M}$ be the empirical outlier associated with $\lambda_k \in \mathcal{K}^\ast$ as in \propref{prop:outlier_characterization}, and $\bm{u}_k(\bm{Y}), \bm{v}_k(\bm{Y})$ the corresponding left and right singular vectors of unit norm.
As $M,N\to\infty$, we have the joint convergence
\begin{align}
\bigl(\,\langle \bm{u}_*,\bm{u}_1(\bm{Y})\rangle\,
     \bm{u}_1(\bm{Y}),\ \ldots,\ 
     \langle \bm{u}_*,\bm{u}_K(\bm{Y})\rangle\,
     \bm{u}_K(\bm{Y})\,\bigr)
&\xrightarrow[]{W}
\bigl(\,\mathsf{U}_1^{\mathrm{OUT}},\ldots,\mathsf{U}_K^{\mathrm{OUT}}\,\bigr)^\UT,
\label{eq:joint-distribution-projector-u}
\\
\bigl(\,\langle \bm{v}_*,\bm{v}_1(\bm{Y})\rangle\,
     \bm{v}_1(\bm{Y}),\ \ldots,\ 
     \langle \bm{v}_*,\bm{v}_K(\bm{Y})\rangle\,
     \bm{v}_K(\bm{Y})\,\bigr)
&\xrightarrow[]{W}
\bigl(\,\mathsf{V}_1^{\mathrm{OUT}},\ldots,\mathsf{V}_K^{\mathrm{OUT}}\,\bigr)^\UT,
\label{eq:joint-distribution-projector-v}
\end{align}
where the random variables appearing on the RHS satisfy the following:
\begin{itemize}
\item[1] \textbf{Signal-plus-noise decomposition.}
For every $k\in\{1,\ldots,K\}$ we have
\BS\label{eq: distribution of projector}
\begin{align} \label{eq: distribution of projector u}
\mathsf{U}_k^{\mathrm{OUT}}
    &=\nu_1(\{\lambda_k\})\,\mathsf{U}_*
      +\sqrt{\nu_1(\{\lambda_k\})-\nu_1^2(\{\lambda_k\})}\,\mathsf{Z}_{u,k},\\
\mathsf{V}_k^{\mathrm{OUT}}
    &=\nu_2(\{\lambda_k\})\,\mathsf{V}_*
      +\sqrt{\nu_2(\{\lambda_k\})-\nu_2^2(\{\lambda_k\})}\,\mathsf{Z}_{v,k},
\end{align}
\ES
where $(\mathsf{U}_*,\mathsf{V}_*)$ are the limiting signal distributions and
$\{\mathsf{Z}_{u,k}\}_{k=1}^K$, $\{\mathsf{Z}_{v,k}\}_{k=1}^K$ are Gaussian
noise variables satisfying
\[
\bigl(\mathsf{Z}_{u,1},\ldots,\mathsf{Z}_{u,K}\bigr)
\indep \mathsf{U}_*,
\qquad
\bigl(\mathsf{Z}_{v,1},\ldots,\mathsf{Z}_{v,K}\bigr)
\indep \mathsf{V}_*.
\]

\item[2] \textbf{Gaussian noise and covariance.}
The vectors
$\bigl(\mathsf{Z}_{u,1},\ldots,\mathsf{Z}_{u,K}\bigr)$ and
$\bigl(\mathsf{Z}_{v,1},\ldots,\mathsf{Z}_{v,K}\bigr)$ are centered jointly
Gaussian. For each $k\in\{1,\ldots,K\}$,
\begin{align}\label{eq:sv-diag-cov}
\mathbb{E}[\mathsf{Z}_{u,k}]
=\mathbb{E}[\mathsf{Z}_{v,k}]
=0,
\qquad
\mathbb{E}[\mathsf{Z}_{u,k}^2]
=\mathbb{E}[\mathsf{Z}_{v,k}^2]
=1,    
\end{align}
and for all $1\le k< \ell\le K$,
\begin{align}\label{eq:sv-parral-cov}
\mathbb{E}[\mathsf{Z}_{u,k}\mathsf{Z}_{u,\ell}]
=-\frac{\nu_1(\{\lambda_k\})\,\nu_1(\{\lambda_\ell\})}
       {\sqrt{\nu_1(\{\lambda_k\})-\nu_1^2(\{\lambda_k\})}\,
        \sqrt{\nu_1(\{\lambda_\ell\})-\nu_1^2(\{\lambda_\ell\})}},    
\end{align}
\begin{align}
\mathbb{E}[\mathsf{Z}_{v,k}\mathsf{Z}_{v,\ell}]
=-\frac{\nu_2(\{\lambda_k\})\,\nu_2(\{\lambda_\ell\})}
       {\sqrt{\nu_2(\{\lambda_k\})-\nu_2^2(\{\lambda_k\})}\,
        \sqrt{\nu_2(\{\lambda_\ell\})-\nu_2^2(\{\lambda_\ell\})}},    
\end{align}
In particular, the covariance matrices
$\Sigma_u^{\mathrm{OUT}}=(\mathbb{E}[\mathsf{Z}_{u,k}\mathsf{Z}_{u,\ell}])_{1\le k<\ell\le K}$
and
$\Sigma_v^{\mathrm{OUT}}=(\mathbb{E}[\mathsf{Z}_{v,k}\mathsf{Z}_{v,\ell}])_{1\le k<\ell\le K}$
are positive definite. For $i\in\{1,2\}$, $\nu_i(\{\lambda_k\})$ denotes the point mass of the
parallel spectral measure defined in Lemma~\ref{lem:spectral_measures_properties}.
\end{itemize}

\end{proposition}

\proofseeapp{\propref{prop:signal_plus_noise_sv}}{app:pf-signal-plus-noise-sv}

\begin{remark}[Heuristic derivation of the outlier limit law]\label{Rem:outlying_distribution}
At a heuristic level, the decompositions
\eqref{eq:joint-distribution-projector-u}--\eqref{eq:joint-distribution-projector-v}
describe each projected outlier component as a deterministic multiple of the
signal plus an asymptotically Gaussian noise term. In the i.i.d.\ Gaussian
noise case, closely related eigenvector asymptotics are well known; see
\cite[Appendix~C]{montanari2021estimation}. In our setting, a convenient
starting point is the singular-equation for the outlying vector
\[
\langle\bm u_*,\bm u_k\rangle\,\bm u_k =
(\lambda_{k,M} \bm{I}_M-\bm{W}\bm{W}^\UT)^{-1}
\biggl(
  \frac{\theta\sigma_{k,M}}{\sqrt{MN}}
  \langle\bm{v}_*,\bm{v}_k\rangle
  \langle\bm{u}_*,\bm{u}_k\rangle\,\bm{u}_*
  + \frac{\theta}{\sqrt{MN}}
    \langle\bm{u}_*,\bm{u}_k\rangle^2\,\bm{W}\bm{v}_*
\biggr),
\]
and the analogous equation for $\langle\bm v_*,\bm v_k\rangle\,\bm v_k$.
The overlaps
$\langle\bm{v}_*,\bm{v}_k\rangle\langle\bm{u}_*,\bm{u}_k\rangle$
and $\langle\bm{u}_*,\bm{u}_k\rangle^2$
can be shown, via standard resolvent and concentration arguments, to converge to
deterministic limits determined by the spectral measures. Moreover, by Haar
invariance and the independence between $\bm W$ and $(\bm u_*,\bm v_*)$, the full
vector $\langle\bm u_*,\bm u_k\rangle\,\bm u_k$ can be analyzed using arguments
similar to those in Appendix \ref{app:general-SE}, yielding convergence of the empirical law to the
signal-plus-Gaussian form stated in Proposition~\ref{prop:signal_plus_noise_sv}.
Its proof is provided in the appendix.
\end{remark}

\medskip

To construct a practical estimator that attains the oracle bound
\eqref{eq:oracle_optimal_spectral_est.}, it suffices to approximate the oracle
linear combination of the informative outlier singular vectors. Since empirical singular vectors are defined only up to a global sign, we fix their orientation via a randomized sign convention (cf.~\cite[Remark~3.6]{feng2022unifying}), which simplifies the theoretical analysis.
For each $k\in\mathcal I_M$, let $\bm u_k(\bm Y)$ and $\bm v_k(\bm Y)$ be any choice of unit outlier singular vectors. Let $\{\xi_k\}_{k\in\mathcal I_M}$ be i.i.d.\ Rademacher random variables, independent of all other random elements in the model \eqref{eq:rectangular spiked model}. Define the $M\times K$ matrix of randomized scaled singular vectors
\begin{align}\label{eq:scaled-outlier-matrix-rand}
  \bm U^\sharp
  \;\bydef\;
  \begin{bmatrix}
    \bm u_{1}^\sharp & \cdots & \bm u_{K}^\sharp
  \end{bmatrix},
  \qquad
  \bm u_{k}^\sharp
  \;\bydef\;
  \sqrt{M}\,\xi_k\,\bm u_k(\bm Y),
  \qquad k\in\mathcal I_M.
\end{align}
Proposition~\ref{prop:signal_plus_noise_sv} shows that the associated asymptotic
signal magnitudes $\sqrt{\nu_1(\{\lambda_k\})}$ and $\sqrt{\nu_2(\{\lambda_k\})}$
are deterministic functions of the noise spectrum (see
Lemma~\ref{lem:spectral_measures_properties}).  The only remaining unknowns are
the \emph{relative signs} of the overlaps
\[
\{\langle\bm u_i^\sharp,\bm u_*\rangle\}_{i\in\mathcal I_M},
\qquad
\{\langle\bm v_i^\sharp,\bm v_*\rangle\}_{i\in\mathcal I_M},
\]
which determine the alignment of the outlier directions with the signal.  Let
$s_i^u,s_i^v\in\{\pm1\}$ denote sign variables (defined up to a global flip in
each channel), and define the practical spectral estimators
\BS\label{eq: optimal PCA estimators}
\begin{align}
 \bm{u}^*_{\mathrm{PCA}}
   \bydef \sum_{i\in\mathcal{I}_M}
   s_i^u \sqrt{\nu_1(\{\lambda_i\})}\,\bm{u}_i^\sharp, \quad
 \bm{v}^*_{\mathrm{PCA}}
   \bydef \sum_{i\in\mathcal{I}_M}
   s_i^v \sqrt{\nu_2(\{\lambda_i\})}\,\bm{v}_i^\sharp .
\end{align}
\ES
The next proposition shows that these estimators attain the oracle performance
whenever the signs $s_i^u,s_i^v$ are chosen consistently with the true overlaps
(up to global sign flips). Its proof can be found in
Appendix~\ref{app:optimal_data_driven_estimators}.

\begin{proposition}[Optimality via Consistent Sign Estimation]
\label{prop:optimal_data_driven_estimators}
Assume the setting of Proposition~\ref{prop:outlier_characterization} with
supercritical $\theta$ and Assumptions~\ref{assump:main}, and let
$\bm{u}^*_{\mathrm{PCA}}$ and $\bm{v}^*_{\mathrm{PCA}}$ be defined in
\eqref{eq: optimal PCA estimators}, with $\mathcal I_M$, $\lambda_i$, and
$\nu_1(\{\lambda_i\}),\nu_2(\{\lambda_i\})$ as above.  Suppose the signs satisfy
\[
s_i^u \ase \sign\langle\bm u_i^\sharp,\bm u_*\rangle,\qquad
s_i^v \ase \sign\langle\bm v_i^\sharp,\bm v_*\rangle,\quad \forall i\in\mathcal{I}_M,
\]
up to a common global flip in each channel. Then
\begin{align*}
 \lim_{M\to \infty} 
 \frac{\langle \bm{u}^*_{\mathrm{PCA}}, \bm{u}_* \rangle^2}
      {\| \bm{u}^*_{\mathrm{PCA}} \|^2 \|\bm{u}_*\|^2}
 &\ase   \sum_{\lambda_i \in \mathcal{K}^*} \nu_1(\{\lambda_i\}), \\
 \lim_{N\to \infty} 
 \frac{\langle \bm{v}^*_{\mathrm{PCA}}, \bm{v}_* \rangle^2}
      {\| \bm{v}^*_{\mathrm{PCA}} \|^2 \|\bm{v}_*\|^2}
 &\ase   \sum_{\lambda_i \in \mathcal{K}^*} \nu_2(\{\lambda_i\}).
\end{align*}
\end{proposition}

Thus, the final step is to find consistent estimators of the relative signs,
which we address next.

%============================
%       revised version
%============================

\subsection{Estimation of Relative Signs}\label{sec:relativesign}

This section addresses the estimation of the relative signs required for the spectral estimators in \eqref{eq: optimal PCA estimators}. We work under the setting of Proposition~\ref{prop:signal_plus_noise_sv}.
Let $\mathcal{I}_M$ denote the set of empirical outlier indices in \eqref{eq:emp-outlier-index}, with $K = \card{\mathcal{I}_M}$. We fix a reference index $r\in\mathcal I_M$ and encode the true relative sign by the vector $ \bm{s}^{\mathrm{R}}_{u,*}\in\{\pm1\}^K$ defined, for
$\ell\in[K]$, as
\begin{align}\label{eq:SRu_def}
[\bm{s}^{\mathrm{R}}_{u,*}]_\ell
\bydef
\sign\!\big(\langle \bm u_{\ell}^\sharp,\bm u_*\rangle\big)\,
\sign\!\big(\langle \bm u_{r}^\sharp,\bm u_*\rangle\big),
\qquad\text{so that }[\bm{s}^{\mathrm{R}}_{u,*}]_r\equiv+1,
\end{align}
which is well-defined in the supercritical regime. Analogously we define \(\bm{s}^{\mathrm{R}}_{v,*}\) as the true relative sign vectors in the \(\bm v\)-channel. We next characterize the row-wise limiting law of the randomized outlying singular vectors. It can be shown that conditioned on $\bm{s}_{u,*}^\mathrm{R} \in \R^K$, the following convergence holds (see Appendix~\ref{sec:NonGaussian-MLE})
\begin{align}\label{eq:limit-matrix-U-rand}
\bigl(\bm{u}_{1}^{\sharp},\ldots,\bm{u}_{K}^{\sharp}\bigr)
  &\wc
  \bigl(
    \mathsf{U}_{1}^{\sharp},
    \ldots,
    \mathsf{U}_{K}^{\sharp}
  \bigr)\bydef
  \bm{\mathsf{U}}^{\sharp}
  \in \mathbb R^K,
\end{align}
with
\begin{align}\label{eq:limit-vector-U-true-body}
  \mathsf{U}_{\ell}^{\sharp}
  \;&\bydef \;
  [\bm{s}^{\mathrm{R}}_{u,*}]_\ell \sqrt{\nu_1(\{\lambda_\ell\})}\,\mathsf{U}_*
  \;+\;
  \sqrt{1-\nu_1(\{\lambda_\ell\})}\,\mathsf{Z}_\ell,
  \qquad \ell \in \mathcal{I}_M,
\end{align}
with $\{\mathsf Z_\ell\}_{\ell\in\mathcal I_M}$ standard Gaussian independent of $\mathsf U_*$.

We consider two estimators of the relative signs: (i) a maximum likelihood estimator (MLE) based on the full prior, and (ii) a non-Gaussian moment contrast (NGMC) estimator which exploits higher order moments and is computationally simpler.

\begin{proposition}[MLE for relative signs]\label{prop:MLE_relative_sign}
 Let $\bm{s} \in \{\pm1\}^K$ and $[\bm{s}]_r=1$ be any fixed relative sign vector. Denote by $P_s$ the joint probability density function of
 \begin{equation}\label{eq:Usharps}
 \bigl([\bm{s}]_\ell \sqrt{\nu_1(\{\lambda_\ell\})}\,\mathsf{U}_*
  \;+\;
  \sqrt{1-\nu_1(\{\lambda_\ell\})}\,\mathsf{Z}_\ell\bigr)_{\ell \in \mathcal{I}_M},    
 \end{equation}
where $\mathsf{U}_*$ and $ (\mathsf{Z}_\ell)_{\ell \in \mathcal{I}_M}$ are distributed as in \eqref{eq:limit-vector-U-true-body}. Denote the $i$-th row of the matrix $\bm U^\sharp \in\mathbb R^{M\times K}$ \eqref{eq:limit-vector-U-true-body} by $\bm{U}^\sharp_{i,:}$. Let $\hat{\bm{s}}_{u}^{\mathrm{MLE}}$ be the maximum likelihood estimator of $\bm{s}^{\mathrm{R}}_{u,*}$
\begin{equation}\label{eq:MLE-est-global}
  \hat{\bm{s}}_{u}^{\mathrm{MLE}}
  \in 
  \underset{\bm{s}\in\mathcal S_r}{\mathrm{argmax}}\ 
  \sum_{i=1}^M
  \log P_s( \bm{U}^\sharp_{i,:} ),
\end{equation}
and analogously $\hat{\bm{s}}_{v}^{\mathrm{MLE}}$ the maximum likelihood estimator of $\bm{s}^{\mathrm{R}}_{v,*}$.
We have:
\begin{enumerate}
    \item If either the law of $\mathsf{U}_*$ or $\mathsf{V}_*$ is not standard Gaussian, then 
    \[
    \hat{\bm{s}}_{u}^{\mathrm{MLE}} \overset{a.s.}{\longrightarrow} \bm{s}^{\mathrm{R}}_{u,*},
    \qquad
    \hat{\bm{s}}_{v}^{\mathrm{MLE}} \overset{a.s.}{\longrightarrow} \bm{s}^{\mathrm{R}}_{v,*}.
    \]
    \item If both $\mathsf{U}_*$ and $\mathsf{V}_*$ are standard Gaussian, then $P_s \sim \mathcal{N}(\bm{0},\bm{I}_K)$ for any $\bm{s}\in\mathcal S_r$, and consistent estimation of the relative signs via MLE is impossible.
\end{enumerate}
\end{proposition}

\proofseeapp{\propref{prop:MLE_relative_sign}}{app:MLE_relative_sign}

\begin{remark}[Well-posedness of the likelihood]\label{rem:LLR-well-posed}
The random vector $\bm{\mathsf{U}}^{\sharp}$ is constructed by adding an independent Gaussian vector with non-degenerate covariance to the signal $\mathsf{U}_*$. Consequently, for any sign configuration $\bm{s}$, the joint law of $\bm{\mathsf{U}}^{\sharp}$ is the convolution of the prior measure of $\mathsf{U}_*$ with a non-degenerate Gaussian distribution on $\mathbb{R}^K$. This ensures that the distribution admits a smooth, strictly positive density $P_s$ with respect to the Lebesgue measure. Hence, the log-likelihood terms in \eqref{eq:MLE-est-global} are well-defined.
\end{remark}

Proposition~\ref{prop:MLE_relative_sign} establishes that any non-Gaussianity in the prior $\mathsf U_*$ renders the relative signs identifiable, yielding an asymptotically consistent MLE. However, minimizing the objective over $\mathcal{S}_r$ can be computationally intensive when the prior lacks a closed-form Gaussian convolution. This motivates a simpler alternative that specifically exploits non-Gaussianity through an appropriate moment contrast, called the \textit{non-Gaussian moment contrast} (NGMC) scheme.

\begin{assumption}[Non-Gaussian even moment]\label{assump: Non-Gaussian Prior}
There exists an even integer \(k \ge 0\) such that
\[
    \mathbb{E}[\,\mathsf{U}_*^{\,k+2}\,] \;\neq\; (k+1)!! .
\]
In other words, at least one even-order moment of \(\mathsf{U}_*\) differs from the corresponding moment of a standard Gaussian random variable.
\end{assumption}

\begin{proposition}[NGMC estimator for relative signs]\label{prop:NGMC-est}
Assume the setting of Proposition~\ref{prop:signal_plus_noise_sv}, and suppose Assumption~\ref{assump: Non-Gaussian Prior} holds. Let $k \ge 2$ be the smallest even integer admissible in Assumption~\ref{assump: Non-Gaussian Prior}, 
and let \(f(x) \bydef x^{k+1}\) be the corresponding entrywise moment-contrast function. Fix an arbitrary reference outlier index \(r \in \mathcal{I}_M\), and for any other outlier \(j \in \mathcal{I}_M \setminus \{r\}\) define
\begin{subequations}\label{eq:NGMC-est}
\begin{align}
    \hat{s}_{u,j}^{\mathrm{NGMC}}
    &\bydef
    \sign\!\Big(
        f\big(\bm u_r^{\sharp}\big)^{\UT}\,
        \bm u_j^{\sharp}
    \Big)
    \cdot
    \sign\!\Big(
        \mathbb{E}[\mathsf U_*^{\,k+2}] - (k+1)!!
    \Big),
    \\
    \hat{s}_{v,j}^{\mathrm{NGMC}}
    &\bydef
    \hat{s}_{u,j}^{\mathrm{NGMC}}
    \cdot
    \sign\!\big(\nu_3(\{\sigma_r\})\,\nu_3(\{\sigma_j\})\big).
\end{align}
\end{subequations}
Then the NGMC estimators are consistent: $\hat{s}_{u,j}^{\mathrm{NGMC}} \overset{a.s.}{\longrightarrow}  [\bm{s}^{\mathrm{R}}_{u,*}]_j$ and $\hat{s}_{v,j}^{\mathrm{NGMC}} \overset{a.s.}{\longrightarrow} [\bm{s}^{\mathrm{R}}_{v,*}]_j$.
\end{proposition}

\proofseeapp{\propref{prop:NGMC-est}}{app:NGMC-est}

%% file: body/spectral_init.tex
\section{OAMP Algorithm with Spectral Initialization} \label{sec: spec OAMP}

The optimal spectral estimator developed in the previous section naturally
suggests a way to initialize iterative algorithms. Here, we study OAMP when
initialized using this spectral estimator, following the construction of
Section~\ref{sec:optimal OAMP}.

\subsection{Spectrally-Initialized Optimal OAMP}
\label{sec:spec-optimal-oamp}

We use a tilde to distinguish the iterates of the spectrally initialized
algorithm from those of the generic OAMP recursion in
Section~\ref{sec:optimal OAMP}. Under the assumptions of
Proposition~\ref{prop:optimal_data_driven_estimators}, the initialization is
given by the unit-variance normalized versions of the optimal spectral
estimators \eqref{eq: optimal PCA estimators}. Specifically, at $t=1$,
\BE \label{eq: spectral init}
\tilde{\bm u}_1 \;=\;
\Big(\sum_{k\in\mathcal I_M}\nu_1(\{\lambda_k\})\Big)^{-1/2} \bm{u}^*_{\mathrm{PCA}}
\quad\text{and}\quad
\tilde{\bm v}_1 \;=\; \Big(\sum_{k\in\mathcal I_M}\nu_2(\{\lambda_k\})\Big)^{-1/2}
\bm{v}^*_{\mathrm{PCA}},
\EE
where $\bm{u}^*_{\mathrm{PCA}}$ and $\bm{v}^*_{\mathrm{PCA}}$ are the optimal
spectral estimates defined in \eqref{eq: optimal PCA estimators}.

For all subsequent iterations $t \ge 2$, and for fixed sign parameters
$s_1,s_2\in\{+1,-1\}$, the algorithm proceeds according to the standard
optimal OAMP update rules in \eqref{eq:u-opt}--\eqref{eq:v-opt}. The update
rules are
\BS\label{eq: spec-OAMP}
\begin{align}
\tilde{\bm{u}}_t^* &=
\frac{1}{\sqrt{\widetilde{w}_{1,t}}}\Big[
F^*_{t}(\bm{Y}\bm{Y}^\UT)\,
\bar{\phi}(\tilde{\bm{u}}_{t-1}^* \mid \widetilde{w}_{1,t-1}, s_{1})
\;+\;
\tilde{F}^*_{t}(\bm{Y}\bm{Y}^\UT)\,
\bm{Y}\,
\bar{\phi}(\tilde{\bm{v}}^*_{t-1} \mid \widetilde{w}_{2,t-1}, s_{2})
\Big], \label{eq:OAMP Algo u si} \\
\tilde{\bm{v}}_t^* &=
\frac{1}{\sqrt{\widetilde{w}_{2,t}}}\Big[
G^*_{t}(\bm{Y}^\UT\bm{Y})\,
\bar{\phi}(\tilde{\bm{v}}^*_{t-1} \mid \widetilde{w}_{2,t-1}, s_{2})
\;+\;
\tilde{G}^*_{t}(\bm{Y}^\UT\bm{Y})\,
\bm{Y}^\UT\,
\bar{\phi}(\tilde{\bm{u}}_{t-1}^* \mid \widetilde{w}_{1,t-1}, s_{1})
\Big], \label{eq:OAMP Algo v si}
\end{align}
\ES
where $F_t^*,\tilde F_t^*,G_t^*,\tilde G_t^*$ (as functions of
$\lambda,\widetilde{\rho}_{1,t},\widetilde{\rho}_{2,t}$) denote the trace-free
spectral matrix denoisers from \eqref{eq:Fstar}--\eqref{eq:tGstar}, and
$\bar\phi(\cdot\mid w,s)$ denotes the signed DMMSE denoiser associated with the
scalar Gaussian channel
\[
\mathsf{X} = s\sqrt{w}\,\mathsf{X}_* + \sqrt{1-w}\,\mathsf{Z},
\qquad s\in\{\pm1\},
\]
defined by
\begin{equation}\label{eq:signed-family}
\bar\phi(x\mid w,s) \;\bydef\; \bar\phi(sx\mid w),
\end{equation}
where $\bar\phi(\cdot\mid w)$ is the DMMSE denoiser for the standard scalar
Gaussian channel in \eqref{eq:mmse-scalar}.

\paragraph{Update of state evolution parameters.}
The scalar parameters used in the denoisers are updated, for $t\ge2$, by
\BS \label{eq:OptimalRecursion for spec}
\begin{align}
\widetilde{\rho}_{1,t}
&= \frac{1}{\mathrm{mmse}_{\mathsf{U}}(\widetilde{w}_{1,t-1})}
   - \frac{1}{1-\widetilde{w}_{1,t-1}},
&\qquad
\widetilde{\rho}_{2,t}
&= \frac{1}{\mathrm{mmse}_{\mathsf{V}}(\widetilde{w}_{2,t-1})}
   - \frac{1}{1-\widetilde{w}_{2,t-1}},\\
\widetilde{w}_{1,t}
&= 1 - \frac{1 - {\langle P_t^*(\lambda;\widetilde{\rho}_{1,t},\widetilde{\rho}_{2,t}) \rangle}_{\muM}}
{{\langle P_t^*(\lambda;\widetilde{\rho}_{1,t},\widetilde{\rho}_{2,t}) \rangle}_{\muM}}
\cdot \frac{1}{\widetilde{\rho}_{1,t}},
&\qquad
\widetilde{w}_{2,t}
&= 1 - \frac{1 - {\langle Q_t^*(\lambda;\widetilde{\rho}_{1,t},\widetilde{\rho}_{2,t}) \rangle}_{\muN}}
{{\langle Q_t^*(\lambda;\widetilde{\rho}_{1,t},\widetilde{\rho}_{2,t}) \rangle}_{\muN}}
\cdot \frac{1}{\widetilde{\rho}_{2,t}}.
\end{align}
\ES
The recursion is initialized with
\[
\widetilde{w}_{1,1} = \sum_{\lambda_i\in\mathcal{K}^*}\nu_1(\{\lambda_i\}),
\qquad
\widetilde{w}_{2,1} = \sum_{\lambda_i\in\mathcal{K}^*}\nu_2(\{\lambda_i\}).
\]

\paragraph{Choice of global sign parameters.}
The spectral initializers $(\tilde{\bm u}_1,\tilde{\bm v}_1)$ inherit a single
global Rademacher ambiguity from the randomization
\eqref{eq:scaled-outlier-matrix-rand}: their overlaps with the true signals,
$\langle\tilde{\bm u}_1,\bm u_*\rangle$ and
$\langle\tilde{\bm v}_1,\bm v_*\rangle$, are determined only up to a $\pm1$
factor. The sign parameters $s_1,s_2\in\{+1,-1\}$ in
\eqref{eq: spec-OAMP} are introduced to resolve this ambiguity; they are chosen
once according to the prior structure:
\begin{itemize}
\item \textit{Asymmetric priors.}
When the priors for $(\mathsf U_*,\mathsf V_*)$ are asymmetric, the global signs
are statistically identifiable. We estimate them using the MLE or moment-based
procedures in Appendix~\ref{app:pf-global_sign_gsmle} and
Appendix~\ref{app:pf-global_sign_GSOC}, and set $(s_1,s_2)$ to these
estimates. With this choice, the OAMP iterates have asymptotically positive
overlap with the true signals; see Fact~\ref{fact:signed-dmmse-sign}(i).

\item \textit{Symmetric priors.}
For symmetric priors, the individual global signs cannot be identified (as
noted in \cite[Remark~3.6]{feng2022unifying}). Nevertheless,
Lemma~\ref{lem:inter_channel_sign_coupling} shows that one may, without loss of
generality, adopt the convention
\[
(s_1,s_2)=\Bigl(1,\sign\bigl(\nu_3(\{\sigma_r\})\bigr)\Bigr),
\]
where $\nu_3(\{\sigma_r\})\neq 0$ denotes the point mass of the (signed) cross
measure $\nu_3$ associated with the reference outlier
$\lambda_r\in\mathcal{K}^*$. This pair is determined only up to a common global
flip; under such a flip, the state evolution is preserved in absolute value.
Equivalently, the SE recursion for the squared overlaps (and hence the cosine
similarities) is invariant; see Fact~\ref{fact:signed-dmmse-sign}(ii).
\end{itemize}
Unless stated otherwise, all subsequent state evolution results are understood
for the sign-resolved iterates obtained by the above choice of $(s_1,s_2)$.
A detailed treatment of the global sign ambiguity is provided in
Appendix~\ref{app:pf-global_sign}.

%==========================================
\subsection{State Evolution of Spectrally-Initialized OAMP}\label{subsec:spec-SE}

Let $(\tilde{\mu}_{u,t}, \tilde{\mathsf{Z}}_{u,t})$ and $(\tilde{\mu}_{v,t}, \tilde{\mathsf{Z}}_{v,t})$ denote the SE parameters for the spectrally-initialized iterates. The associated scalar random variables are
\begin{align}\label{eq:spec-SE-rvs}
(\mathsf{U}_{*}, \mathsf{A}) &\sim \pi_u, \qquad \tilde{\mathsf{U}}_t = \tilde{\mu}_{u,t}\,\mathsf{U}_{*} + \tilde{\mathsf{Z}}_{u,t}, \qquad \forall t\in\mathbb{N},\\
(\mathsf{V}_{*}, \mathsf{B}) &\sim \pi_v, \qquad \tilde{\mathsf{V}}_t = \tilde{\mu}_{v,t}\,\mathsf{V}_{*} + \tilde{\mathsf{Z}}_{v,t}, \qquad \forall t\in\mathbb{N},
\end{align}
where, for each $t$, the noise variables $\tilde{\mathsf{Z}}_{u,t}$ and $\tilde{\mathsf{Z}}_{v,t}$ are centered Gaussian and independent of $(\mathsf{U}_{*},\mathsf{A})$ and $(\mathsf{V}_{*},\mathsf{B})$, respectively.

\paragraph{Initialization and the global-sign issue.}
The SE is initialized at $t=1$ using the spectral estimators from~\eqref{eq: spectral init}. As usual, spectral singular vectors are only defined up to a global sign, and (prior to any convention) this sign is not identifiable from the data under symmetric priors. To make the initialization amenable to state evolution, we adopt the randomized sign convention used in the construction of the spectral initializer (cf.~\eqref{eq:scaled-outlier-matrix-rand}). Using this method, we may write the realized orientations as
\[
\sign\!\big(\langle \tilde{\bm u}_1,\bm u_*\rangle\big)=:\mathsf{S}_u\in\{\pm1\},\qquad
\sign\!\big(\langle \tilde{\bm v}_1,\bm v_*\rangle\big)=:\mathsf{S}_v\in\{\pm1\},
\]
where $(\mathsf{S}_u,\mathsf{S}_v)$ are Rademacher variables independent of $(\bm{u}_\ast,\bm{v}_\ast,\bm{W})$ and hence can be treated as fixed when conditioning. Hereafter, we work conditionally on $(\mathsf{S}_u,\mathsf{S}_v)$. Accordingly, we define the initial SE parameters in terms of the \emph{realized} orientation:
\begin{align}
\tilde{\mu}_{u,1}
&= \mathsf{S}_u\cdot
\Big(\sum_{\lambda_i\in\mathcal{K}^*}\nu_1(\{\lambda_i\})\Big)^{1/2},
\qquad
\Var(\tilde{\mathsf{Z}}_{u,1})
= 1 - \sum_{\lambda_i\in\mathcal{K}^*}\nu_1(\{\lambda_i\}),
\label{eq:spec-SE-init-u}\\
\tilde{\mu}_{v,1}
&= \mathsf{S}_v\cdot
\Big(\sum_{\lambda_i\in\mathcal{K}^*}\nu_2(\{\lambda_i\})\Big)^{1/2},
\qquad
\Var(\tilde{\mathsf{Z}}_{v,1})
= 1 - \sum_{\lambda_i\in\mathcal{K}^*}\nu_2(\{\lambda_i\}),
\label{eq:spec-SE-init-v}
\end{align}
where the point masses $\nu_1(\{\lambda\})$ and $\nu_2(\{\lambda\})$ on the limiting outliers $\lambda\in\mathcal{K}^*$ are characterized in~\propref{prop:outlier_characterization}.

\paragraph{Recursion for $t\ge 2$.}
For $t\ge 2$, the SE parameters $\{(\tilde{\mu}_{u,t},\tilde{\mathsf{Z}}_{u,t})\}$ and $\{(\tilde{\mu}_{v,t},\tilde{\mathsf{Z}}_{v,t})\}$ follow from the general update rules in~\eqref{eq: mean value of SE1} and~\eqref{eq: Cov of SE1}. Substituting the denoisers in~\eqref{eq: spec-OAMP} yields the compact scalar-channel form
\begin{align}\label{eq:spec-SE-recursion-compact}
\tilde{\mu}_{u,t}
&=
\sign\!\Biggl(
  \E\Bigl[
    \mathsf U_*\,
    \bar{\phi}\Bigl(
      \tilde{\mathsf U}_{t-1}\,\Bigm|\,
      \tilde w_{1,t-1},\, s_1
    \Bigr)
  \Bigr]
\Biggr)\,
\tilde w_{1,t}^{1/2},
\qquad
\Var(\tilde{\mathsf Z}_{u,t}) = 1-\tilde w_{1,t},
\\
\tilde{\mu}_{v,t}
&=
\sign\!\Biggl(
  \E\Bigl[
    \mathsf V_*\,
    \bar{\phi}\Bigl(
      \tilde{\mathsf V}_{t-1}\,\Bigm|\,
      \tilde w_{2,t-1},\, s_2
    \Bigr)
  \Bigr]
\Biggr)\,
\tilde w_{2,t}^{1/2},
\qquad
\Var(\tilde{\mathsf Z}_{v,t}) = 1-\tilde w_{2,t}.
\end{align}
where $\bar{\phi}(\cdot\,|\,w,s)$ is the signed DMMSE denoiser (cf.~\eqref{eq:signed_DMMSE}), and $(s_1,s_2)$ encode the convention used to fix the global-sign ambiguity (see the discussion around the initialization).

Having defined the SE parameters in~\eqref{eq:spec-SE-init-u}--\eqref{eq:spec-SE-recursion-compact}, we now state the state evolution characterization for spectrally-initialized OAMP.

\begin{theorem}[State Evolution of Spectrally-Initialized OAMP]\label{thm: spec-SE}
Consider the rectangular spiked model $\bm Y$ in~\eqref{eq:rectangular spiked model} with super-critical $\theta$, satisfying \assumpref{assump:main} and~\ref{assump:lip-MMSE} under the settings of \propref{prop:optimal_data_driven_estimators}. Let $\{(\tilde{\bm{u}}_t,\tilde{\bm{v}}_t)\}_{t\ge 1}$ be the iterates generated by~\eqref{eq: spectral init}--\eqref{eq: spec-OAMP}, and let $\{(\tilde{\mathsf U}_t,\tilde{\mathsf V}_t)\}_{t\ge1}$ be the SE variables defined by~\eqref{eq:spec-SE-rvs} with initialization~\eqref{eq:spec-SE-init-u}--\eqref{eq:spec-SE-init-v} and recursion~\eqref{eq:spec-SE-recursion-compact}.

Then, for any fixed $t\in\mathbb{N}$, conditionally on the realized phase variables $(\mathsf{S}_u,\mathsf{S}_v)$, the empirical joint distribution of the iterates converges in Wasserstein-2 distance to the law of the SE variables:
\begin{align}
(\bm{u}_*, \tilde{\bm{u}}_1, \ldots, \tilde{\bm{u}}_t; \bm{a})
&\xrightarrow{W_2}
(\mathsf{U}_*, \tilde{\mathsf{U}}_1, \ldots, \tilde{\mathsf{U}}_t; \mathsf{A}),\\
(\bm{v}_*, \tilde{\bm{v}}_1, \ldots, \tilde{\bm{v}}_t; \bm{b})
&\xrightarrow{W_2}
(\mathsf{V}_*, \tilde{\mathsf{V}}_1, \ldots, \tilde{\mathsf{V}}_t; \mathsf{B}).
\end{align}
\end{theorem}
\proofseeapp{\thref{thm: spec-SE}}{app:pf-thm-spec-se}

\begin{remark}[Relation to existing spectral initialization results]
Our approach differs from prior work on spectral initialization \cite{mondelli2021pca,zhong2021approximate} in the following respects. First, to accommodate the multi-outlier setting, we employ an optimally weighted combination of all informative outlier components, rather than relying solely on the leading eigenvector. Second, we formulate the initialization phase as a \emph{one-shot} OAMP update based on singular-vector equations, avoiding the auxiliary iterative AMP constructions used in \cite{mondelli2021pca,zhong2021approximate}. Third, this direct formulation bypasses restrictive technical conditions required for the convergence of auxiliary AMPs to sample PCs (such as nonnegative free cumulants \cite{mondelli2021pca} or sufficiently large signal-to-noise ratios \cite{zhong2021approximate}), thereby establishing validity for general rotationally invariant models under Assumption~\ref{assump:main}.
\end{remark}

%% file: body/simulation.tex
\section{Simulation Results}
\label{sec:simulation}

In this section, we provide numerical evidence to validate our theoretical results. We first evaluate the finite-sample performance of the proposed spectral estimators for relative-sign recovery. Subsequently, we investigate the dynamics of the spectrally-initialized OAMP algorithm under both Gaussian and general rotationally invariant noise models.

\subsection{Spectral Estimator}

A key feature of the rectangular RI model is that a single rank-one signal typically generates \emph{multiple} informative outlier singular values once the signal strength exceeds a certain threshold. In such regimes, standard PCA (which relies solely on the top singular vector) is suboptimal because it discards the signal energy contained in secondary outliers. Our proposed method aims to remedy this by optimally combining all informative outliers.

\begin{figure}[htbp]
\centering
%------------------ Row 1: MLE ------------------%
\begin{minipage}[t]{0.48\columnwidth}
    \centering
    \includegraphics[width=\linewidth]{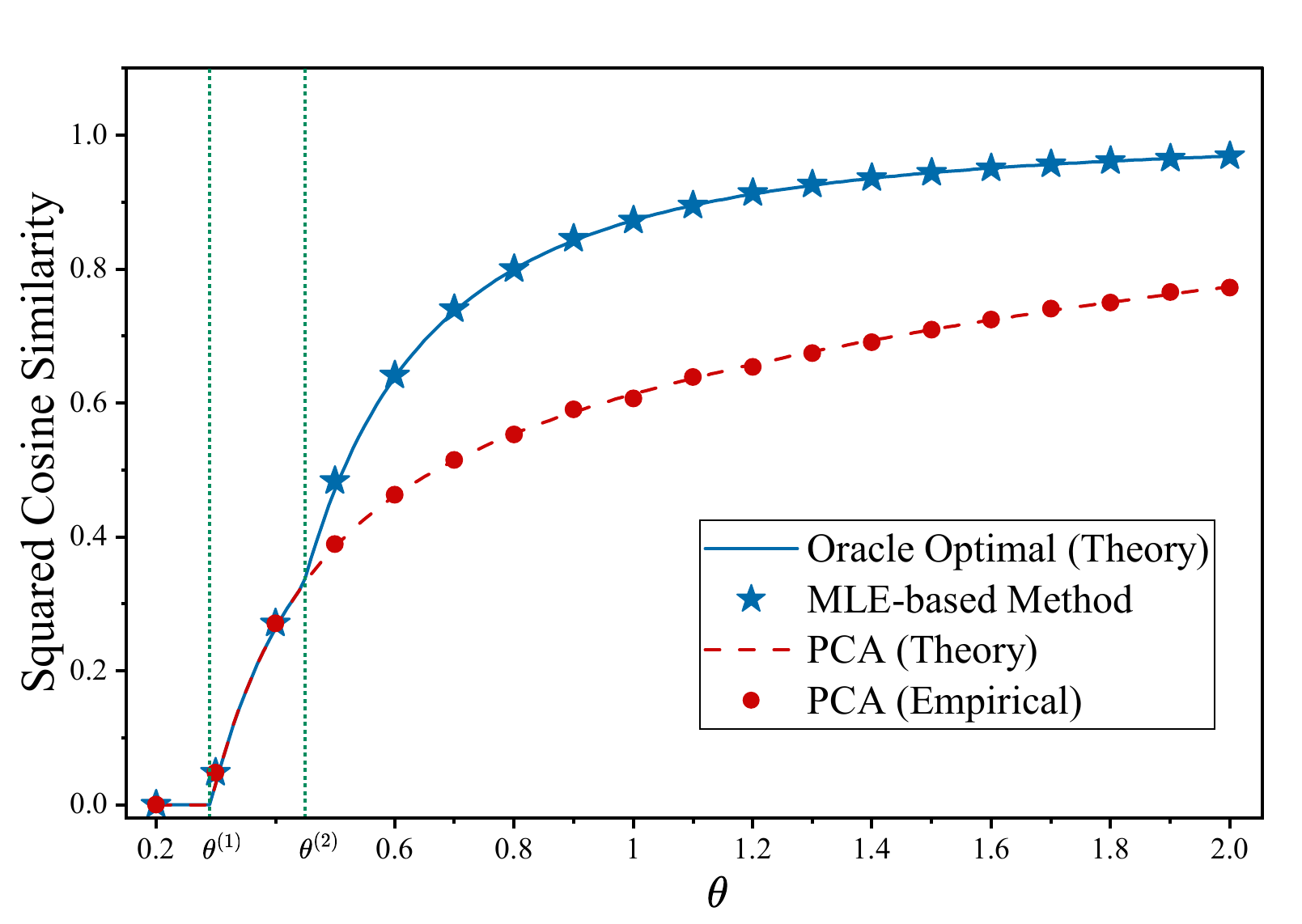}
    {\scriptsize (a) MLE: \(\bm{u}\)-channel}
\end{minipage}\hfill
\begin{minipage}[t]{0.48\columnwidth}
    \centering
    \includegraphics[width=\linewidth]{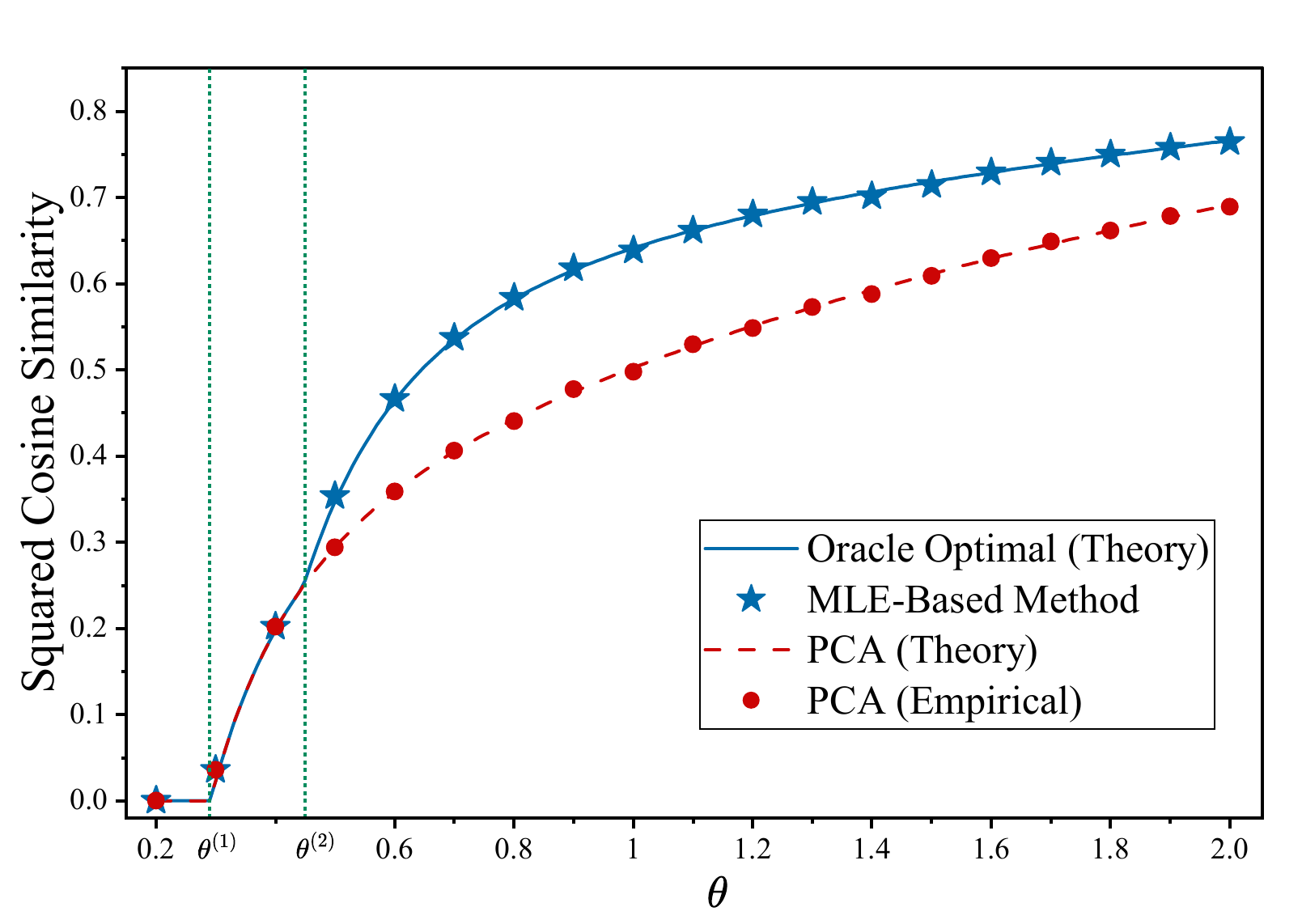}
    {\scriptsize (b) MLE: \(\bm{v}\)-channel}
\end{minipage}
\vspace{1em}
%------------------ Row 2: NGMC ------------------%
\begin{minipage}[t]{0.48\columnwidth}
    \centering
    \includegraphics[width=\linewidth]{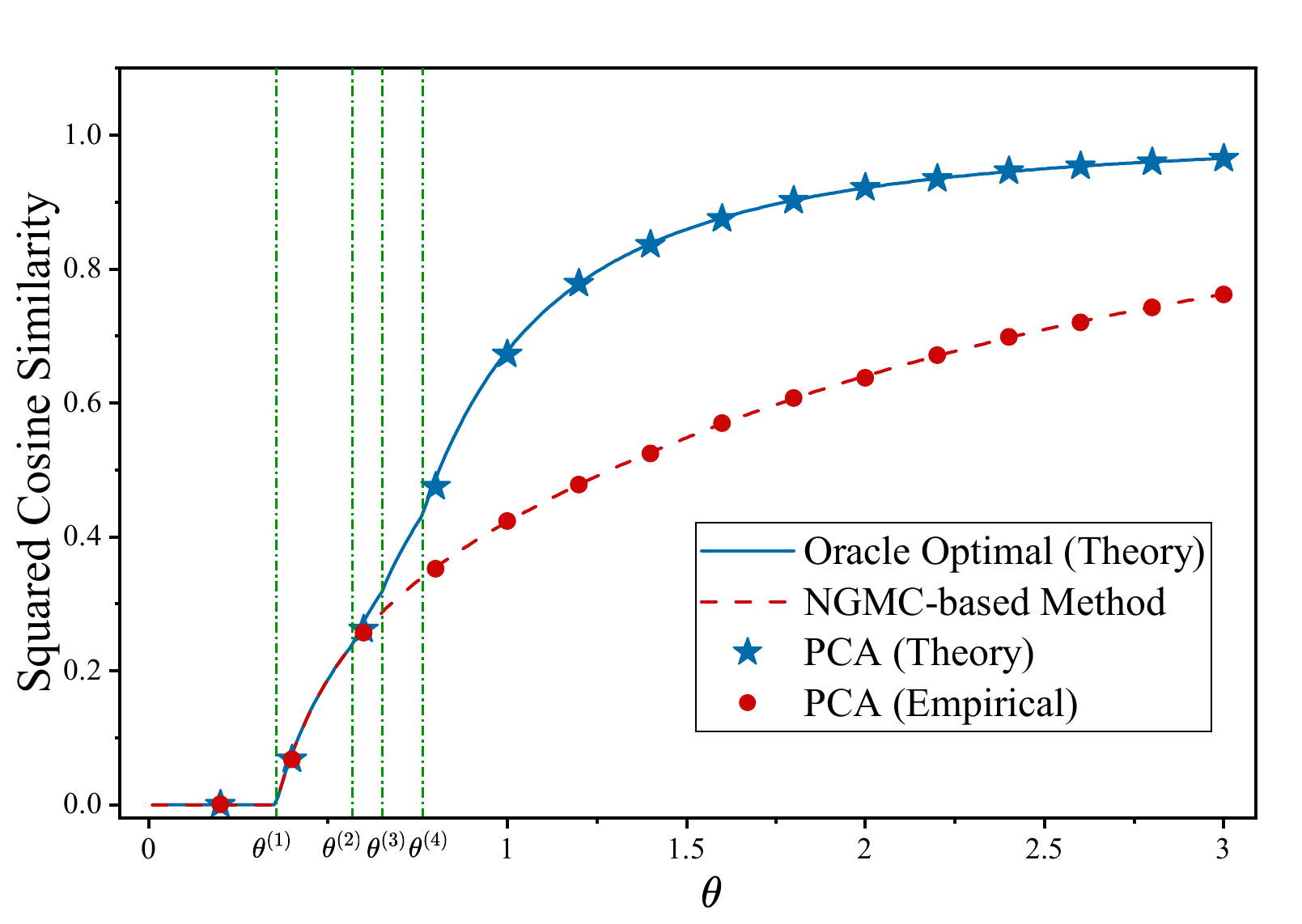}
    {\scriptsize (c) NGMC: \(\bm{u}\)-channel}
\end{minipage}\hfill
\begin{minipage}[t]{0.48\columnwidth}
    \centering
    \includegraphics[width=\linewidth]{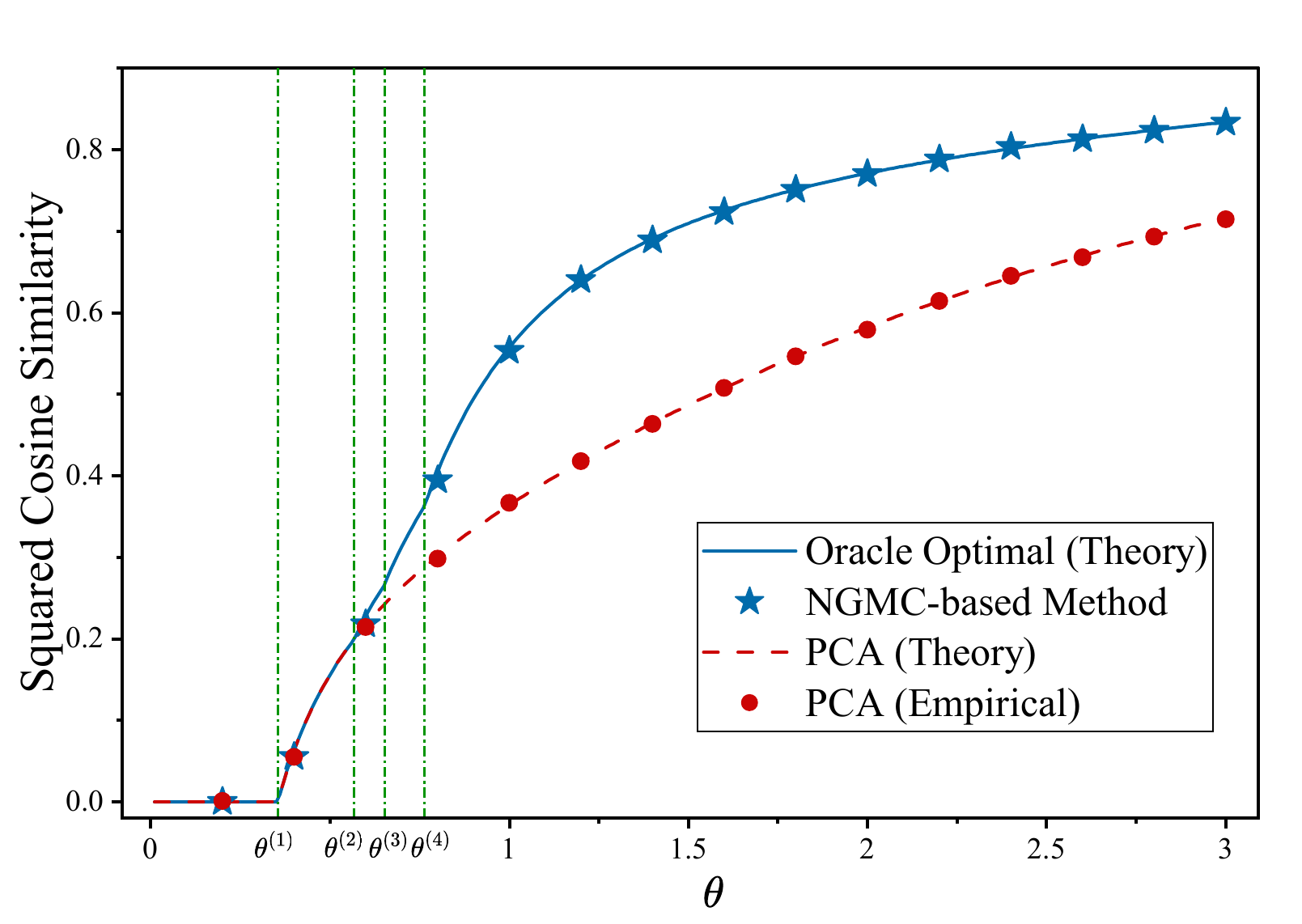}
    {\scriptsize (d) NGMC: \(\bm{v}\)-channel}
\end{minipage}
\caption{\textbf{Relative-sign estimation via MLE and NGMC.}
(a)--(b): MLE with Rademacher $\mathsf U_*$ and Gaussian $\mathsf V_*$ under noise law $\mu_2$ ($\delta=0.7$).
(c)--(d): NGMC with Student-$t$ $\mathsf U_*$ (df$=5$) under noise law $\mu_3$ ($\delta=0.8$).
Vertical green lines indicate the SNR phase transitions where successive outliers detach from the bulk (cf.~Fig.~\ref{fig:combined-outliers}).
Across all experiments, $N=5000$ and results are averaged over $50$ trials. Stars denote the proposed estimators; circles denote PCA; the solid line denotes the oracle bound.}
\label{fig:merged_estimation}
\end{figure}

We evaluate the maximum-likelihood (MLE) and non-Gaussian moment-contrast (NGMC)
sign estimators under two representative prior settings; the results are
summarized in Fig.~\ref{fig:merged_estimation}.
\begin{itemize}[leftmargin=*]
\item \emph{MLE with a Rademacher prior.}
We consider the setting of Fig.~\ref{fig:combined-outliers} with
$\mathsf U_*\sim\mathrm{Rad}(\pm1)$, $\mathsf V_*\sim\mathcal N(0,1)$, and noise
law
\[
\mu_2(\lambda)=\frac{2}{\pi}\sqrt{(\lambda-2)(4-\lambda)}\,\mathbf{1}_{[2,4]}(\lambda).
\]
The resulting MLE performance is reported in the top row of
Fig.~\ref{fig:merged_estimation}.  We take the largest outlier (index $1$) as
the reference.  For each $j\in\mathcal I_M\setminus\{1\}$, define the coordinate
pairs
$
(x_i,y_i)\bydef\bigl([\bm u_1^\sharp]_i,[\bm u_j^\sharp]_i\bigr).
$
The pairwise relative-sign MLE is then
\begin{align}\label{eq:MLE-Rad-opt}
\hat s_{u,j}^{\mathrm{MLE}}
\in \argmax_{s\in\{\pm1\}}
\sum_{i=1}^M \log p_s(x_i,y_i),
\end{align}
where, in terms of the spectral atoms $\nu_1(\{\lambda\})$ from
Lemma~\ref{lem:spectral_measures_properties}, the per-coordinate likelihood
satisfies
\begin{align}\label{eq:ps-Rad}
p_s(x,y)\ \propto\
\cosh\!\left(
\frac{\sqrt{\nu_1(\{\lambda_1\})}}{1-\nu_1(\{\lambda_1\})}\,x
+s\,\frac{\sqrt{\nu_1(\{\lambda_j\})}}{1-\nu_1(\{\lambda_j\})}\,y
\right).
\end{align}

\item \emph{NGMC with a Student-$t$ prior.}
We take $\mathsf U_*$ to be a rescaled Student-$t$ prior (df$=5$, unit variance)
and $\mathsf V_*\sim\mathcal N(0,1)$ under the noise law
\[
\mu_3(\lambda)=\frac{1}{\pi}\sqrt{(\lambda-2)(4-\lambda)}\,\mathbf{1}_{[2,4]}(\lambda)
+\frac{1}{\pi}\sqrt{(\lambda-6)(8-\lambda)}\,\mathbf{1}_{[6,8]}(\lambda),
\]
as in Fig.~\ref{fig:combined-outliers}.  Since the Gaussian convolution of this
prior is not available in closed form, we use the NGMC estimator in
Proposition~\ref{prop:NGMC-est} with the cubic contrast $f(x)=x^3$.  The bottom
row of Fig.~\ref{fig:merged_estimation} reports the resulting performance.
\end{itemize}

Fig.~\ref{fig:merged_estimation} shows that aggregating informative outliers
strictly improves upon standard PCA and closely tracks the oracle benchmark,
indicating accurate recovery of the relative signs.

\subsection{Performance of OAMP}

% \paragraph{i.i.d. Gaussian noise.}
We first validate the theory in the classical i.i.d.\ Gaussian setting. The true signals $\bm u_*$ and $\bm v_*$ have i.i.d.\ Rademacher entries. Figure~\ref{fig:res} reports the squared cosine similarities achieved by PCA, AMP, and OAMP, together with the corresponding state evolution (SE) predictions. The OAMP iterates closely match SE and converge to the same fixed point as standard AMP, in agreement with the equivalence in Proposition~\ref{thm:OAMP_Wigner_FP}.

\begin{figure}[htbp]
\begin{minipage}[b]{.48\linewidth}
  \centering
  \centerline{\includegraphics[width=\linewidth]{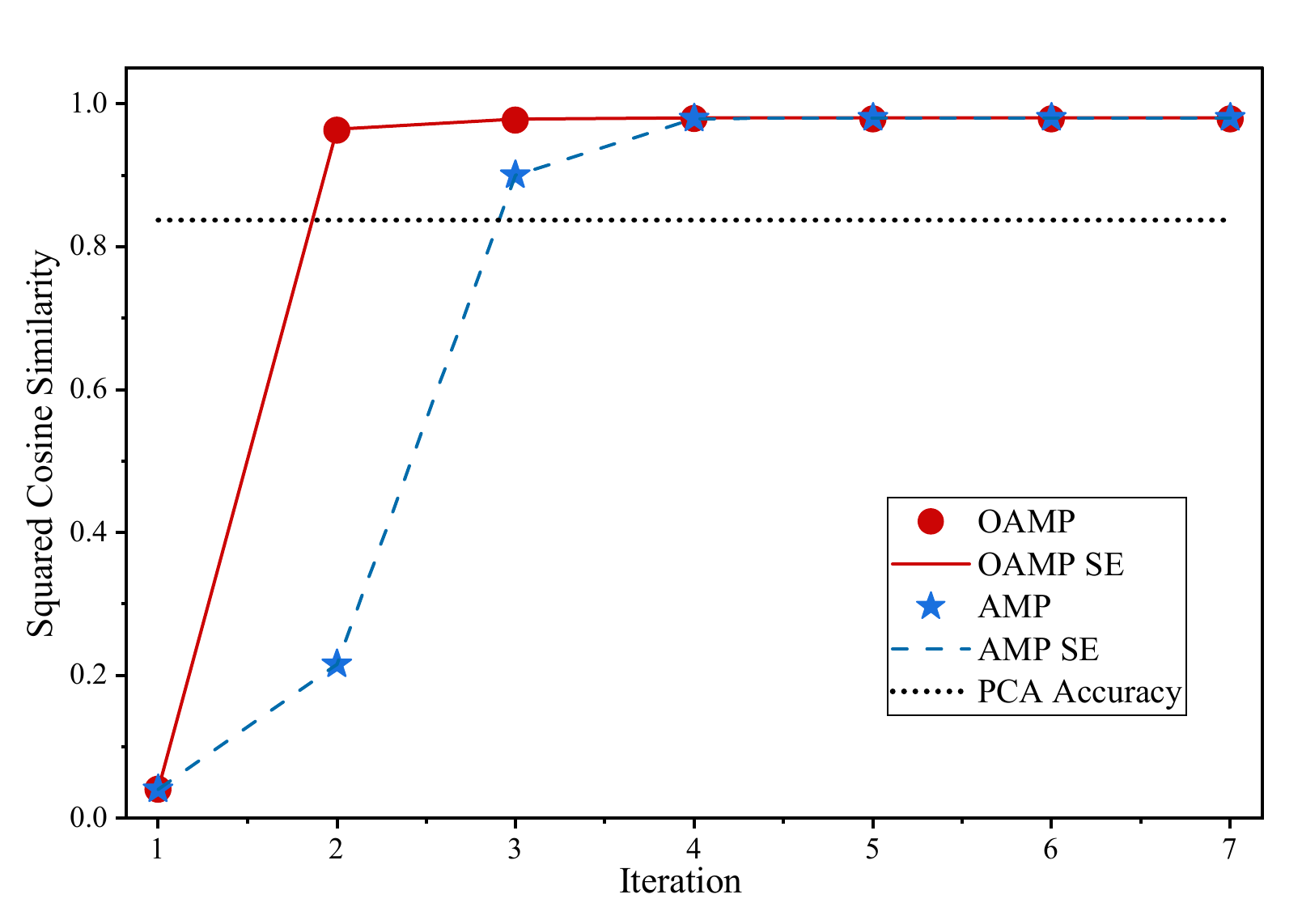}}
  \centerline{(a) $\bm u$-channel}\medskip
\end{minipage}
\hfill
\begin{minipage}[b]{0.48\linewidth}
  \centering
  \centerline{\includegraphics[width=\linewidth]{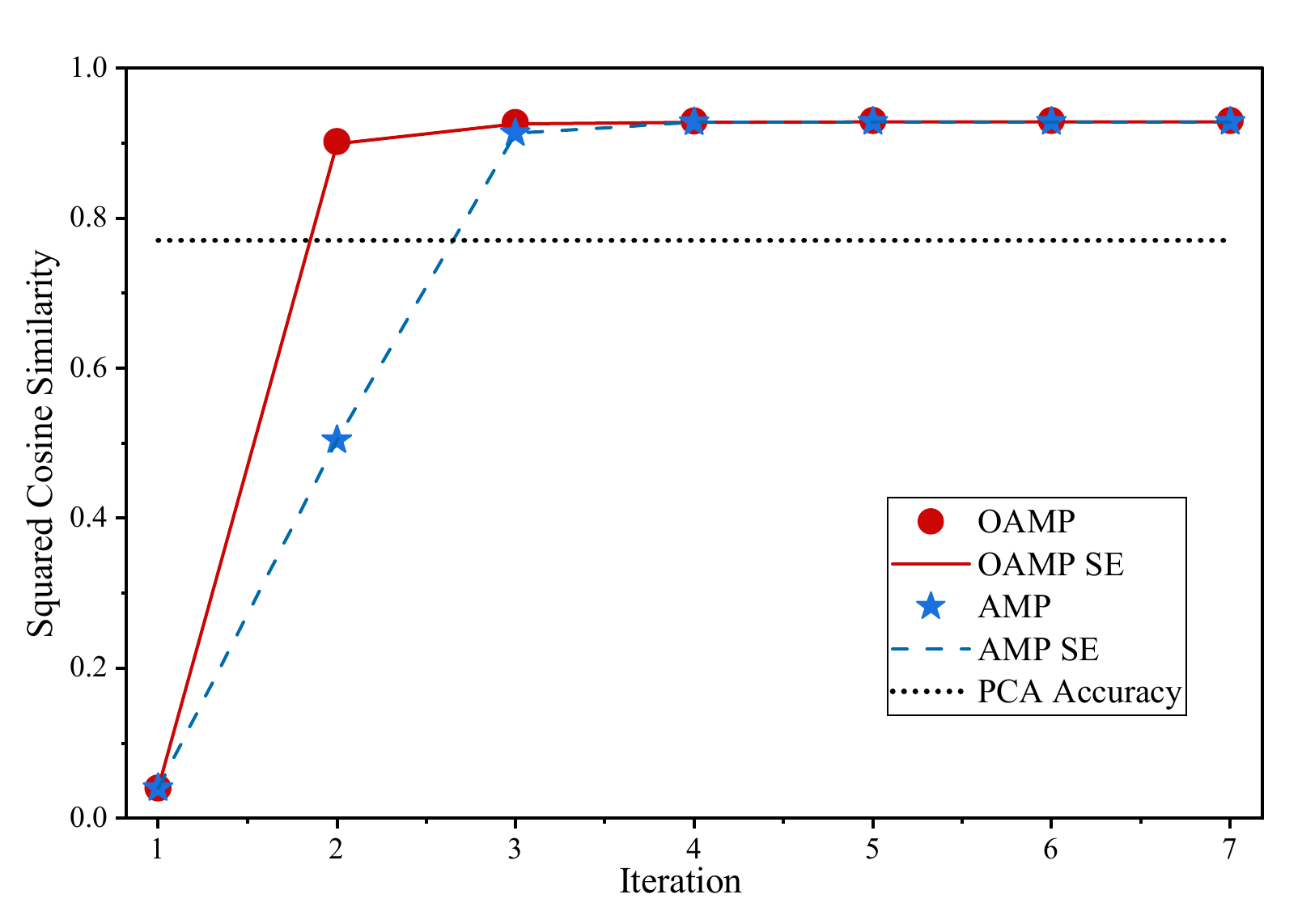}}
  \centerline{(b) $\bm v$-channel}\medskip
\end{minipage}
\caption{Performance under i.i.d.\ Gaussian noise ($N=8000,\delta=0.6$, $\theta=2$). Markers are empirical averages over $50$ trials. All methods are initialized with cosine similarity $0.2$ in both channels.}
\label{fig:res}
\end{figure}

% \paragraph{Non-Gaussian rotationally invariant noise: optimal spectral initialization.}
We next study spectrally-initialized OAMP under rotationally invariant (RI) noise with bulk density $\mu_2$ in Fig.~\ref{fig:combined-outliers}. Throughout, OAMP uses the optimal spectral initialization from Section~\ref{sec: spec OAMP}. Figure~\ref{fig:threepoint-vs-rach} plots the \emph{signed} cosine similarity to highlight the intrinsic global sign ambiguity of the spectral initializer. To illustrate the impact of global sign ambiguity, we explicitly realize two possible global orientations for the spectral initializer: one with a positive initial overlap with $(\bm u_*,\bm v_*)$, and one with a negative overlap.

We compare two signal priors that necessitate different strategies for determining the denoiser signs $(s_1,s_2)$, as discussed in Section \ref{sec:spec-optimal-oamp}:

\begin{itemize}
  \item \emph{Asymmetric three-point prior.}
  Both $\bm u_*$ and $\bm v_*$ follow a three-point mixture supported on $\{-1,\,1.5,\,-0.5\}$ with probabilities $(0.2,\,0.3,\,0.5)$. The asymmetry renders the realized global signs $(s_u,s_v)$ identifiable. Using the global-sign estimators in \appref{app:pf-global_sign}, we align the denoiser signs $(s_1,s_2)$ with $(s_u,s_v)$ at initialization. Consequently, the iterates rapidly correct the orientation: the trajectories started from the two initial orientations coincide for $t\ge 2$; see Fig.~\ref{fig:threepoint-vs-rach}(a).

  \item \emph{Symmetric Rademacher prior.}
  Both $\bm u_*$ and $\bm v_*$ are Rademacher. Here the individual global signs are unidentifiable and only the \emph{relative} sign is recoverable. We adopt the convention
  \[
    (s_1,s_2)=\bigl(1,\ \mathrm{sign}(\nu_3(\{\sqrt{\lambda_r}\}))\bigr),
  \]
  where $r$ is the reference outlier index in Lemma~\ref{lem:inter_channel_sign_coupling}. Under this convention, the signed overlap tracks the realized orientation of the initializer: a positive (resp.\ negative) initial overlap remains positive (resp.\ negative), and the two trajectories are exact sign-mirrors; see Fig.~\ref{fig:threepoint-vs-rach}(b). In particular, the squared cosine similarity is invariant to the global sign.
\end{itemize}

\begin{figure}[htbp]
\centering
\vspace{0.5em}
\begin{minipage}[b]{.48\linewidth}
  \centering
  \includegraphics[width=\linewidth]{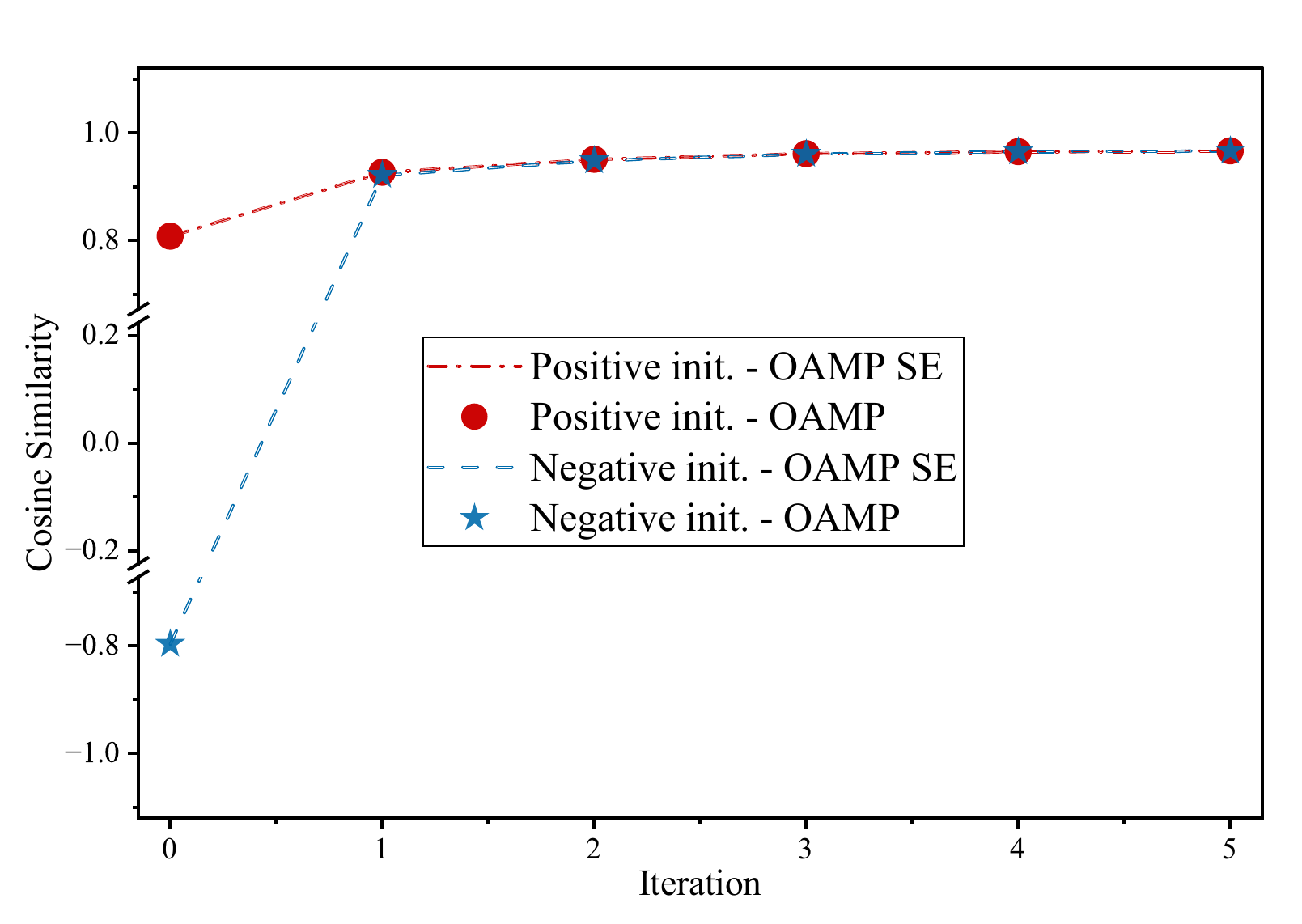}
  \centerline{(a) Three-point prior, $\bm v$-channel}
\end{minipage}
\hfill
\begin{minipage}[b]{.48\linewidth}
  \centering
  \includegraphics[width=\linewidth]{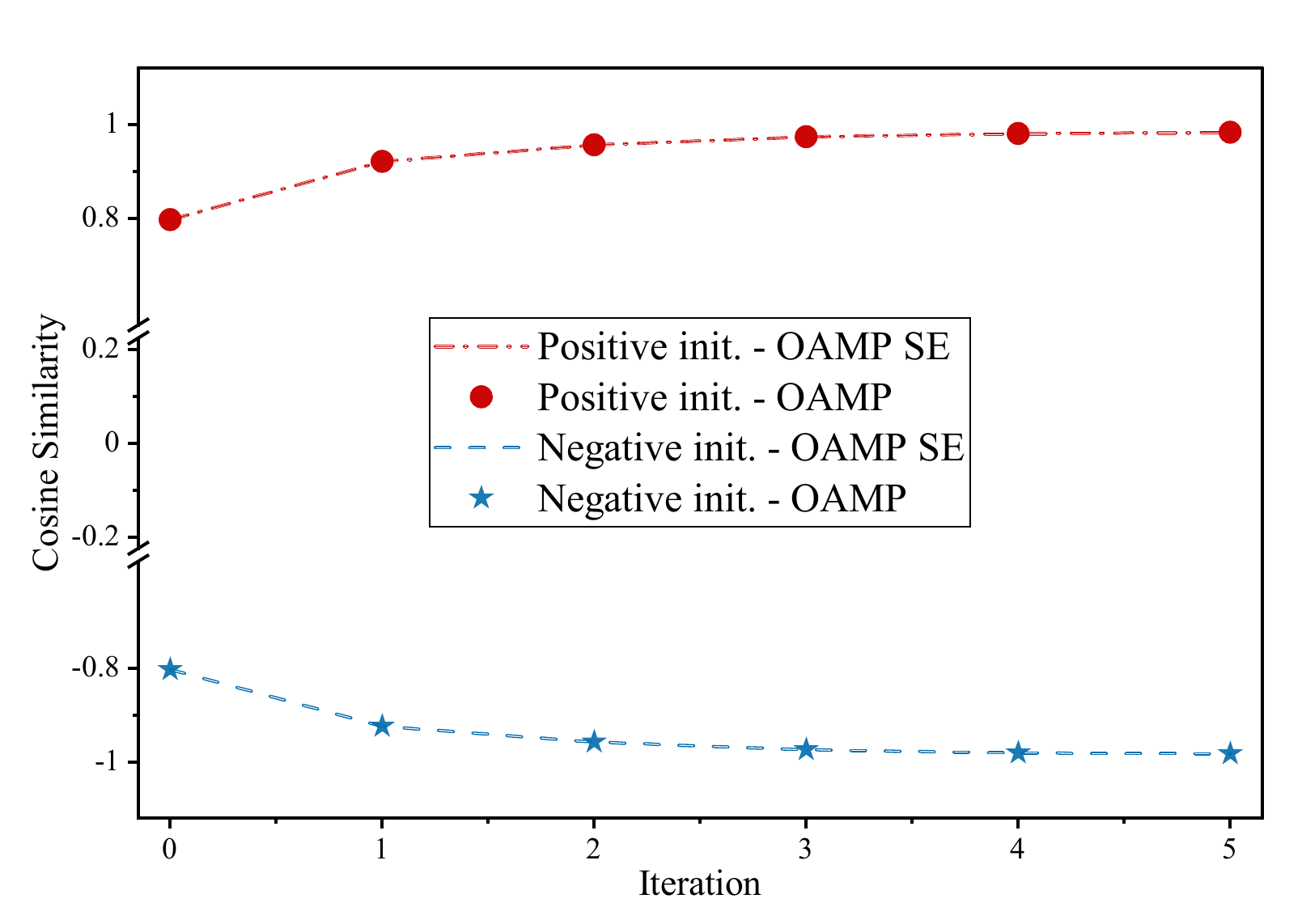}
  \centerline{(b) Rademacher prior, $\bm v$-channel}
\end{minipage}
\caption{Signed cosine similarity of spectrally-initialized OAMP under RI noise with bulk density $\mu_2(\lambda)=\frac{2}{\pi}\sqrt{(\lambda-2)(4-\lambda)}\,\mathbf{1}_{[2,4]}(\lambda)$ (cf.~Fig.~\ref{fig:combined-outliers}), with $\delta=0.7$ and $\theta=1$. Left: asymmetric three-point prior. Right: symmetric Rademacher prior. In both cases $N=20000$. Stars and circles represent the two global orientations of the spectral initializer (initial overlap positive vs.~negative).}
\label{fig:threepoint-vs-rach}
\end{figure}

% \paragraph{Non-Gaussian rotationally invariant noise: comparison to RI-AMP.}
We finally compare OAMP to the RI-AMP framework of \cite{fan2022approximate}. We report a single-iterate variant (AMP-S) and a multi-iterate variant (AMP-M). AMP-S applies the scalar MMSE denoiser to the current iterate, whereas AMP-M first forms the optimal linear combination of past signal-plus-noise observations using their covariance and then applies a single-iterate MMSE denoiser; see \cite[Remark~3.3]{fan2022approximate}. Figure~\ref{fig:beta_spectrum} shows that the spectrally-initialized optimal OAMP consistently outperforms PCA and both RI-AMP variants. 

\begin{figure}[htbp]
\begin{minipage}[b]{.48\linewidth}
  \centering
  \centerline{\includegraphics[width=\linewidth]{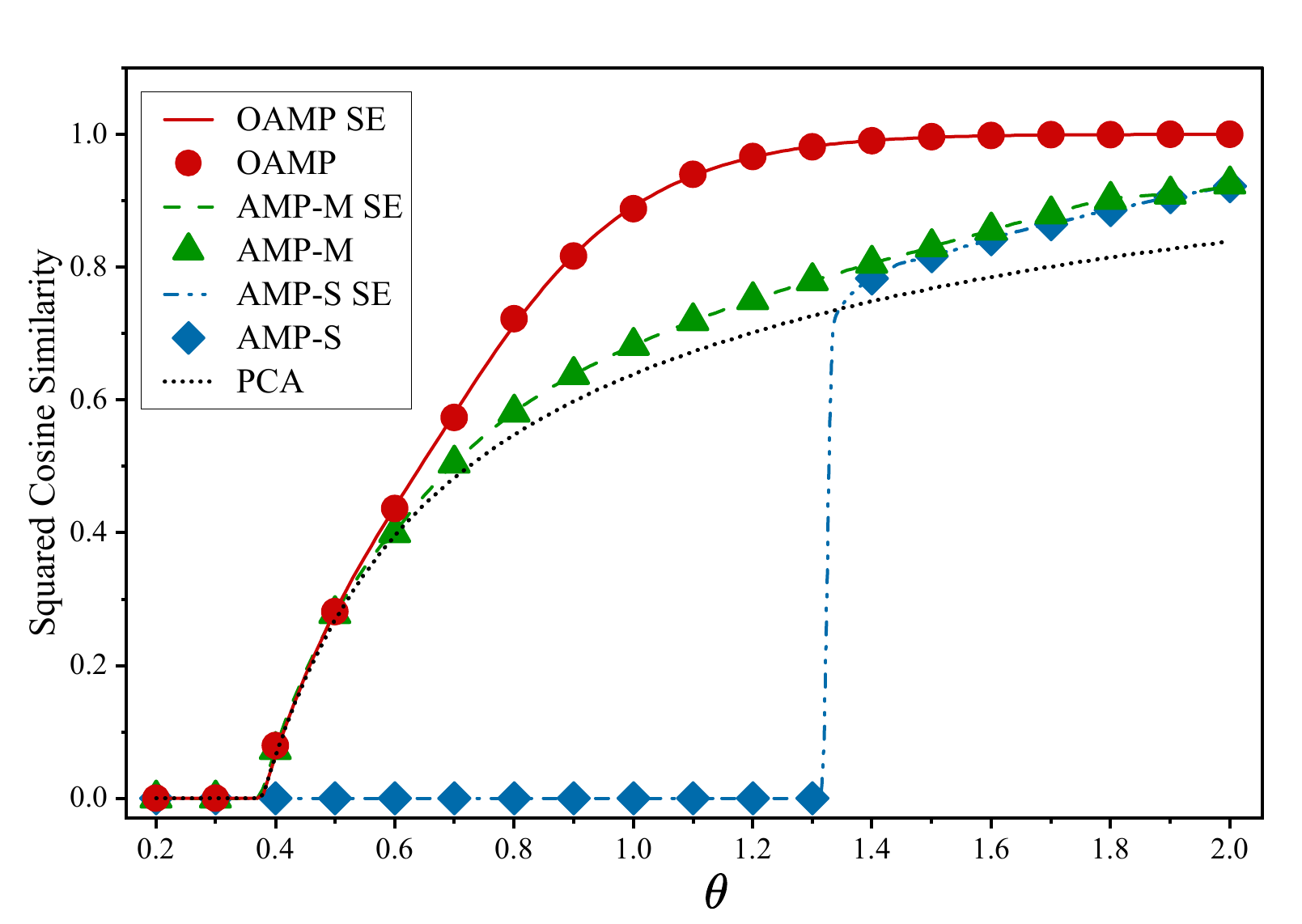}}
  \centerline{(a) $\bm u$-channel}\medskip
\end{minipage}
\hfill
\begin{minipage}[b]{0.48\linewidth}
  \centering
  \centerline{\includegraphics[width=\linewidth]{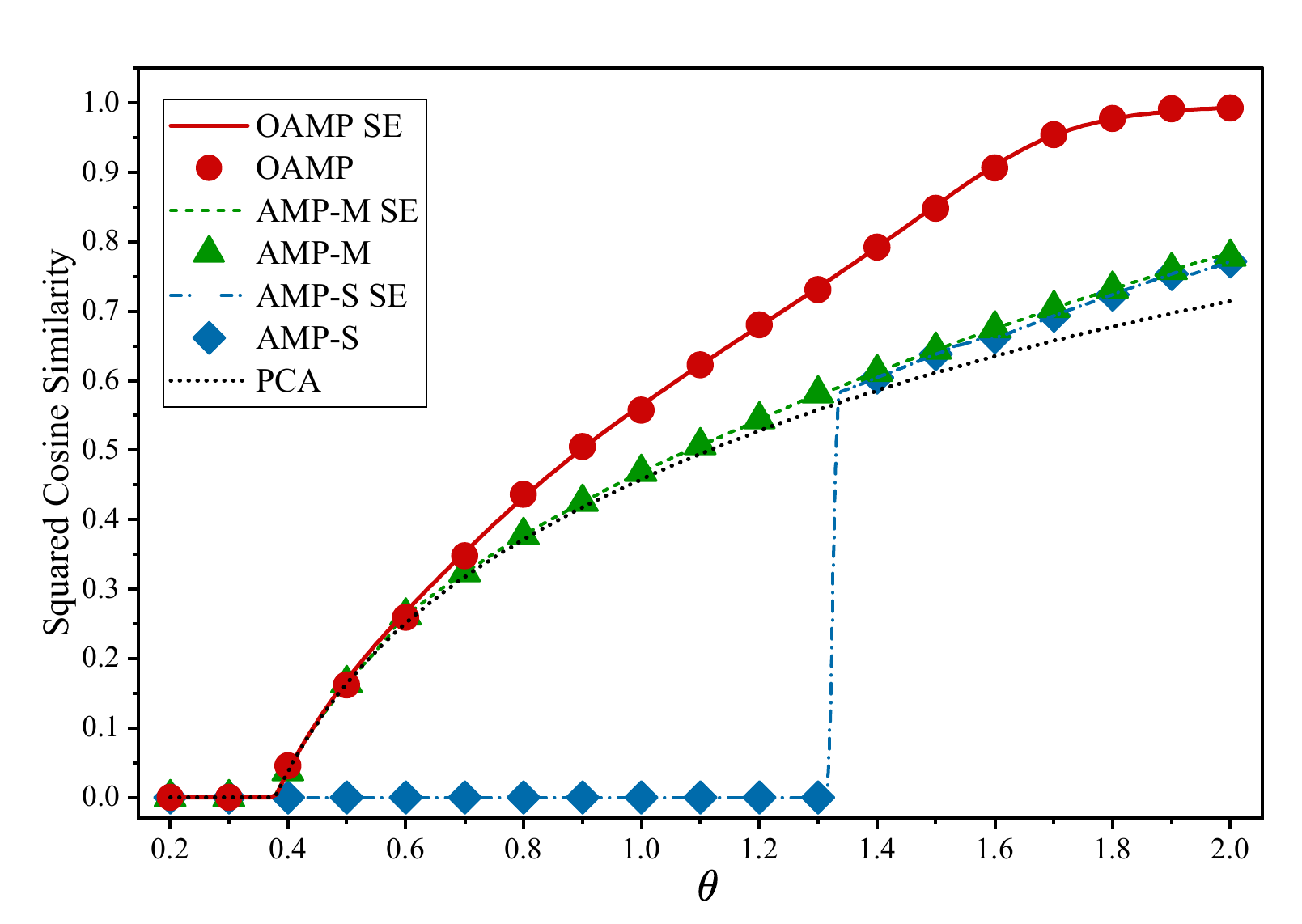}}
  \centerline{(b) $\bm v$-channel}\medskip
\end{minipage}
\caption{Non-Gaussian RI noise with bulk density $\mu_4(\lambda)=\frac{2}{\pi}\sqrt{(\lambda-1)(3-\lambda)}\,\mathbf{1}_{[1,3]}(\lambda)$ and Rademacher priors ($N=10000$, $\delta=0.5$). Markers are empirical averages over $50$ trials. When outlier eigenvectors are present, \textup{OAMP} denotes the spectrally-initialized optimal OAMP in Theorem~\ref{thm: spec-SE}, whereas \textup{AMP} uses the spectral initialization of \cite{zhong2021approximate} based only on the top eigenvector. When no outlier eigenvectors are available, all methods are initialized with cosine similarity $0.1$.}
\label{fig:beta_spectrum}
\end{figure}

\section{Conclusions and Future Work}

This paper established an optimal orthogonal approximate message passing (OAMP) framework for rectangular spiked matrix models under general rotationally invariant noise. We demonstrated that the proposed algorithm admits a rigorous state evolution characterization and incorporates an optimal spectral initialization that effectively aggregates information from multiple outlier eigenvalues. Furthermore, the algorithm achieves asymptotic performance consistent with replica-symmetric Bayes-optimal predictions \cite{Chenqun2025}, provided the model operates in a regime where there is no statistical-computational gap. These results provide a robust, computationally efficient approach to inference in high-dimensional settings where the classical i.i.d. noise assumption does not hold.

Several directions for future research emerge from this work:

\begin{itemize}
\item \textbf{Optimality:} Establishing rigorous optimality within a general class of iterative algorithms, extending recent results for symmetric models \cite{dudeja2024optimality} to the rectangular case.

\item \textbf{Finite-Rank Generalizations:} Extending the framework to finite-rank spikes to characterize the algorithmic limits of multi-signal inference.

\item \textbf{Rigorous Bayes Risk:} Providing a formal proof of the Bayes-optimal error under general rotationally invariant noise, potentially utilizing the adaptive interpolation method \cite{barbier2019adaptive} or related techniques.
\end{itemize}

\section*{Acknowledgement}
We are grateful to Rishabh Dudeja for numerous insightful comments on this work.

%% file: appendix/ap-prelims3.tex
\section{Proofs for Preliminaries and Spectral Analysis}

This section is dedicated to proving the characterizations of the signal-eigenspace spectral measures and illustrating their applications. We first proves Lemma~\ref{lem:Gamma-analytic} in \sref{App:master_eqn_analytic}, which is concerned with the analytic aspects of the master equation. We then address the properties of the measures in the signal direction in \sref{sec:spectral_measures_properties}. Finally, \sref{sec:spectral_analysis_apps} provides applications of these measures to the spectral analysis of the rectangular spiked model.

%--------------------------------------------------------
\subsection{Proof of Lemma~\ref{lem:Gamma-analytic}}\label{App:master_eqn_analytic}
%--------------------------------------------------------

\begin{proof}
We prove the four claims in order.

\paragraph{1. Analyticity and behavior at infinity.}
By Definition~\ref{Def:C_transform}, the function $\CT(z)$ is a polynomial combination of $\Smu(z)$ and $z^{-1}$.
Since $\Smu$ is holomorphic on $\C\setminus\supp(\mu)$, it follows that $\CT$ (and hence
$\Gamma(z)=1-\theta^2\CT(z)$) is holomorphic on $\C\setminus\supp(\mu)$. Moreover, as $|z|\to\infty$ we have the
standard asymptotic expansion for the Stieltjes transform of a probability measure,
\[
\Smu(z)=\frac{1}{z}+O\!\left(\frac{1}{z^2}\right),
\qquad |z|\to\infty,
\]
so $z\,\Smu(z)\to 1$ and
\[
\delta\,\Smu(z)+\frac{1-\delta}{z}=\frac{1}{z}+O\!\left(\frac{1}{z^2}\right).
\]
Therefore
\[
\CT(z)
= z\,\Smu(z)\Bigl(\delta\,\Smu(z)+\frac{1-\delta}{z}\Bigr)
= O\!\left(\frac{1}{z}\right),
\qquad |z|\to\infty,
\]
and hence $\Gamma(z)\to 1$ as $|z|\to\infty$. In particular, $\Gamma$ is not identically zero.

\paragraph{2. Isolated and simple real zeros outside $\supp(\mu)$.}
Let $\lambda_*\in\R\setminus\supp(\mu)$ satisfy $\Gamma(\lambda_*)=0$.
Since $\Gamma$ is holomorphic and not identically zero, its zeros are isolated; hence $\lambda_*$ is isolated.

To prove simplicity, recall from Lemma~\ref{lem:spectral_measures_properties} that the limiting Stieltjes transform
of $\nu_1$ is
\[
\Su(z)=\frac{\Smu(z)}{\Gamma(z)}.
\]
The function $\Su$ is the Stieltjes transform of the finite measure $\nu_1$. In particular, all poles of $\Su$ are
simple: a point mass $m\,\delta_{\lambda_*}$ contributes exactly the term $m/(z-\lambda_*)$, and higher-order poles
cannot arise from a finite measure. Since $\Smu(\lambda_*)\in\R$ is finite for $\lambda_*\notin\supp(\mu)$, the only
possible singularity of $\Su$ at $\lambda_*$ is via the denominator $\Gamma(z)$. Therefore the pole of $\Su$ at
$z=\lambda_*$ must be simple, which forces $\Gamma'(\lambda_*)\neq 0$.

\paragraph{3. Absence of real zeros in the interior of $\supp(\mu)$ and $\Gamma(0) \neq0$.}
Let $\lambda$ lie in the interior of $\supp(\mu)$. Since $\mu$ has a H\"older continuous density, the boundary value
$\Smu(\lambda-\mathrm i0^+)$ exists for Lebesgue-a.e.\ such $\lambda$ and satisfies the Sokhotski--Plemelj formula (see, e.g., \cite[Section 2.1]{pastur2011eigenvalue})
\begin{align}\label{eq:sp-lemma-ap}
\Smu(\lambda-\mathrm i0^+)=\pi\mathcal H(\lambda)+\mathrm i\pi\,\mu(\lambda).    
\end{align}
Write
\[
\Gamma(\lambda-\mathrm i0^+)=\mathcal A(\lambda)-\mathrm i\,\mathcal B(\lambda),
\]
where $\mathcal A(\lambda)$ and $\mathcal B(\lambda)$ are the real and imaginary parts of
$1-\theta^2\CT(\lambda-\mathrm i0^+)$. Substituting the boundary value of $\Smu$ into \eqref{eq:Ctrans} and
separating real and imaginary parts yields
\[
\mathcal B(\lambda)
= \pi\theta^2\mu(\lambda)\Bigl[(1-\delta)+2\delta\pi\lambda\mathcal H(\lambda)\Bigr],
\]
and $\mathcal A(\lambda)$ is the corresponding real-part expression appearing in \eqref{eq: norm of the master equation}.

Suppose $\mu(\lambda)>0$ and $\Gamma(\lambda-\mathrm i0^+)=0$. Then $\mathcal B(\lambda)=0$, which implies
\[
\pi\mathcal H(\lambda)= -\frac{1-\delta}{2\delta\lambda}.
\]
Substituting this identity into $\mathcal A(\lambda)$ gives
\[
\mathcal A(\lambda)
= 1+\theta^2\frac{(1-\delta)^2}{4\delta\lambda}+\delta\theta^2\pi^2\lambda\,\mu(\lambda)^2 \;>\;0,
\]
a contradiction. Hence $\Gamma(\lambda-\mathrm i0^+)\neq 0$ at every interior point for which $\mu(\lambda)>0$
(in particular, for Lebesgue-a.e.\ $\lambda$ in the interior of $\supp(\mu)$).  Moreover,  we have $\Gamma(0) \neq 0$ since for $z=-\mathrm{i}\varepsilon$:
\begin{align}
\lim_{\varepsilon\downarrow0}\Gamma(-\mathrm{i}\varepsilon)
&\stackrel{(a)}{=}\lim_{\varepsilon\downarrow0}\Bigl\{1-\theta^2(1-\delta)\Smu(-\mathrm{i}\varepsilon)\Bigr\}
\stackrel{(b)}{=}1-\theta^2(1-\delta)\,\pi\mathcal H(0)
\stackrel{(c)}{>}0,
\end{align}
where (a) uses the definition of $\Gamma$ in \eqref{Eqn:master0} and the non-tangential limit
$z\Smu(z)\to0$ as $z\to0$ for $\mu$ with absolutely continuous density (cf.~Fact~\ref{fact: Non-tangential limit});
(b) follows from the Sokhotski--Plemelj formula \eqref{eq:sp-lemma-ap} at $\lambda=0$; and
(c) holds since
$\pi\mathcal H(0)=\mathrm{P.V.}\!\int_{\R_+}\tfrac{1}{-t}\,\mu(t)\,\mathrm dt<0$. Consequently, any real solution of the
master equation that produces an isolated spectral component must lie in $\mathbb{R} \setminus (\operatorname{supp}(\mu) \cup \{0\})$.

\paragraph{4. Sign of the derivative at real zeros.}
Let $\lambda$ satisfy the stated conditions. Differentiating $\Gamma(\lambda)=1-\theta^2\CT(\lambda)$ yields
\begin{equation}\label{eq: C and Ctrans}
\Gamma'(\lambda)=-\theta^2\,\CT'(\lambda), \qquad
\CT'(z) = \delta \Smu^2(z) + \bigl( 2\delta z \Smu(z) + 1-\delta \bigr)\Smu'(z).
\end{equation}
If $\lambda\in\R\setminus\supp(\mu)$ satisfies $\Gamma(\lambda)=0$, then the master equation gives
\[
\CT(\lambda)=\frac{1}{\theta^2}
=\Smu(\lambda)\bigl(\delta\lambda\Smu(\lambda)+1-\delta\bigr).
\]
Since $\CT(\lambda)\neq0$, we have $\Smu(\lambda)\neq0$, and thus
\begin{equation}\label{eq: sub S}
\delta\lambda\Smu(\lambda)+1-\delta
=\frac{1}{\theta^2\Smu(\lambda)}.
\end{equation}
Substituting \eqref{eq: sub S} into \eqref{eq: C and Ctrans} yields
\begin{equation}\label{eq:final_c_prime}
\CT'(\lambda)
=\frac{1}{\Smu(\lambda)}
\left(
\delta \Smu^2(\lambda)\bigl[\Smu(\lambda)+\lambda\Smu'(\lambda)\bigr]
+\frac{\Smu'(\lambda)}{\theta^2}
\right).
\end{equation}
For any $\lambda\notin\supp(\mu)$,
\begin{align}\label{eq: S and Sprime}
\Smu'(\lambda) &= \int \frac{-1}{(\lambda - t)^2}\,d\mu(t) < 0,\\
\Smu(\lambda) + \lambda\Smu'(\lambda) &= \int \frac{-t}{(\lambda - t)^2}\,d\mu(t) < 0.
\end{align}
Hence the parenthetical term in \eqref{eq:final_c_prime} is strictly negative. Since $\Smu(\lambda)\neq0$, it follows that
\[
\operatorname{sign}\!\big(\CT'(\lambda)\big)=-\,\operatorname{sign}\!\big(\Smu(\lambda)\big),
\qquad
\operatorname{sign}\!\big(\Gamma'(\lambda)\big)
=\operatorname{sign}\!\big(-\theta^2\CT'(\lambda)\big)
=\operatorname{sign}\!\big(\Smu(\lambda)\big).
\]
In particular, $\Gamma'(\lambda)\neq0$.

This completes the proof.
\end{proof}

%--------------------------------------------------------

\subsection{Limits of Stieltjes Transforms}\label{App:limit_transform}

Before we present the proof of Lemma~\ref{lem:spectral_measures_properties} (which we collect in Section \ref{subsec:proof_lemma2}), we first derive the limiting Stieltjes transforms that will be used throughout the proof.

Define the normalized vectors
\[
\bar{\bm u}_*\bydef \frac{\bm u_*}{\sqrt M}\in\R^M,
\qquad
\bar{\bm v}_*\bydef \frac{\bm v_*}{\sqrt N}\in\R^N.
\]
By Assumption~\ref{assump:main}(b), $\|\bar{\bm u}_*\|^2\to 1$ and $\|\bar{\bm v}_*\|^2\to 1$ almost surely.
With this notation, the model reads
\[
\bm Y = \theta\,\bar{\bm u}_* \bar{\bm v}_*^\UT + \bm W.
\]
Moreover, the Stieltjes transforms of the empirical measures in
Definition~\ref{def:spectral_measures} admit the resolvent representations
\begin{align}
\mathcal S_{\nu_{M,1}}(z)
&= \frac{1}{M}\bm u_*^\UT(z\bm I_M-\bm Y\bm Y^\UT)^{-1}\bm u_*
  = \bar{\bm u}_*^\UT(z\bm I_M-\bm Y\bm Y^\UT)^{-1}\bar{\bm u}_*,
\label{eq:SM1_resolvent_norm}\\
\mathcal S_{\nu_{N,2}}(z)
&= \frac{1}{N}\bm v_*^\UT(z\bm I_N-\bm Y^\UT\bm Y)^{-1}\bm v_*
  = \bar{\bm v}_*^\UT(z\bm I_N-\bm Y^\UT\bm Y)^{-1}\bar{\bm v}_*,
\label{eq:SN2_resolvent_norm}\\
\mathcal S_{\nu_{L,3}}(z)
&= \frac{1}{L}\widehat{\bm u}_*^\UT(z\bm I_L-\widehat{\bm Y})^{-1}\widehat{\bm v}_*,
\qquad L=M+N.
\label{eq:SL3_resolvent}
\end{align}
We prove the claimed limit for $\mathcal S_{\nu_{N,2}}$; the derivation for $\mathcal S_{\nu_{M,1}}$
is completely analogous. Expanding $\bm Y^\UT\bm Y$ gives
\begin{align*}
\bm Y^\UT\bm Y
&= \bm W^\UT\bm W
 + \theta\,\bar{\bm v}_* \bar{\bm u}_*^\UT \bm W
 + \theta\,\bm W^\UT \bar{\bm u}_* \bar{\bm v}_*^\UT
 + \theta^2\,\|\bar{\bm u}_*\|^2\,\bar{\bm v}_*\bar{\bm v}_*^\UT \\
&= \bm W^\UT\bm W + \widetilde{\bm S}_e + \bm R_M,
\end{align*}
where
\[
\widetilde{\bm S}_e
\bydef \theta^2\,\bar{\bm v}_*\bar{\bm v}_*^\UT
      +\theta\,\bar{\bm v}_* \bar{\bm u}_*^\UT \bm W
      +\theta\,\bm W^\UT \bar{\bm u}_* \bar{\bm v}_*^\UT,
\qquad
\bm R_M \bydef \theta^2(\|\bar{\bm u}_*\|^2-1)\,\bar{\bm v}_*\bar{\bm v}_*^\UT.
\]
Note that $\bm R_M$ is rank one and
\[
\|\bm R_M\|_{\op}\le \theta^2\,|\|\bar{\bm u}_*\|^2-1|\,\|\bar{\bm v}_*\|^2 \xrightarrow{\text{a.s.}} 0.
\]
Let $\bm C\bydef z\bm I_N-\bm W^\UT\bm W$. Then
\[
z\bm I_N-\bm Y^\UT\bm Y = \bm C - \widetilde{\bm S}_e - \bm R_M.
\]
The resolvent identity gives
\[
(\bm C-\widetilde{\bm S}_e-\bm R_M)^{-1}
= (\bm C-\widetilde{\bm S}_e)^{-1}
 +(\bm C-\widetilde{\bm S}_e-\bm R_M)^{-1}\bm R_M(\bm C-\widetilde{\bm S}_e)^{-1}.
\]
Taking the quadratic form in direction $\bar{\bm v}_*$ and using the resolvent bounds,
\begin{align}
\Big|\mathcal S_{\nu_{N,2}}(z)
 - \bar{\bm v}_*^\UT(\bm C-\widetilde{\bm S}_e)^{-1}\bar{\bm v}_*\Big|
&\le \|\bar{\bm v}_*\|^2 \,\|(\bm C-\widetilde{\bm S}_e-\bm R_M)^{-1}\|_{\op}\,\|\bm R_M\|_{\op}\,
      \|(\bm C-\widetilde{\bm S}_e)^{-1}\|_{\op} \notag\\
&\le \frac{\|\bar{\bm v}_*\|^2}{\eta^2}\,\|\bm R_M\|_{\op}
 \xrightarrow{\text{a.s.}} 0.
\label{eq:remove_RM}
\end{align}
Hence it suffices to analyze $\bar{\bm v}_*^\UT(\bm C-\widetilde{\bm S}_e)^{-1}\bar{\bm v}_*$. Write $\widetilde{\bm S}_e=\bm A\bm B$ with $\bm A\in\R^{N\times 2}$ and $\bm B\in\R^{2\times N}$ defined by
\[
\bm A \bydef \bigl[\;\theta\,\bar{\bm v}_*,\ \bm W^\UT\bar{\bm u}_*\;\bigr],
\qquad
\bm B \bydef
\begin{bmatrix}
\theta\,\bar{\bm v}_*^\UT + \bar{\bm u}_*^\UT \bm W \\
\theta\,\bar{\bm v}_*^\UT
\end{bmatrix}.
\]
Then $\bm A\bm B=\widetilde{\bm S}_e$ is verified by direct multiplication. The Sherman--Morrison--Woodbury
formula gives
\begin{equation}
(\bm C-\bm A\bm B)^{-1}
= \bm C^{-1} + \bm C^{-1}\bm A\bigl(\bm I_2-\bm B\bm C^{-1}\bm A\bigr)^{-1}\bm B\bm C^{-1}.
\label{eq:SMW_parallel}
\end{equation}
Consequently,
\begin{equation}
\bar{\bm v}_*^\UT(\bm C-\widetilde{\bm S}_e)^{-1}\bar{\bm v}_*
= \bar{\bm v}_*^\UT\bm C^{-1}\bar{\bm v}_*
 + \bar{\bm v}_*^\UT\bm C^{-1}\bm A\bigl(\bm I_2-\bm B\bm C^{-1}\bm A\bigr)^{-1}\bm B\bm C^{-1}\bar{\bm v}_*.
\label{eq:Sv_decomp}
\end{equation}
By rotational invariance of $\bm W$ and $(\bar{\bm u}_*,\bar{\bm v}_*)\indep \bm W$, the limits follow by the same Haar-rotation/continuous-mapping argument as in \eqref{Eqn:app_SE_base_1}--\eqref{Eqn:app_SE_base_3} (cf.~\cite[Proposition E.2]{fan2022approximate}); in particular,\ for each fixed $z\in\C\setminus\R_+$,
\begin{align}
\bar{\bm v}_*^\UT(z\bm I_N-\bm W^\UT\bm W)^{-1}\bar{\bm v}_*
&\ac \delta\,\Smu(z)+\frac{1-\delta}{z},
\label{eq:qform_v_limit}\\
\bar{\bm u}_*^\UT\bm W(z\bm I_N-\bm W^\UT\bm W)^{-1}\bm W^\UT\bar{\bm u}_*
&\ac \int_{\R_+}\frac{\lambda}{z-\lambda}\,d\mu(\lambda)
  = z\,\Smu(z)-1,
\label{eq:qform_uWu_limit}\\
\bar{\bm u}_*^\UT\bm W(z\bm I_N-\bm W^\UT\bm W)^{-1}\bar{\bm v}_*
&\ac 0,
\qquad
\bar{\bm v}_*^\UT(z\bm I_N-\bm W^\UT\bm W)^{-1}\bm W^\UT\bar{\bm u}_*\ac 0.
\label{eq:cross_qform_0}
\end{align}
Define the deterministic quantities
\[
d_1(z)\bydef \delta\,\Smu(z)+\frac{1-\delta}{z},
\qquad
d_2(z)\bydef z\,\Smu(z)-1.
\]
Using \eqref{eq:qform_v_limit}--\eqref{eq:cross_qform_0}, we obtain
\begin{align*}
\bm B\bm C^{-1}\bm A
&=
\begin{bmatrix}
(\theta\bar{\bm v}_*^\UT+\bar{\bm u}_*^\UT\bm W)\bm C^{-1}(\theta\bar{\bm v}_*) &
(\theta\bar{\bm v}_*^\UT+\bar{\bm u}_*^\UT\bm W)\bm C^{-1}(\bm W^\UT\bar{\bm u}_*)\\[0.3em]
\theta\bar{\bm v}_*^\UT\bm C^{-1}(\theta\bar{\bm v}_*) &
\theta\bar{\bm v}_*^\UT\bm C^{-1}(\bm W^\UT\bar{\bm u}_*)
\end{bmatrix} \\
&\xrightarrow{\text{a.s.}}
\begin{bmatrix}
\theta^2 d_1(z) & d_2(z)\\
\theta^2 d_1(z) & 0
\end{bmatrix}.
\end{align*}
Likewise,
\[
\bar{\bm v}_*^\UT\bm C^{-1}\bm A
\xrightarrow{\text{a.s.}}
\bigl[\;\theta\,d_1(z),\ 0\;\bigr],
\qquad
\bm B\bm C^{-1}\bar{\bm v}_*
\xrightarrow{\text{a.s.}}
\begin{bmatrix}\theta\,d_1(z)\\ \theta\,d_1(z)\end{bmatrix}.
\]
A direct $2\times 2$ calculation then yields
\begin{equation}
\bar{\bm v}_*^\UT(\bm C-\widetilde{\bm S}_e)^{-1}\bar{\bm v}_*
\xrightarrow{\text{a.s.}}
\frac{d_1(z)}{1-\theta^2 d_1(z)\bigl(1+d_2(z)\bigr)}.
\label{eq:Sv_limit_intermediate}
\end{equation}
Since $1+d_2(z)=z\,\Smu(z)$, we recognize
\[
d_1(z)\bigl(1+d_2(z)\bigr)
= \Bigl(\delta\,\Smu(z)+\frac{1-\delta}{z}\Bigr)\cdot z\,\Smu(z)
= z\,\Smu(z)\Bigl(\delta\,\Smu(z)+\frac{1-\delta}{z}\Bigr)
= \CT(z).
\]
Therefore, combining \eqref{eq:remove_RM} and \eqref{eq:Sv_limit_intermediate},
\BE\label{Eqn:app_v2}
\mathcal S_{\nu_{N,2}}(z)
\xrightarrow{\text{a.s.}}
\Sv(z)\bydef \frac{\delta\,\Smu(z)+(1-\delta)/z}{1-\theta^2\CT(z)},
\qquad z\in\C\setminus\R.
\EE
The proof for $\mathcal S_{\nu_{M,1}}(z)$ is analogous and gives
\BE\label{Eqn:app_v1}
\mathcal S_{\nu_{M,1}}(z)
\xrightarrow{\text{a.s.}}
\Su(z)\bydef \frac{\Smu(z)}{1-\theta^2\CT(z)}.
\EE
This establishes the first two displays in \eqref{eq:nu-Stieltjes}.

Recall the symmetric dilation
\[
\widehat{\bm Y}\bydef
\begin{bmatrix}
\bm 0 & \bm Y\\
\bm Y^\UT & \bm 0
\end{bmatrix}\in\R^{L\times L},
\qquad L=M+N,
\]
and define $\widehat{\bm W}$ analogously. Define the {embedded normalized} signal directions
\[
\bar{\bm u}_*^{(L)} \bydef \begin{bmatrix}\bar{\bm u}_*\\ \bm 0\end{bmatrix}\in\R^L,
\qquad
\bar{\bm v}_*^{(L)} \bydef \begin{bmatrix}\bm 0\\ \bar{\bm v}_*\end{bmatrix}\in\R^L.
\]
Then $\widehat{\bm u}_*=\sqrt{M}\,\bar{\bm u}_*^{(L)}$ and $\widehat{\bm v}_*=\sqrt{N}\,\bar{\bm v}_*^{(L)}$, and
\[
\widehat{\bm Y}
= \widehat{\bm W} + \theta\bigl(\bar{\bm u}_*^{(L)}(\bar{\bm v}_*^{(L)})^\UT
                              +\bar{\bm v}_*^{(L)}(\bar{\bm u}_*^{(L)})^\UT\bigr).
\]
Let $\bm C_b\bydef z\bm I_L-\widehat{\bm W}$ and define
\[
\bm X\bydef \bigl[\bar{\bm u}_*^{(L)},\,\bar{\bm v}_*^{(L)}\bigr],
\qquad
\bm J\bydef
\begin{bmatrix}
0 & 1\\
1 & 0
\end{bmatrix}.
\]
Applying Sherman--Morrison--Woodbury to this rank-two perturbation yields
\begin{equation}
(z\bm I_L-\widehat{\bm Y})^{-1}
= \bm C_b^{-1} + \bm C_b^{-1}\bm X
\bigl(\bm I_2-\theta\,\bm J\,\bm X^\UT \bm C_b^{-1}\bm X\bigr)^{-1}
\theta\,\bm J\,\bm X^\UT\bm C_b^{-1}.
\label{eq:SMW_cross}
\end{equation}
By \eqref{eq:SL3_resolvent} and the relations above,
\begin{equation}
\mathcal S_{\nu_{L,3}}(z)
= \frac{1}{L}\widehat{\bm u}_*^\UT(z\bm I_L-\widehat{\bm Y})^{-1}\widehat{\bm v}_*
= \frac{\sqrt{MN}}{L}\,(\bar{\bm u}_*^{(L)})^\UT(z\bm I_L-\widehat{\bm Y})^{-1}\bar{\bm v}_*^{(L)}.
\label{eq:SnuL3_scaled}
\end{equation}
A Schur complement calculation gives the standard identity
\[
\bm C_b^{-1}=
\begin{bmatrix}
z^{-1}\bigl(\bm I_M+\bm W(z^2\bm I_N-\bm W^\UT\bm W)^{-1}\bm W^\UT\bigr) &
\bm W(z^2\bm I_N-\bm W^\UT\bm W)^{-1}\\[0.3em]
(z^2\bm I_N-\bm W^\UT\bm W)^{-1}\bm W^\UT &
z(z^2\bm I_N-\bm W^\UT\bm W)^{-1}
\end{bmatrix}.
\]
Hence
\[
(\bar{\bm u}_*^{(L)})^\UT \bm C_b^{-1}\bar{\bm v}_*^{(L)}
= \bar{\bm u}_*^\UT \bm W(z^2\bm I_N-\bm W^\UT\bm W)^{-1}\bar{\bm v}_*
\xrightarrow{\text{a.s.}} 0,
\]
by the same cross-term argument as \eqref{eq:cross_qform_0}. Define
\[
e_1(z)\bydef z\Bigl(\delta\,\Smu(z^2)+\frac{1-\delta}{z^2}\Bigr),
\qquad
e_2(z)\bydef z\,\Smu(z^2).
\]
The same quadratic-form concentration used above yields almost surely
\[
(\bar{\bm v}_*^{(L)})^\UT \bm C_b^{-1}\bar{\bm v}_*^{(L)} \to e_1(z),
\qquad
(\bar{\bm u}_*^{(L)})^\UT \bm C_b^{-1}\bar{\bm u}_*^{(L)} \to e_2(z),
\qquad
(\bar{\bm u}_*^{(L)})^\UT \bm C_b^{-1}\bar{\bm v}_*^{(L)} \to 0.
\]
A direct $2\times 2$ computation in \eqref{eq:SMW_cross} then gives
\[
(\bar{\bm u}_*^{(L)})^\UT(z\bm I_L-\widehat{\bm Y})^{-1}\bar{\bm v}_*^{(L)}
\xrightarrow{\text{a.s.}}
\frac{\theta\,e_1(z)e_2(z)}{1-\theta^2 e_1(z)e_2(z)}.
\]
Finally,
\[
e_1(z)e_2(z)
= z^2\,\Smu(z^2)\Bigl(\delta\,\Smu(z^2)+\frac{1-\delta}{z^2}\Bigr)
= \CT(z^2).
\]
Since $\sqrt{MN}/L\to \sqrt{\delta}/(1+\delta)$, \eqref{eq:SnuL3_scaled} yields
\begin{equation}
\label{Eqn:app_v3}
\mathcal S_{\nu_{L,3}}(z)
\xrightarrow{\text{a.s.}}
\Suv(z)\bydef \frac{\sqrt{\delta}}{1+\delta}\cdot\frac{\theta\,\CT(z^2)}{1-\theta^2\CT(z^2)},
\qquad z\in\C\setminus\R,
\end{equation}
which is the last identity in \eqref{eq:nu-Stieltjes}.

\subsection{Proof of Lemma~\ref{lem:spectral_measures_properties}}
\label{subsec:proof_lemma2}
We prove each claim separately.
%==========================
\paragraph{1. Weak convergence.}
%==========================
We have pointwise a.s.\ convergence of Stieltjes transforms on $\C\setminus\R$.
Moreover, by Assumption~\ref{assump:main}(c), $\|\bm W\|_{\op}$ is uniformly bounded. As a consequence, $\{\nu_{M,1}\}_M$, $\{\nu_{N,2}\}_N$, and $\{\nu_{L,3}\}_L$
are uniformly compactly supported and (in particular, tight). It follows from the Stieltjes continuity theorem
(and its signed-measure analogue) that $\nu_{M,1},\nu_{N,2},\nu_{L,3}$ converges weakly (almost surely) to $\nu_1,\nu_2,\nu_3$, respectively. Finally, $\nu_1$ and $\nu_2$ are probability measures since
$\lim_{y\to\infty}(-iy)\Su(iy)=1$ and $\lim_{y\to\infty}(-iy)\Sv(iy)=1$
(using $y\Smu(iy)\to -i$ and $\CT(iy)=O(1/y)$). Likewise, $\nu_3$ is finite since
$\limsup_{y\to\infty} y|\Suv(iy)|<\infty$.

%==========================
\paragraph{2. Limiting Stieltjes transform.} They have been established in Section \ref{App:limit_transform}, specifically, \eqref{Eqn:app_v1}, \eqref{Eqn:app_v2}, and \eqref{Eqn:app_v3}.

%==========================
\paragraph{3. Absolutely continuous parts.}
%==========================
Let $\nu_i=\nu_i^{\parallel}+\nu_i^{\perp}$ be the Lebesgue decomposition for $i\in\{1,2,3\}$.
We use the Stieltjes inversion principle for finite signed measures (cf.\ Fact~\ref{fact:boundary-signed}):
if $\chi=\chi^\parallel+\chi^\perp$, then
\begin{equation}
\label{eq:stieltjes_inversion_finite_vs_infinite}
\begin{aligned}
\lim_{\epsilon\downarrow 0}\ \Im\bigl\{\mathcal S_\chi(\lambda-\mathrm i\epsilon)\bigr\}
&=\pi\,\frac{d\chi^\parallel}{d\lambda}(\lambda),
&& \text{for Lebesgue-a.e.\ }\lambda\in\R,\\
\lim_{\epsilon\downarrow 0}\ \Bigl|\Im\bigl\{\mathcal S_\chi(\lambda-\mathrm i\epsilon)\bigr\}\Bigr|
&=+\infty,
&& \text{for }|\chi^\perp|\text{-a.e.\ }\lambda\in\R.
\end{aligned}
\end{equation}
Moreover, since $\mu$ has a H\"older continuous density (Assumption~\ref{assump:main}(c)),
its Stieltjes transform admits boundary values given by the Sokhotski--Plemelj formula:
for Lebesgue-a.e.\ $\lambda\in\R$,
\begin{equation}
\label{eq:SP_mu}
\lim_{\epsilon\downarrow 0}\Smu(\lambda-\mathrm i\epsilon)
= \pi\,\mathcal H(\lambda)+\mathrm i\pi\,\mu(\lambda),
\end{equation}
where $\mathcal H$ is the Hilbert transform of $\mu$. Applying \eqref{eq:stieltjes_inversion_finite_vs_infinite} to $\chi=\nu_1$ and $\chi=\nu_2$, and using
\eqref{eq:nu-Stieltjes}, we obtain for Lebesgue-a.e.\ $\lambda\in\R$,
\[
\frac{d\nu_1^\parallel}{d\lambda}(\lambda)
= \frac{1}{\pi}\lim_{\epsilon\downarrow 0}\Im\{\Su(\lambda-\mathrm i\epsilon)\},
\qquad
\frac{d\nu_2^\parallel}{d\lambda}(\lambda)
= \frac{1}{\pi}\lim_{\epsilon\downarrow 0}\Im\{\Sv(\lambda-\mathrm i\epsilon)\}.
\]
Substituting the boundary value \eqref{eq:SP_mu} into \eqref{eq:nu-Stieltjes} and separating real and
imaginary parts yields the expressions in \eqref{eq:nu-ac} with the shrinkage functions
$\varphi_1$ and $\varphi_2$ defined in \eqref{eq:shrikage-function}. The required algebra is exactly the
one encoded in \eqref{eq: norm of the master equation} and \eqref{eq:shrikage-function}.

We apply \eqref{eq:stieltjes_inversion_finite_vs_infinite} to the signed measure $\chi=\nu_3$ and use the
closed-form transform $\Suv$ in \eqref{eq:nu-Stieltjes}. For $\sigma\neq 0$, the boundary value of
$\Smu((\sigma-\mathrm i0^+)^2)$ is given by Fact~\ref{fact:special_SP_hilbert_form}, which yields
\[
\frac{d\nu_3^\parallel}{d\sigma}(\sigma)
= \frac{1}{\pi}\lim_{\epsilon\downarrow 0}\Im\{\Suv(\sigma-\mathrm i\epsilon)\}
= \frac{\sqrt{\delta}}{1+\delta}\,\mathrm{sign}(\sigma)\,\mu(\sigma^2)\,\varphi_3(\sigma^2),
\]
for Lebesgue-a.e.\ $\sigma\in\R$, and this is precisely the third identity in \eqref{eq:nu-ac}.

%==========================
\paragraph{4. Singular parts and atomic masses.}
%==========================
Let
\[
\mathcal K_* \bydef \{\lambda\in\R\setminus\supp(\mu):\Gamma(\lambda)=0\},
\qquad
\Gamma(z)\bydef 1-\theta^2\CT(z).
\]
Notice that $\mathcal K_*$ is assumed a \textit{finite} set. By Fact~\ref{fact:boundary-signed}, the singular component $\chi^\perp$ of a finite (signed) measure
$\chi=\chi^\parallel+\chi^\perp$ can only assign mass to points where the boundary imaginary part of
$\mathcal S_\chi(\lambda-\mathrm i\epsilon)$ diverges as $\epsilon\downarrow 0$.
We apply this criterion to $\nu_1,\nu_2,\nu_3$ using the explicit Stieltjes transforms
\eqref{eq:nu-Stieltjes}.

Fix $\lambda_0\in\R\setminus\supp(\mu)$. Since $\Smu$ is analytic on $\C\setminus\supp(\mu)$, the
functions $\CT$ and $\Gamma$ are analytic at $\lambda_0$. If $\Gamma(\lambda_0)\neq 0$, then by continuity
there exists $\epsilon_0>0$ such that $\inf_{0<\epsilon<\epsilon_0}|\Gamma(\lambda_0-\mathrm i\epsilon)|>0$.
Moreover, $\Smu(\lambda_0-\mathrm i\epsilon)$ remains bounded for $0<\epsilon<\epsilon_0$.
Hence,
\[
\Su(\lambda_0-\mathrm i\epsilon)
=\frac{\Smu(\lambda_0-\mathrm i\epsilon)}{\Gamma(\lambda_0-\mathrm i\epsilon)}
\quad\text{remains bounded as }\epsilon\downarrow 0,
\]
so $\Im\Su(\lambda_0-\mathrm i\epsilon)$ cannot diverge at such a point. Therefore, outside $\supp(\mu)$,
any singular mass of $\nu_1$ can only occur at real zeros of $\Gamma$, i.e., at points in $\mathcal K_*$. The same conclusion holds for $\nu_2$ at any $\lambda_0\neq 0$, since
\[
\Sv(z)=\frac{\delta\,\Smu(z)+(1-\delta)/z}{\Gamma(z)}
\]
is analytic at $\lambda_0$ whenever $\lambda_0\in\R\setminus\supp(\mu)$ and $\Gamma(\lambda_0)\neq 0$.
When $\delta<1$, the term $(1-\delta)/z$ may create an additional pole at $z=0$, yielding a possible
atom at the origin even if $0\notin\mathcal K_*$. For $\nu_3$, we use
\[
\Suv(z)=\frac{\sqrt{\delta}}{1+\delta}\cdot\frac{\theta\,\CT(z^2)}{1-\theta^2\CT(z^2)}
       =\frac{\sqrt{\delta}}{1+\delta}\cdot\frac{\theta\,\CT(z^2)}{\Gamma(z^2)}.
\]
Thus any singularity of $\Suv$ away from the bulk can only occur when $\Gamma(z^2)=0$, i.e., at
$z=\pm\sqrt{\lambda_*}$ with $\lambda_*\in\mathcal K_*$.

Since $\mathcal K_*$ is finite, the preceding localization implies that the singular components are purely
atomic and admit the representations
\begin{align*}
\nu_1^\perp
&=\sum_{\lambda_*\in\mathcal K_*}\nu_1(\{\lambda_*\})\,\delta_{\lambda_*},\\
\nu_2^\perp
&=\sum_{\lambda_*\in\mathcal K_*}\nu_2(\{\lambda_*\})\,\delta_{\lambda_*}
  \;+\;\mathbf 1_{\{\delta<1\}}\,\nu_2(\{0\})\,\delta_0,\\
\nu_3^\perp
&=\sum_{\lambda_*\in\mathcal K_*}
  \Bigl(\nu_3(\{\sigma_*\})\,\delta_{\sigma_*}+\nu_3(\{-\sigma_*\})\,\delta_{-\sigma_*}\Bigr),
\qquad \sigma_*\bydef\sqrt{\lambda_*}.
\end{align*}
Fix $\lambda_*\in\mathcal K_*$. Then $\lambda_*\notin\supp(\mu)$, so $\Smu$ and $\CT$ are analytic at
$\lambda_*$. By Lemma~\ref{lem:Gamma-analytic}, $\lambda_*$ is a simple zero of $\Gamma$, hence
$\Gamma'(\lambda_*)\neq 0$. Therefore $\Su$ and $\Sv$ have simple poles at $z=\lambda_*$, and since a Dirac
mass $m\,\delta_{\lambda_*}$ contributes $m/(z-\lambda_*)$ to the Stieltjes transform, the corresponding atom
masses are precisely the residues:
\[
\nu_1(\{\lambda_*\})
=\operatorname*{Res}_{z=\lambda_*}\Su(z)
=\frac{\Smu(\lambda_*)}{\Gamma'(\lambda_*)}
=-\frac{\Smu(\lambda_*)}{\theta^2 C'(\lambda_*)},
\]
and
\[
\nu_2(\{\lambda_*\})
=\operatorname*{Res}_{z=\lambda_*}\Sv(z)
=\frac{\delta\,\Smu(\lambda_*)+(1-\delta)/\lambda_*}{\Gamma'(\lambda_*)}
=-\frac{\delta\,\Smu(\lambda_*)+(1-\delta)/\lambda_*}{\theta^2 C'(\lambda_*)}.
\]
These residues are nonzero: otherwise $\CT(\lambda_*)=0$, contradicting $\Gamma(\lambda_*)=0$. For $\nu_3$, the poles occur at $z=\pm\sigma_*$ with $\sigma_*=\sqrt{\lambda_*}>0$. A direct residue
calculation from $\Suv(z)=\frac{\sqrt{\delta}}{1+\delta}\frac{\theta\,\CT(z^2)}{\Gamma(z^2)}$ yields
\[
\nu_3(\{\pm\sigma_*\})
=\operatorname*{Res}_{z=\pm\sigma_*}\Suv(z)
=\mp\,\frac{\sqrt{\delta}}{1+\delta}\,\frac{1}{2\theta^3\sigma_*\,C'(\sigma_*^2)}.
\]
Finally, when $\delta<1$, $\Sv$ has a simple pole at $z=0$ due to the term $(1-\delta)/z$, and its residue is
\[
\nu_2(\{0\})
=\operatorname*{Res}_{z=0}\Sv(z)
=\lim_{z\to 0} z\,\Sv(z)
=\frac{1-\delta}{1-\theta^2(1-\delta)\pi\mathcal H(0)}.
\]
This completes the proof.

\subsection{Proof of Proposition~\ref{prop:outlier_characterization}}\label{sec:spectral_measures_properties}

\begin{proof}
We prove Claims~1–4 in order.

\medskip
\paragraph{Proof of Claim~1 (Isolation of population outliers).}
Since $\mathcal K_*\subset\R\setminus\supp(\mu)$ is finite, define
\[
d_\mu \bydef \min_{\lambda\in\mathcal K_*} \mathrm{d}(\lambda,\supp(\mu))>0,
\qquad
d_K \bydef \min_{\substack{\lambda,\lambda'\in\mathcal K_*\\ \lambda\neq\lambda'}}|\lambda-\lambda'|\in(0,\infty],
\]
with $d_K=+\infty$ if $|\mathcal K_*|=1$. Set
\begin{equation}\label{eq:eps_choice_common}
\varepsilon \bydef \min\Big\{\frac{d_\mu}{4},\frac{d_K}{4}\Big\}>0.
\end{equation}
Then for each $\lambda_k\in\mathcal K_*$,
\[
(\lambda_k-\varepsilon,\lambda_k+\varepsilon)\cap\supp(\mu)=\varnothing,
\qquad
(\lambda_k-\varepsilon,\lambda_k+\varepsilon)\cap\mathcal K_*=\{\lambda_k\},
\]
and the intervals are pairwise disjoint. Define the outlier neighborhood
\[
\mathcal{K}^*_{\epsilon} \bydef \bigcup_{\lambda_k\in\mathcal K_*}(\lambda_k-\varepsilon,\lambda_k+\varepsilon).
\]
By construction, $\mathrm{d}(\mathcal{K}^*_{\epsilon},\supp(\mu))\ge 3\varepsilon$, hence
$\mathcal{K}^*_{\epsilon}\cap\supp_{\varepsilon/2}(\mu)=\varnothing$.
This proves Claim~1.

\medskip
\paragraph{General setup: confinement of roots and uniform convergence.}
Let
\[
\bm S_1(z)\bydef(z\bm I_M-\bm W\bm W^\UT)^{-1},
\qquad
\bm S_2(z)\bydef(z\bm I_N-\bm W^\UT\bm W)^{-1},
\]
and recall the empirical master function $\Gamma_M(z)$.

Fix $\eta>0$ and a compact set $\mathcal E\subset\C\setminus\supp_\eta(\mu)$.
By Assumption~\ref{assump:main}(d), almost surely for all $M$ large enough,
$\spp(\bm W\bm W^\UT)\subset\supp_{\eta/2}(\mu)$, hence
\[
\|\bm S_1(z)\|_{\op},\ \|\bm S_2(z)\|_{\op}\le 2/\eta
\qquad\forall z\in\mathcal E.
\]
Therefore the quadratic forms appearing in $\Gamma_M$ define uniformly bounded
holomorphic families on $\mathcal E$.

For each fixed $z\in\mathcal E$, standard quadratic-form limits for Haar singular
vectors (e.g., \cite[Proposition~8.2]{Benaych-Georges2012singular}) yield almost surely
\begin{subequations}\label{eq:quadform_pointwise}
\begin{align}
\frac{\theta}{M}\bm u_*^\UT\bm S_1(z)\bm u_* &\to \theta\,\Smu(z),\\
\frac{\theta}{N}\bm v_*^\UT\bm S_2(z)\bm v_* &\to \theta\Bigl(\delta\,\Smu(z)+\frac{1-\delta}{z}\Bigr),\\
\frac{\theta}{\sqrt{MN}}\bm v_*^\UT\bm W^\UT\bm S_1(z)\bm u_* &\to 0,\\
\frac{\theta}{\sqrt{MN}}\bm u_*^\UT\bm W\bm S_2(z)\bm v_* &\to 0.
\end{align}
\end{subequations}
By Montel’s theorem and uniqueness of the holomorphic limit, the convergence is
uniform on $\mathcal E$. Substituting into $\Gamma_M$ yields
\[
\Gamma_M(z)\to \Gamma(z)\bydef 1-\theta^2\CT(z)
\qquad\text{uniformly on }\mathcal E.
\]
Since $\eta>0$ and $\mathcal E$ are arbitrary, $\Gamma_M\to\Gamma$ almost surely
locally uniformly on $\C\setminus\supp(\mu)$.

Moreover, since $\Smu(z)\sim 1/z$ as $|z|\to\infty$, we have $\Gamma(z)\to 1$.
Thus there exists $R>0$ such that $|\Gamma(z)|>0$ for $|z|\ge R$.
By local uniform convergence, the same holds for $\Gamma_M$ for all large $M$.
Hence all real zeros of $\Gamma$ and $\Gamma_M$ lie in $(-R,R)$ eventually.

\medskip
\paragraph{Proof of Claim~3 (Existence and convergence of empirical outliers).}
Fix $\lambda_k\in\mathcal K_*$. Since $\lambda_k$ is an isolated and simple real zero of $\Gamma$,
there exists $\rho\in(0,\varepsilon)$ such that the closed disk
\[
D_k \bydef \{z\in\C:\ |z-\lambda_k|\le \rho\}
\]
satisfies
\[
D_k\cap\supp(\mu)=\varnothing,
\qquad
D_k\cap\mathcal K_*=\{\lambda_k\},
\qquad
\inf_{z\in\partial D_k}|\Gamma(z)|>0.
\]
Define the contour
\[
\gamma_k \bydef \partial D_k
= \{z\in\C:\ |z-\lambda_k|=\rho\},
\]
which is the positively oriented circle of radius $\rho$ centered at $\lambda_k$.

Since $\gamma_k\subset\C\setminus\supp(\mu)$, the function $\Gamma$ is holomorphic in a neighborhood of
$\gamma_k$. By the argument principle,
\begin{equation}\label{eq:winding-Gamma}
\frac{1}{2\pi i}\oint_{\gamma_k}\frac{\Gamma'(z)}{\Gamma(z)}\,dz
=
\mathrm{Wind}\bigl(\Gamma(\gamma_k),0\bigr)
=1,
\end{equation}
where the equality follows from the simplicity of the zero at $\lambda_k$.

By the local uniform convergence $\Gamma_M\to\Gamma$ on $\C\setminus\supp(\mu)$,
there exists $M_k<\infty$ such that for all $M\ge M_k$,
\[
\inf_{z\in\gamma_k}|\Gamma_M(z)| \ge \frac12 \inf_{z\in\gamma_k}|\Gamma(z)| > 0.
\]
Moreover, by Assumption~\ref{assump:main}(d),
$\spp(\bm W\bm W^\UT)\subset\supp_{\varepsilon/2}(\mu)$ for all large $M$, hence
$\gamma_k\cap\spp(\bm W\bm W^\UT)=\varnothing$ and $\Gamma_M$ is holomorphic in a neighborhood of $\gamma_k$.
Uniform convergence of $\Gamma_M$ and $\Gamma_M'$ on $\gamma_k$ then implies
\[
\frac{\Gamma_M'(z)}{\Gamma_M(z)} \;\to\;
\frac{\Gamma'(z)}{\Gamma(z)}, \qquad \text{uniformly in }z\in\gamma_k.
\]
Therefore,
\begin{equation}\label{eq:winding-GammaM}
\lim_{M\to\infty}
\frac{1}{2\pi \mathrm{i}}\oint_{\gamma_k}\frac{\Gamma_M'(z)}{\Gamma_M(z)}\,dz
=
\frac{1}{2\pi \mathrm{i}}\oint_{\gamma_k}\frac{\Gamma'(z)}{\Gamma(z)}\,dz
=1.
\end{equation}
Since the left-hand side is integer-valued for each $M$, it follows that for all sufficiently large $M$,
\[
\frac{1}{2\pi \mathrm{i}}\oint_{\gamma_k}\frac{\Gamma_M'(z)}{\Gamma_M(z)}\,dz = 1.
\]
Hence $\Gamma_M$ has exactly one zero (counted with multiplicity) inside $D_k$.

\medskip
\paragraph{Proof of Claim~2 (Exact spectral separation).}
Let $\epsilon>0$ be fixed as in Claim~1, and recall the definitions of the outlier neighborhood $\mathcal K_\epsilon^*$ and the bulk neighborhood $\supp_\epsilon(\mu)$. By Assumption~\ref{assump:main}(d), we have $\spp(\bm W\bm W^{\UT})\subset\supp_{\epsilon/2}(\mu)$ almost surely for large $M$. Since Claim~1 ensures $\mathcal K_\epsilon^*\cap\supp_{\epsilon/2}(\mu)=\varnothing$, it follows that $\spp(\bm W\bm W^{\UT})\cap\mathcal K_\epsilon^*=\varnothing$, which proves \eqref{eq:sep-W}.

We next exclude spurious roots of the empirical master equation. Since $\Gamma(z)\to1$ as $|z|\to\infty$, there exists $R>0$ such that $|\Gamma(z)|>0$ for $|z|\ge R$. By definition of $\mathcal K_*$, the function $\Gamma$ has no real zeros on
$
[-R,R]\setminus\bigl(\supp_\epsilon(\mu)\cup\mathcal K_\epsilon^*\bigr).
$
Hence there exists $\alpha>0$ such that $|\Gamma(x)|\ge\alpha
$, $\forall x\in[-R,R]\setminus\bigl(\supp_\epsilon(\mu)\cup\mathcal K_\epsilon^*\bigr)$. By the almost sure local uniform convergence $\Gamma_M\to\Gamma$ on $\C\setminus\supp(\mu)$, for all sufficiently large $M$, $|\Gamma_M(x)|\ge\alpha/2$ on the same set, and thus $\Gamma_M$ has no real zeros outside $\supp_\epsilon(\mu)\cup\mathcal K_\epsilon^*$.

Now fix $\lambda\in\spp(\bm Y\bm Y^{\UT})\setminus\supp_\epsilon(\mu)$. We first claim that
$\lambda>0$ almost surely. If $0\in\supp(\mu)$, then $0\in\supp_\epsilon(\mu)$, and the
assumption $\lambda\notin\supp_\epsilon(\mu)$ immediately implies $\lambda\neq 0$. If instead
$0\notin\supp(\mu)$, then $\bm W\bm W^{\UT}\succ 0$ almost surely. Let
$\bm W=\bm U_W\bm\Sigma_W\bm V_W^{\UT}$ be the SVD of $\bm W$. In this case one can also show
that $\mathbb{P}\bigl(0\in\spp(\bm Y\bm Y^{\UT})\bigr)=0.$ Indeed, conditional on all randomness other than the Haar matrix $\bm V_W$, the event
$\det(\bm Y\bm Y^{\UT})=0$ is characterized by the vanishing of a nontrivial polynomial in
the entries of $\bm V_W$; equivalently, it defines a proper algebraic subset of the
orthogonal group, which has Haar measure zero. Consequently, $\lambda>0$ almost surely.
Moreover, since $\lambda\notin\supp_\epsilon(\mu)$, we also have
$\lambda\notin\spp(\bm W\bm W^{\UT})$. Fact~\ref{fact:loc of emp out} therefore applies and
yields $\Gamma_M(\lambda)=0$. On the other hand, we have shown in the above that $\Gamma_M$ has no real zeros outside $\supp_\epsilon(\mu)\cup\mathcal K_\epsilon^*$, this forces
$\lambda\in\mathcal K_\epsilon^*$, which proves~\eqref{eq:no-extra}.

% Finally, for each $\lambda_k\in\mathcal K_*$, Claim~3 ensures that $\Gamma_M(z)=0$ admits exactly one real root in the interval $(\lambda_k-\epsilon,\lambda_k+\epsilon)$. Since this interval is disjoint from $\spp(\bm W\bm W^{\UT})$, Fact~\ref{fact:loc of emp out} implies the existence of exactly one eigenvalue of $\bm Y\bm Y^{\UT}$ therein. This proves~\eqref{eq:uniq-outlier} and completes the proof.

Finally, fix $\lambda_k\in\mathcal K_*$. Since $0\notin\mathcal K_*$, we have $\lambda_k>0$.
Choose $\rho\in(0,\epsilon)$ such that $\rho<\lambda_k/2$, and let
$D_k=\{z\in\C:\ |z-\lambda_k|\le \rho\}$.
By Claim~3, for all sufficiently large $M$ the equation $\Gamma_M(z)=0$ has exactly one zero
(counted with multiplicity) in $D_k$. Since $\Gamma_M(\overline z)=\overline{\Gamma_M(z)}$,
this zero must be real; denote it by
$\widehat\lambda_{k,M}\in(\lambda_k-\rho,\lambda_k+\rho)\subset(0,\infty)$. Moreover, $D_k\cap\spp(\bm W\bm W^{\UT})=\varnothing$ for all large $M$. Hence
Fact~\ref{fact:loc of emp out} implies that $\bm Y\bm Y^{\UT}$ has exactly one eigenvalue in
$(\lambda_k-\rho,\lambda_k+\rho)$ (counted with multiplicity). This proves~\eqref{eq:uniq-outlier}
and completes the proof.

\medskip
\paragraph{Proof of Claim~4 (Limiting overlaps).}
We prove the first convergence; the others are analogous.

Fix $\lambda_k\in\mathcal K_*$. By Claims~1–3, $\lambda_k$ is isolated and
$\lambda_{k,M}\to\lambda_k$.
Let $h$ be continuous, supported near $\lambda_k$, with $h(\lambda_k)=1$.
By weak convergence of $\nu_{M,1}$,
\[
\int h\,d\nu_{M,1}\to\nu_1(\{\lambda_k\}),
\]
and for all large $M$ exactly one term contributes, yielding
\[
\frac{1}{M}\langle\bm u_k(\bm Y),\bm u_*\rangle^2\to\nu_1(\{\lambda_k\}).
\]
For the mixed overlap, apply the same localization argument to the symmetric
dilation of $\bm Y$, whose spectrum contains the pair
$\pm\sigma_{k,M}=\pm\sqrt{\lambda_{k,M}}$.
The signal component splits equally between the symmetric and anti-symmetric
modes, producing an additional factor of $2$.
After normalization by $L=M+N$ and using $N/M\to\delta$, this yields
\[
\frac{1}{\sqrt{MN}}
\langle\bm u_k(\bm Y),\bm u_*\rangle
\langle\bm v_k(\bm Y),\bm v_*\rangle
\;\to\;
2\frac{1+\delta}{\sqrt{\delta}}\,
\nu_3(\{\sigma_k\}),
\]
as claimed. 
\end{proof}

%===================================
%        Integrals
%%===================================
\subsection{Integrals of Spectral Measures}\label{sec:spectral_analysis_apps}

The signal--eigenspace spectral measures yield integral representations for the quadratic and bilinear
forms that arise in the spectral initialization and state-evolution analysis. The following proposition
records these representations; the limits follow from the weak convergence in
\lemref{lem:spectral_measures_properties}.

\begin{proposition}\label{prop:integral_representation}
Let $\bm{Y}$ follow \eqref{eq:rectangular spiked model}, and let $\nu_1,\nu_2,\nu_3$ be the limiting
measures from Definition~\ref{def:spectral_measures}.
\begin{enumerate}
\item[(1)] Let $h:\R_+\to\R$ be bounded, and continuous on $\supp(\mu) \cup \mathcal{K}^*$. Then, almost surely we have:
\BS
\begin{align}
\frac{1}{M}\bm{u}_*^\UT h(\bm{Y}\bm{Y}^\UT)\bm{u}_*
&\xrightarrow{\text{a.s.}} \int h(\lambda)\,d\nu_1(\lambda)
= \langle h(\lambda)\rangle_{\nu_1}, \label{eq: conv_u}\\
\frac{1}{N}\bm{v}_*^\UT h(\bm{Y}^\UT\bm{Y})\bm{v}_*
&\xrightarrow{\text{a.s.}} \int h(\lambda)\,d\nu_2(\lambda)
= \langle h(\lambda)\rangle_{\nu_2}. \label{eq: conv_v}
\end{align}
\ES

\item[(2)] Let $f:\R\to\R$ be bounded, continuous, and odd. Define
$f(\bm{Y})\bydef \bm{U}_{\bm{Y}}\mathrm{diag}(f(\sigma_i(\bm{Y})))\bm{V}_{\bm{Y}}^\UT$.
With $L=M+N$, almost surely we have:
\begin{align}
\frac{1}{L}\bm{u}_*^\UT f(\bm{Y})\bm{v}_*
&\xrightarrow{\text{a.s.}} \int f(\sigma)\,d\nu_3(\sigma)
= \langle f(\sigma)\rangle_{\nu_3}. \label{eq: conv_uv}
\end{align}
\end{enumerate}
\end{proposition}

\begin{proof}
\noindent \textbf{Proof of Claim~(1).} We first prove \eqref{eq: conv_u}.
\begin{align}\label{eq:uconv-proof}
\frac{1}{M}\bm{u}_*^\UT h(\bm{Y}\bm{Y}^\UT)\bm{u}_*
&= \frac{1}{M}\sum_{i=1}^{M} h(\lambda_i(\bm{Y}\bm{Y}^\UT))
   \,\langle \bm{u}_i(\bm{Y}\bm{Y}^\UT),\bm{u}_*\rangle^2 \nonumber \\
&\stackrel{(a)}{=} \int h(\lambda)\,d\nu_{M,1}(\lambda)
\underset{\text{a.s.}}{\overset{(b)}{\longrightarrow}}
 \int h(\lambda)\,d\nu_{1}(\lambda)
 \stackrel{}{=} \langle h(\lambda)\rangle_{\nu_1},
\end{align}
where (a) is Definition~\ref{def:spectral_measures}, and
(b) follows by weak convergence on $\supp(\mu)$ and the convergence of the finitely many outlier atoms
$\lambda_{k,M}\xrightarrow{\mathrm{a.s.}}\lambda_k$ for each $\lambda_k\in\mathcal K^*$
(cf.~\lemref{lem:spectral_measures_properties}, \propref{prop:outlier_characterization}).
Finally, \eqref{eq: conv_v} follows in the same spirit, using in addition the atomic convergence
$\nu_{N,2}(\{0\})\xrightarrow{\mathrm{a.s.}}\nu_2(\{0\})$, which is implied by the H\"older continuity of $\mu$
in \assumpref{assump:main}, together with the weak convergence $\nu_{N,2}\to \nu_2$.

\noindent \textbf{Proof of Claim~(2).} Let $\widehat{\bm{Y}}\bydef\begin{bmatrix}0&\bm{Y}\\ \bm{Y}^\UT&0\end{bmatrix}$ be the symmetric dilation, and write
\[
\bm{Y}
= \bm{U}_{\bm{Y}}\,[\bm{\Sigma}_{\bm{Y}}\mid \bm{0}]\,[\bm{V}_{\bm{Y},1}\mid \bm{V}_{\bm{Y},2}]^\UT.
\]
Then $\widehat{\bm{Y}}=\bm{Q}_{\widehat{\bm{Y}}}\bm{\Lambda}_{\widehat{\bm{Y}}}\bm{Q}_{\widehat{\bm{Y}}}^\UT$ with
\[
\bm{\Lambda}_{\widehat{\bm{Y}}}
= \mathrm{diag}\!\big(\bm{\Sigma}_{\bm{Y}},-\bm{\Sigma}_{\bm{Y}},\bm{0}_{N-M}\big),
\qquad
\bm{Q}_{\widehat{\bm{Y}}}
=
\left[
\begin{array}{c|c|c}
\frac{1}{\sqrt2}\bm{U}_{\bm{Y}} & -\frac{1}{\sqrt2}\bm{U}_{\bm{Y}} & \bm{0}_{M\times (N-M)}\\
\hline
\frac{1}{\sqrt2}\bm{V}_{\bm{Y},1} & \frac{1}{\sqrt2}\bm{V}_{\bm{Y},1} & \bm{V}_{\bm{Y},2}
\end{array}
\right].
\]
Substituting into Definition~\ref{def:spectral_measures}(b) gives the explicit form
\begin{align}\label{eq:nuL3_explicit}
\nu_{L,3}
&\stackrel{(c)}{=}
\frac{1}{2L}\sum_{i=1}^{M}
\langle \bm{u}_i(\bm{Y}),\bm{u}_* \rangle \langle \bm{v}_i(\bm{Y}),\bm{v}_* \rangle
\big(\delta_{\sigma_i(\bm{Y})}-\delta_{-\sigma_i(\bm{Y})}\big),
\end{align}
where (c) uses the block forms of the eigenvectors in $\bm{Q}_{\widehat{\bm{Y}}}$ and the orthogonality of the
$N-M$ null-space eigenvectors to $\widehat{\bm{u}}_*$. Moreover,
\BS
\begin{align*}
\frac{1}{L}\bm{u}_*^\UT f(\bm{Y})\bm{v}_*
&= \frac{1}{L} \sum_{i=1}^{M} f(\sigma_i(\bm{Y}))
   \langle \bm{u}_i(\bm{Y}), \bm{u}_* \rangle \langle \bm{v}_i(\bm{Y}), \bm{v}_* \rangle \\
&\stackrel{(d)}{=}
\frac{1}{2L}\sum_{i=1}^{M}
\langle \bm{u}_i(\bm{Y}),\bm{u}_* \rangle \langle \bm{v}_i(\bm{Y}),\bm{v}_* \rangle
\big(f(\sigma_i(\bm{Y}))-f(-\sigma_i(\bm{Y}))\big) \\
&\stackrel{(e)}{=} \int f(\sigma)\,d\nu_{L,3}(\sigma)
\underset{\text{a.s.}}{\overset{(g)}{\longrightarrow}}
 \int f(\sigma)\,d\nu_3(\sigma)
 = \langle f(\sigma)\rangle_{\nu_3},
\end{align*}
\ES
where (d) uses that $f$ is odd; (e) follows from \eqref{eq:nuL3_explicit}; and (g) uses
weak convergence $\nu_{L,3}\to \nu_3$ from \lemref{lem:spectral_measures_properties} and mimics step (b) in \eqref{eq:uconv-proof}.
\end{proof}

%% file: appendix/ap-SE.tex
\section{General OAMP Algorithm with Rotationally-Invariant Matrices}\label{app:general-SE}

The proof of the main result (\thref{Thm: State Evolution}) follows a reduction strategy similar to that in \cite{dudeja2024optimality}. Specifically, we transform the OAMP iteration \eqref{def:OAMP}, which depends on the signal matrix $\bm{Y}$, into an asymptotically equivalent iteration that depends only on the random matrix $\bm{W}$. The resulting algorithm admits a state evolution characterization, which can be established using standard conditioning techniques \cite{rangan2019vector,takeuchi2019rigorous,fan2022approximate}.

As we are not aware of prior work that directly addresses our specific setting, we include in this appendix the formulation of the general OAMP iteration with random $\bm{W}$ and its associated state evolution for completeness. Our derivation closely follows \cite{fan2022approximate}, and we therefore omit many technical details, emphasizing instead the key differences from that work.

\subsection{General OAMP Iteration}\label{Sec:general_OAMP}

% \subsubsection*{Algorithm and Assumptions}
We introduce a general OAMP iteration with bi-rotationally-invariant random matrix $\bm{W}$.

\begin{definition}[General OAMP algorithm]\label{def: Noise OAMP}
For $t\in\mathbb{N}$, the general OAMP algorithm generates the iterates $(\bm{x}_t)_{t\in\mathbb{N}}$ and $(\bm{z}_t)_{t\in\mathbb{N}}$ via
\BS\label{Eqn:OAMP_app_def}
\begin{align}
\bm{x}_t &=  {\Psi}_{t}(\bm{W}\bm{W}^\UT) m_{t}(\bm{x}_{\le t-1};\bm{a}) +  \tilde{\Psi}_{t}(\bm{WW}^\UT)\,\bm{W} q_{t}(\bm{z}_{\le t-1};\bm{b}) ,\\
%&=\mathscr{G}_{t}(\bm{W}\bm{W}^\UT) \bm{q}_t + G_{t}(\bm{WW}^\UT)\,\bm{W} \bm{m}_t ,\\
\bm{z}_t &=  \Phi_{t}(\bm{W}^\UT \bm{W} ) q_{t}(\bm{z}_{\le t-1};\bm{b}) + \tilde{\Phi}_{t}(\bm{W}^\UT\bm{W})\bm{W}^\UT m_{t}(\bm{x}_{\le t-1};\bm{a}) ,
\end{align}
\ES
where the matrix denoising functions $\Psi_t$ and $\Phi_t$ satisfy (while $\tilde{\Psi}_t$ and $\tilde{\Phi}_t$ do not) the following \textit{trace-free} conditions as $M,N\to\infty$ with $M/N\to\delta\in(0,1]$
\begin{align}
\mathbb{E}\left[\Psi_t(\mathsf{D}_M^2)\right]=0,\quad \mathsf{D}_M^2\sim\mu,\\
\mathbb{E}\big[\Phi_t(\mathsf{D}_N^2)\big]=0,\quad \mathsf{D}_N^2\sim\muN.
\end{align}
and the signal denoisers $(m_t)_{t\ge1}$ and $(q_t)_{t\ge1}$ are \textit{divergence-free}:
\begin{align}
\mathbb{E}\left[\partial_i m_{t}(\mathsf{X}_{\le t-1};\mathsf{A})\right] &=0,\quad \forall t\in\mathbb{N}, i\in [t],\\
\mathbb{E}\left[\partial_i q_{t}(\mathsf{Z}_{\le t-1};\mathsf{B})\right] &=0,\quad \forall t\in\mathbb{N}, i\in [t].
\end{align}
The random variables $(\mathsf{D}_M,\mathsf{D}_N)$ and $(\mathsf{X}_{t},\mathsf{Z}_{t})_{t\in\mathbb{N}}$ are to be defined in Definition \ref{Def:SE}.
\end{definition}

Let the singular value decomposition of $\bm{W}$ be $\bm{W}=\bm{U}_W\bm{\Sigma}\bm{V}_W^\UT$. We make a change of variables:
\BE
\tilde{\bm{x}}_t \bydef \bm{U}_W^\UT\bm{x}_t\quad\text{and}\quad \tilde{\bm{z}}_t\bydef\bm{V}_W^\UT\bm{z}_t.
\EE
%Using the new variables, we can write the OAMP iteration into the following form:
%\BS\label{Eqn:OAMP_transformed}
%\begin{align}
%{\tilde{\bm{x}}_t} &= \Psi_t(\bm{DD}^\UT)\bm{U}_W^\UT f_{t-1}(\bm{U}_W{\tilde{\bm{x}}_{\le t-1}};\bm{a}) +  \tilde{\Psi}_t(\bm{D}\bm{D}^\UT)\,\bm{D}\bm{V}_W^\UT  g_{t-1}(\bm{V}_W{\tilde{\bm{z}}_{\le t-1}};\bm{b}) ,\\
%{\tilde{\bm{z}}_t} &= \Phi_t(\bm{D}^\UT\bm{D})\bm{V}_W^\UT g_{t-1}(\bm{V}_W{\tilde{\bm{z}}_{\le t-1}};\bm{b}) + \tilde{\Phi}_t(\bm{D}^\UT\bm{D})\,\bm{D}^\UT\bm{x}^\UT f_{t}(\bm{U}_W{\tilde{\bm{x}}_{\le t-1}};\bm{a}).
%\end{align}
%\ES
Using the new variables, we can write the OAMP iteration into the following factorized form (see Fig.~\ref{Fig:OAMP_diagram} for an illustration). 

\begin{definition}[General OAMP algorithm: factorized form]\label{Def:OAMP_factorized}
The factorized form OAMP algorithm proceeds as follows ($\forall t\in\mathbb{N}$):
\BS\label{Eqn:OAMP_app_full}
\begin{align}
\text{(Orthogonal transform)}\qquad &\bm{r}_t = \bm{U}_W^\UT\bm{m}_t,\quad \bm{s}_t=\bm{V}_W^\UT\bm{q}_t,\\[5pt]
\text{(Matrix denoising)}\qquad &\tilde{\bm{x}}_{t} = \psi_t(\bm{r}_t,\bm{D}\bm{s}_t|\bm{D}\bm{D}^\UT),\quad\tilde{\bm{z}}_t=\phi_t(\bm{s}_t,\bm{D}^\UT\bm{r}_t|\bm{D}^\UT\bm{D}),\\[5pt]
\text{(Orthogonal transform)}\qquad &\bm{x}_t = \bm{U}_W\tilde{\bm{x}}_t,\quad \bm{z}_t = \bm{V}_W\tilde{\bm{z}}_t,\\[5pt]
\text{(Iterate denoising)}\qquad &\bm{m}_{t+1}=m_{t+1}(\bm{x}_{\le t};\bm{a}),\quad \bm{q}_{t+1} = q_{t+1}(\bm{z}_{\le t};\bm{b}),
\end{align}
where $\psi_t$ and $\phi_t$ are defined as
\begin{align}
\psi_t(\bm{r}_t,\bm{D}\bm{s}_t|\bm{DD}^\UT) &\bydef \Psi_t(\bm{DD}^\UT)\bm{r}_t + \tilde{\Psi}_t(\bm{DD}^\UT)\bm{D}\bm{s}_t,\\
\phi_t(\bm{s}_t,\bm{D}^\UT\bm{r}_t|\bm{D}^\UT\bm{D})) &\bydef \Phi_t(\bm{D}^\UT\bm{D}))\bm{s}_t + \tilde{\Phi}_t(\bm{D}^\UT\bm{D})\bm{D}^\UT\bm{r}_t.
\end{align}
\ES
\end{definition}
\vspace{5pt}

\begin{figure}[htbp]
 \centering
 \includegraphics[width=0.65\textwidth]{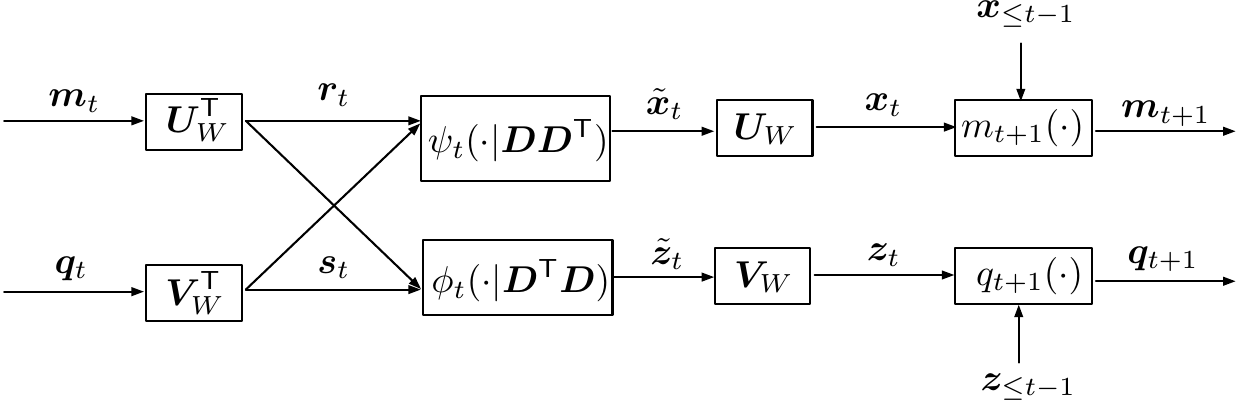}
 \caption{Diagram of the OAMP algorithm \eqref{Eqn:OAMP_app_full}. $m_t$ and $q_t$ are divergence-free. $\psi_t$ and $\phi_t$ are divergence-free with respect to the direct inputs (but not necessarily with respect to the cross input terms).}\label{Fig:OAMP_diagram}
 \end{figure}

%===========================================================
\subsection{State Evolution of General OAMP Iteration}

To establish a high-dimensional asymptotic characterization of the OAMP algorithm, we impose the following assumptions.

\begin{assumption}\label{Ass:OAMP_app}
The following conditions hold for the factorized OAMP algorithm defined in Def.~\ref{Def:OAMP_factorized}:
\begin{itemize}
\item[(1)] The matrix $\bm{W}$ satisfies Assumption \ref{assump:main}-(c).
\item[(2)] The initialization vectors $\bm{m}_1\in\mathbb{R}^N$ and $\bm{q}_1\in\mathbb{R}^N$ are independent of $\bm{U}_W$ and $\bm{V}_W$, respectively. Moreover, $(\bm{m}_1,\bm{a})\overset{W}{\longrightarrow}(\mathsf{X}_1,\mathsf{A})$ and $(\bm{q}_1,\bm{b})\overset{W}{\longrightarrow}(\mathsf{Q}_1,\mathsf{B})$, where the limit random variables possess finite moments of all orders.
\item[(3)] The matrix denoisers $\psi_t$ and $\phi_t$ are continuous, and the iterate denoisers $m_t$  and $q_t$ are continuously differentiable and Lipschitz.
\end{itemize}
\end{assumption}

The iterates of the general OAMP algorithm introduced in Definition \ref{Def:OAMP_factorized} admit an exact asymptotic characterization via a state evolution. Before presenting the formal result, we first describe the corresponding state evolution recursion.

\begin{definition}[State evolution of OAMP: factorized form]\label{Def:SE}

Set $\paraB{\Omega}_1^{u}=\mathbb{E}[\mathsf{M}_1^2]$ and $\paraB{\Omega}_1^{v}=\mathbb{E}[\mathsf{Q}_1^2]$. Iterate the following steps for $t=1,2,\ldots$
\hfill
\begin{enumerate}[label=(\roman*)]
\item \textit{Gaussian random variables:}
\BS
\begin{align}
\mathsf{R}_{\le t}\sim\mathcal{N}(\bm{0},\paraB{\Omega}_t^u),\quad \mathsf{S}_{\le t}\sim\mathcal{N}(\bm{0},\paraB{\Omega}_t^v).
\end{align}
\item \textit{Matrix denoising:}
\BE
        \tilde{\mathsf X}_t =
          \psi_t\!\bigl(\mathsf R_t,\mathsf D_M \mathsf S_t|\mathsf{D}_M^2\bigr),
        \quad
        \tilde{\mathsf Z}_t =
          \phi_t\!\bigl(\mathsf S_t,\mathsf D_N \mathsf R_t|\mathsf{D}_N^2\bigr),
\EE
where $\mathsf{D}_M$, $\mathsf{R}_{\le t}$ and $\mathsf{S}_{\le t}$ are mutually independent. Moreover, $\mathsf D_N\explain{d}{=} \mathsf H\,\mathsf D_M$, where $\mathsf H\sim\text{Ber}(\delta)$ is independent of other random variables.
\item \textit{Covariance update:}
\BE\label{Eqn:app_SE_Sigma}
        \paraB{\Sigma}_t^{u}
          = \mathbb E\left[\tilde{\mathsf X}_{\le t}
                         \tilde{\mathsf X}_{\le t}^\UT\right],
        \quad
        \paraB{\Sigma}_t^{v}
          = \mathbb E\left[\tilde{\mathsf Z}_{\le t}
                         \tilde{\mathsf Z}_{\le t}^\UT\right].
\EE

\item \textit{Gaussian random variables:}
%      \BE
%        \mathsf{X}_{\le t} \sim \mathcal{N}(\bm{0}_t,\paraB{\Sigma}_t^{u}),
%        \quad
%       \mathsf{Z}_{\le t} \sim \mathcal{N}(\bm{0}_t,\paraB{\Sigma}_t^{v}).
%      \EE
\begin{align}
\mathsf{X}_{\le t}\sim\mathcal{N}(\bm{0},\paraB{\Sigma}_t^u),\quad \mathsf{Z}_{\le t}\sim\mathcal{N}(\bm{0},\paraB{\Sigma}_t^v).
\end{align}

\item \textit{Iterate denoising:}
      \BE\label{Eqn:app_SE_M_Q}
        \mathsf M_{t+1} = m_{t+1}(\mathsf X_{\le t};\mathsf A),
        \quad
        \mathsf Q_{t+1} = q_{t+1}(\mathsf Z_{\le t};\mathsf B),
      \EE
where $\mathsf X_{\le t}\sim\mathcal{N}(\bm{0}_t,\paraB{\Sigma}_t^u)\indep (\mathsf{A},\mathsf{M}_1)$ and $\mathsf Z_{\le t}\sim\mathcal{N}(\bm{0}_t,\paraB{\Sigma}_t^v)\indep (\mathsf{B},\mathsf{Q}_1)$.
\item \textit{Covariance update:}
      \BE\label{Eqn:app_SE_Omega}
        \paraB{\Omega}_{t+1}^{u}
     =\mathbb E[\mathsf M_{\le t+1}\mathsf M_{\le t+1}^{\!\UT}],
        \quad
        \paraB{\Omega}_{t+1}^{v}
          = \mathbb E[\mathsf Q_{\le t+1}\mathsf Q_{\le t+1}^{\!\UT}].
      \EE
    \ES 
\end{enumerate}
In the above equations, with slight abuse of notations, $\psi_t$ and $\phi_t$ are defined as
\BS
\begin{align}
\psi_t\left(\mathsf{R}_t,\mathsf{D}_M\mathsf{S}_t|\mathsf{D}_M^2\right)&\;\bydef\;  \Psi_t(\mathsf{D}_M^2)\mathsf{R}_t + \tilde{\Psi}_t(\mathsf{D}_M^2)\mathsf{D}_M\mathsf{S}_t,\\
\phi_t(\mathsf{S}_t,\mathsf{D}_N \mathsf{R}_t|\mathsf{D}_N^2) &\;\bydef\; \Phi_t(\mathsf{D}_N^2)\mathsf{S}_t + \tilde{\Phi}_t(\mathsf{D}_N^2)\mathsf{D}_N\mathsf{R}_t.
\end{align}
\ES
%Moreover, it is understood that:
%\begin{itemize}
%\item The expectations in \eqref{Eqn:app_SE_Sigma} are taken by expanding $\tilde{\mathsf{X}}_{\le t}$ and $\tilde{\mathsf{Z}}_{\le t}$, and assuming $(\mathsf{R}_1,\ldots,\mathsf{R}_t)$, $(\mathsf{S}_1,\ldots,\mathsf{S}_t)$ and $(\mathsf{D}_N,\mathsf{D}_N)$ to be mutually independent;
%\item In \eqref{Eqn:app_SE_Omega}, $\mathsf{M}_1$ is independent of $\mathsf{X}_{\le t}$, and $\mathsf{Q}_1$ is independent of $\mathsf{Z}_{\le t}$.
%\end{itemize}
\noindent
%All expectations in \eqref{Eqn:app_SE_Sigma} and \eqref{Eqn:app_SE_Omega} are taken with respect to the joint law generated by steps (i)–(vi), under which
%\[
%(\mathsf R_1,\dots,\mathsf R_t),\;(\mathsf S_1,\dots,\mathsf S_t),\;\mathsf D_M,\;\mathsf A,\;\mathsf B,\;\mathsf H
%\]
%are mutually independent, and $\mathsf D_N\explain{d}{=} \mathsf H\,\mathsf D_M$.  Moreover, for \eqref{Eqn:app_SE_Omega}, $\mathsf M_1$ is independent of $\mathsf X_{\le t}$ and $\mathsf Q_1$ is independent of $\mathsf Z_{\le t}$. 
\end{definition}
%\textbf{\textcolor{green}{We can even study the joint distribution of all the dimension-$M$ vectors:
%\[
%(\bm{r}_{\le t},\tilde{\bm{x}}_{\le t},\tilde{\bm{x}}_{\le t}, \tilde{\bm{m}}_{\le t+1})
%\]
%From the conditioning technique, it seems that the two groups of random variables $(\bm{r}_{\le t},\tilde{\bm{x}}_{\le t})$ and $(\tilde{\bm{x}}_{\le t}, \tilde{\bm{m}}_{\le t+1})$ are asymptotically independent, although there are intra-group corrections.
%\begin{itemize}
%\item Check whether this is correct
%\item Think how to introduce the SE!
%\item Perhaps we revise the main theorem accordingly!
%\end{itemize}
%}}
% \vspace{5pt}
\begin{remark}
Below are some remarks about Definition \ref{Def:SE}:
\begin{itemize}
\item Whenever the collections of random variables $\mathsf{R}_{\le t}$, $\mathsf{S}_{\le t}$ and $(\mathsf{D}_M,\mathsf{D}_N)$ appear jointly, they are understood as mutually independent. This independence is used in the definitions of the covariance matrices $\paraB{\Sigma}_t^u$ and $\paraB{\Sigma}_t^v$.
\item Explicit formulas for $\{\paraB{\Omega}_{t+1}^{u}[i,j],i\le j\}$ are given below:
\BS
\begin{align}
\paraB{\Omega}_{t+1}^u[i,j]=
\begin{cases}
\mathbb{E}[\mathsf{M}_1^2] & \text{ for } i=j=1,\\
\mathbb{E}[\mathsf{M}_1\cdot f_{j-1}(\mathsf{X}_{\le j-1};\mathsf{A})] & \text{ for } i=1,j>1,\\
%\mathbb{E}[\mathsf{M}_1]\cdot\mathbb{E}[f_{i-1}(\mathsf{X}_{\le i-1};\mathsf{A})] & \text{ for } i>1,j=1\\
\mathbb{E}[f_{i-1}(\mathsf{X}_{\le i-1};\mathsf{A})f_{j-1}(\mathsf{X}_{\le j-1};\mathsf{A})] & \text{ for } 1<i\le j\le t+1.\\
\end{cases}
\end{align}
\ES
In the second expectation, $(\mathsf{M}_1,\mathsf{A}) \indep \mathsf{X}_{\le t}$. Similar formulas apply analogously to $\paraB{\Omega}_{t+1}^v$.
\end{itemize}
\end{remark}

The orthogonality property stated below explains the term ``Orthogonal AMP” and serves as a key ingredient in the proof of Theorem \ref{Sec:OAMP_general_SE_proof}. This property arises from the divergence-free and trace-free constraints imposed on the OAMP denoisers.

\begin{lemma}[Orthogonality]\label{Lem:orthogonality}
Suppose that the covariance matrices $(\paraB{\Sigma}_t^u,\paraB{\Sigma}_t^v)_{t\in\mathbb{N}}$ in Definition \ref{Def:SE} are non-singular. Then, the state evolution random variables in Definition \ref{Def:SE} satisfy
\BS
\begin{align}
\mathbb{E}\left[\mathsf{R}_i\tilde{\mathsf{X}}_j\right] &=0,\quad \mathbb{E}\left[\mathsf{S}_i\tilde{\mathsf{Z}}_j\right] =0,\quad\forall i,j\in\mathbb{N},\\ 
\mathbb{E}\left[\mathsf{X}_i{\mathsf{M}}_{j}\right] &=0,\quad \mathbb{E}\left[\mathsf{Z}_i{\mathsf{Q}}_{j}\right] =0,\quad\forall i,j\in\mathbb{N}.
\end{align}
\ES
\end{lemma}
\begin{proof}
For $\mathbb{E}\left[\mathsf{R}_i\tilde{\mathsf{X}}_j\right]=0$, we substitute the definition of $\tilde{\mathsf{X}}_j$:
\BS
\begin{align*}
\mathbb{E}\left[\mathsf{R}_i\tilde{\mathsf{X}}_j\right] &=\mathbb{E}\left[\mathsf{R}_i\left(\Psi_j(\mathsf{D}_M^2)\mathsf{R}_j + \tilde{\Psi}_j(\mathsf{D}_M^2)\mathsf{D}_M\mathsf{S}_j\right)\right],\quad\forall i,j\in\mathbb{N},\\
&\explain{a}{=}\mathbb{E}\left[\mathsf{R}_i\mathsf{R}_j\right]\cdot\mathbb{E}\left[\Psi_j(\mathsf{D}_M^2)\right]+\mathbb{E}[\mathsf{R}_i]\cdot \mathbb{E}[\mathsf{S}_j]\cdot\mathbb{E}\left[\Psi_j(\mathsf{D}_M^2)\mathsf{D}_M\right]\\
&\explain{b}{=}0,
\end{align*}
\ES
where step (a) is due to the independence of $(\mathsf{R}_i,\mathsf{R}_j)$, $\mathsf{S}_j$ and $\mathsf{D}_M$, and step (b) is due to the trace-free property $\mathbb{E}\left[\Psi_i(\mathsf{D}_M^2)\right]=0$ and the fact that $\mathsf{R}_i,\mathsf{S}_j$ have zero mean. The term $\mathbb{E}\left[\mathsf{X}_i{\mathsf{M}}_{j}\right]$ can be computed as
\BS
\begin{align*}
\mathbb{E}\left[\mathsf{X}_i{\mathsf{M}}_{j}\right] &=
\begin{cases}
\mathbb{E}\left[\mathsf{X}_i\right] \cdot\mathbb{E}\left[{\mathsf{M}}_{1}\right]\explain{a}{=}0, & \text{if } j=1,\\
\sum_{k=1}^j \mathbb{E}\left[\mathsf{X}_i\mathsf{X}_k\right] \cdot \mathbb{E}\left[\partial_k m_{j}(\mathsf{X}_{\le j-1};\mathsf{A})\right]\explain{b}{=}0, & \text{if } j>1,
\end{cases}
\end{align*}
\ES
where step (a) is due to the independence between $\mathsf{X}_i$ and $\mathsf{M}_1$, and step (b) is a consequence of the multivariate Stein's lemma and the divergence-free properties of $m_{j}(\cdot)$. The other two properties can be proved analogously and omitted for brevity.
\end{proof}

The following theorem shows that, under Assumption \ref{Ass:OAMP_app}, the performance of the OAMP algorithm is governed by the corresponding state evolution equations. All convergence statements are understood in the limit $M,N\to\infty$ with $M/N\to\delta\in(0,1]$. Its proof is deferred to Section \ref{Sec:OAMP_general_SE_proof}.

\begin{theorem}[State evolution characterization of OAMP: factorized form]\label{The:SE_app}
Consider the OAMP algorithm in Definition \ref{Def:OAMP_factorized} with initialization $\bm{m}_1\in\mathbb{R}^M$ and $\bm{q}_t\in\mathbb{R}^N$. Suppose Assumption \ref{Ass:OAMP_app} holds. Let the covariance matrices $(\paraB{\Omega}_t^u,\paraB{\Omega}_t^v,\paraB{\Sigma}_{t+1}^u,\paraB{\Sigma}_{t+1}^v)_{t\in\mathbb{N}}$ be defined as in Definition \ref{Def:SE}. {Assume additionally that these covariance matrices are non-singular for all fixed $t\in\mathbb{N}$.} The following hold for all fixed $t\in\mathbb{N}$:
\begin{itemize}
\item[(a)]
\BS\label{Th:main_a}
\begin{align}
\left(\bm{r}_{\le t},\bm{D}{\bm{s}}_{\le t},\bm{d}_M\right) &\overset{W}{\longrightarrow}\left(\mathsf{R}_{\le t},\mathsf{D}_M{\mathsf{S}}_{\le t},\mathsf{D}_M\right),\\
\left(\bm{s}_{\le t},\bm{D}^\UT{\bm{z}}_{\le t},\bm{d}_N\right) &\overset{W}{\longrightarrow}\left(\mathsf{S}_{\le t},\mathsf{D}_N{\mathsf{Z}}_{\le t},\mathsf{D}_N\right),
\end{align}
\ES
where $\mathsf{R}_{\le t}\sim\mathcal{N}(\bm{0}_t,\paraB{\Omega}_t^u)$, $\mathsf{S}_{\le t}\sim\mathcal{N}(\bm{0}_t,\paraB{\Omega}_t^v)$; and $\mathsf{R}_{\le t}$, $\mathsf{S}_{\le t}$ and $(\mathsf{D}_M,\mathsf{D}_N)$ are mutually independent.
\item[(b)]
\BS\label{Th:main_b}
\begin{align}
\left(\bm{r}_{\le t},\tilde{\bm{x}}_{\le t},\bm{d}_M\right) &\overset{W}{\longrightarrow}\left(\mathsf{R}_{\le t},\tilde{\mathsf{X}}_{\le t},\mathsf{D}_M\right),\\
\left(\bm{s}_{\le t},\tilde{\bm{z}}_{\le t},\bm{d}_N\right) &\overset{W}{\longrightarrow}\left(\mathsf{S}_{\le t},\tilde{\mathsf{Z}}_{\le t},\mathsf{D}_N\right),
\end{align}
\ES
where $\mathsf{R}_{\le t}\sim\mathcal{N}(\bm{0}_t,\paraB{\Omega}_t^u)$, $\mathsf{S}_{\le t}\sim\mathcal{N}(\bm{0}_t,\paraB{\Omega}_t^v)$; $\tilde{\mathsf{X}}_t=\psi_t(\mathsf{R}_t,\mathsf{D}_M\mathsf{S}_t)$, $\tilde{\mathsf{Z}}_t=\phi_t(\mathsf{S}_t,\mathsf{D}_N\mathsf{R}_t)$ with $\mathsf{R}_{\le t}$, $\mathsf{S}_{\le t}$ and $(\mathsf{D}_M,\mathsf{D}_N)$ mutually independent.
\item[(c)]
\BS\label{Th:main_c}
\begin{align}
\left(\bm{x}_{\le t},\bm{m}_{\le t+1};\bm{a}\right) &\overset{W}{\longrightarrow}\left(\mathsf{X}_{\le t},{\mathsf{M}}_{\le t+1},\mathsf{A}\right),\\
\left(\bm{z}_{\le t},{\bm{q}}_{\le t+1};\bm{b}\right) &\overset{W}{\longrightarrow}\left(\mathsf{Z}_{\le t},{\mathsf{Q}}_{\le t+1};\mathsf{B}\right),
\end{align}
\ES
where $\mathsf{X}_{\le t}\sim\mathcal{N}(\bm{0}_t,\paraB{\Sigma}_t^u)$, $\mathsf{M}_{t+1}=m_{t+1}(\mathsf{X}_{\le t};\mathsf{A})$, $\mathsf{Z}_{\le t}\sim\mathcal{N}(\bm{0}_t,\paraB{\Sigma}_t^v)$, $\mathsf{Q}_{t+1}=q_{t+1}(\mathsf{Z}_{\le t};\mathsf{B})$. Moreover, $\mathsf{X}_{\le t}\indep (\mathsf{M}_1,\mathsf{A})$, and $\mathsf{Z}_{\le t}\indep (\mathsf{Q}_1,\mathsf{B})$.
\end{itemize}
\end{theorem}

\vspace{5pt}

\begin{remark}[Additional asymptotic independence]
The proof of Theorem~\ref{The:SE_app} actually shows that for each fixed $t$,
%\[
%(\bm{r}_t,\tilde{\bm{x}}_t)\;\perp\;(\bm{x}_t,\bm{m}_{t+1})
%\quad\text{and}\quad
%(\bm{s}_t,\tilde{\bm{z}}_t)\;\perp\;(\bm{z}_t,\bm{q}_{t+1})\,,
%\]
%so that
\[
(\bm{r}_t,\tilde{\bm{x}}_t,\bm{x}_t,\bm{m}_{t+1})
\;\overset{W}{\longrightarrow}\;
(\mathsf{R}_t,\tilde{\mathsf{X}}_t,\mathsf{X}_t,\mathsf{M}_{t+1}),
\]
with $(\mathsf{R}_t,\tilde{\mathsf{X}}_t)\indep(\mathsf{X}_t,\mathsf{M}_{t+1})$, and analogously for $(\bm{s}_t,\tilde{\bm{z}}_t,\bm{z}_t,\bm{q}_{t+1})$. This refinement is not used in this paper and is therefore omitted from Theorem~\ref{The:SE_app}.
\end{remark}

Theorem \ref{The:SE_app} relies on the assumption that the covariance matrices 
$(\paraB{\Omega}_t^u,\paraB{\Omega}_t^v,\paraB{\Sigma}_{t+1}^u,\paraB{\Sigma}_{t+1}^v)_{t\in\mathbb{N}}$ are non-singular, which can be hard to check. However, by the perturbation argument in \cite{berthier2020state} (see also \cite[Corollary~4.4]{fan2022approximate} and \cite[Appendix~D.1]{dudeja2024spectral}), one may drop this assumption at the cost of weakening convergence from $W_p$ for all $p\ge1$ to $W_2$. We state this result in the next corollary; its proof is almost identical to that of \cite[Corollary.~4.4]{fan2022approximate} and is omitted.

\begin{corollary}[Removing non‐degeneracy assumption]\label{Cor:app_degeneracy}
Under Assumption~\ref{Ass:OAMP_app}, the statements \eqref{Th:main_a}–\eqref{Th:main_c} continue to hold when convergence is replaced by $W_2$.
\end{corollary}

%======================================

Finally, the state evolution associated with the factorized-form OAMP (Definition \ref{Def:OAMP_factorized}) can be transformed back to the original OAMP algorithm (Definition \ref{def: Noise OAMP}). For ease of reference, we record it below.

\begin{theorem}[State evolution characterization of general OAMP: original form]\label{Th:OAMP_general_SE_app}
Consider the OAMP algorithm in Definition \ref{def: Noise OAMP} with initialization $\bm{m}_1:=m_1(\bm{x}_0;\bm{a})$ and $\bm{q}_1:= q_1(\bm{z}_0;\bm{b})$. Suppose Assumption \ref{Ass:OAMP_app} holds. The following hold for all fixed $t\in\mathbb{N}$:
\BS
\begin{align}
\left(\bm{x}_1,\ldots,\bm{x}_t,\bm{a}\right)\overset{W_2}{\longrightarrow}\left(\mathsf{X}_1,\ldots,\mathsf{X}_t,\mathsf{A}\right),\\
\left(\bm{z}_1,\ldots,\bm{z}_t,\bm{b}\right)\overset{W_2}{\longrightarrow}\left(\mathsf{Z}_1,\ldots,\mathsf{Z}_t,\mathsf{B}\right),
\end{align}
where $(\mathsf{X}_1,\ldots,\mathsf{X}_t)\sim\mathcal{N}(\bm{0},\paraB{\Sigma}_t^u)$ is independent of $\mathsf{A}$, $(\mathsf{Z}_1,\ldots,\mathsf{Z}_t)\sim\mathcal{N}(\bm{0},\paraB{\Sigma}_t^v)$ is independent of $\mathsf{B}$. Define $\Sigma_{u,st}  \bydef \mathbb{E}\left[\mathsf{X}_s\mathsf{X}_t\right]$ and $\Sigma_{v,st}  \bydef \mathbb{E}\left[\mathsf{Z}_s\mathsf{Z}_t\right]$, for $s,t\in\mathbb{N}$. Then,
\begin{align}
\Sigma_{u,st} &=\mathbb{E}\left[\Psi_s(\mathsf{D}_M^2)\Psi_t(\mathsf{D}_M^2)\right]\cdot \mathbb{E}\left[\mathsf{M}_s\mathsf{M}_t\right]+\mathbb{E}\left[\tilde{\Psi}_s(\mathsf{D}_M^2)\tilde{\Psi}_t(\mathsf{D}_M^2)\mathsf{D}_M^2\right]\cdot \mathbb{E}\left[\mathsf{Q}_s\mathsf{Q}_t\right],\\
\Sigma_{v,st} &=\mathbb{E}\left[\Phi_s(\mathsf{D}_N^2)\Phi_t(\mathsf{D}_N^2)\right]\cdot \mathbb{E}\left[\mathsf{Q}_s\mathsf{Q}_t\right]+\mathbb{E}\left[\tilde{\Phi}_s(\mathsf{D}_N^2)\tilde{\Phi}_t(\mathsf{D}_N^2)\mathsf{D}_N^2\right]\cdot \mathbb{E}\left[\mathsf{M}_s\mathsf{M}_t\right],
\end{align}
where the random variables $(\mathsf{D}_M,\mathsf{D}_N)$ are defined as in Definition \ref{Def:SE}, and $\forall s\in\mathbb{N}$
\begin{align}
\mathsf{M}_{s+1} &\bydef m_{s+1}(\mathsf{X}_1,\ldots,\mathsf{X}_{s};\mathsf{A}),\\
\mathsf{Q}_{s+1} &\bydef q_{s+1}(\mathsf{Z}_1,\ldots,\mathsf{Z}_{s};\mathsf{B}).
\end{align}
\ES
\end{theorem}
Theorem~\ref{Th:OAMP_general_SE_app} follows directly from Theorem~\ref{The:SE_app} and Corollary~\ref{Cor:app_degeneracy}, and its proof is therefore omitted.

\subsection{Proof of Theorem \ref{The:SE_app}}\label{Sec:OAMP_general_SE_proof}

Let $(t.a)$ and $(t.b)$ denote the claims $(a)$ and $(b)$ for iteration $t$. We prove by induction on $t=1,2,\ldots$

\paragraph{Base case: proof of claim $(1.a)-(1.c)$} Recall that $\bm{r}_1=\bm{U}_W^\UT\bm{m}_1$, $\bm{s}_1=\bm{V}_W^\UT\bm{q}_1$, where $\bm{U}_W\in\mathbb{O}(M)$ and $\bm{V}_W\in\mathbb{O}(N)$ are independent Haar random matrices. Based on standard properties of $W$ convergence of empirical probability measures and Haar random matrix (see \cite[Appendix E and Appendix F]{fan2022approximate}), together with the fact that $\bm{U}_W$ and $\bm{V}_W$ are independent and the entries of $(\bm{d}_M,\bm{d}_N)$ are bounded by dimension-independent constants, {we obtain}
\BS\label{Eqn:app_SE_base_1}
\begin{align}
(\bm{r}_1,\bm{D} \bm{s}_{1},\bm{d}_M) &\overset{W}{\longrightarrow} \left(\mathsf{R}_1,\mathsf{D}_M\mathsf{S}_{1},\mathsf{D}_M\right),\\
(\bm{s}_1,\bm{D}^\UT \bm{r}_{1},\bm{d}_N) &\overset{W}{\longrightarrow} \left(\mathsf{S}_1,\mathsf{D}_N\mathsf{R}_{1},\mathsf{D}_N\right),
\end{align}
\ES
where $\mathsf{R}_{1}\sim\mathcal{N}(0,\mathbb{E}[\mathsf{M}_1^2])$, $\mathsf{S}_{1}\sim\mathcal{N}(0,\mathbb{E}[\mathsf{Q}_1^2])$ and $(\mathsf{D}_M,\mathsf{D}_N)$ are mutually independent. Recall that
\BS\label{Eqn:app_SE_base_2}
\begin{align}
\tilde{\bm{x}}_1 &= \Psi_1(\bm{d}_M^2)\circ \bm{r}_1+\tilde{\Psi}_1(\bm{d}_M^2)\circ (\bm{Ds}_1),\label{Eqn:app_SE_prove_x1}\\
\tilde{\bm{z}}_1 &= \Phi_1(\bm{d}_N^2)\circ \bm{s}_1+\tilde{\Phi}_1(\bm{d}_N^2)\circ (\bm{D}^\UT\bm{r}_1),
\end{align}
\ES
where $\circ$ denotes Hadamard product. Combining \eqref{Eqn:app_SE_base_1} and \eqref{Eqn:app_SE_base_2}, and further noting that $(\Psi_1,\tilde{\Psi}_1,\Phi_1,\tilde{\Phi}_1)$ are continuous and the entries of $(\bm{d}_M,\bm{d}_N)$ are uniformly bounded, {it is straightforward to show that} (cf.~\cite[Proposition E.2]{fan2022approximate})
\BS\label{Eqn:app_SE_base_3}
\begin{align}
(\bm{r}_1,\tilde{\bm{x}}_1,\bm{d}_M)\overset{W}{\longrightarrow}(\mathsf{R}_1,\tilde{\mathsf{X}}_1,\mathsf{D}_M),\\
(\bm{s}_1,\tilde{\bm{z}}_1,\bm{d}_N)\overset{W}{\longrightarrow}(\mathsf{S}_1,\tilde{\mathsf{Z}}_1,\mathsf{D}_N),
\end{align}
\ES
where the joint distributions of the random variables appeared above are described in Definition \ref{Def:SE}. This proves claim (1.a) and (1.b).

To prove claim (1.c), we identify the conditional distribution of $(\bm{U}_W,\bm{V}_W)$
given $(\bm{r}_1,\bm{s}_1)$, equivalently given the linear constraints 
$\bm{r}_1 = \bm{U}_W^\UT \bm{m}_1$ and $\bm{s}_1 = \bm{V}_W^\UT \bm{q}_1$.
Let $\mathcal{G}_1 \bydef \sigma(\bm{r}_1,\bm{s}_1) $
be the $\sigma$–algebra generated by $(\bm{r}_1,\bm{s}_1)$.
For all sufficiently large $M,N$, by \cite[Lemma 4]{rangan2019vector}~\cite{takeuchi2019rigorous}
there exist independent Haar matrices 
$\tilde{\bm{U}}_W \in \mathbb{O}(M-1)$ and $\tilde{\bm{V}}_W \in \mathbb{O}(N-1)$,
independent of $\mathcal{G}_1$, such that under the conditional law
$\mathbb{P}(\,\cdot\,\mid \mathcal{G}_1)$ we have
\begin{align}
\bm{U}_W &\explain{d}{=} \frac{\bm{m}_1 \bm{r}_1^\UT}{\|\bm{m}_1\|^2}
     + \Pi_{m_1}^\perp \tilde{\bm{U}}_W \Pi_{r_1}^\perp, \\
\bm{V}_W &\explain{d}{=} \frac{\bm{q}_1 \bm{s}_1^\UT}{\|\bm{q}_1\|^2}
     + \Pi_{q_1}^\perp \tilde{\bm{V}}_W \Pi_{s_1}^\perp,
\label{eq:conditional-U-V-base}
\end{align}
where $\Pi_x^\perp$ denotes the orthogonal projector onto $\bm{x}^\perp$. Consequently, still under $\mathbb{P}(\,\cdot\,\mid \mathcal{G}_1)$, the
iterates $\bm{x}_1 = \bm{U}_W \tilde{\bm{x}}_1$ and $\bm{z}_1 = \bm{V}_W\tilde{\bm{z}}_1$ admit the decompositions
\begin{align}
\bm{x}_1
&\explain{d}{=} \frac{\bm{m}_1 \bm{r}_1^\UT \tilde{\bm{x}}_1}{\|\bm{m}_1\|^2}
  + \Pi_{m_1}^\perp \tilde{\bm{U}}_W \Pi_{r_1}^\perp \tilde{\bm{x}}_1, \label{eq:x1-decomp}\\
\bm{z}_1
&\explain{d}{=} \frac{\bm{q}_1 \bm{s}_1^\top \tilde{\bm{z}}_1}{\|\bm{q}_1\|^2}
  + \Pi_{q_1}^\perp \tilde{\bm{V}}_W \Pi_{s_1}^\perp \tilde{\bm{z}}_1. \label{eq:z1-decomp}
\end{align}

% To prove claim (1.c), we derive the conditional law of $(\bm{U}_W,\bm{V}_W)$ given the constraint $\bm{r}_1=\bm{U}_W^\UT\bm{m}_1$ and $\bm{s}_1=\bm{V}_W^\UT\bm{q}_1$. Let $\mathcal{G}_1$ be this event. For all large $M,N$, this is given by \cite[Lemma 4]{rangan2019vector}~\cite{takeuchi2019rigorous}:
% \BE
% \begin{split}
% \bm{U}_W\left.\right|_{\mathcal{G}_1} &\explain{d}{=}\frac{\bm{m}_1\bm{r}_1^\UT}{\|\bm{m}_1\|^2}+\Pi^\perp_{\bm{m}_{1}}\tilde{\bm{U}}_W\Pi^\perp_{\bm{r}_1},\\
% \bm{V}_W\left.\right|_{\mathcal{G}_1} &\explain{d}{=}\frac{\bm{q}_1\bm{s}_1^\UT}{\|\bm{q}_1\|^2}+\Pi^\perp_{\bm{q}_{1}}\tilde{\bm{V}}_W\Pi^\perp_{\bm{s}_1},
% \end{split}
% \EE
% where $\tilde{\bm{U}}_W\in\mathbb{O}(M-1)$ and $\tilde{\bm{V}}_W\in\mathbb{O}(N-1)$ are two independent Haar random matrices, which are further independent of $\mathcal{G}_1$. Therefore, conditional on $\mathcal{G}_1$, the iterates $\bm{x}_1=\bm{U}_W\tilde{\bm{x}}_1$ and $\bm{z}_1=\bm{V}_W_1\tilde{\bm{z}}_1$ are distributed as
% \BS
% \begin{align}
% \bm{x}_1&\explain{d}{=}\frac{\bm{m}_1\bm{r}_1^\UT\tilde{\bm{x}}_1}{\|\bm{m}_1\|^2}+\Pi^\perp_{\bm{m}_{1}}\tilde{\bm{U}}_W\Pi^\perp_{\bm{r}_1}\tilde{\bm{x}}_1,\\
% \bm{z}_1&\explain{d}{=}\frac{\bm{q}_1\bm{s}_1^\UT\tilde{\bm{z}}_1}{\|\bm{q}_1\|^2}+\Pi^\perp_{\bm{q}_{1}}\tilde{\bm{V}}_W\Pi^\perp_{\bm{s}_1}\tilde{\bm{z}}_1.
% \end{align}
% \ES
From claim (1.a), we have
\BS
\begin{align}
\frac{1}{M}\bm{r}_1^\UT\tilde{\bm{x}}_1 & \overset{\ConvMode}{\longrightarrow}\mathbb{E}[\mathsf{R}_1\tilde{\mathsf{X}}_1]\explain{a}{=}0,\\
\frac{1}{M}\bm{s}_1^\UT\tilde{\bm{z}}_1 & \overset{\ConvMode}{\longrightarrow}\mathbb{E}[\mathsf{S}_1\tilde{\mathsf{Z}}_1]\explain{b}{=}0,
\end{align}
\ES
where step (a) and step (b) are a consequence of Lemma \ref{Lem:orthogonality}. Hence, by arguments analogous to those employed in the proof of \cite[Lemma A.4]{fan2022approximate}, we obtain:
\BS
\begin{align}
(\bm{x}_1,\bm{m}_1,\bm{m}_2,\bm{a})&\overset{W}{\longrightarrow}\left(\mathsf{X}_1,\mathsf{M}_1,\mathsf{M}_2,\mathsf{A}\right),\\
(\bm{z}_1,\bm{q}_1,\bm{q}_2,\bm{b})&\overset{W}{\longrightarrow}\left(\mathsf{Z}_1,\mathsf{Q}_1,\mathsf{Q}_2,\mathsf{B}\right),
\end{align}
\ES
where $\mathsf{X}_1\sim\mathcal{N}\left(0,\mathbb{E}\left[\tilde{\mathsf{X}}_1^2\right]\right)$, $\mathsf{M}_2=m_2(\mathsf{X}_1;\mathsf{A})$, $\mathsf{X}_1\indep(\mathsf{M}_1,\mathsf{A})$, and $\mathsf{Z}_1\sim\mathcal{N}\left(0,\mathbb{E}\left[\tilde{\mathsf{Z}}_1^2\right]\right)$, $\mathsf{Q}_2=q_2(\mathsf{Z}_1;\mathsf{B})$, $\mathsf{Z}_1\indep(\mathsf{Q}_1,\mathsf{B})$.
%
%
%Based on standard properties of $W$ convergence of empirical probability measures and Haar random matrix (see \cite[Appendix E and Appendix F.1]{fan2022approximate}), and noting the asymptotic orthogonality properties $\langle\bm{r}_1,\tilde{\bm{x}}_1\rangle\overset{\ConvMode}{\longrightarrow}0$ and $\langle\bm{s}_1,\tilde{\bm{z}}_1\rangle\overset{\ConvMode}{\longrightarrow}0$, we obtain
%\BS
%\begin{align}
%(\bm{x}_1,\bm{a})\overset{W}{\longrightarrow}(\mathsf{X}_1,\mathsf{A}),\\
%(\bm{z}_1,\bm{b})\overset{W}{\longrightarrow}(\mathsf{Z}_1,\mathsf{B}).
%\end{align}
%\ES

\paragraph{Induction step: proof of claim $(t+1.a)$} We shall assume that the claims hold up to $(t.e)$. In what follows, we analyze the distribution of $\bm{r}_{t+1}$. We introduce the following matrix notations for the iterates: 
\[
\bm{M}_{t} \bydef [\bm{m}_1,\ldots,\bm{m}_{t}].
\]
We define the matrices $\bm{X}_{t}$, $\bm{R}_t$, $\tilde{\bm{X}}_{t}$, $\bm{Q}_{t}$, $\bm{Z}_t$, $\bm{S}_t$ and $\tilde{\bm{Z}}_t$ analogously. Using the matrix notations, the OAMP iterates \eqref{Eqn:OAMP_app_full} can be written as follows
\BS
\begin{align}
[\bm{M}_t,\bm{X}_{t}] &= \bm{U}_W[\bm{R}_t,\tilde{\bm{X}}_{t}],\\
[\bm{Q}_t,\bm{Z}_{t}] &=\bm{V}_W[\bm{S}_t,\tilde{\bm{Z}}_{t}].
\end{align}
\ES
Let $\mathcal{G}_t$ be the $\sigma$-algebra generated by the iterates up to $\bm{x}_t$. Note that by claim $(t.a)$ and Lemma \ref{Lem:orthogonality},
\BS\label{Eqn:app_SE_t1_inner}
\begin{align}
\frac{1}{M}
\begin{bmatrix}
{\bm{M}}_{t}^\UT{\bm{M}}_{t} & {\bm{M}}_{t}^\UT{\bm{X}}_t\\
{\bm{X}}_t^\UT\bm{M}_t&  {\bm{X}}_{t}^\UT  {\bm{X}}_{t} 
\end{bmatrix} &\overset{\ConvMode}{\longrightarrow}
\begin{bmatrix}
\paraB{\Omega}_t^u & \bm{0}_{t\times t}\\
\bm{0}_{t\times t}&  \paraB{\Sigma}_t^u
\end{bmatrix},\quad
\frac{1}{M}
\begin{bmatrix}
\bm{M}_t^\UT\bm{m}_{t+1} \\
{\bm{X}}_{t}^\UT\bm{m}_{t+1}
\end{bmatrix}
\overset{\ConvMode}{\longrightarrow}
\begin{bmatrix}
\paraB{\omega}^u_{t+1}\\
\bm{0}_{t\times1}
\end{bmatrix},
\end{align}
where the covariance matrices $\paraB{\Omega}^u_t$ and $\paraB{\Sigma}_t^u$ are defined in Definition \ref{Def:SE}, and $\paraB{\omega}^u_{t+1}$ is defined by
 \begin{align}
\paraB{\omega}^u_{t+1}[i] &\bydef \mathbb{E}\left[{\mathsf{M}}_{i}{\mathsf{M}}_{t+1}\right],\quad\forall i\in[t].
\end{align}
\ES
Moreover, $\paraB{\Omega}_t^u$ and $\paraB{\Sigma}_t^u$ are invertible by assumption. Hence, the following matrix is invertible for all sufficiently large $M,N$:
\[
\begin{bmatrix}
{\bm{M}}_{t}^\UT{\bm{M}}_{t} & {\bm{M}}_{t}^\UT{\bm{X}}_t\\
{\bm{X}}_t^\UT\bm{M}_t&  {\bm{X}}_{t}^\UT  {\bm{X}}_{t} 
\end{bmatrix} .
\]
By Lemma 4 in \cite{rangan2019vector} and \cite{takeuchi2019rigorous},
the conditional laws of $\bm{U}_W$ and $\bm{V}_W$ given $\mathcal{G}_t$ can be represented as
\BS \label{eq:conditioning-haar}
\begin{align}
\bm{U}_W &\explain{d}{=} [\bm{M}_t,\bm{X}_{t}] \begin{bmatrix}
{\bm{M}}_{t}^\UT{\bm{M}}_{t} & {\bm{M}}_{t}^\UT{\bm{X}}_t\\
{\bm{X}}_t^\UT\bm{M}_t&  {\bm{X}}_{t}^\UT  {\bm{X}}_{t} 
\end{bmatrix}
^{-1}
\begin{bmatrix}
\bm{R}_t^\UT \\
\tilde{\bm{X}}_{t}^\UT
\end{bmatrix}
+\Pi^\perp_{[\bm{M}_{t},\bm{X}_t]}\tilde{\bm{U}}_W\Pi^\perp_{[{\bm{R}}_{t},\tilde{\bm{X}}_t]},\\
%--------------------------------------------------------------------------------
\bm{V}_W& \explain{d}{=}[\bm{Q}_t,\bm{Z}_{t}]
\begin{bmatrix}
 \bm{Q}_t^\UT\bm{Q}_t & \bm{Q}_t^\UT {\bm{Z}}_{t}\\
{\bm{Z}}_{t}^\UT\bm{Q}_t & {\bm{Z}}_{t}^\UT{\bm{Z}}_{t}
\end{bmatrix}
^{-1}
\begin{bmatrix}
\bm{S}_t^\UT \\
\tilde{\bm{Z}}_{t}^\UT
\end{bmatrix}
+\Pi^\perp_{[\bm{Q}_{t},\bm{Z}_t]}\tilde{\bm{V}}_W\Pi^\perp_{[{\bm{S}}_{t},\tilde{\bm{Z}}_t]},
\end{align}
\ES
where $\tilde{\bm{U}}_W\in\mathbb{O}(M-2t)$ and $\tilde{\bm{V}}_W\in\mathbb{O}(N-2t)$ are Haar-distributed orthogonal matrices, which are mutually independent and independent of $\mathcal{G}_t$. $\Pi^\perp_{[\bm{M}_{t},\bm{X}_t]}\in\mathbb{R}^{M\times(M-2t)}$ is a matrix whose columns form an orthonormal basis for $(\text{col}[\bm{M}_{t},\bm{X}_t])^\perp$. Other projection matrices are defined analogously. Hence, conditional on $\mathcal{G}_t$, the iterate $\bm{r}_{t+1}=\bm{U}_W^\UT\bm{m}_{t+1}$ can be written as
\BS
\BE
\bm{r}_{t+1} \explain{d}{=} \bm{r}_{t+1}^{||}+ \bm{r}_{t+1}^{\perp},
\EE
with
\begin{align}
\bm{r}_{t+1}^{||} 
&\bydef[\bm{R}_t,\tilde{\bm{X}}_t]
\begin{bmatrix}
{\bm{M}}_{t}^\UT{\bm{M}}_{t} & {\bm{M}}_{t}^\UT{\bm{X}}_t\\
{\bm{X}}_t^\UT\bm{M}_t&  {\bm{X}}_{t}^\UT  {\bm{X}}_{t} 
\end{bmatrix}
^{-1}
\begin{bmatrix}
\bm{M}_t^\UT\bm{m}_{t+1} \\
{\bm{X}}_{t}^\UT\bm{m}_{t+1}
\end{bmatrix},\label{Eqn:r_parallel}\\
\bm{r}_{t+1}^{\perp}  & \bydef  \Pi^\perp_{[{\bm{R}}_{t},\tilde{\bm{X}}_t]}\tilde{\bm{U}}_W^\UT\Pi^\perp_{[\bm{M}_{t},\bm{X}_t]}\bm{m}_{t+1}\label{Eqn:r_perp}.
\end{align}
\ES
From \eqref{Eqn:app_SE_t1_inner} and the asymptotic orthogonality stated in Lemma \ref{Lem:orthogonality}, and by arguments analogous to those in the proof of \cite[Lemma A.4]{fan2022approximate}, one obtains:
\BS\label{Eqn:r_para_app}
\begin{align}
\bm{r}_{t+1}^{||} & \overset{W}{\longrightarrow}\mathsf{R}_{t+1}^{||}\bydef[\mathsf{R}_1,\ldots,\mathsf{R}_t](\paraB{\Omega}_t^u)^{-1}\paraB{\omega}^u_{t+1},\\
\bm{r}_{t+1}^{\perp} &\overset{W}{\longrightarrow} R_{t+1}^\perp\sim\mathcal{N}\left(0,\mathbb{E}[(\mathsf{M}_{t+1}^\perp)^2]\right),
\end{align}
\ES
with $\mathsf{R}_{t+1}^\perp$ is independent of $(\mathsf{R}_1,\ldots,\mathsf{R}_t)$, and $\mathsf{M}_{t+1}^\perp$ denotes the projection (in the $L_2$ sense) of $\mathsf{M}_{t+1}$ onto the orthogonal space of $\text{span}(\mathsf{M}_1,\ldots,\mathsf{M}_t)$: $\mathsf{M}_{t+1}^\perp\bydef  \Pi_{(\mathsf{M}_1,\ldots,\mathsf{M}_t)}^\perp\left(\mathsf{M}_{t+1}\right)$.
The variance of $R_{t+1}^\perp$ can be further expressed as $\mathbb{E}[(\mathsf{R}_{t+1}^\perp)^2] =\mathbb{E}[(\mathsf{M}_{t+1}^\perp)^2] =\mathbb{E}[\mathsf{M}_{t+1}^2] -\left(\paraB{\omega}^u_{t+1}\right)^\UT (\paraB{\Omega}_t^u)^{-1}\paraB{\omega}^u_{t+1}
$. One can further obtain the convergence of the joint empirical law of $(\bm{r}_{\le t},\bm{r}_{t+1},\bm{d}_M)$ based on the same reasoning as those in \cite{fan2022approximate}. First, by induction hypothesis
\BE
(\bm{r}_{\le t},\bm{d}_M)\overset{W}{\longrightarrow}\left(\mathsf{R}_{\le t},\mathsf{D}_M\right).
\EE
Note that $\bm{r}_{t+1}^{||}$ is a linear transform of $(\bm{r},\ldots,\bm{r}_t,\bm{d}_M)$ up to an error term that vanish in $W_p$ for every $p\ge1$. Applying \cite[Proposition E.4]{fan2022approximate} yields
\BE
(\bm{r}_{\le t},\bm{r}_{t+1}^{||},\bm{d}_M)\overset{W}{\longrightarrow}\left(\mathsf{R}_{\le t},\mathsf{R}_{t+1}^{||},\mathsf{D}_M\right).
\EE
Using \cite[Proposition F.2]{fan2022approximate}, we obtain
\BE
(\bm{r}_{\le t},\bm{r}_{t+1}^{||}+\bm{r}_{t+1}^{\perp},\bm{d}_M)\overset{W}{\longrightarrow}\left(\mathsf{R}_{\le t},\mathsf{R}_{t+1}^{||}+\mathsf{R}_{t+1}^{\perp},\mathsf{D}_M\right).
\EE
Let $\mathsf{R}_{t+1}\bydef\mathsf{R}_{t+1}^{||}+\mathsf{R}_{t+1}^\perp$. It is straightforward to check that $\mathsf{R}_{\le t+1}\sim\mathcal{N}(\bm{0}_{t+1},\paraB{\Omega}_{t+1}^u)$. We can apply exactly the same arguments to $\bm{s}_{t+1}$. Note that the fresh Haar random matrix $\tilde{\bm{V}}_W$ is independent of $\tilde{\bm{U}}_W$. We can repeat the above reasoning to conclude that
\BE
(\bm{r}_{\le t},\bm{r}_{t+1},\bm{D}\bm{s}_{t+1},\bm{d}_M)\overset{W}{\longrightarrow}\left(\mathsf{R}_{\le t},\mathsf{R}_{t+1},\mathsf{D}_M\mathsf{S}_{t+1},\mathsf{D}_M\right).
\EE
The analysis of the dimension-$N$ vectors are similar. The proof of claim $(t+1.a)$ is complete.

\paragraph{Induction step: proof of claim $(t+1.b)$} 

To prove claim $(t+1.b)$, we note that the map from the rows of $(\bm{r}_{\le t+1},\bm{Ds}_{\le {t+1}},\bm{d}_N)$ to those of $(\bm{r}_{\le t+1},\tilde{\bm{x}}_{\le t+1},\bm{d}_M)$ is polynomially bounded; cf.~\eqref{Eqn:app_SE_prove_x1}. Then, applying \cite[Proposition E.2]{fan2022approximate} together with claim $(t+1.a)$ shows that the joint empirical law of $(\bm{r}_{\le t+1},\tilde{\bm{x}}_{\le t+1},\bm{d}_M)$ converges. The analysis of $(\bm{s}_{\le t+1},\tilde{\bm{z}}_{\le t+1},\bm{d}_N)$ is similar.

\paragraph{Induction step: proof of claim $(t+1.c)$} 

To analyze $\bm{x}_{t+1}$ and $\bm{z}_{t+1}$, we derive the law of $\bm{U}_W$ and $\bm{V}_W$ conditional on the OAMP iterates up to $\bm{r}_{t+1}$ and $\bm{s}_{t+1}$, i.e.,
\BS
\begin{align}
[\bm{M}_{t+1},\bm{X}_{t}] &= \bm{U}_W[\bm{R}_{t+1},\tilde{\bm{X}}_{t}],\\
[\bm{Q}_{t+1},\bm{Z}_{t}] &=\bm{V}_W[\bm{S}_{t+1},\tilde{\bm{Z}}_{t}].
\end{align}
\ES
Let $\mathcal{G}_t^+$ be the $\sigma$-algebra generated by the iterates up to $\bm{r}_{t+1}$ and $\bm{s}_{t+1}$. The conditional law of $\bm{U}_W$ and $\bm{V}_W$ for large $M,N$ are given by
\BS
\begin{align}
\bm{U}_W&\explain{d}{=} [\bm{M}_{t+1},\bm{X}_{t}] \begin{bmatrix}
{\bm{M}}_{t+1}^\UT{\bm{M}}_{t+1} & {\bm{M}}_{t+1}^\UT{\bm{X}}_t\\
{\bm{X}}_t^\UT\bm{M}_{t+1}&  {\bm{X}}_{t}^\UT  {\bm{X}}_{t} 
\end{bmatrix}
^{-1}
\begin{bmatrix}
\bm{R}_{t+1}^\UT \\
\tilde{\bm{X}}_{t}^\UT
\end{bmatrix}
+\Pi^\perp_{[\bm{M}_{t+1},\bm{X}_t]}\tilde{\bm{U}}_W\Pi^\perp_{[{\bm{R}}_{t+1},\tilde{\bm{X}}_t]},\\
%--------------------------------------------------------------------------------
\bm{V}_W& \explain{d}{=}[\bm{Q}_{t+1},\bm{Z}_{t}]
\begin{bmatrix}
 \bm{Q}_{t+1}^\UT\bm{Q}_{t+1} & \bm{Q}_{t+1}^\UT {\bm{Z}}_{t}\\
{\bm{Z}}_{t}^\UT\bm{Q}_{t+1} & {\bm{Z}}_{t}^\UT{\bm{Z}}_{t}
\end{bmatrix}
^{-1}
\begin{bmatrix}
\bm{S}_{t+1}^\UT \\
\tilde{\bm{Z}}_{t}^\UT
\end{bmatrix}
+\Pi^\perp_{[\bm{Q}_{t+1},\bm{Z}_t]}\tilde{\bm{V}}_W\Pi^\perp_{[{\bm{S}}_{t+1},\tilde{\bm{Z}}_t]},
\end{align}
\ES
Hence, the conditional law of $\bm{x}_{t+1}=\bm{U}_W\tilde{\bm{x}}_{t+1}$ is
\BS
\begin{equation}
\bm{x}_{t+1}\explain{d}{=}\bm{x}_{t+1}^{||}+\bm{x}_{t+1}^{\perp},
\end{equation}
with
\begin{align}
\bm{x}_{t+1}^{||} &\bydef 
[\bm{M}_{t+1},\bm{X}_{t}] \begin{bmatrix}
{\bm{M}}_{t+1}^\UT{\bm{M}}_{t+1} & {\bm{M}}_{t+1}^\UT{\bm{X}}_t\\
{\bm{X}}_t^\UT\bm{M}_{t+1}&  {\bm{X}}_{t}^\UT  {\bm{X}}_{t} 
\end{bmatrix}
^{-1}
\begin{bmatrix}
\bm{R}_{t+1}^\UT\tilde{\bm{x}}_{t+1} \\
\tilde{\bm{X}}_{t}^\UT\tilde{\bm{x}}_{t+1}
\end{bmatrix},\\
\bm{x}_{t+1}^{\perp} &\bydef \Pi^\perp_{[\bm{M}_{t+1},\bm{X}_t]}\tilde{\bm{U}}_W\Pi^\perp_{[{\bm{R}}_{t+1},\tilde{\bm{X}}_t]}\tilde{\bm{x}}_{t+1}.
\end{align}
\ES
%Similar to the analysis of $\bm{r}_{t+1}^{\parallel}$ and $\bm{r}_{t+1}^{\perp}$ (see \eqref{Eqn:r_parallel}-\eqref{Eqn:r_perp}), and noting in particular the following orthogonality (which follows from the induction hypothesis):
From claim $(t+1.b)$, and appealing to the orthogonality properties in Lemma \ref{Lem:orthogonality}, we obtain
\BS
\begin{align}
\begin{bmatrix}
{\bm{M}}_{t+1}^\UT{\bm{M}}_{t+1} & {\bm{M}}_{t+1}^\UT{\bm{X}}_t\\
{\bm{X}}_t^\UT\bm{M}_{t+1}&  {\bm{X}}_{t}^\UT  {\bm{X}}_{t} 
\end{bmatrix}
^{-1}
\begin{bmatrix}
\bm{R}_{t+1}^\UT\tilde{\bm{x}}_{t+1} \\
\tilde{\bm{X}}_{t}^\UT\tilde{\bm{x}}_{t+1}
\end{bmatrix}
&\overset{\ConvMode}{\longrightarrow}
\begin{bmatrix}
\paraB{\Omega}_{t+1}^u & \bm{0}_{(t+1)\times t}\\
\bm{0}_{(t+1)\times t}&  \paraB{\Sigma}_{t}^u 
\end{bmatrix}
^{-1}
\begin{bmatrix}
\bm{0}_{(t+1)\times 1} \\
\paraB{\sigma}^u_{t+1}
\end{bmatrix}\\
&=
\begin{bmatrix}
\bm{0}_{(t+1)\times 1} \\
(\paraB{\Sigma}_t^u)^{-1}\paraB{\sigma}^u_{t+1}
\end{bmatrix},
\end{align}
\ES
where $\paraB{\sigma}^u_{t+1}[i]\bydef\mathbb{E}\left[\tilde{\mathsf{X}}_i\tilde{\mathsf{X}}_{t+1}\right]$ and $\paraB{\Sigma}_t^u[i,j]\bydef \mathbb{E}\left[\tilde{\mathsf{X}}_i\tilde{\mathsf{X}}_{j}\right]$, $\forall i,j\in[t]$. Based on claim $(t.c)$ and similar to the analysis of $\bm{r}_{t+1}$, we obtain
\BS
\BE\label{Eqn:x_para_app}
(\bm{x}_{\le t},\bm{x}_{t+1}^{||},\bm{x}_{t+1}^\perp,\bm{a})\overset{W}{\longrightarrow}(\mathsf{X}_{\le t},\mathsf{X}_{t+1}^{\parallel},\mathsf{X}_{t+1}^{\perp},\mathsf{A}),
\EE
where 
\begin{align}
\mathsf{X}_{t+1}^{\parallel}&=[\mathsf{X}_1,\ldots,\mathsf{X}_t] (\paraB{\Sigma}_t^u)^{-1}\paraB{\sigma}^u_{t+1},\\
\mathsf{X}_{t+1}^{\perp}&\sim\mathcal{N}\left(0,\mathbb{E}\left[\mathsf{X}_{t+1}^2\right]-\left(\paraB{\sigma}^u_{t+1}\right)^\UT(\paraB{\Sigma}_t^u)^{-1}\paraB{\sigma}^u_{t+1}\right).
\end{align}
\ES
Moreover, $\mathsf{X}_{t+1}^{\perp}$ is independent of $(\mathsf{X}_1,\ldots,\mathsf{X}_t,\mathsf{A})$. Finally, since $\bm{m}_{t+1}$ is a Lipschitz continuous function of $(\bm{x}_{\le t},\bm{a})$, applying \cite[Proposition E.2]{fan2022approximate} and using the induction hypothesis $(t.c)$ yields
\BE
(\bm{x}_{\le t+1},\bm{m}_{\le t+2},\bm{a})\overset{W}{\longrightarrow}(\mathsf{X}_{\le t+1},\mathsf{M}_{\le t+2},\mathsf{A}),
\EE
where the state evolution random variable appeared on the above equation are distributed as described in Definition \ref{Def:SE}. The analysis of $(\bm{z}_{\le t+1},\bm{q}_{t+2};\bm{b})$ is similar and omitted.

%================================================================
\section{State Evolution of OAMP for Spiked Models (\thref{Thm: State Evolution})}\label{Sec:OAMP_SE_proof}

% \jjm{We adopt the proof strategy outlined in \cite{dudeja2024optimality}: we transform the OAMP iteration \eqref{def:OAMP}, which depends on the signal matrix $\bm{Y}$, into an asymptotically equivalent OAMP iteration that only depends on the random matrix $\bm{W}$. The resulting algorithm admits a state evolution characterization, which can be derived using established conditioning techniques \cite{rangan2019vector,takeuchi2019rigorous,fan2022approximate}. }

% \paragraph{Polynomial Approximation.} 
% Consider the OAMP algorithm of the general form in \eqref{eq:OAMP Algo u} and \eqref{eq:OAMP Algo v}:
% \BS
% \begin{align}
%  \bm{u}_t &= F_{t}(\bm{Y}\bm{Y}^\UT) f_{t}(\bm{u}_{1},...,\bm{u}_{t-1}; \bm{a}) +\tilde{F}_{t}(\bm{Y}\bm{Y}^\UT)\bm{Y} g_{t}(\bm{v}_{1},...,\bm{v}_{t-1}; \bm{a}) , \label{eq:OAMP Algo u}\\
%  \bm{v}_t &=  G_{t}(\bm{Y}^\UT\bm{Y}) g_{t}(\bm{v}_{1},...,\bm{v}_{t-1}; \bm{a})+\tilde{G}_{t}(\bm{Y}^\UT\bm{Y})\bm{Y}^\UT f_{t}(\bm{u}_{1},...,\bm{u}_{t}; \bm{a}).\label{eq:OAMP Algo v}
% \end{align}
% \ES

Recall that the OAMP algorithm for spiked matrix models consists of the following iterations (Definition \ref{def:OAMP})
\BS
\begin{align}
\bm{u}_t &= F_{t}(\bm{Y}\bm{Y}^\UT) f_{t}(\bm{u}_{<t}; \bm{a}) + \tilde{F}_{t}(\bm{Y}\bm{Y}^\UT)\bm{Y}g_{t}(\bm{v}_{<t}; \bm{b}), \\
\bm{v}_t &= G_{t}(\bm{Y}^\UT\bm{Y}) g_{t}(\bm{v}_{<t}; \bm{b}) + \tilde{G}_{t}(\bm{Y}^\UT\bm{Y})\bm{Y}^\UT f_{t}(\bm{u}_{<t}; \bm{a}).
\end{align}
\ES
A major difficulty in analyzing the above OAMP algorithm is that the matrix denoisers act on the observation matrix~$\bm{Y}$ rather than the random matrix~$\bm{W}$. Our strategy for proving \thref{Thm: State Evolution} parallels that of~\cite[Theorem~1]{dudeja2024optimality} and proceeds through the following steps:
\begin{enumerate}
\item We approximate the matrix denoisers in the OAMP algorithm by polynomial functions, which is justified by the Weierstrass approximation theorem.
\item The OAMP algorithm with polynomial matrix denoisers acting on~$\bm{Y}$ can be reformulated as an auxiliary OAMP algorithm that depends only on~$\bm{W}$, whose dynamics are characterized by existing results (cf.~Theorem~\ref{Th:OAMP_general_SE_app}).
\end{enumerate}
%The remainder of this appendix is organized as follows. Section~\ref{Sec:polynomial} summarizes results on polynomial approximation. Section~\ref{Sec:auxiliary} presents the auxiliary OAMP algorithm based on~$\bm{W}$. The main proof of Theorem~\ref{Thm: State Evolution} is provided in Section~\ref{Sec:main_proof_SE_sub}, and the final sections contain the proofs of intermediate results.

\paragraph{Polynomial Approximation}
Following the approach in \cite[Lemma~5]{dudeja2024optimality}, 
we can assume that the matrix denoisers 
\(F_{t}, \tilde{F}_{t}, G_{t}, \tilde{G}_{t}\) 
are polynomial functions, which is justified by the Weierstrass approximation theorem. 
The result is formalized in the following lemma, whose proof—being analogous to that of 
\cite[Lemma~5]{dudeja2024optimality}—is omitted.

\begin{lemma}\label{lem: poly approx}
It is sufficient to prove \thref{Thm: State Evolution} under the additional assumption that for each $t \in \N$, the
matrix denoisers $F_{t}(\cdot), \tilde{F}_{t}(\cdot), G_{t}(\cdot), \tilde{G}_{t}(\cdot): \R \to \R$  are polynomials.   
\end{lemma}

% \paragraph{Orthogonal Decomposition.} 

To analyze the behavior of these iterations, we decompose the functions \( f_t \) and \( g_t \) into two components: one that is aligned with the ground-truth signal and an orthogonal residual. Specifically, we write:
\BS\label{eq: ort.decomp.ap}
\begin{equation}
f_t(\bm{u}_{<t}; \bm{a}) = \alpha_t \bm{u}_* + \bm{f}_t^\perp, \quad g_t(\bm{v}_{<t}; \bm{b}) = \beta_t \bm{v}_* + \bm{g}_t^\perp.
\end{equation}
The signal alignment parameters, \( \alpha_t \) and \( \beta_t \), and the residual vectors, \( \bm{f}_t^\perp \) and \( \bm{g}_t^\perp \), are defined as follows:
\begin{align}
  \alpha_t & \bydef \mathbb{E} [\mathsf{U}_\ast f_t(\mathsf{U}_{<t};\mathsf{A})], \quad \bm{f}_t^\perp \bydef f_t(\bm{u}_{<t}; \bm{a}) - \alpha_t \bm{u}_*, \\
  \beta_t  & \bydef \mathbb{E} [\mathsf{V}_\ast g_t(\mathsf{V}_{<t};\mathsf{B})], \quad \bm{g}_t^\perp \bydef g_t(\bm{v}_{<t}; \bm{b}) - \beta_t  \bm{v}_*,
\end{align}
\ES
where $(\mathsf{U}_\ast,\mathsf{V}_\ast,\mathsf{U}_{<t},\mathsf{V}_{<t})$ are state evolution random variables defined in Section \ref{Sec:OAMP_main}.
%By the orthogonal construction together with the independency in \assumpref{assump:main}, we have:
%\BS \label{eq: residue-signal perp}
%\begin{align}
%\lim_{M\to \infty} \frac{1}{M} \langle\bm{u}_*,\bm{f}_t^\perp\rangle & = 0, \lim_{M\to \infty} \frac{1}{M} \langle\bm{W}^\UT\bm{u}_*,\bm{g}_t^\perp\rangle  = 0,\\
%\lim_{N\to \infty} \frac{1}{N} \langle\bm{v}_*,\bm{g}_t^\perp\rangle & = 0, \lim_{N\to \infty} \frac{1}{N} \langle\bm{W}\bm{v}_*,\bm{f}_t^\perp\rangle  = 0,
%\end{align}
%\ES
Substituting these decompositions into the update rules yields the following expressions for \( \bm{u}_t \) and \( \bm{v}_t \):
\BS\label{eq:original-OAMP-decomp}
\begin{align}
\bm{u}_t &= \alpha_t F_{t}(\bm{Y}\bm{Y}^\UT) \bm{u}_* + \beta_t \tilde{F}_{t}(\bm{Y}\bm{Y}^\UT) \bm{Y} \bm{v}_* + F_{t}(\bm{Y}\bm{Y}^\UT) \bm{f}_t^\perp + \tilde{F}_{t}(\bm{Y}\bm{Y}^\UT) \bm{Y} \bm{g}_t^\perp, \label{eq:original-OAMP-decomp-u} \\
\bm{v}_t &= \beta_t G_{t}(\bm{Y}^\UT\bm{Y}) \bm{v}_* + \alpha_t \tilde{G}_{t}(\bm{Y}^\UT\bm{Y}) \bm{Y}^\UT \bm{u}_* + G_{t}(\bm{Y}^\UT\bm{Y}) \bm{g}_t^\perp + \tilde{G}_{t}(\bm{Y}^\UT\bm{Y}) \bm{Y}^\UT \bm{f}_t^\perp. \label{eq:original-OAMP-decomp-v}
\end{align}
\ES
A key challenge in analyzing this expression is that the matrix 
\( \bm{Y}\bm{Y}^\UT \) is not rotationally invariant. 
The following lemma, which parallels \cite[Lemma~6]{dudeja2024optimality}, 
provides a crucial tool for addressing this issue by relating the terms 
to expressions involving the rotationally invariant matrix 
\( \bm{W}\bm{W}^\UT \). Its proof is deferred to Section \ref{Sec:proof_lemma_aux1}.

\begin{lemma}\label{lem:aux1}
Let \( F, \tilde{F}, G, \tilde{G}: \mathbb{R} \to \mathbb{R} \) be dimension-independent polynomials.
\begin{enumerate}
    \item There exist polynomial functions \( \Psi^u, \tilde{\Psi}^u, \Psi^v, \tilde{\Psi}^v: \mathbb{R} \mapsto \mathbb{R} \) associated with \( F, \tilde{F}, G, \tilde{G}: \mathbb{R} \to \mathbb{R} \)  such that the following asymptotic equivalences hold:
    \BS\label{Eqn:lemma6_1}
    \begin{align}
    F(\bm{Y}\bm{Y}^\UT) \bm{u}_* & \explain{$M \rightarrow \infty$}{\simeq} \Psi^u(\bm{W}\bm{W}^\UT) \bm{u}_{*} + \Psi^v(\bm{W}\bm{W}^\UT) \bm{W}\bm{v}_{*}, \label{Eqn:lemma6_a}\\
    \tilde{F}(\bm{Y}\bm{Y}^\UT) \bm{Y} \bm{v}_* & \explain{$M \rightarrow \infty$}{\simeq} \tilde{\Psi}^v(\bm{W}\bm{W}^\UT)\bm{W} \bm{v}_{*} + \tilde{\Psi}^u(\bm{W}\bm{W}^\UT)\bm{u}_{*},\label{Eqn:lemma6_b}\\
    G(\bm{Y}^\UT\bm{Y}) \bm{v}_* & \explain{$N \rightarrow \infty$}{\simeq} \Phi^v(\bm{W}^\UT\bm{W}) \bm{v}_{*} + \Phi^u(\bm{W}^\UT\bm{W}) \bm{W}^\UT \bm{u}_{*}, \label{Eqn:lemma6_c}\\
    \tilde{G}(\bm{Y}^\UT\bm{Y}) \bm{Y}^\UT \bm{u}_* & \explain{$N \rightarrow \infty$}{\simeq} \tilde{\Phi}^u(\bm{W}^\UT\bm{W})\bm{W}^\UT\bm{u}_{*} + \tilde{\Phi}^v(\bm{W}^\UT\bm{W})\bm{v}_{*},\label{Eqn:lemma6_d}
    \end{align}
    \ES
where \( \explain{$M \rightarrow \infty$}{\simeq} \) denotes asymptotic equivalence between random vectors as defined in \defref{def:eq}. 
    \item Let $\bm{u}\in\mathbb{R}^M$ and $\bm{v}\in\mathbb{R}^N$ be two random vectors such that the following hold for all $i\in\mathbb{N}\cup\{0\}$:
%   \BS\label{Eqn:lemma6_app_orthgonality}
%   \begin{align}
%      \frac{\left\langle \bm{u}_\ast, (\bm{WW}^\UT)^{i}\bm{u}\right\rangle}{M}\overset{a.s.}{\longrightarrow}0,\qquad \frac{\left\langle\bm{u}_\ast,(\bm{WW}^\UT)^i\bm{W}\bm{v}\right\rangle}{M}\overset{a.s.}{\longrightarrow}0,\\
%      \frac{\left\langle \bm{v}_\ast, (\bm{W}^\UT\bm{W})^{i}\bm{v}\right\rangle}{N}\overset{a.s.}{\longrightarrow}0,\qquad \frac{\left\langle\bm{v}_\ast,(\bm{W}^\UT\bm{W})^i\bm{W}^\UT\bm{u}\right\rangle}{N}\overset{a.s.}{\longrightarrow}0.
%   \end{align}
%   \ES
%   \BS\label{Eqn:lemma6_app_orthgonality}
%   \begin{align}
%     & \frac{\left\langle \bm{u}_\ast, (\bm{WW}^\UT)^{i}\bm{u}\right\rangle}{M}\overset{a.s.}{\longrightarrow}0,\quad \frac{\left\langle \bm{v}_\ast, (\bm{W}^\UT\bm{W})^{i}\bm{v}\right\rangle}{N}\overset{a.s.}{\longrightarrow}0,\\
%  &    \frac{\left\langle\bm{u}_\ast,(\bm{WW}^\UT)^i\bm{W}\bm{v}\right\rangle}{M}= \frac{\left\langle\bm{u}_\ast,\bm{W}(\bm{W}^\UT\bm{W})^i\bm{v}\right\rangle}{M}\overset{a.s.}{\longrightarrow}0,\\
%&\frac{\left\langle\bm{v}_\ast,(\bm{W}^\UT\bm{W})^i\bm{W}^\UT\bm{u}\right\rangle}{N}=\frac{\left\langle\bm{v}_\ast,\bm{W}^\UT(\bm{W}\bm{W}^\UT)^i\bm{u}\right\rangle}{N}\overset{a.s.}{\longrightarrow}0.
%   \end{align}
%   \ES
   \BS\label{Eqn:lemma6_app_orthgonality}
   \begin{align}
     & \frac{\left\langle \bm{u}_\ast, (\bm{WW}^\UT)^{i}\bm{u}\right\rangle}{M}\overset{a.s.}{\longrightarrow}0,\quad  \frac{\left\langle\bm{u}_\ast,(\bm{WW}^\UT)^i\bm{W}\bm{v}\right\rangle}{M}= \frac{\left\langle\bm{u}_\ast,\bm{W}(\bm{W}^\UT\bm{W})^i\bm{v}\right\rangle}{M}\overset{a.s.}{\longrightarrow}0,\\
&\frac{\left\langle \bm{v}_\ast, (\bm{W}^\UT\bm{W})^{i}\bm{v}\right\rangle}{N}\overset{a.s.}{\longrightarrow}0,\quad \frac{\left\langle\bm{v}_\ast,(\bm{W}^\UT\bm{W})^i\bm{W}^\UT\bm{u}\right\rangle}{N}=\frac{\left\langle\bm{v}_\ast,\bm{W}^\UT(\bm{W}\bm{W}^\UT)^i\bm{u}\right\rangle}{N}\overset{a.s.}{\longrightarrow}0.
   \end{align}
   \ES
Then, the following asymptotic equivalence holds
\BS
\begin{align}
    F(\bm{Y}\bm{Y}^\UT) \bm{u} & \explain{$M\rightarrow \infty$}{\simeq} F(\bm{W}\bm{W}^\UT)\bm{u}, & \tilde{F}(\bm{Y}\bm{Y}^\UT) \bm{Y} \bm{v}& \explain{$M\rightarrow \infty$}{\simeq} \tilde{F}(\bm{W}\bm{W}^\UT) \bm{W} \bm{v}, \\
    G(\bm{Y}^\UT\bm{Y}) \bm{v} & \explain{$N\rightarrow \infty$}{\simeq} G(\bm{W}^\UT\bm{W})\bm{v}, & \tilde{G}(\bm{Y}^\UT\bm{Y}) \bm{Y}^\UT \bm{u}& \explain{$N\rightarrow \infty$}{\simeq} \tilde{G}(\bm{W}^\UT\bm{W}) \bm{W}^\UT \bm{u}.
\end{align}
\ES
\end{enumerate}
\end{lemma}

Importantly, Lemma \ref{lem:aux1} shows that as long as the random vectors $\bm{u}\in\mathbb{R}^M$ and $\bm{v}\in\mathbb{R}^N$ satisfy the orthogonality conditions \eqref{Eqn:lemma6_app_orthgonality}, then they do not interact with the signal components in $\bm{Y}$. Later we shall show that the component $\bm{f}_t^\perp\bydef f_t(\bm{u}_{<t};\bm{a})-\mathbb{E}[\mathsf{U}_\ast f_t(\mathsf{U}_{<t};\mathsf{A})]\cdot\bm{u}_\ast$ and $\bm{g}_t^\perp\bydef g_t(\bm{v}_{<t};\bm{b})-\mathbb{E}[\mathsf{V}_\ast g_t(\mathsf{V}_{<t};\mathsf{B})]\cdot\bm{v}_\ast$ satisfy these orthogonality conditions. Combined with Lemma~\ref{lem:aux1}, this would yield the following asymptotic equivalence:
\BS\label{Eqn:lemma6_2}
\begin{align}
F_t(\bm{Y}\bm{Y}^\UT) \bm{f}_t^\perp({\bm{u}}_{< t};\bm{a}) & \explain{$M\rightarrow \infty$}{\simeq}F_t(\bm{W}\bm{W}^\UT) \bm{f}_t^\perp({\bm{u}}_{< t};\bm{a}) ,\\
 \widetilde{F}_{t}(\bm{Y}\bm{Y}^\UT) \bm{Y} \bm{g}_t^\perp({\bm{v}}_{< t};\bm{b})  & \explain{$M\rightarrow \infty$}{\simeq}\widetilde{F}_{t}(\bm{W}\bm{W}^\UT) \bm{W} \bm{g}_t^\perp({\bm{v}}_{< t};\bm{b}),\\
G_t(\bm{Y}^\UT\bm{Y}) \bm{g}_t^\perp({\bm{v}}_{< t};\bm{b}) & \explain{$M\rightarrow \infty$}{\simeq} G_t(\bm{W}^\UT\bm{W}) \bm{g}_t^\perp({\bm{v}}_{< t};\bm{b}) ,\\
 \widetilde{G}_{t}(\bm{Y}^\UT\bm{Y}) \bm{Y}^\UT \bm{f}_t^\perp({\bm{u}}_{< t};\bm{a}) & \explain{$M\rightarrow \infty$}{\simeq}\widetilde{G}_{t}(\bm{W}^\UT\bm{W}) \bm{W}^\UT \bm{f}_t^\perp({\bm{u}}_{< t};\bm{a}).
\end{align}
\ES

\paragraph{Auxiliary OAMP Algorithm}\label{Sec:auxiliary}

By replacing the terms in the original OAMP algorithm \eqref{eq:original-OAMP-decomp} using the corresponding asymptotic equivalence as established in \eqref{Eqn:lemma6_1} and \eqref{Eqn:lemma6_2}, we introduce the following auxiliary OAMP algorithm: 
\BS\label{eq:aux-OAMP}
\begin{align}
\tilde{\bm{u}}_t &= \Psi_t(\bm{W}\bm{W}^\UT) \bm{u}_* + \widetilde{\Psi}_{t}(\bm{W}\bm{W}^\UT) \bm{W} \bm{v}_*  + F_t(\bm{W}\bm{W}^\UT) \bm{f}_t^\perp(\tilde{\bm{u}}_{< t};\bm{a}) + \widetilde{F}_{t}(\bm{W}\bm{W}^\UT) \bm{W} \bm{g}_t^\perp(\tilde{\bm{v}}_{< t};\bm{b}), \\
\tilde{\bm{v}}_t & = \Phi_t(\bm{W}^\UT\bm{W}) \bm{v}_* + \widetilde{\Phi}_{t}(\bm{W}^\UT\bm{W}) \bm{W}^\UT \bm{u}_* +G_t(\bm{W}^\UT\bm{W}) \bm{g}_t^\perp(\tilde{\bm{v}}_{< t};\bm{b}) +\widetilde{G}_{t}(\bm{W}^\UT\bm{W}) \bm{W}^\UT \bm{f}_t^\perp(\tilde{\bm{u}}_{< t};\bm{a}),
\end{align}
\ES
where the matrix denoisers $\{\Psi_t, \widetilde{\Psi}_{t},\Phi_t,\widetilde{\Phi}_{t}\}$ are defined as linear combinations of the transformed polynomials $\{\Psi^u_t,\widetilde{\Psi}^u_t,\Phi^u_t,\widetilde{\Phi}^u_t\}$ introduced in \lemref{lem:aux1}\footnote{Note that our convention slightly differs from that in~\cite{dudeja2024optimality}: 
the deterministic scalars \(\alpha_t\) and \(\beta_t\), which appear in the orthogonal decomposition of the iterates 
(cf.~\eqref{eq: ort.decomp.ap}), are absorbed into the definitions of 
\(\{\Psi_t, \widetilde{\Psi}_t, \Phi_t, \widetilde{\Phi}_t\}\). 
}:
\BS\label{eq: aux-denoiser-decomp}
\begin{align}
    \Psi_t(\lambda)   &\bydef \alpha_t \Psi^u_t(\lambda) + \beta_t \widetilde{\Psi}^u_t(\lambda), & \widetilde{\Psi}_t(\lambda)   &\bydef \alpha_t \Psi^v_t(\lambda) + \beta_t \widetilde{\Psi}^v_t(\lambda), \\
    \Phi_t(\lambda)   &\bydef \beta_t \Phi^v_t(\lambda) + \alpha_t \widetilde{\Phi}^v_t(\lambda), & \widetilde{\Phi}_t(\lambda)   &\bydef \beta_t \Phi^u_t(\lambda) + \alpha_t \widetilde{\Phi}^u_t(\lambda).
\end{align}
\ES
%
%\begin{remark}
%Note that our convention slightly differs from that in~\cite{dudeja2024optimality}: 
%the deterministic scalars \(\alpha_t\) and \(\beta_t\), which appear in the orthogonal decomposition of the iterates 
%(cf.~\eqref{eq: ort.decomp.ap}), are absorbed into the definitions of 
%\(\{\Psi_t, \widetilde{\Psi}_t, \Phi_t, \widetilde{\Phi}_t\}\). 
%We emphasize that \((\alpha_t, \beta_t)\) involve the state evolution random variables of the original OAMP algorithm, which depends on the original matrix denoising functions $\{F,\tilde{F},G,\tilde{G}\}$.
%\end{remark}
The advantage of this auxiliary algorithm is that its dynamics are governed by the rotationally invariant matrix $\bm{W}$ instead of the observation matrix $\bm{Y}$. The dynamics of this system is tractable using existing techniques (cf.~\thref{The:SE_app}), which we summarize in the following lemma, whose proof is deferred to Section~\ref{App:lemma7_proof}.

\begin{lemma}\label{lem:aux-OAMP-dyn}
The following hold for any $t \in \mathbb{N}$:
\begin{enumerate}
\item The iterates generated by the auxiliary OAMP algorithm in \eqref{eq:aux-OAMP} satisfy
\BS\label{Eqn:aux-OAMP-dyn_SE}
\begin{align}
 (\bm{u}_{*}, \tilde{\bm{u}}_1, \dotsc, \tilde{\bm{u}}_t; \bm{a}) &\overset{W_2}{\longrightarrow} (\serv{U}_{*}, \serv{U}_1, \dotsc, \serv{U}_t; \serv{A}),\\
 (\bm{v}_{*}, \tilde{\bm{v}}_1, \dotsc, \tilde{\bm{v}}_t; \bm{b}) &\overset{W_2}{\longrightarrow} (\serv{V}_{*}, \serv{V}_1, \dotsc, \serv{V}_t; \serv{B}),
\end{align}
\ES
where $ (\serv{U}_{*}, \serv{U}_1, \dotsc, \serv{U}_t; \serv{A})$ and $ (\serv{V}_{*}, \serv{V}_1, \dotsc, \serv{V}_t; \serv{B})$ are the state evolution random variables defined for the original OAMP algorithm in \eqref{eq:SEuchannel} and \eqref{eq:SEvchannel}.
\item Denote $\bm{f}_t^\perp\bydef{f}_t^\perp(\tilde{\bm{u}}_{<t};\bm{a}) $ and $\bm{g}_t^\perp\bydef{g}_t^\perp(\tilde{\bm{v}}_{<t};\bm{v}) $. The follow holds for all $i\in\mathbb{N}$:
   \BS\label{Eqn:lemma6_app_orthgonality2}
   \begin{align}
     & \frac{\left\langle \bm{u}_\ast, (\bm{WW}^\UT)^{i}\bm{f}_t^\perp\right\rangle}{M}\overset{a.s.}{\longrightarrow}0,\quad  \frac{\left\langle\bm{u}_\ast,(\bm{WW}^\UT)^i\bm{W}\bm{g}_t^\perp\right\rangle}{M}= \frac{\left\langle\bm{u}_\ast,\bm{W}(\bm{W}^\UT\bm{W})^i\bm{g}_t^\perp\right\rangle}{M}\overset{a.s.}{\longrightarrow}0,\label{Eqn:lemma6_app_orthgonality2-u}\\
&\frac{\left\langle \bm{v}_\ast, (\bm{W}^\UT\bm{W})^{i}\bm{g}_t^\perp\right\rangle}{N}\overset{a.s.}{\longrightarrow}0,\quad \frac{\left\langle\bm{v}_\ast,(\bm{W}^\UT\bm{W})^i\bm{W}^\UT\bm{f}_t^\perp\right\rangle}{N}=\frac{\left\langle\bm{v}_\ast,\bm{W}^\UT(\bm{W}\bm{W}^\UT)^i\bm{f}_t^\perp\right\rangle}{N}\overset{a.s.}{\longrightarrow}0.\label{Eqn:lemma6_app_orthgonality2-v}
   \end{align}
   \ES
\end{enumerate}
\end{lemma}

\begin{remark}
Note that a direct application of Theorem~\ref{The:SE_app} would yield a state evolution expressed in terms of the transformed polynomials  
\(\{\Psi^u, \Psi^v, \tilde{\Psi}^u, \tilde{\Psi}^v, \Phi^u, \Phi^v, \tilde{\Phi}^u, \tilde{\Phi}^v\}\),  
which themselves depend on the original functions  
\(\{F, \tilde{F}, G, \tilde{G}\}\) in an implicit and complicated manner.  
Fortunately, by invoking the asymptotic equivalence established in Lemma~\ref{lem:aux1},  
the state evolution can be reformulated directly in terms of the original functions  
\(\{F, \tilde{F}, G, \tilde{G}\}\) and the probability measures  
\(\{\nu_1, \nu_2, \nu_3\}\), as presented in the original OAMP algorithm in  
\eqref{eq:SEuchannel} and~\eqref{eq:SEvchannel}.
\end{remark}

\vspace{5pt}

\paragraph{Proof of \thref{Thm: State Evolution}} Building on the preceding results, we are now ready to prove Theorem~\ref{Thm: State Evolution}.
Invoking the state evolution of the auxiliary OAMP algorithm in~\eqref{Eqn:aux-OAMP-dyn_SE},
it suffices to show that the auxiliary OAMP algorithm approximates the original OAMP algorithm in the following sense:
\begin{align} \label{eq:OAMP-to-auxOAMP-induction}
\bm{u}_t \explain{$N \rightarrow \infty$}{\simeq} \tilde{\bm{u}}_t \quad \text{and} \quad \bm{v}_t \explain{$N \rightarrow \infty$}{\simeq} \tilde{\bm{v}}_t,\quad\forall t\in\mathbb{N}.
\end{align}
We prove this via induction on $t$. The base case is trivial. Suppose \eqref{eq:OAMP-to-auxOAMP-induction} holds up to iteration $t-1$. We next show \eqref{eq:OAMP-to-auxOAMP-induction} holds for $t$. We have
\begin{align*}
\bm{u}_{t} & \;\; \explain{\eqref{eq:original-OAMP-decomp}}{=} \; \; \alpha_t F_{t}(\bm{Y}\bm{Y}^\UT) \bm{u}_* + \beta_t \tilde{F}_{t}(\bm{Y}\bm{Y}^\UT) \bm{Y} \bm{v}_* + F_{t}(\bm{Y}\bm{Y}^\UT) \bm{f}_t^\perp(\bm{u}_{<t};\bm{a}) + \tilde{F}_{t}(\bm{Y}\bm{Y}^\UT) \bm{Y} \bm{g}_t^\perp(\bm{v}_{<t};\bm{b}) \\
 &\explain{$N \rightarrow \infty$}{\simeq}\; \alpha_t F_{t}(\bm{Y}\bm{Y}^\UT) \bm{u}_* + \beta_t \tilde{F}_{t}(\bm{Y}\bm{Y}^\UT) \bm{Y} \bm{v}_* + F_{t}(\bm{Y}\bm{Y}^\UT) \bm{f}_t^\perp(\tilde{\bm{u}}_{<t};\bm{a}) + \tilde{F}_{t}(\bm{Y}\bm{Y}^\UT) \bm{Y} \bm{g}_t^\perp(\tilde{\bm{v}}_{<t};\bm{b}) \\
 &\explain{$N \rightarrow \infty$}{\simeq}\; \Psi_t(\bm{W}\bm{W}^\UT) \bm{u}_* + \widetilde{\Psi}_{t}(\bm{W}\bm{W}^\UT) \bm{W} \bm{v}_* + F_t(\bm{W}\bm{W}^\UT) \bm{f}_t^\perp(\tilde{\bm{u}}_{<t};\bm{a}) + \widetilde{F}_{t}(\bm{W}\bm{W}^\UT) \bm{W} \bm{g}_t^\perp(\tilde{\bm{v}}_{<t};\bm{b}) \\
 & \;\; \explain{\eqref{eq:aux-OAMP}}{=} \; \; \tilde{\bm{u}}_t,
\end{align*}
where the second step follows from the inductive hypothesis that $\bm{u}_s\explain{$N \rightarrow \infty$}{\simeq}  \tilde{\bm{u}}_s$ and $\bm{v}_s \explain{$N \rightarrow \infty$}{\simeq}  \tilde{\bm{v}}_s$ for all $s < t$ as well as the Lipschitz continuity of the iterate denoiser and the operator norm bound on matrix denoisers (cf.~\eqref{Eqn:YY_operator}), and the third step is due to Lemma \ref{lem:aux1} and Lemma \ref{lem:aux-OAMP-dyn} (which guarantees that the orthogonality conditions \eqref{Eqn:lemma6_app_orthgonality} are met).

The analysis of $\bm{v}_t$ is completely analogous and omitted. The proof  is complete.

\subsection{Proof of Lemma \ref{lem:aux1}}\label{Sec:proof_lemma_aux1}

The proof is presented in two parts, corresponding to the two claims in the lemma statement. 

\paragraph{Proof of Claim (1).} 

In what follows, we prove~\eqref{Eqn:lemma6_a}, which we recall below for convenience:
\BE
F(\bm{Y}\bm{Y}^\UT)\bm{u}_\ast 
\explain{$M \to \infty$}{\simeq} 
\Psi^u(\bm{W}\bm{W}^\UT)\bm{u}_\ast 
+ \Psi^v(\bm{W}\bm{W}^\UT)\bm{W}\bm{v}_\ast.
\EE
Since $F$ is a polynomial, it suffices to consider a monomial term and show that, for all $d \in \mathbb{N}$, there exist polynomials $(Q_u^d, Q_v^d)$ such that the following asymptotic equivalence holds:
\begin{align} \label{eq:aux-lemma-goal1}
(\bm{Y}\bm{Y}^\UT)^d \bm{u}_\ast 
\explain{$M \to \infty$}{\simeq} 
Q_u^d(\bm{W}\bm{W}^\UT)\bm{u}_\ast 
+ Q_v^d(\bm{W}\bm{W}^\UT)\bm{W}\bm{v}_\ast.
\end{align}
We prove \eqref{eq:aux-lemma-goal1} via induction on $d$. The base case $d=0$ is immediate. We now consider the induction step. 
Assume that~\eqref{eq:aux-lemma-goal1} holds for some integer $d \ge 0$; we will show that it also holds for $d+1$. 
Left-multiplying the expression for $(\bm{Y}\bm{Y}^\UT)^d \bm{u}_\ast$ by $\bm{Y}\bm{Y}^\UT$ gives
\BS\label{Eqn:C4_tmp1}
\begin{align}
(\bm{Y}\bm{Y}^\UT)^{d+1}\bm{u}_\ast
  &= (\bm{Y}\bm{Y}^\UT)\big((\bm{Y}\bm{Y}^\UT)^d \bm{u}_\ast\big) \\
  &\bydef (\bm{Y}\bm{Y}^\UT)\Big(Q_u^d(\bm{W}\bm{W}^\UT)\bm{u}_\ast 
       + Q_v^d(\bm{W}\bm{W}^\UT)\bm{W}\bm{v}_\ast + \bm{\epsilon}\Big) \\
  &= (\bm{Y}\bm{Y}^\UT)\Big(Q_u^d(\bm{W}\bm{W}^\UT)\bm{u}_\ast 
       + Q_v^d(\bm{W}\bm{W}^\UT)\bm{W}\bm{v}_\ast\Big) 
       + (\bm{Y}\bm{Y}^\UT)\bm{\epsilon}.
\end{align}
\ES
By the induction hypothesis~\eqref{eq:aux-lemma-goal1}, we have $\|\bm{\epsilon}\|^2/M \to 0$ almost surely.
Using the definition of $\bm{Y}$, we see that $\bm{Y}\bm{Y}^\UT$ is a rank-two perturbation of $\bm{W}\bm{W}^\UT$:
\BE\label{Eqn:YYT_expansion}
\bm{Y}\bm{Y}^\UT 
= \frac{\theta^2}{M}\bm{u}_\ast\bm{u}_\ast^\UT
  + \frac{\theta}{\sqrt{MN}}\!\left(\bm{u}_\ast\bm{v}_\ast^\UT\bm{W}^\UT 
  + \bm{W}\bm{v}_\ast\bm{u}_\ast^\UT\right)
  + \bm{W}\bm{W}^\UT.
\EE
Since $\bm{W}$ has bounded operator norm, as shown in~\cite[Appendix~B.1]{dudeja2024optimality}, we have
\BE\label{Eqn:YY_operator}
\limsup_{N\to\infty}\|\bm{Y}\bm{Y}^\UT\|_{\mathrm{op}} < \infty, 
\quad \text{almost surely.}
\EE
Combining~\eqref{Eqn:C4_tmp1}--\eqref{Eqn:YY_operator} yields
\BE\label{Eqn:C4_YY_d1_u}
\begin{split}
(\bm{Y}\bm{Y}^\UT)^{d+1}\bm{u}_\ast 
&\explain{$M \to \infty$}{\simeq} 
(\bm{Y}\bm{Y}^\UT)
\Big(Q_u^d(\bm{W}\bm{W}^\UT)\bm{u}_\ast 
+ Q_v^d(\bm{W}\bm{W}^\UT)\bm{W}\bm{v}_\ast\Big) \\
&\explain{(a)}{=} 
\Bigg(\frac{\theta^2}{M}\bm{u}_\ast\bm{u}_\ast^\UT
+ \frac{\theta}{\sqrt{MN}}\!\left(\bm{u}_\ast\bm{v}_\ast^\UT\bm{W}^\UT 
+ \bm{W}\bm{v}_\ast\bm{u}_\ast^\UT\right)
+ \bm{W}\bm{W}^\UT\Bigg)
\Big(Q_u^d(\bm{W}\bm{W}^\UT)\bm{u}_\ast 
+ Q_v^d(\bm{W}\bm{W}^\UT)\bm{W}\bm{v}_\ast\Big),
\end{split}
\EE
where step~(a) uses~\eqref{Eqn:YYT_expansion}. 
Expanding the product yields six terms, which we analyze below. 
For this purpose, define the following scalar quantities arising from inner products that converge almost surely in the same spirit of \eqref{eq:qform_v_limit}-\eqref{eq:cross_qform_0}:
\BS\label{Eqn:appC4_alphabeta}
\begin{align}
 \alpha_{uu} &\triangleq \lim_{M\to\infty}\frac{1}{M}\bm{u}_\ast^\UT Q_u^d(\bm{W}\bm{W}^\UT)\bm{u}_\ast 
   = \langle Q_u^d(\lambda) \rangle_{\mu}, \\
 \alpha_{uv} &\triangleq \lim_{M\to\infty}\frac{1}{M}\bm{u}_\ast^\UT Q_v^d(\bm{W}\bm{W}^\UT)\bm{W}\bm{v}_\ast = 0, \\
 \beta_{uv} &\triangleq \lim_{M\to\infty}\frac{1}{N}\bm{v}_\ast^\UT\bm{W}^\UT Q_u^d(\bm{W}\bm{W}^\UT)\bm{u}_\ast = 0, \\
 \beta_{vv} &\triangleq \lim_{M\to\infty}\frac{1}{N}\bm{v}_\ast^\UT\bm{W}^\UT Q_v^d(\bm{W}\bm{W}^\UT)\bm{W}\bm{v}_\ast 
   = \langle  \lambda Q_v^d(\lambda) \rangle_{\widetilde{\mu}},
\end{align}
\ES
where the spectral measures $\mu$ and $\widetilde{\mu}$ are defined in~\propref{prop:integral_representation}. 
Using~\eqref{Eqn:C4_YY_d1_u} and~\eqref{Eqn:appC4_alphabeta}, it follows that
\BS
\begin{align}
(\bm{Y}\bm{Y}^\UT)^{d+1}\bm{u}_\ast 
&\explain{$M \to \infty$}{\simeq}
\Big(\bm{W}\bm{W}^\UT Q_u^d(\bm{W}\bm{W}^\UT)
   + \theta^2\alpha_{uu}\bm{I}
   + \frac{\theta}{\sqrt{\delta}}\beta_{vv}\bm{I}\Big)\bm{u}_\ast \nonumber\\
&\quad + \Big(\bm{W}\bm{W}^\UT Q_v^d(\bm{W}\bm{W}^\UT)
   + \theta\sqrt{\delta}\alpha_{uu}\bm{I}\Big)\bm{W}\bm{v}_\ast \\
&= Q_u^{d+1}(\bm{W}\bm{W}^\UT)\bm{u}_\ast 
   + Q_v^{d+1}(\bm{W}\bm{W}^\UT)\bm{W}\bm{v}_\ast, \label{Eqn:appC4_tmp2}
\end{align}
\ES
where the degree-$(d+1)$ polynomials $(Q_u^{d+1}, Q_v^{d+1})$ are defined as
\begin{align}
Q_u^{d+1}(\lambda) 
  &\triangleq \lambda Q_u^d(\lambda) 
     + \theta^2\langle Q_u^d(\lambda) \rangle_{\mu}
     + \frac{\theta}{\sqrt{\delta}}\langle \lambda Q_v^d(\lambda) \rangle_{\widetilde{\mu}}, \\
Q_v^{d+1}(\lambda) 
  &\triangleq \lambda Q_v^d(\lambda) 
     + \theta\sqrt{\delta}\langle Q_u^d(\lambda) \rangle_{\mu}.
\end{align}
Hence,~\eqref{Eqn:appC4_tmp2} establishes the induction step for~\eqref{eq:aux-lemma-goal1}, 
thereby completing the proof of~\eqref{Eqn:lemma6_a} in Claim~(1).

The analyses of~\eqref{Eqn:lemma6_b}--\eqref{Eqn:lemma6_d} in Claim~(1) are analogous and are therefore omitted.

\paragraph{Proof of Claim (2).} Again, it suffices to prove that the following holds for all $d\in\mathbb{N}\cup\{0\}$:
\BS
\begin{align}
(\bm{YY}^\UT)^d \bm{u}&\explain{$M\rightarrow \infty$}{\simeq} (\bm{WW}^\UT)^d \bm{u}, \label{Eqn:app_YYd_u}\\
(\bm{Y}\bm{Y}^\UT)^d \bm{Y} \bm{v}& \explain{$M\rightarrow \infty$}{\simeq} (\bm{W}\bm{W}^\UT)^d \bm{W} \bm{v},\\
    (\bm{Y}^\UT\bm{Y})^d \bm{v} & \explain{$N\rightarrow \infty$}{\simeq} (\bm{W}^\UT\bm{W})^d\bm{v}, \label{Eqn:app_YYd_v}\\
    (\bm{Y}^\UT\bm{Y})^d \bm{Y}^\UT \bm{u}& \explain{$N\rightarrow \infty$}{\simeq} (\bm{W}^\UT\bm{W})^d \bm{W}^\UT \bm{u}.\label{Eqn:app_YYd_v2}
\end{align}
\ES

\textit{Analysis of $(\bm{Y}\bm{Y}^\UT)^d\bm{u}$:} We prove by induction on $d$. Assume that the claim is true for an integer $d \ge 0$ and proceed to the inductive step for $d+1$. We have:
\begin{align*}
    (\bm{Y}\bm{Y}^\UT)^{d+1} \bm{u} &= (\bm{Y}\bm{Y}^\UT) \left( (\bm{Y}\bm{Y}^\UT)^{d}  \bm{u}  \right) \explain{$M\rightarrow \infty$}{\simeq}  (\bm{Y}\bm{Y}^\UT) \left( (\bm{W}\bm{W}^\UT)^d \bm{u}  \right),
\end{align*}
where the second step follows from the induction hypothesis together with \eqref{Eqn:YY_operator}.
Next, we substitute the decomposition of $\bm{Y}\bm{Y}^\UT$ in \eqref{Eqn:YYT_expansion} and expand the product:
\begin{align*}
   (\bm{Y}\bm{Y}^\UT)^{d+1}\bm{u} \explain{$M\rightarrow \infty$}{\simeq}  & \left(\frac{\theta^2}{M} \bm{u}_* \bm{u}_*^\UT + \frac{\theta}{\sqrt{MN}} (\bm{u}_* \bm{v}_*^\UT \bm{W}^\UT + \bm{W} \bm{v}_* \bm{u}_*^\UT) + \bm{W}\bm{W}^\UT\right) \left( (\bm{W}\bm{W}^\UT)^d\bm{u}\right) \\
    =\; & \underbrace{\theta^2 \left( \frac{1}{M} \bm{u}_*^\UT (\bm{W}\bm{W}^\UT)^d \bm{u} \right) \bm{u}_*}_{\text{Term (i)}} 
    + \underbrace{\frac{\theta\sqrt{N}}{\sqrt{M}} \left( \frac{1}{N} \bm{v}_*^\UT \bm{W}^\UT (\bm{W}\bm{W}^\UT)^d\bm{u} \right) \bm{u}_*}_{\text{Term (ii)}} \\
    & + \underbrace{\frac{\theta\sqrt{M}}{\sqrt{N}} \left( \frac{1}{M} \bm{u}_*^\UT (\bm{W}\bm{W}^\UT)^d \bm{u} \right) \bm{W}\bm{v}_*}_{\text{Term (iii)}} 
    + \underbrace{(\bm{W}\bm{W}^\UT)^{d+1}\bm{u}}_{\text{Term (iv)}}.
\end{align*}
Based on the assumptions \eqref{Eqn:lemma6_app_orthgonality}, and using the facts $\|\bm{u}_\ast\|^2/M\overset{a.s.}{\longrightarrow}1$ and $\|\bm{Wv}_\ast\|^2/M\overset{a.s.}{\longrightarrow}C<\infty$ (where $C$ is a constant), it is straightforward to show that 
\[
\frac{\|\text{Term (i)}\|^2+\|\text{Term (ii)}\|^2+\|\text{Term (iii)}\|^2}{N}\overset{a.s.}{\longrightarrow}0.
\]
Hence, only Term (iv) survives, yielding the desired result:
\BE
   (\bm{Y}\bm{Y}^\UT)^{d+1}\bm{u}  \explain{$M\rightarrow \infty$}{\simeq} (\bm{W}\bm{W}^\UT)^{d+1}\bm{u}.
\EE

\textit{Analysis of $ (\bm{Y}\bm{Y}^\UT)^{d} \bm{Y} \bm{v} $:} We first note that
\BS\label{Eqn:Yv}
\begin{align}
 \bm{Y} \bm{v}  
 &= \left(\frac{\theta}{\sqrt{MN}}\bm{u}_\ast\bm{v}_\ast^\UT + \bm{W}\right)\bm{v} \\
 &= \left(\frac{\theta}{\sqrt{MN}}\bm{v}_\ast^\UT\bm{v}\right)\bm{u}_\ast + \bm{W}\bm{v} \\
 &\explain{$M \to \infty$}{\simeq} \bm{W}\bm{v}.
\end{align}
\ES
Note that the assumptions in~\eqref{Eqn:lemma6_app_orthgonality} imply the asymptotic orthogonality 
\(\bm{v}_\ast^\UT\bm{v}/N \overset{a.s.}{\to} 0\). 
Using~\eqref{Eqn:Yv} and \eqref{Eqn:YY_operator}, we obtain
\BE
 (\bm{Y}\bm{Y}^\UT)^{d}\bm{Y}\bm{v} 
 \explain{$M \to \infty$}{\simeq} 
   (\bm{Y}\bm{Y}^\UT)^{d}\bm{W}\bm{v}.
\EE
A careful inspection shows that the vector 
\(\hat{\bm{u}} \bydef \bm{W}\bm{v}\) 
satisfies the requirements on \(\bm{u}\) in~\eqref{Eqn:lemma6_app_orthgonality}. 
Hence, applying~\eqref{Eqn:app_YYd_u}, which we have just established above, yields
\BE
 (\bm{Y}\bm{Y}^\UT)^{d}\bm{W}\bm{v} 
 \explain{$M \to \infty$}{\simeq} 
   (\bm{W}\bm{W}^\UT)^{d}\bm{W}\bm{v}.
\EE
Combining the above two results gives
\BE
 (\bm{Y}\bm{Y}^\UT)^{d}\bm{Y}\bm{v} 
 \explain{$M \to \infty$}{\simeq} 
 (\bm{W}\bm{W}^\UT)^{d}\bm{W}\bm{v}.
\EE

The proofs of~\eqref{Eqn:app_YYd_v} and~\eqref{Eqn:app_YYd_v2} are analogous and are therefore omitted.

\subsection{Proof of \lemref{lem:aux-OAMP-dyn}}\label{App:lemma7_proof}

\paragraph{Proof of Claim (1):} We will prove the convergence for the $\bm{u}$-channel only, as the argument for the $\bm{v}$-channel is analogous. Our goal is to show that for any $t \in \mathbb{N}$, the iterates of the auxiliary OAMP algorithm converge weakly to the state evolution random variables:
\begin{align*}
(\bm{u}_{*}, \tilde{\bm{u}}_1, \dotsc, \tilde{\bm{u}}_{t}; \bm{a}) \overset{W_2}{\longrightarrow} (\serv{U}_{*}, \serv{U}_1, \dotsc, \serv{U}_t ; \serv{A}).
\end{align*}
The strategy is to demonstrate that the auxiliary OAMP algorithm can be rewritten in a canonical \textit{signal plus noise} form, whose state evolution has already been characterized in \thref{The:SE_app}.
\BS\label{eq:decomp-aux-u}
\begin{align} 
\tilde{\bm{u}}_t &=\Psi_t(\bm{W}\bm{W}^\UT) \bm{u}_* + \widetilde{\Psi}_{t}(\bm{W}\bm{W}^\UT) \bm{W} \bm{v}_*+ F_t(\bm{W}\bm{W}^\UT) \bm{f}_t^\perp(\tilde{\bm{u}}_{< t};\bm{a}) + \widetilde{F}_{t}(\bm{W}\bm{W}^\UT) \bm{W} \bm{g}_t^\perp(\tilde{\bm{v}}_{< t};\bm{b}) \nonumber \\
&=\mathbb{E}\left[ \Psi_t(\mathsf{D}_M^2)\right]\cdot\bm{u}_\ast +\bm{e}_t+\bm{h}_t , 
\end{align}
where $\mathsf{D}_M^2\sim\mu$ and 
\begin{align}
\bm{e}_t &\bydef \left(\Psi_t(\bm{W}\bm{W}^\UT) -\mathbb{E}\left[ \Psi_t(\mathsf{D}_M^2)\right]\cdot \bm{I}_M \right)\bm{u}_\ast ,\\
\bm{h}_t &\bydef  \widetilde{\Psi}_{t}(\bm{W}\bm{W}^\UT) \bm{W} \bm{v}_*+ F_t(\bm{W}\bm{W}^\UT) \bm{f}_t^\perp(\tilde{\bm{u}}_{< t};\bm{a}) + \widetilde{F}_{t}(\bm{W}\bm{W}^\UT) \bm{W} \bm{g}_t^\perp(\tilde{\bm{v}}_{< t};\bm{b}).
\end{align}
\ES
By viewing $\bm{u}_\ast$ as a side information, the above iteration an instance of the general OAMP algorithm introduced in Definition \ref{def: Noise OAMP}. By Theorem \ref{Th:OAMP_general_SE_app}, we have
\BE
(\bm{e}_{\le t},\bm{h}_{\le t};\bm{a},\bm{u}_\ast) \overset{W_2}{\longrightarrow}\left(\mathsf{E}_{\le t},\mathsf{H}_{\le t};\mathsf{A},\mathsf{U}_\ast\right),
\EE
where
\begin{enumerate}
\item The random variables $\mathsf{E}_{\le t}$, $\mathsf{H}_{\le t}$ and $(\mathsf{A},\mathsf{U}_\ast)$ are mutually independent;
\item $\mathsf{E}_{\le t}$, $\mathsf{W}_{\le t}$ are zero-mean Gaussian with 
\begin{align}\label{eq:aux-cov-structure}
\mathbb{E}\left[\mathsf{E}_s\mathsf{E}_t\right] &=\mathbb{E}\left[ \Psi_s(\mathsf{D}_M^2)\Psi_t(\mathsf{D}_M^2)\right] -\mathbb{E}\left[ \Psi_s(\mathsf{D}_M^2)\right]\cdot \mathbb{E}\left[ \Psi_t(\mathsf{D}_M^2)\right],\quad \forall s,t\in\mathbb{N},\\
\mathbb{E}\left[\mathsf{H}_s\mathsf{H}_t\right] &=\mathbb{E}\left[ \tilde{\Psi}_s(\mathsf{D}_M^2) \tilde{\Psi}_t(\mathsf{D}_M^2) \mathsf{D}_M^2\right] + \mathbb{E}\left[F_s(\mathsf{D}_M^2)F_t(\mathsf{D}_M^2)\right]\cdot \mathbb{E}[\mathsf{F}_s^\perp\mathsf{F}_t^\perp]+ \mathbb{E}\left[\tilde{F}_s(\mathsf{D}_M^2)\tilde{F}_t(\mathsf{D}_M^2)\mathsf{D}_M^2\right]\cdot \mathbb{E}[\mathsf{G}_s^\perp\mathsf{G}_t^\perp]\nonumber
\end{align}
where
\begin{align*}
\mathsf{F}_t &\bydef  f_t\Big( \mathbb{E}\left[ \Psi_{1}(\mathsf{D}_M^2)\right]\mathsf{U}_\ast +\mathsf{E}_{1}+\mathsf{H}_{1},\ldots, \mathbb{E}\left[ \Psi_{t-1}(\mathsf{D}_M^2)\right]\mathsf{U}_\ast +\mathsf{E}_{t-1}+\mathsf{H}_{t-1}; \mathsf{A}\Big),\\
\mathsf{F}_t^\perp &\bydef \mathsf{F}_t - \mathbb{E}\left[\mathsf{U}_\ast \mathsf{F}_t \right]\cdot \mathsf{U}_\ast.
\end{align*}
The random variables $(\mathsf{G}_t^\perp)_{t\ge1}$ are similarly defined. 
\end{enumerate}
The proof for the above claims are similar to those in \cite[Appendix B.4]{dudeja2024optimality} and hence omitted.

Next, we express the inner products involving $(\Psi_t, \tilde{\Psi}_t)$ in terms of $(F, \tilde{F}, G, \tilde{G})$, in the same spirit as the treatment in \cite[Appendix~B.4]{dudeja2024optimality}. As an initial step, we identify the coefficient of the signal components in \eqref{eq:decomp-aux-u}. Recall
\begin{align} \label{eq:aux-u-signal}
\mathbb{E}\left[ \Psi_t(\mathsf{D}_M^2)\right] &\explain{(a)}{=}\lim_{N\to\infty} \frac{\langle \Psi_t(\bm{WW}^\UT)\bm{u}_\ast, \bm{u}_\ast\rangle}{M} \nonumber \\
&\explain{(b)}{=}\lim_{N\to\infty} \frac{\langle  \alpha_t F_{t}(\bm{Y}\bm{Y}^\UT) \bm{u}_* + \beta_t \tilde{F}_{t}(\bm{Y}\bm{Y}^\UT) \bm{Y} \bm{v}_*, \bm{u}_\ast\rangle}{M}\nonumber \\
&=\alpha_t\cdot \lim_{N\to\infty} \frac{\langle   F_{t}(\bm{Y}\bm{Y}^\UT) \bm{u}_* , \bm{u}_\ast\rangle}{M} + \beta_t\cdot\lim_{N\to\infty}\frac{\langle    \tilde{F}_{t}(\bm{Y}\bm{Y}^\UT) \bm{Y} \bm{v}_*, \bm{u}_\ast\rangle}{M} \\
&\explain{(c)}{=} \alpha_t \cdot \langle F_t(\lambda) \rangle_{\nu_1} + \beta_t \cdot (1+\delta^{-1}) \langle \sigma \widetilde{F}_t(\sigma^2) \rangle_{\nu_3},\nonumber
\end{align}
where step (a) follows from \eqref{Eqn:app_SE_base_1}, and all limits are taken in the almost sure sense; step (b) substitutes the definition of the transformed polynomial $\Psi_t$ (cf.~\eqref{eq: aux-denoiser-decomp}); step (c) is due to \propref{prop:integral_representation} in the following manner
\begin{itemize}
        \item The first term, $\frac{1}{M} \bm{u}_*^\UT F_t(\bm{Y}\bm{Y}^\UT)\bm{u}_*$, is a direct application of \propref{prop:integral_representation} Claim~(1) with $h(\cdot) = F_t(\cdot)$.
        \item The second term involves the bilinear form $\frac{1}{M} \bm{u}_*^\UT \widetilde{F}_{t}(\bm{Y}\bm{Y}^\UT) \bm{Y} \bm{v}_*$. To use \propref{prop:integral_representation}, we define an operator $f(\bm{Y}) = \widetilde{F}_{t}(\bm{Y}\bm{Y}^\UT)\bm{Y}$. The function $f$ acts on the singular values $\sigma$ of $\bm{Y}$ as $f(\sigma) = \sigma \widetilde{F}_t(\sigma^2)$, which is  odd by construction. We can therefore apply \propref{prop:integral_representation} Claim~(2), which gives the limit:
        $ \lim_{M\to\infty} \frac{L}{M} \cdot \frac{1}{L} \bm{u}_*^\UT \widetilde{F}_{t}(\bm{Y}\bm{Y}^\UT) \bm{Y} \bm{v}_* \ac (1+\delta^{-1})\int \sigma \widetilde{F}_t(\sigma^2) \,d \nu_3(\sigma). $
    \end{itemize}

\medskip

We now examine the covariance structure \eqref{eq:aux-cov-structure} in the sequel.
\begin{align}\label{eq:aux-main-noise}
&\mathbb{E}\!\left[ \Psi_s(\mathsf{D}_M^2)\Psi_t(\mathsf{D}_M^2)\right] + \mathbb{E}\!\left[ \widetilde{\Psi}_s(\mathsf{D}_M^2)\,\widetilde{\Psi}_t(\mathsf{D}_M^2)\,\mathsf{D}_M^2\right] \nonumber \\
&\explain{(d)}{=} \lim_{M\to\infty} \frac{1}{M}\Big( \big\langle \Psi_s(\bm{W}\bm{W}^\UT)\bm{u}_\ast,\; \Psi_t(\bm{W}\bm{W}^\UT)\bm{u}_\ast\big\rangle  +  \big\langle \widetilde{\Psi}_s(\bm{W}\bm{W}^\UT)\bm{W}\bm{v}_\ast,\; \widetilde{\Psi}_t(\bm{W}\bm{W}^\UT)\bm{W}\bm{v}_\ast\big\rangle \Big)\nonumber\\
&\explain{(e)}{=} \lim_{M\to\infty} \frac{1}{M}\Big\langle
\big(\alpha_s \Psi^u_s(\bm{W}\bm{W}^\UT)+\beta_s \widetilde{\Psi}^u_s(\bm{W}\bm{W}^\UT)\big)\bm{u}_\ast
+\big(\alpha_s \Psi^v_s(\bm{W}\bm{W}^\UT)+\beta_s \widetilde{\Psi}^v_s(\bm{W}\bm{W}^\UT)\big)\bm{W}\bm{v}_\ast,\\
&\qquad\qquad\qquad\qquad\quad
\big(\alpha_t \Psi^u_t(\bm{W}\bm{W}^\UT)+\beta_t \widetilde{\Psi}^u_t(\bm{W}\bm{W}^\UT)\big)\bm{u}_\ast
+\big(\alpha_t \Psi^v_t(\bm{W}\bm{W}^\UT)+\beta_t \widetilde{\Psi}^v_t(\bm{W}\bm{W}^\UT)\big)\bm{W}\bm{v}_\ast
\Big\rangle\nonumber\\
&\explain{(f)}{=} \lim_{M\to\infty} \frac{1}{M}\Big\langle
\alpha_s F_s(\bm{Y}\bm{Y}^\UT)\bm{u}_\ast+\beta_s \widetilde{F}_s(\bm{Y}\bm{Y}^\UT)\bm{Y}\bm{v}_\ast,\;
\alpha_t F_t(\bm{Y}\bm{Y}^\UT)\bm{u}_\ast+\beta_t \widetilde{F}_t(\bm{Y}\bm{Y}^\UT)\bm{Y}\bm{v}_\ast
\Big\rangle\nonumber\\
& \explain{(g)}{=}   \alpha_s\alpha_t \langle F_s(\lambda) F_t(\lambda) \rangle_{\nu_1}
+ \beta_s\beta_t \delta^{-1}\langle \lambda \tilde{F}_s(\lambda)\tilde{F}_t(\lambda) \rangle_{\nu_2}
+ (1+\delta^{-1})\left(
  \alpha_s\beta_t \langle \sigma F_s(\sigma^2)\tilde{F}_t(\sigma^2) \rangle_{\nu_3}
  + \alpha_t\beta_s \langle \sigma F_t(\sigma^2)\tilde{F}_s(\sigma^2) \rangle_{\nu_3}
\right)
,\nonumber
\end{align}
where step (d) follows from  \eqref{Eqn:app_SE_base_1}-\eqref{Eqn:app_SE_base_3}, and all limits are taken in the almost sure sense; step (e) substitutes the expansion of the transformed polynomial $\Psi_t,\tilde{\Psi}_t$ (cf.~\eqref{eq: aux-denoiser-decomp}) and uses the independence between $\bm{u}_*,\bm{v}_*$ to absorb the cross term; step (f) utilizes the asymptotical equivalence \eqref{Eqn:lemma6_1} in \lemref{lem:aux1}; step (g) repeats the same procedure as in step (c) in the light of \propref{prop:integral_representation}. 
Next, we investigate 
\begin{align}\label{eq:aux-pure-noise-u}
\mathbb{E}\left[F_s(\mathsf{D}_M^2)F_t(\mathsf{D}_M^2)\right]\cdot \mathbb{E}[\mathsf{F}_s^\perp\mathsf{F}_t^\perp] & = \E[(\mathsf{F}_s - \alpha_s \mathsf{U}_*)(\mathsf{F}_t - \alpha_t \mathsf{U}_*)] \cdot \lim_{M\to \infty} \frac{\langle F_s(\bm{WW^\UT})\bm{u}_*, F_t(\bm{WW^\UT})\bm{u}_* \rangle}{M}\nonumber \\
&= \big(  \E[\mathsf{F}_s \mathsf{F}_t]- \alpha_s\alpha_t  \big) \cdot  \langle F_s(\lambda) F_t(\lambda) \rangle_{\mu} \bydef \sigma_{f,st}^2 \cdot  \langle F_s(\lambda) F_t(\lambda) \rangle_{\mu}.
\end{align}
Similarly we have
\begin{align}\label{eq:aux-pure-noise-v}
\mathbb{E}\left[\tilde{F}_s(\mathsf{D}_M^2)\tilde{F}_t(\mathsf{D}_M^2)\mathsf{D}_M^2\right]\cdot &\mathbb{E}[\mathsf{G}_s^\perp\mathsf{G}_t^\perp]  = \E[(\mathsf{G}_s - \beta_s \mathsf{V}_*)(\mathsf{G}_t - \beta_t \mathsf{V}_*)] \cdot \lim_{M\to \infty} \frac{\langle \tilde{F}_s(\bm{WW^\UT})\bm{W}\bm{v}_*, \tilde{F}_t(\bm{WW^\UT})\bm{W}\bm{v}_* \rangle}{M} \nonumber \\
&\quad = \big(  \E[\mathsf{G}_s \mathsf{G}_t]- \beta_s\beta_t  \big) \cdot \delta^{-1} \langle \lambda\tilde{F}_s(\lambda) \tilde{F}_t(\lambda) \rangle_{\muN}  \bydef  \sigma_{g,st}^2 \cdot \delta^{-1} \langle \lambda\tilde{F}_s(\lambda) \tilde{F}_t(\lambda) \rangle_{\muN}.
\end{align}
Finally, let us compute the total covariance \eqref{eq:aux-cov-structure} by gathering \eqref{eq:aux-u-signal} to \eqref{eq:aux-pure-noise-v}
\begin{align}
\Sigma_{u,st} & \bydef \mathbb{E}[\mathsf{Z}_{u,s}\mathsf{Z}_{u,t}]  = \mathbb{E}\left[\mathsf{E}_s\mathsf{E}_t\right] + \mathbb{E}\left[\mathsf{H}_s\mathsf{H}_t\right]\\
&= \alpha_s\alpha_t \langle F_s(\lambda)F_t(\lambda) \rangle_{\nu_1}
  + \beta_s\beta_t \delta^{-1}\langle \lambda\tilde{F}_s(\lambda)\tilde{F}_t(\lambda) \rangle_{\nu_2}
  - \mu_{u,s}\mu_{u,t} \nonumber\\
&\quad+ (1+\delta^{-1})\left(
    \alpha_s\beta_t \langle \sigma F_s(\sigma^2)\tilde{F}_t(\sigma^2) \rangle_{\nu_3}
    + \alpha_t\beta_s \langle \sigma F_t(\sigma^2)\tilde{F}_s(\sigma^2) \rangle_{\nu_3}
  \right) \nonumber\\
&\quad+ \sigma_{f,st}^2\langle F_s(\lambda)F_t(\lambda) \rangle_{\mu}
  + \delta^{-1} \sigma_{g,st}^2 \langle \lambda\tilde{F}_s(\lambda)\tilde{F}_t(\lambda) \rangle_{\widetilde{\mu}},
\end{align}
which is precisely the claimed covariance structure in \eqref{eq: Cov of SE1}.
\vspace{10pt}

\paragraph{Proof of Claim (2).} The proof follows the same strategy as in \cite[Appendix B.4]{dudeja2024optimality}. We briefly outline the argument for \eqref{Eqn:lemma6_app_orthgonality2-u}; the proof of \eqref{Eqn:lemma6_app_orthgonality2-v} is analogous and therefore omitted.

For any fixed $t \in \mathbb{N}$ and $i \in \mathbb{N}$, the terms
\BE
    (\bm{W}\bm{W}^\UT)^{i}\bm{f}_t^\perp 
    \quad \text{and} \quad 
    (\bm{W}\bm{W}^\UT)^{i}\bm{W}\bm{g}_t^\perp
\EE
may be interpreted as post-processing steps of an OAMP algorithm (cf.~\eqref{eq:decomp-aux-u}), whose dynamics are characterized by state evolution. By (1) a simple re-indexing, (2) viewing $\bm{u}_\ast$ and $\bm{v}_\ast$ as side information, and (3) an appropriate specification of the matrix-denoising functions, these post-processed terms can be identified with the iterates of another OAMP algorithm of the form \eqref{def: Noise OAMP}. Consequently, the state-evolution results in Theorem~\ref{Th:OAMP_general_SE_app} apply.

In particular, we establish that the empirical distributions of these terms converge to Gaussian random variables that are independent of the side information (in this case, $\mathsf{U}_\ast$ and $\mathsf{V}_\ast$, which represent the underlying signals). The claimed asymptotic orthogonality therefore follows.

%% file: appendix/ap-Denoisers2.tex
\section{Derivations for Optimal Denoisers} 

We adopt a greedy per--iteration design for the denoisers, focusing on the
$\bm{v}$--channel (the $\bm{u}$--channel is analogous). At iteration $t$, the
state evolution in~\eqref{eq:SEvchannel} yields a scalar Gaussian channel
$\mathsf V_t=\mu_{v,t}\mathsf V_*+\mathsf Z_{v,t}$ with variance
$\sigma_{v,t}^2\bydef\Var(\mathsf Z_{v,t})$. We design the denoisers to locally
maximize the squared cosine similarity
\begin{equation}\label{eq:local_sqcos_v}
w_{2,t}\bydef \frac{\mu_{v,t}^2}{\mu_{v,t}^2+\sigma_{v,t}^2}.
\end{equation}
The design at iteration $t$ is decoupled into two conditional optimizations:
\begin{itemize}

\item With $(f_t,g_t)$ fixed, the pair $(\mu_{v,t},\sigma_{v,t}^2)$ depends on the
matrix denoisers $(G_t,\widetilde G_t)$. Optimizing $(G_t,\widetilde G_t)$ under
the trace--free constraint~\eqref{eq:trace free} yields the optimal spectral
denoisers (Appendix~\ref{app:optim_matrix}).

\item With $(G_t^\ast,\widetilde G_t^\ast)$ fixed, the pair
$(\mu_{v,t},\sigma_{v,t}^2)$ depends on $(f_t,g_t)$ only through the residual
covariances in state evolution. Optimizing $(f_t,g_t)$ under the divergence--free
constraints~\eqref{eq: divergence free constraint} yields the DMMSE iterate
denoisers (Appendix~\ref{app:optim_iterate}).

\end{itemize}

\subsection{Optimal Matrix Denoisers}\label{app:optim_matrix}

At iteration $t$, we treat the iterate denoisers $(f_t,g_t)$ as given and optimize the spectral
denoisers $(G,\widetilde G)$ in the $\bm v$--update. Recall the SE variables and suppress $t$ for simplicity \eqref{eq:longmemory-F-G}-\eqref{eq:SE-sigma-fg}
\[
\mathsf F\bydef f_t(\mathsf U_1,\ldots,\mathsf U_{t-1};\mathsf A),\qquad
\mathsf G\bydef g_t(\mathsf V_1,\ldots,\mathsf V_{t-1};\mathsf B),\qquad
\alpha\bydef \E[\mathsf U_*\mathsf F],\qquad
\beta\bydef \E[\mathsf V_*\mathsf G],
\]
\[
\sigma_f^2\bydef \E(\mathsf F^2)-\alpha^2,\qquad
\sigma_g^2\bydef \E(\mathsf G^2)-\beta^2.
\]
Introduce the effective precisions
\begin{equation}\label{eq:rho_def_greedy}
\rho_1\bydef \frac{\alpha^2}{\sigma_f^2},\qquad
\rho_2\bydef \frac{\beta^2}{\sigma_g^2},
\end{equation}
With $(f_t,g_t)$ fixed (hence $(\alpha,\beta,\rho_1,\rho_2)$ fixed), we choose
$(G,\widetilde G)$ to maximize~\eqref{eq:local_sqcos_v} subject to the trace--free
constraint. Writing $(\mu_{v,t},\sigma_{v,t}^2)$ as functionals of $(G,\widetilde G)$,
we consider
\begin{align}\label{eq:optimal-matrix-snr}
&\max_{G,\widetilde G}\ 
  \frac{\bigl[\mu_{v,t}(G,\widetilde G)\bigr]^2}
       {\bigl[\mu_{v,t}(G,\widetilde G)\bigr]^2+\sigma^2_{v,t}(G,\widetilde G)}
&&\text{s.t.}\quad \langle G\rangle_{\widetilde\mu}=0\\
&=
\min_{G,\widetilde G}\ 
  \biggl\{
    1-
    \frac{\bigl[\mu_{v,t}(G,\widetilde G)\bigr]^2}
         {\bigl[\mu_{v,t}(G,\widetilde G)\bigr]^2+\sigma^2_{v,t}(G,\widetilde G)}
  \biggr\}
&&\text{s.t.}\quad \langle G\rangle_{\widetilde\mu}=0 \nonumber\\
&\stackrel{(a)}{=}
\min_{G,\widetilde G}\min_{c\in\R}\ 
  \bigl[1-c\,\mu_{v,t}(G,\widetilde G)\bigr]^2
  + c^2\sigma^2_{v,t}(G,\widetilde G)
&&\text{s.t.}\quad \langle cG\rangle_{\widetilde\mu}=0 \nonumber\\
&\stackrel{(b)}{=}
\min_{G,\widetilde G}\ 
  \bigl[1-\mu_{v,t}(G,\widetilde G)\bigr]^2
  + \sigma^2_{v,t}(G,\widetilde G)
&&\text{s.t.}\quad \langle G\rangle_{\widetilde\mu}=0 \nonumber\\
&\stackrel{(c)}{=}
\min_{G,\widetilde G}\ 
  \Big\langle (\beta G-1)^2\Big\rangle_{\nu_2}
  + \delta\Big\langle \lambda(\alpha\widetilde G)^2\Big\rangle_{\nu_1}
  + \frac{1}{\rho_2}\Big\langle (\beta G)^2\Big\rangle_{\widetilde\mu}
\label{eq:optimal-matrix-snr-final}
&&\\
&\qquad
  + \frac{\delta}{\rho_1}\Big\langle \lambda(\alpha\widetilde G)^2\Big\rangle_{\mu}
  + 2(1+\delta)\Big\langle (\beta G(\sigma^2)-1)\,(\alpha\widetilde G(\sigma^2))\,\sigma\Big\rangle_{\nu_3}
&&\text{s.t.}\quad \langle G\rangle_{\widetilde\mu}=0. \nonumber
\end{align}
Here $\widetilde\mu \bydef \delta \cdot \mu + (1-\delta)\cdot \delta_{\{0\}}$ is the noise measure in~\eqref{eq:trace free}. The reductions are:
\begin{enumerate}
\item[(a)] Apply $1-a^2/(a^2+b)=\min_{c\in\R}(1-ca)^2+c^2b$ and use $\langle cG\rangle_{\widetilde\mu}=c\langle G\rangle_{\widetilde\mu}$.
\item[(b)] Use the homogeneity of $(\mu_{v,t},\sigma_{v,t}^2)$ to absorb $c$ into $(G,\widetilde G)$.
\item[(c)] Substitute~\eqref{eq:SEvchannel} and~\eqref{eq:original-OAMP-decomp-v}, and eliminate $(\sigma_f^2,\sigma_g^2)$ via~\eqref{eq:rho_def_greedy}.
\end{enumerate}
The variational problem \eqref{eq:optimal-matrix-snr-final} is convex. We introduce a
Lagrange multiplier $\xi_v$ for the trace-free constraint:
\begin{align}\label{eq: Lag of v}
\mathcal{L}_v
&\bydef \langle (\beta G-1)^2\rangle_{\nu_2}
   + \delta\langle \lambda(\alpha\widetilde G)^2\rangle_{\nu_1}
   + \frac{1}{\rho_2}\langle (\beta G)^2\rangle_{\widetilde\mu}
   + \frac{\delta}{\rho_1}\langle \lambda(\alpha\widetilde G)^2\rangle_{\mu}
   + 2(1+\delta)\langle (\beta G(\sigma^2)-1)(\alpha\widetilde G(\sigma^2))\,\sigma\rangle_{\nu_3}
   - 2\xi_v\langle G\rangle_{\widetilde\mu} \nonumber\\
&= \mathcal{L}_v^{(0)} + \mathcal{L}_v^{(\mathrm{bulk})}
  + \mathcal{L}_v^{(\mathrm{out})}, \nonumber
\end{align}
where in spirit of
\lemref{lem:spectral_measures_properties} we decompose $\mathcal{L}_v$ into the contributions from the atom at $0$, the
bulk $\supp(\mu)$, and the limiting outliers $\mathcal{K}_*$ as:
\begin{equation}\label{eq:Lv0_def}
\begin{aligned}
\mathcal{L}_v^{(0)}
\bydef\ &
\nu_2(\{0\})\bigl(\beta G(0)-1\bigr)^2
+ (1-\delta)\frac{1}{\rho_2}\bigl(\beta G(0)\bigr)^2
- 2\xi_v(1-\delta)G(0).
\end{aligned}
\end{equation}
\begin{equation}\label{eq:Lvout_def}
\begin{aligned}
\mathcal{L}_v^{(\mathrm{out})}
\bydef\ &
\sum_{\lambda_*\notin\supp(\mu)\cup\{0\}}
\Big\{
\nu_2(\{\lambda_*\})\bigl(\beta G(\lambda_*)-1\bigr)^2
+ \delta\,\nu_1(\{\lambda_*\})\,\lambda_*\,\bigl(\alpha\widetilde G(\lambda_*)\bigr)^2 \\
&\qquad\qquad\qquad
+ 4(1+\delta)\,\nu_3(\{\sigma_*\})\,\sigma_*\,
\bigl(\beta G(\lambda_*)-1\bigr)\bigl(\alpha\widetilde G(\lambda_*)\bigr)
\Big\}.
\end{aligned}
\end{equation}
\begin{align}
\mathcal{L}_v^{(\mathrm{bulk})}
\bydef\ &
\int_{\supp(\mu)} \bigl(\beta G(\lambda)-1\bigr)^2\,\mathrm{d}\nu_2(\lambda)
+ \delta\int_{\supp(\mu)} \lambda\,\bigl(\alpha\widetilde G(\lambda)\bigr)^2\,\mathrm{d}\nu_1(\lambda) \nonumber \\
&\quad
+ \frac{1}{\rho_2}\int_{\supp(\mu)} \bigl(\beta G(\lambda)\bigr)^2\,\mathrm{d}\{\delta\cdot\mu(\lambda)\}
+ \frac{\delta}{\rho_1}\int_{\supp(\mu)} \lambda\,\bigl(\alpha\widetilde G(\lambda)\bigr)^2\,\mathrm{d}\mu(\lambda) \label{eq:Lvbulk_def} \\
&\quad
+ 2(1+\delta)\int_{\{\sigma:\,\sigma^2\in\supp(\mu)\}}
\bigl(\beta G(\sigma^2)-1\bigr)\bigl(\alpha\widetilde G(\sigma^2)\bigr)\sigma\,\mathrm{d}\nu_3(\sigma)
- 2\xi_v\int_{\supp(\mu)} G(\lambda)\,\mathrm{d}\{\delta\cdot\mu(\lambda)\}. \nonumber
\end{align}
\paragraph{Optimal Denoisers in $\supp(\mu)$.}
We minimize $\mathcal{L}_v^{(\mathrm{bulk})}$ pointwise over $\supp(\mu)$. Using \lemref{lem:spectral_measures_properties} and the characterization of the absolutely continuous parts
$\nu_i^{\parallel}$ via the shrinkage functions $\varphi_i$, we have
\begin{align}\label{eq:Lvbulk-mu-repr}
\mathcal{L}_v^{(\mathrm{bulk})}
&=
\Big\langle
(\beta G(\lambda)-1)^2\,\varphi_2(\lambda)
+ \frac{\delta}{\rho_2}\bigl(\beta G(\lambda)\bigr)^2
- 2\xi_v\delta\,G(\lambda)
\nonumber\\
&\qquad\quad
+ \delta\lambda\bigl(\alpha\widetilde G(\lambda)\bigr)^2\Big(\varphi_1(\lambda)+\frac{1}{\rho_1}\Big)
+ 2\sqrt{\delta}\,(\beta G(\lambda)-1)\bigl(\alpha\widetilde G(\lambda)\bigr)\varphi_3(\lambda)
\Big\rangle_{\mu,\ \supp(\mu)} .
\end{align}
Here we use the oddness of $\nu_3^{\parallel}$ and the change of variables $\lambda=\sigma^2$ on $\{\sigma:\sigma^2\in\supp(\mu)\}$.

\noindent Taking pointwise first--order conditions yields, for $\mu$--a.e.\ $\lambda\in\supp(\mu)$,
the following equations hold:
\begin{align}\label{eq:FOC-v-bulk}
\bigl(\rho_2\varphi_2(\lambda)+\delta\bigr)\bigl(\beta G(\lambda)-1\bigr)
\;+\;\sqrt{\delta}\,\rho_2\varphi_3(\lambda)\,\bigl(\alpha\widetilde G(\lambda)\bigr)
&= -\,\delta\Bigl(1-\frac{\xi_v\rho_2}{\beta}\Bigr), \nonumber\\
\sqrt{\delta}\,\rho_1\varphi_3(\lambda)\,\bigl(\beta G(\lambda)-1\bigr)
\;+\;\delta\lambda\bigl(\rho_1\varphi_1(\lambda)+1\bigr)\bigl(\alpha\widetilde G(\lambda)\bigr)
&= 0.
\end{align}
Solving \eqref{eq:FOC-v-bulk} yields the bulk minimizers
\begin{align}\label{eq:bulk-optimal-G-Gtilde}
\beta G_{\mathrm{bulk}}^*(\lambda)
&= 1-\Bigl(1-\frac{\xi_v\rho_2}{\beta}\Bigr)\,
      \frac{\delta\bigl[\rho_1\varphi_1(\lambda)+1\bigr]\lambda}
           {\bigl[\rho_1\varphi_1(\lambda)+1\bigr]\bigl[\rho_2\varphi_2(\lambda)+\delta\bigr]\lambda
            - \rho_1\rho_2\varphi_3^2(\lambda)}
   \;\bydef\; 1-\Bigl(1-\frac{\xi_v\rho_2}{\beta}\Bigr)Q^*(\lambda), \nonumber\\
\alpha\widetilde G_{\mathrm{bulk}}^*(\lambda)
&= \Bigl(1-\frac{\xi_v\rho_2}{\beta}\Bigr)\,
    \frac{\sqrt{\delta}\,\rho_1\varphi_3(\lambda)}
         {\bigl[\rho_1\varphi_1(\lambda)+1\bigr]\bigl[\rho_2\varphi_2(\lambda)+\delta\bigr]\lambda
          - \rho_1\rho_2\varphi_3^2(\lambda)}
   \;\bydef\; \Bigl(1-\frac{\xi_v\rho_2}{\beta}\Bigr)\widetilde Q^*(\lambda).
\end{align}

\paragraph{Optimal denoisers at the origin.}
We next minimize the contribution of the atom at $\{0\}$. From \eqref{eq:Lv0_def},
\begin{align}\label{eq:LV0-expand}
\mathcal{L}_v^{(0)}
= \Big[\nu_2(\{0\})+(1-\delta)\tfrac{1}{\rho_2}\Big]\bigl(\beta G(0)\bigr)^2
 - 2\Big[\nu_2(\{0\})+(1-\delta)\tfrac{\xi_v}{\beta}\Big]\bigl(\beta G(0)\bigr)
 + \nu_2(\{0\}).
\end{align}
Minimizing this quadratic gives
\begin{align}\label{eq:optimal_betaG0}
\beta G_0^*(0)
&\stackrel{(a)}{=} 1-\Bigl(1-\frac{\xi_v\rho_2}{\beta}\Bigr)\,
     \frac{1-\delta}{\rho_2\nu_2(\{0\})+(1-\delta)}
\stackrel{(b)}{=} 1-\Bigl(1-\frac{\xi_v\rho_2}{\beta}\Bigr)\,
   \frac{1-\theta^2(1-\delta)\pi\mathcal{H}(0)}
        {\rho_2+1-\theta^2(1-\delta)\pi\mathcal{H}(0)}.
\end{align}
Here (a) is the closed-form minimizer of \eqref{eq:LV0-expand}, and (b) substitutes
$\nu_2(\{0\})$ from \lemref{lem:spectral_measures_properties}.

To match the bulk formula \eqref{eq:bulk-optimal-G-Gtilde} at the origin \eqref{eq:optimal_betaG0}, note that $\beta G_0^*(0)$ is exactly the $\lambda\downarrow 0$ limit of \eqref{eq:bulk-optimal-G-Gtilde} by the explicit form of $\varphi_2$ in \eqref{eq:phiv}:
\[
\beta G_0^*(0)
= \lim_{\lambda\to 0}\beta G_{\mathrm{bulk}}^*(\lambda)
\qquad\text{since}\qquad
\lim_{\lambda\to 0}Q^*(\lambda)
= \frac{\delta}{\rho_2\varphi_2(0)+\delta}
= \frac{1-\theta^2(1-\delta)\pi\mathcal{H}(0)}
       {\rho_2+1-\theta^2(1-\delta)\pi\mathcal{H}(0)}.
\]
To enforce the trace--free constraint $\langle G\rangle_{\widetilde\mu}=0$ with
$\widetilde\mu=\delta\mu+(1-\delta)\delta_{\{0\}}$, we expand
\begin{align}\label{eq:xi_v_normalization}
0=\langle G\rangle_{\widetilde\mu}
&=\Big\langle \frac{1}{\beta}\Big(1-\Bigl(1-\frac{\xi_v\rho_2}{\beta}\Bigr)Q^*(\lambda)\Big)\Big\rangle_{\widetilde\mu}
\ \implies\
1-\frac{\xi_v\rho_2}{\beta}
=\Big\langle Q^*(\lambda)\Big\rangle_{\widetilde\mu}^{-1},
\end{align}
Substituting into \eqref{eq:bulk-optimal-G-Gtilde} gives the unified bulk/zero expressions via $Q^*(\lambda), \widetilde Q^*(\lambda)$ in \eqref{eq:bulk-optimal-G-Gtilde}
\[
\beta G^*(\lambda)
= 1-\Big\langle Q^*(\lambda)\Big\rangle_{\widetilde\mu}^{-1}Q^*(\lambda),
\qquad
\alpha\widetilde G^*(\lambda)
= \Big\langle Q^*(\lambda)\Big\rangle_{\widetilde\mu}^{-1}\widetilde Q^*(\lambda),
\]
for all $\lambda\in\supp(\mu)\cup\{0\}$.
\paragraph{Optimal denoisers at non-zero outliers.}
Fix a nonzero outlier $\lambda_*\in\mathcal{K}_*$ and let $\sigma_*=\sqrt{\lambda_*}$.
The outlier contribution to the Lagrangian is
\begin{align}\label{eq:Lv-outlier-loss}
\mathcal{L}_v^{(\{\lambda_*\})}
&=\nu_2(\{\lambda_*\})\bigl(\beta G(\lambda_*)-1\bigr)^2
  +\delta\,\nu_1(\{\lambda_*\})\,\lambda_*\bigl(\alpha\widetilde G(\lambda_*)\bigr)^2 \nonumber \\
 &\qquad  +4(1+\delta)\,\nu_3(\{\sigma_*\})\,\sigma_*\,
      \bigl(\beta G(\lambda_*)-1\bigr)\bigl(\alpha\widetilde G(\lambda_*)\bigr),
\end{align}
where we use the atomic characterizations of $\nu_1,\nu_2,\nu_3$ in \lemref{lem:spectral_measures_properties}
and the oddness $\nu_3(\{\sigma_*\})=-\nu_3(\{-\sigma_*\})$.

\noindent The first--order conditions of \eqref{eq:Lv-outlier-loss} are
\begin{align}
\nu_2(\{\lambda_*\})\bigl(\beta G(\lambda_*)-1\bigr)
+2(1+\delta)\,\nu_3(\{\sigma_*\})\,\sigma_*\,\bigl(\alpha\widetilde G(\lambda_*)\bigr)
&=0, \label{eq:FOC-v-outlier1}\\
\delta\,\nu_1(\{\lambda_*\})\,\lambda_*\,\bigl(\alpha\widetilde G(\lambda_*)\bigr)
+2(1+\delta)\,\nu_3(\{\sigma_*\})\,\sigma_*\,\bigl(\beta G(\lambda_*)-1\bigr)
&=0,\label{eq:FOC-v-outlier2}
\end{align}
and substituting the point masses from \lemref{lem:spectral_measures_properties} shows this linear system is underdetermined. Hence all solutions admit the parametrization
\begin{align}\label{eq:v-outlier-general-solution}
\bigl(\beta G(\lambda_*)-1,\ \alpha\widetilde G(\lambda_*)\bigr)
=\tau\Bigl(-2(1+\delta)\,\nu_3(\{\sigma_*\})\,\sigma_*,\ \nu_2(\{\lambda_*\})\Bigr),
\qquad \tau\in\R .
\end{align}
To match the bulk minimizer at $\lambda_*$, we verify that the bulk ratio satisfies \eqref{eq:FOC-v-outlier1}; indeed, taking $\lambda\to\lambda_*$ in the first equation of \eqref{eq:FOC-v-bulk} gives
\begin{align}\label{eq:v-outlier-matching}
\lim_{\lambda\to\lambda_*}\frac{1-\beta G_{\mathrm{bulk}}^*(\lambda)}{\alpha\widetilde G_{\mathrm{bulk}}^*(\lambda)}
=\frac{2(1+\delta)\,\nu_3(\{\sigma_*\})\,\sigma_*}{\nu_2(\{\lambda_*\})}
= \sqrt{\delta}\,\theta\,\lambda_*\,\Smu(\lambda_*),
\end{align}
and the second equation in \eqref{eq:FOC-v-bulk} yields the analogous consistency for \eqref{eq:FOC-v-outlier2}.

Combining \eqref{eq:bulk-optimal-G-Gtilde} and \eqref{eq:v-outlier-matching} yields the unified form for $ \lambda\in\supp(\mu)\cup\{0\}\cup\mathcal{K}_*,$
\begin{align}\label{eq:unfied-G-tG-ap}
 \beta G^*(\lambda)=1-\Big\langle Q^*(\lambda;\rho_1,\rho_2)\Big\rangle_{\widetilde\mu}^{-1}Q^*(\lambda),
\qquad
\alpha\widetilde G^*(\lambda)=\Big\langle Q^*(\lambda;\rho_1,\rho_2)\Big\rangle_{\widetilde\mu}^{-1}\widetilde Q^*(\lambda),   
\end{align}
where $(\alpha,\beta,\rho_1,\rho_2 )$ are induced by the iterate denoisers $(f_t,g_t)$ and will be optimized in the next subsection.

\subsection{Optimal Iterate Denoisers}\label{app:optim_iterate}
With the optimal spectral denoisers $(G_t^*,\widetilde G_t^*)$ fixed at iteration $t$,
the objective in~\eqref{eq:optimal-matrix-snr}--\eqref{eq:optimal-matrix-snr-final}
depends on the iterate denoisers $(f_t,g_t)$ only through the $(\sigma_f^2,\sigma_g^2)$ in \eqref{eq:SE-sigma-fg}
and, for fixed $(G_t^*,\widetilde G_t^*)$, only via
\begin{align}\label{eq:greedy-iterative}
\sigma_g^2\int \bigl[G_t^*(\lambda)\bigr]^2\,
\mathrm{d}\{\delta\,\mu(\lambda)+(1-\delta)\delta_0\},
\qquad
\delta\sigma_f^2\int \lambda\,\bigl[\widetilde G_t^*(\lambda)\bigr]^2\,
\mathrm{d}\mu(\lambda).
\end{align}
Therefore, maximizing the squared cosine similarity is equivalent to minimizing $(\sigma_f^2,\sigma_g^2)$ subject to the
divergence--free constraints \eqref{eq: divergence free constraint}. This yields two decoupled scalar programs:
\[
\min_{f_t}\ \E\!\left[\bigl(\mathsf U_*-f_t(\mathsf U_{t-1};\mathsf A)\bigr)^2\right]
\quad\text{s.t.}\quad \E\!\left[f_t'(\mathsf U_{t-1};\mathsf A)\right]=0,
\qquad
\min_{g_t}\ \E\!\left[\bigl(\mathsf V_*-g_t(\mathsf V_{t-1};\mathsf B)\bigr)^2\right]
\quad\text{s.t.}\quad \E\!\left[g_t'(\mathsf V_{t-1};\mathsf B)\right]=0,
\]
where $(\mathsf U_*,\mathsf A)\sim\pi_u$ and $(\mathsf V_*,\mathsf B)\sim\pi_v$ are the
scalar priors induced by state evolution. By~\cite[Definition~3]{dudeja2024optimality},
each problem is solved by the corresponding DMMSE estimator, hence
\begin{align}
f^{*}_{t}(\bm{u}; \bm{a})
&= \bar{\phi}\!\left(
      \bm{u} \big/ \sqrt{\mu_{u,t-1}^{2} + \Var(\mathsf Z_{u,t-1})};
      \bm{a} \,\middle|\, w_{1,t-1}
   \right),
&\qquad
w_{1,t-1}
&= \frac{\mu_{u,t-1}^{2}}{\mu_{u,t-1}^{2} + \Var(\mathsf Z_{u,t-1})},
\\
g^{*}_{t}(\bm{v}; \bm{b})
&= \bar{\phi}\!\left(
      \bm{v} \big/ \sqrt{\mu_{v,t-1}^{2} + \Var(\mathsf Z_{v,t-1})};
      \bm{b} \,\middle|\, w_{2,t-1}
   \right),
&\qquad
w_{2,t-1}
&= \frac{\mu_{v,t-1}^{2}}{\mu_{v,t-1}^{2} + \Var(\mathsf Z_{v,t-1})}.
\end{align}
To express the SE effective precision in closed form, define the standardized
observation
\[
\mathsf X_{v,t-1}\bydef
\frac{\mathsf V_{t-1}}{\sqrt{\mu_{v,t-1}^{2}+\Var(\mathsf Z_{v,t-1})}}
=\sqrt{w_{2,t-1}}\,\mathsf V_*+\sqrt{1-w_{2,t-1}}\,\mathsf Z,
\qquad \mathsf Z\sim\mathcal N(0,1).
\]
Since $\mathsf G_t=g_t^*(\mathsf V_{t-1};\mathsf B)=\bar\phi(\mathsf X_{v,t-1};\mathsf B\,|\,w_{2,t-1})$,
\cite[Lemma~2]{dudeja2024optimality} yields
$\E[\mathsf G_t^2]=\E[\mathsf V_*\mathsf G_t]=\beta_t$ and hence
\begin{align}\label{eq:beta-rho-expres}
\rho_{2,t}
&= \frac{\beta_t}{1-\beta_t}
 = \frac{1}{\mathrm{mmse}_{\pi_v}(w_{2,t-1})}-\frac{1}{1-w_{2,t-1}},
\qquad
\beta_t
= \frac{\rho_{2,t}}{1+\rho_{2,t}}.
\end{align}
The same identities hold for the $\bm u$--channel with similar forms for
$(\mathsf U_*,\mathsf X_{u,t-1},\mathsf A,w_{1,t-1},\alpha_t,\rho_{1,t},\pi_u)$.

%% file: appendix/ap-OptimalSE.tex
\section{Proof of \propref{prop: optimal SE}}\label{sec:optimal_SE}
The state evolution parameters are defined by the following recursion, initialized with $w_{1,0} = w_{2,0} \in (0,1)$
\BS
\begin{align}
 \rho_{1,t} &= \mathcal{F}_1(w_{1,t-1}) \bydef \frac{1}{\mathrm{mmse}_{\mathsf{U}}(w_{1,t-1})} - \frac{1}{1-w_{1,t-1}}, \label{eq:Optimal_Recursion_rho1_pf} \\
 \rho_{2,t} &= \mathcal{F}_2(w_{2,t-1}) \bydef \frac{1}{\mathrm{mmse}_{\mathsf{V}}(w_{2,t-1})} - \frac{1}{1-w_{2,t-1}}, \label{eq:Optimal_Recursion_rho2_pf} \\
 w_{1,t} &= \mathcal{F}_3(\rho_{1,t},\rho_{2,t}) \bydef 1 - \frac{1 - {\langle P_t^*(\lambda;\rho_{1,t},\rho_{2,t}) \rangle}_{\muM}}{{\langle P_t^*(\lambda;\rho_{1,t},\rho_{2,t}) \rangle}_{\muM}}\cdot \frac{1}{\rho_{1,t}}, \label{eq:Optimal_Recursion_w1_pf} \\
 w_{2,t} &= \mathcal{F}_4(\rho_{1,t},\rho_{2,t}) \bydef 1 - \frac{1 - {\langle Q_t^*(\lambda;\rho_{1,t},\rho_{2,t}) \rangle}_{\muN}}{{\langle Q_t^*(\lambda;\rho_{1,t},\rho_{2,t}) \rangle}_{\muN}}\cdot \frac{1}{\rho_{2,t}}, \label{eq:Optimal_Recursion_w2_pf}
\end{align}
\ES
The proof relies on the properties of these recursive functions, summarized in the following lemma which we defer its proof in Section~\ref{app:Increasing_and_Decreasing_Function_Lemma_pf}.
\begin{lemma}
\label{lem:Increasing_and_Decreasing_Function_Lemma_pf}
The functions defining the state evolution recursion satisfy:
\begin{enumerate}
    \item $\mathcal{F}_1(w)$ and $\mathcal{F}_2(w)$ are continuous and non-decreasing functions mapping $[0, 1)$ to $[0, \infty)$, with $\lim_{w \to 1^-} \mathcal{F}_{1,2}(w) = \infty$.
    \item $\mathcal{F}_3(\rho_1, \rho_2)$ and $\mathcal{F}_4(\rho_1, \rho_2)$ are continuous and non-decreasing in both arguments on $(0, \infty)^2$, and map to $[0, 1)$, with $\lim_{\rho_1, \rho_2 \to \infty} \mathcal{F}_{3,4}(\rho_1, \rho_2) < 1$.
\end{enumerate}
\end{lemma}

\begin{proof}
\noindent\textbf{Proof of Claim~(1).}
We prove Claim~(1) for the $\bm v$--channel; the $\bm u$--channel is analogous.
We show by induction that, for every $t\ge 1$,
\begin{equation}\label{eq:claim1_channel_form_v}
w_{2,t}\in(0,1),\quad \rho_{2,t}>0,\quad
\mathsf V_t^{*}\mid \mathsf V_* \sim \mathcal{N}\!\bigl(\sqrt{w_{2,t}}\ \mathsf V_*,\,1-w_{2,t}\bigr).
\end{equation}
Assuming \eqref{eq:claim1_channel_form_v}, the MMSE postprocessing
$\widehat{\bm v}_t^*=\phi(\bm v_t^*;\bm b\,|\,w_{2,t})$ satisfies, by
Theorem~\ref{Thm: State Evolution},
\begin{align}\label{eq:claim1_mse_v}
\frac{\|\widehat{\bm v}_t^*-\bm v_*\|^2}{N}
\ac
\E\Bigl[\bigl(\phi(\mathsf V_t^{*};\mathsf B\,|\,w_{2,t})-\mathsf V_*\bigr)^2\Bigr]
=\mathrm{mmse}_{\mathsf V}(w_{2,t}),
\end{align}
Hence it remains to prove
\eqref{eq:claim1_channel_form_v} by induction.\\
\smallskip
\noindent \emph{Induction Steps.} Assume \eqref{eq:claim1_channel_form_v} holds at step $t-1$, i.e.
\begin{equation}\label{eq:IH_v}
\mathsf V_{t-1}^{*}\mid \mathsf V_* \sim \mathcal{N}\!\bigl(\sqrt{w_{2,t-1}}\ \mathsf V_*,\,1-w_{2,t-1}\bigr).
\end{equation}
At iteration $t$, state evolution yields the scalar representation
\[
\mathsf V_t^{*}=\mu_{v,t}\,\mathsf V_*+\mathsf Z_{v,t},\qquad
\mathsf Z_{v,t}\sim\mathcal N(0,\sigma_{v,t}^2)\ \indep\ \mathsf V_*,
\]
and our goal is to verify that this channel has squared cosine similarity $w_{2,t}$,
i.e.
\begin{equation}\label{eq:goal_t_v}
\mu_{v,t}=\sqrt{w_{2,t}},\qquad \sigma_{v,t}^2=1-w_{2,t}.
\end{equation}
We only prove the mean identity $\mu_{v,t}=\sqrt{w_{2,t}}$; the variance identity follows
analogously from \eqref{eq:cov_Z_v}.

Under the induction hypothesis \eqref{eq:IH_v}, the iterate denoiser equals
$g_t^*(\cdot;\bm b)=\bar\phi(\cdot;\bm b\,|\,w_{2,t-1})$. Hence the DMMSE precision identities \eqref{eq:beta-rho-expres} hold; see Appendix~\ref{app:optim_iterate}. Substituting the explicit forms of $G_t^*$ and $\widetilde G_t^*$ in \eqref{eq:tGstar} into the general SE formula \eqref{eq: mean value of SE1} yields the following representation
\begin{align}\label{eq:mu_v_Qrepr}
\mu_{v,t}
=\frac{1}{\sqrt{w_{2,t}}}\Bigg[
1-\frac{1}{\langle Q_t^*\rangle_{\widetilde\mu}}
\Big\{
\big\langle Q_t^*(\lambda)\big\rangle_{\nu_2}
-(1+\delta)\big\langle \sigma\,\widetilde Q_t^*(\sigma^2)\big\rangle_{\nu_3}
\Big\}
\Bigg].
\end{align}
We claim that for every $t\ge1$,
\begin{equation}\label{eq:SE-normalization-claim}
\Big\{  \big\langle Q_t^*(\lambda)\big\rangle_{\nu_2}
-(1+\delta)\big\langle \sigma\,\widetilde Q_t^*(\sigma^2)\big\rangle_{\nu_3}\Big\}
=\frac{1}{\rho_{2,t}}\Bigl(1-\big\langle Q_t^*(\lambda)\big\rangle_{\widetilde\mu}\Bigr).
\end{equation}
Assuming \eqref{eq:SE-normalization-claim}, \eqref{eq:mu_v_Qrepr} and recalling the definition of $w_{2,t}$ in the recursion
\eqref{eq:OptimalRecursion} give
\[
\mu_{v,t}
=\frac{1}{\sqrt{w_{2,t}}}\Bigl(1-\frac{1-\langle Q_t^*\rangle_{\widetilde\mu}}{\langle Q_t^*\rangle_{\widetilde\mu}}\cdot\frac{1}{\rho_{2,t}}\Bigr)
=\frac{w_{2,t}}{\sqrt{w_{2,t}}}
=\sqrt{w_{2,t}}.
\]
The variance identity $\sigma_{v,t}^2=1-w_{2,t}$ follows analogously from \eqref{eq:cov_Z_v}, completing the induction step and hence Claim~(1).\\
\noindent \emph{It remains to prove \eqref{eq:SE-normalization-claim}.}
Fix $t$ and abbreviate $\rho_2=\rho_{2,t}$, $Q=Q_t^*$, $\widetilde Q=\widetilde Q_t^*$. Then
% requires: \usepackage{mathtools}  % for \mathclap
\begin{align}\label{eq:SE-normalization-claim-proof}
\big\langle Q(\lambda)\big\rangle_{\nu_2}
-(1+\delta)\big\langle \sigma\,\widetilde Q(\sigma^2)\big\rangle_{\nu_3}
&\stackrel{(a)}{=}
\underbrace{\Big\langle \varphi_2(\lambda)\,Q(\lambda)\Big\rangle_{\mu}
-\sqrt{\delta}\Big\langle \varphi_3(\lambda)\,\widetilde Q(\lambda)\Big\rangle_{\mu}
\vphantom{\sum_{\lambda_*\in\mathcal K_*}\Big\{\nu_2(\{\lambda_*\})Q(\lambda_*)
-2(1+\delta)\nu_3(\{\sigma_*\})\sigma_*\,\widetilde Q(\lambda_*)\Big\}}}_{\mathclap{\text{a.c.\ part}}}
\nonumber\\[-0.2em]
&\qquad
+\underbrace{\sum_{\lambda_*\in\mathcal K_*}\Big\{\nu_2(\{\lambda_*\})Q(\lambda_*)
-2(1+\delta)\nu_3(\{\sigma_*\})\sigma_*\,\widetilde Q(\lambda_*)\Big\}
\vphantom{\Big\langle \varphi_2(\lambda)\,Q(\lambda)\Big\rangle_{\mu}}}_{\mathclap{\text{outliers}}}
+\underbrace{\nu_2(\{0\})Q(0)
\vphantom{\sum_{\lambda_*\in\mathcal K_*}\Big\{\nu_2(\{\lambda_*\})Q(\lambda_*)
-2(1+\delta)\nu_3(\{\sigma_*\})\sigma_*\,\widetilde Q(\lambda_*)\Big\}}}_{\mathclap{\text{atom at }0}}
\nonumber\\
&\stackrel{(b)}{=}
\frac{\delta}{\rho_2}\Bigl(1-\langle Q(\lambda)\rangle_{\mu}\Bigr)+\nu_2(\{0\})Q(0)
\stackrel{(c)}{=}
\frac{1}{\rho_2}\Bigl(1-\langle Q(\lambda)\rangle_{\widetilde\mu}\Bigr).
\end{align}
where: (a) applies Lemma~\ref{lem:spectral_measures_properties} for Lebesgue decomposition of $\nu_2,\nu_3$ and the change of variables $\lambda=\sigma^2$ on $\nu_3^\parallel$.
Step (b) simplifies the a.c.\ part by substituting the explicit formulas \eqref{eq:bulk-optimal-G-Gtilde} of $Q,\widetilde Q$ and integrating the resulting pointwise identity against $\mu$, and uses \eqref{eq:v-outlier-matching} to cancel the outlier terms for each $\lambda_*\in\mathcal K_*$.
Step (c) uses \eqref{eq:optimal_betaG0} to eliminate $\nu_2(\{0\})$ and rewrites $ \delta\langle Q\rangle_\mu+(1-\delta)Q(0)=\langle Q\rangle_{\widetilde\mu}$.

\noindent \textbf{Proof of Claim~(2).} With respect to Claim~(2), we now prove that the sequences $\{w_{2,t}\}$ and $\{\rho_{2,t}\}$ (and their $\bm{u}$-channel counterparts) are non-decreasing and converge to the specified fixed point. 
\begin{itemize} 
\item \textbf{Monotonicity:} We proceed by induction. We initialize with $w_{1,0}=w_{2,0}=0$. Then $\rho_{1,1} = \mathcal{F}_1(w_{1,0})$ and $\rho_{2,1} = \mathcal{F}_2(w_{2,0})$. The next iterate is $w_{1,1} = \mathcal{F}_3(\rho_{1,1}, \rho_{2,1})$. Since $\mathcal{F}_3 \ge 0$, we have $w_{1,1} \ge w_{1,0}=0$. Now, assume $w_{1,t-1} \le w_{1,t}$ and $w_{2,t-1} \le w_{2,t}$. The monotonicity of the functions $\mathcal{F}_i$ given by \lemref{lem:Increasing_and_Decreasing_Function_Lemma_pf} ensures that $\rho_{1,t} \le \rho_{1,t+1}$ and $\rho_{2,t} \le \rho_{2,t+1}$. Consequently, $w_{2,t+1} = \mathcal{F}_4(\rho_{1,t+1}, \rho_{2,t+1}) \ge \mathcal{F}_4(\rho_{1,t}, \rho_{2,t}) = w_{2,t}$. The same logic applies to $w_{1,t}$. Therefore, all four sequences are non-decreasing.
\item \textbf{Convergence:} The sequences $\{w_{1,t}\}$ and $\{w_{2,t}\}$ are non-decreasing and bounded above by 1. By the monotone convergence theorem, they must converge to limits $w_1^*$ and $w_2^*$. \lemref{lem:Increasing_and_Decreasing_Function_Lemma_pf} guarantees that the limits are strictly less than 1. Consequently, the sequences $\{\rho_{1,t}\}$ and $\{\rho_{2,t}\}$ also converge to finite limits $\rho_1^* = \mathcal{F}_1(w_1^*)$ and $\rho_2^* = \mathcal{F}_2(w_2^*)$. By the continuity of all functions, the limit point $(\rho_1^*, \rho_2^*, w_1^*, w_2^*)$ must be a solution to the fixed-point equations given in the proposition. 
\end{itemize}

\end{proof}

\subsection{Proof of \lemref{lem:Increasing_and_Decreasing_Function_Lemma_pf}}\label{app:Increasing_and_Decreasing_Function_Lemma_pf}

The properties of $\mathcal{F}_1$ and $\mathcal{F}_2$ follow from standard scalar Gaussian MMSE arguments; see \cite[Lemma~4]{dudeja2024optimality}. We focus on $\mathcal{F}_3$ and $\mathcal{F}_4$, and only treat $\mathcal{F}_3(\rho_1,\rho_2)$, as the argument for $\mathcal{F}_4$ is identical.

Recall that $\mathcal{F}_3$ is defined via $P^*(\lambda;\rho_1,\rho_2)$. For $\rho_1,\rho_2>0$ and $\lambda\in\supp(\mu)$, set
\BS\label{eq:abD}
\begin{align}
    a(\rho_2;\lambda)
    &\bydef \rho_2\lambda\varphi_2(\lambda)+\delta\lambda,\\
    b(\rho_2;\lambda)
    &\bydef \rho_2\bigl[\lambda\varphi_1(\lambda)\varphi_2(\lambda)-\varphi_3^2(\lambda)\bigr]
          +\delta\lambda\varphi_1(\lambda),\\
    D(\rho_1,\rho_2;\lambda)
    &\bydef(\rho_1\varphi_1(\lambda)+1)(\rho_2\varphi_2(\lambda)+\delta)\lambda
        -\rho_1\rho_2\varphi_3^2(\lambda)
      =\rho_1 b(\rho_2;\lambda)+a(\rho_2;\lambda).
\end{align}
\ES
By the explicit formulas for $\varphi_1,\varphi_2,\varphi_3$ in
\eqref{eq:phiu}–\eqref{eq:phiv}, one verifies that
\[
a(\rho_2;\lambda)>0,\qquad b(\rho_2;\lambda)>0,\qquad D(\rho_1,\rho_2;\lambda)>0
\quad\text{for all }\lambda\in\supp(\mu),\ \rho_1,\rho_2>0.
\]
Hence
\[
P^*(\lambda)
=\frac{a(\rho_2;\lambda)}{D(\rho_1,\rho_2;\lambda)},
\qquad
\frac{1}{\rho_1}\bigl(1-P^*(\lambda)\bigr)
=\frac{b(\rho_2;\lambda)}{D(\rho_1,\rho_2;\lambda)}.
\]
We define
\[
N(\rho_1,\rho_2)\bydef\Big\langle\frac{b(\rho_2;\lambda)}{D(\rho_1,\rho_2;\lambda)}\Big\rangle_{\mu},
\qquad
M(\rho_1,\rho_2)\bydef\Big\langle\frac{a(\rho_2;\lambda)}{D(\rho_1,\rho_2;\lambda)}\Big\rangle_{\mu},
\]
so that
\begin{align*}
\mathcal{F}_3(\rho_1,\rho_2)
&=1-\frac{\left\langle b(\rho_2;\lambda)/D(\rho_1,\rho_2;\lambda)\right\rangle_{\mu}}
          {\left\langle a(\rho_2;\lambda)/D(\rho_1,\rho_2;\lambda)\right\rangle_{\mu}}
=1-\frac{N(\rho_1,\rho_2)}{M(\rho_1,\rho_2)}.
\end{align*}

\noindent \textbf{Continuity.}
The functions $a(\rho_2;\lambda)$, $b(\rho_2;\lambda)$, and $D(\rho_1,\rho_2;\lambda)$
are continuous in $(\rho_1,\rho_2)$ for $\rho_1,\rho_2>0$, and
$D(\rho_1,\rho_2;\lambda)>0$ by construction. Since $\supp(\mu)$ is compact,
the ratios $a(\rho_2;\lambda)/D(\rho_1,\rho_2;\lambda)$ and
$b(\rho_2;\lambda)/D(\rho_1,\rho_2;\lambda)$ are uniformly bounded on compact
subsets of $(0,\infty)^2$. Dominated convergence then yields continuity of
$M$ and $N$, and hence $\mathcal{F}_3$ is continuous on $(0,\infty)^2$.

\noindent  \textbf{Monotonicity.}
We show that $\mathcal{F}_3$ is non-decreasing in each coordinate.

\noindent  \emph{(i) Monotonicity in $\rho_2$.}
Recalling $\mathcal{F}_3=1-N/M$, we obtain
\[
\frac{\partial\mathcal{F}_3}{\partial\rho_2}
=\frac{N(\rho_1,\rho_2)\,M_{\rho_2}'(\rho_1,\rho_2)
      -N_{\rho_2}'(\rho_1,\rho_2)\,M(\rho_1,\rho_2)}
      {M(\rho_1,\rho_2)^2}.
\]
Differentiating $b/D$ and $a/D$ with respect to $\rho_2$, we obtain
\begin{align*}
& N(\rho_1,\rho_2)\,M_{\rho_2}'(\rho_1,\rho_2)
- N_{\rho_2}'(\rho_1,\rho_2)\,M(\rho_1,\rho_2)
=\Big\langle\frac{\delta\lambda\varphi_3^2(\lambda)}{D^2(\rho_1,\rho_2;\lambda)}\Big\rangle_{\mu}\ge0,
\end{align*}
since
$b(\rho_2;\lambda)a'(\rho_2;\lambda)-a(\rho_2;\lambda)b'(\rho_2;\lambda)
=\delta\lambda\varphi_3^2(\lambda)$ and $D(\rho_1,\rho_2;\lambda)>0$ on
$\supp(\mu)$. Consequently,
\[
\frac{\partial\mathcal{F}_3}{\partial\rho_2}\;\ge 0,
\]
and $\mathcal{F}_3$ is non-decreasing in $\rho_2$.

\noindent  \emph{(ii) Monotonicity in $\rho_1$.}
Again writing $\mathcal{F}_3=1-N/M$, we have
\[
\frac{\partial\mathcal{F}_3}{\partial\rho_1}
=\frac{N(\rho_1,\rho_2)\,M_{\rho_1}'(\rho_1,\rho_2)
      -N_{\rho_1}'(\rho_1,\rho_2)\,M(\rho_1,\rho_2)}
      {M(\rho_1,\rho_2)^2}.
\]
Differentiating $b/D$ and $a/D$ with respect to $\rho_1$ and taking
$\langle\cdot\rangle_\mu$ yields
\begin{align*}
&N(\rho_1,\rho_2)\,M_{\rho_1}'(\rho_1,\rho_2)
- N_{\rho_1}'(\rho_1,\rho_2)\,M(\rho_1,\rho_2)
\\&\qquad=
\Big\langle\frac{b^2(\rho_2;\lambda)}{D^2(\rho_1,\rho_2;\lambda)}\Big\rangle_{\mu}
  \Big\langle\frac{a(\rho_2;\lambda)}{D(\rho_1,\rho_2;\lambda)}\Big\rangle_{\mu}
 -\Big\langle\frac{a(\rho_2;\lambda)b(\rho_2;\lambda)}{D^2(\rho_1,\rho_2;\lambda)}\Big\rangle_{\mu}
  \Big\langle\frac{b(\rho_2;\lambda)}{D(\rho_1,\rho_2;\lambda)}\Big\rangle_{\mu}.
\end{align*}
To sign this quantity, set
\[
X\bydef\frac{a(\rho_2;\lambda)}{b(\rho_2;\lambda)},
\qquad
Z\bydef\frac{b^2(\rho_2;\lambda)}{D^2(\rho_1,\rho_2;\lambda)}\ge0,
\]
and take the non-decreasing functions $f(x)=x$ and $g(x)=\rho_1+x$.
Chebyshev's association inequality \cite[Theorem~2.14]{Boucheron2013Concentration}
states that for non-decreasing $f,g$ and a non-negative random variable $Z$,
\[
\mathbb{E}[Z]\cdot\mathbb{E}[f(X)g(X)Z]
\;\ge\;
\mathbb{E}[f(X)Z]\cdot\mathbb{E}[g(X)Z].
\]
Applied to the law of $X$ induced by $\mu$ and the above choice of $Z$, this
gives exactly
\[
\Big\langle\frac{b^2(\rho_2;\lambda)}{D^2(\rho_1,\rho_2;\lambda)}\Big\rangle_{\mu}
\Big\langle\frac{a(\rho_2;\lambda)}{D(\rho_1,\rho_2;\lambda)}\Big\rangle_{\mu}
\;\ge\;
\Big\langle\frac{a(\rho_2;\lambda)b(\rho_2;\lambda)}{D^2(\rho_1,\rho_2;\lambda)}\Big\rangle_{\mu}
\Big\langle\frac{b(\rho_2;\lambda)}{D(\rho_1,\rho_2;\lambda)}\Big\rangle_{\mu}.
\]
Consequently, $\mathcal{F}_3$ is non-decreasing in $\rho_1$ since
\[
\frac{\partial\mathcal{F}_3}{\partial\rho_1}
=\frac{N M_{\rho_1}'-N_{\rho_1}'M}{M^2}\;\ge 0.
\]

\noindent \textbf{Range.}
By definition $M(\rho_1,\rho_2)>0$, hence $\mathcal F_3(\rho_1,\rho_2)=1-N/M<1$ for all $\rho_1<\infty$.
For the lower bound, using the monotonicity of $\mathcal F_3$ in $(\rho_1,\rho_2)$,
\begin{align}\label{eq:F3_range_TIT}
\mathcal F_3(\rho_1,\rho_2)
&\;\ge\;\lim_{\rho_2\to0}\,\lim_{\rho_1\to0}\mathcal F_3(\rho_1,\rho_2)\nonumber\\
&\stackrel{(a)}{=} 1-\lim_{\rho_2\to0}\Big\langle \frac{b(\rho_2;\lambda)}{a(\rho_2;\lambda)}\Big\rangle_\mu
\stackrel{(b)}{=} 1-\langle \varphi_1(\lambda)\rangle_\mu
\stackrel{(c)}{=} 1-\nu_1^\parallel(\R)\;\ge\;0.
\end{align}
\noindent where (a) lets $\rho_1\to0$ so that $D(\rho_1,\rho_2;\lambda)\to a(\rho_2;\lambda)$ and applies dominated convergence;
(b) uses $a(\rho_2;\lambda)\to\delta\lambda$ and $b(\rho_2;\lambda)\to\delta\lambda\varphi_1(\lambda)$ as $\rho_2\to0$;
(c) uses $\mathrm d\nu_1^\parallel/\mathrm d\lambda=\varphi_1(\lambda)\mu(\lambda)$ from Lemma~\ref{lem:spectral_measures_properties}
and $\nu_1^\parallel(\R)\le1$ since $\nu_1$ is a probability measure.

Finally, we compute the limit at infinity:
\begin{align}\label{eq:F3_limit_infty_TIT}
\lim_{\rho_2\to\infty}\,\lim_{\rho_1\to\infty}\mathcal F_3(\rho_1,\rho_2)
&\stackrel{(a)}{=} \lim_{\rho_2\to\infty}\,\lim_{\rho_1\to\infty} 1-
\left\langle \frac{a(\rho_2;\lambda)}{D(\rho_1,\rho_2;\lambda)}\right\rangle_\mu^{-1}
\left\langle \frac{b(\rho_2;\lambda)}{D(\rho_1,\rho_2;\lambda)}\right\rangle_\mu \nonumber\\
&\stackrel{(b)}{=}\lim_{\rho_2\to\infty}1-
\left\langle \frac{a(\rho_2;\lambda)}{b(\rho_2;\lambda)}\right\rangle_\mu^{-1} \nonumber\\
&\stackrel{(c)}{=}1-
\left\langle
\frac{\lambda\varphi_2(\lambda)}
{\lambda\varphi_1(\lambda)\varphi_2(\lambda)-\varphi_3^2(\lambda)}
\right\rangle_\mu^{-1}
\;<\;1 .
\end{align}
\noindent where (a) expands $\mathcal F_3=1-N/M$ with $N=\langle b/D\rangle_\mu$ and $M=\langle a/D\rangle_\mu$;
(b) lets $\rho_1\to\infty$ in $D=\rho_1 b+a$, so by dominated convergence $\langle b/D\rangle_\mu/\langle a/D\rangle_\mu\to\langle a/b\rangle_\mu^{-1}$; (c) lets $\rho_2\to\infty$ and uses the fact that $\lambda\varphi_1\varphi_2-\varphi_3^2>0$.

%% file: appendix/ap-Example.tex
\section{Proof of \propref{thm:OAMP_Wigner_FP} (I.I.D.\ Gaussian Noise Model)}\label{sec:example}

This section analyzes a special case of the model from \eqref{eq:rectangular spiked model}, where the noise matrix entries are IID Gaussian random variables: $W_{i,j} \sim \mathcal{N}\Bigl(0, {1}/{N}\Bigr), \quad \text{for any} \quad 1 \le i \le M, 1 \le j \le N$, with an aspect ratio $\lim_{M,N \to \infty}\delta = \frac{M}{N} \in (0,1)$.

\subsection{Spectral Analysis of I.I.D.\ Gaussian Noise Model}\label{app:pf-MP-right-only}

We first recall that when $\bm{W}$ has i.i.d.\ $\mathcal{N}(0,1/N)$ entries and $M/N\to\delta\in(0,1)$, the empirical spectral distribution of $\bm{W}\bm{W}^\UT$ converges almost surely to the Mar\v{c}henko--Pastur law with density
\[
  \mu_{\mathrm{MP}}(\lambda)
  = \frac{\sqrt{(b_+ - \lambda)(\lambda - a_-)}}{2\pi\delta\,\lambda}\,
    \mathbf{1}_{[a_-,b_+]}(\lambda),
  \qquad
  a_-\bydef(1-\sqrt{\delta})^2,\;\; b_+\bydef(1+\sqrt{\delta})^2,
\]
see, e.g.,~\cite[Theorem~3.6]{Bai2010SpectralAO}. In the rectangular spiked model \eqref{eq:rectangular spiked model} with i.i.d.\ Gaussian noise $\bm{W}_{ij}\sim \mathcal{N}(0,1/N)$, the largest eigenvalue of $\bm{Y}\bm{Y}^\UT$ is known to exhibit a phase transition(\,cf.~\cite{paul2007asymptotics,Bai2010SpectralAO}).  There is a critical value $\theta^2=\sqrt{\delta}$ such that
\begin{align}
  &\theta^2\le\sqrt{\delta}
    \;\Rightarrow\;
    \lambda_1(\bm{Y}\bm{Y}^\UT)\xrightarrow{\text{a.s.}} b_+,\nonumber\\
  &\theta^2>\sqrt{\delta}
    \;\Rightarrow\;
    \lambda_1(\bm{Y}\bm{Y}^\UT)\xrightarrow{\text{a.s.}} \lambda_* \bydef 1+\delta+\theta^2+\frac{\delta}{\theta^2}\geq b_+,
    \label{eq:MP-right-outlier}
\end{align}
where equality holds if and only if $\theta^2=\delta^{1/2}$. We next examine the behavior of the master equation \eqref{Eqn:master0} on the left of the Mar\v{c}henko--Pastur bulk.

\begin{lemma}\label{lem:MP-right-only}
Let $\bm{Y}$ follow the rectangular spiked model \eqref{eq:rectangular spiked model} with $\bm{W}$ having i.i.d.\ $\mathcal{N}(0,1/N)$ entries, and let $M,N\to\infty$ with $M/N\to\delta\in(0,1)$. Then the master equation
\[
  \Gamma(\lambda)=1-\theta^2\CT(\lambda)=0
\]
admits no real solution on $[0,a_-)$.
\end{lemma}

\begin{proof}
It suffices to show that $\Gamma(\lambda)\neq 0$ for all $\lambda<a_-$. For the Mar\v{c}henko--Pastur law, the Stieltjes transform $\Smu$ satisfies (cf.~\cite[Lemma~3.11]{Bai2010SpectralAO}),
\[
\Smu(z) =\frac{1}{2\delta z}\left[z + \delta - 1 - \sqrt{(z - \delta - 1)^2 - 4\delta}\right], \quad \forall\, z\notin[a_-,b_+]
\]
substituting into the definition of $\CT(\lambda)$ \eqref{eq:Ctrans} yields:
\[
  \CT(\lambda)= \lambda\Smu(\lambda) -1 = \int \frac{t}{\lambda-t}\,\mu_{\mathrm{MP}}(\mathrm dt),
  \qquad \lambda\notin[a_-,b_+].
\]
If $\lambda<a_-$, then $t\in[a_-,b_+]$ implies $\lambda-t<0$, so the integrand $t/(\lambda-t)$ is strictly negative and hence $\CT(\lambda)<0$. Since $\theta^2>0$, it follows that
\[
  \Gamma(\lambda)=1-\theta^2\CT(\lambda)>1,
\]
and therefore $\Gamma(\lambda)\neq 0$ for all $\lambda<a_-$. This proves that the master equation has no real solution on $[0,a_-)$.
\end{proof}

% \begin{remark}\label{rem:no-left-empirical}
% We conjecture that, for $M$ sufficiently large, the empirical spectrum of $\bm{Y}\bm{Y}^\UT$ contains no eigenvalues below $a_-$. Our analysis concerns only limiting spectral behavior as $M,N\to\infty$, so this conjecture has no bearing on our results. Establishing such a statement is beyond the scope of this work.
% \end{remark}

\subsection{Optimal Denoisers and its OAMP Recursion}

The spectral measures admit an outlier atom at $\lambda_*$ in \eqref{eq:MP-right-outlier} and, for the $\bm v$-channel, an additional atom at $0$ (cf.~\lemref{lem:spectral_measures_properties}). Under \lemref{lem:MP-right-only}, there is no atom on $[0,a_-)$, and the point masses are
\begin{align}\label{eq: point masses in IID Gaussian Case}
\nu_1(\{\lambda_*\}) = 1 - \frac{\delta(1+\theta^2)}{\theta^2(\delta+\theta^2)},\quad
\nu_2(\{\lambda_*\}) = 1 - \frac{\delta+\theta^2}{\theta^2(1+\theta^2)},\quad \quad \nu_2(\{0\}) &= \frac{1-\delta}{1+\theta^2}.
\end{align}
The absolutely continuous parts follow directly from \lemref{lem:spectral_measures_properties}:
\BS\label{eq: bulk densities in IID Gaussian Case}
\begin{align}
\nu_1^{\parallel}(\lambda)& = {\mu(\lambda)}{\varphi_1(\lambda)} = \frac{\delta + \frac{\delta}{\theta^2}}{\lambda_* -\lambda} \mu(\lambda),\\
\nu_2^{\parallel}(\lambda) & = {\mu(\lambda)}{\varphi_2(\lambda)} = \frac{\delta + \frac{\delta^2}{\theta^2}}{\lambda_* -\lambda} \mu(\lambda),\\
\nu_3^{\parallel}(\sigma) & = \sign(\sigma) \frac{\sqrt{\delta}}{1+\delta}\cdot{\mu(\sigma^2)}{\varphi_{3} (\sigma^2)} = \sign(\sigma) \frac{\sqrt{\delta}}{1+\delta}\cdot\frac{ \frac{\delta}{\theta} \sigma^2}{\lambda_* -\sigma^2} \mu(\sigma^2).
\end{align}
\ES
\noindent Substituting \eqref{eq: point masses in IID Gaussian Case}--\eqref{eq: bulk densities in IID Gaussian Case} into \eqref{eq:Fstar} yields the $\bm u$-channel denoisers
\BS
\begin{align*}
P_t^*(\lambda;\rho_{1,t},\rho_{2,t}) &=  \frac{\rho_{2,t}(1+\frac{\delta}{\theta^2})+\lambda_* - \lambda}{\rho_{1,t}(\delta+\frac{\delta}{\theta^2})+\rho_{2,t}(1+\frac{\delta}{\theta^2})+\rho_{1,t}\rho_{2,t}\cdot\frac{\delta}{\theta^2} +\lambda_*- \lambda},\\
\widetilde{P}_t^*(\lambda;\rho_{1,t},\rho_{2,t})  &=   \frac{\frac{\sqrt{\delta}}{\theta}\rho_{2,t}}{\rho_{1,t}(\delta+\frac{\delta}{\theta^2})+\rho_{2,t}(1+\frac{\delta}{\theta^2})+\rho_{1,t}\rho_{2,t}\cdot\frac{\delta}{\theta^2} +\lambda_*- \lambda }, 
\end{align*}
\ES
and, analogously, \eqref{eq:tGstar} gives the $\bm v$-channel denoisers
\BS\label{eq:MP-Q-star}
\begin{align}
Q_t^*(\lambda;\rho_{1,t},\rho_{2,t}) &=  \frac{\rho_{1,t}(\delta+\frac{\delta}{\theta^2})+\lambda_* - \lambda}{\rho_{1,t}(\delta+\frac{\delta}{\theta^2})+\rho_{2,t}(1+\frac{\delta}{\theta^2})+\rho_{1,t}\rho_{2,t}\cdot\frac{\delta}{\theta^2} +\lambda_*- \lambda},\\
\widetilde{Q}_t^*(\lambda;\rho_{1,t},\rho_{2,t})  &=   \frac{\frac{\sqrt{\delta}}{\theta}\rho_{1,t}}{\rho_{1,t}(\delta+\frac{\delta}{\theta^2})+\rho_{2,t}(1+\frac{\delta}{\theta^2})+\rho_{1,t}\rho_{2,t}\cdot\frac{\delta}{\theta^2} +\lambda_*- \lambda },
\end{align}
\ES
Having established the close form of optimal denoisers, we have the following representation of limit squared cosine similarities.
\begin{lemma}\label{lem:limit_cos_MP}
In the rectangular spiked matrix model \eqref{eq:rectangular spiked model} with I.I.D. Gaussian noise $\bm{W}\sim \mathcal{N}(0,{1}/{N})$, let $T(\lambda_*) = \rho_1(\delta+\frac{\delta}{\theta^2})+\rho_2(1+\frac{\delta}{\theta^2})+\rho_1\rho_2\cdot\frac{\delta}{\theta^2} +\lambda_*$, and let $\Smu(\cdot)$ be the Stieltjes transform of the Mar\v{c}henko-Pastur law. The limit squared cosine similarities $w_1$ and $w_2$ are determined by the fixed-point equations
\BS
\begin{align}
w_1 &= 1 - \frac{\left( \frac{ \delta}{\theta^2}\rho_2 + \frac{\delta}{\theta^2}+ \delta \right) \Smu(T(\lambda_*))}{1 - \rho_1 \left( \frac{ \delta}{\theta^2}\rho_2 + \frac{\delta}{\theta^2}+ \delta \right) \Smu(T(\lambda_*))}, \label{eq:relationship:etau and wu,wv} \\
w_2 &= 1 - \frac{\left( \frac{\delta}{\theta^2}\rho_1 + \frac{\delta}{\theta^2} + 1 \right)\left[ \delta \Smu(T(\lambda_*)) + {(1-\delta)}/{T(\lambda_*)} \right]}{1 - \rho_2\left( \frac{\delta}{\theta^2}\rho_1 + \frac{\delta}{\theta^2} + 1 \right) \left[ \delta \Smu(T(\lambda_*)) + {(1-\delta)}/{T(\lambda_*)} \right]}. \label{eq:relationship: etav and wu,wv}
\end{align}
\ES
\end{lemma}

\begin{proof}[Proof]
We prove the $\bm v$-channel identity \eqref{eq:relationship: etav and wu,wv}; the proof for $w_1$ is analogous. 
By the fixed-point update \eqref{eq:Optimal_Recursion_w2_pf},
\begin{align}
w_2
&\stackrel{(c)}{=}1-\frac{1-{\langle Q^*(\lambda;\rho_1,\rho_2)\rangle}_{\muN}}{\rho_2\,{\langle Q^*(\lambda;\rho_1,\rho_2)\rangle}_{\muN}}.\nonumber
\end{align}
Thus it suffices to compute ${\langle Q^*(\lambda;\rho_1,\rho_2)\rangle}_{\muN}$. Recalling $T(\lambda_*)$ and $\muN$ from \lemref{lem:spectral_measures_properties}, we have
\begin{align}
{\langle Q^*(\lambda;\rho_1,\rho_2)\rangle}_{\muN}
&\stackrel{(a)}{=}{\left\langle 1-\frac{\rho_2(1+\frac{\delta}{\theta^2})+\rho_1\rho_2\frac{\delta}{\theta^2}}{T(\lambda_*)-\lambda}\right\rangle}_{\muN}\nonumber\\
&=1-\rho_2\left(\tfrac{\delta}{\theta^2}\rho_1+\tfrac{\delta}{\theta^2}+1\right){\left\langle \frac{1}{T(\lambda_*)-\lambda}\right\rangle}_{\muN}\nonumber\\
&\stackrel{(b)}{=}1-\rho_2\left(\tfrac{\delta}{\theta^2}\rho_1+\tfrac{\delta}{\theta^2}+1\right)\left[\delta\,\Smu\!\bigl(T(\lambda_*)\bigr)+\tfrac{1-\delta}{T(\lambda_*)}\right].\label{eq: Wigner EQ}
\end{align}
\noindent where (a) uses \eqref{eq:MP-Q-star}; and (b) uses $\muN=\delta\mu+(1-\delta)\delta_0$ from \lemref{lem:spectral_measures_properties} and
${\langle (T(\lambda_*)-\lambda)^{-1}\rangle}_\mu=\Smu(T(\lambda_*))$, hence
${\langle (T(\lambda_*)-\lambda)^{-1}\rangle}_{\muN}=\delta\,\Smu(T(\lambda_*))+\frac{1-\delta}{T(\lambda_*)}$.
Substituting \eqref{eq: Wigner EQ} into \eqref{eq:Optimal_Recursion_w2_pf} yields \eqref{eq:relationship: etav and wu,wv}.
\end{proof}

For a clear comparison, the next Lemma~\ref{lem:bayes-amp-fp} reformulates the fixed point equation of AMP \cite[Theorem~3]{venkataramanan2022estimation} for the rectangular spiked model with I.I.D. Gaussian noise.
\begin{lemma}\label{lem:bayes-amp-fp}
Under \assumpref{assump:main}, run the AMP iteration of \cite[Alg.\,(4.2)–(4.3)]{venkataramanan2022estimation}
with posterior-mean iterative denoisers for the two channels. Let $(\bar\mu_{u,t},\bar\sigma_{u,t},\bar\mu_{v,t},\bar\sigma_{v,t})_{t\ge 0}$ be the
SE parameters, and define the squared cosine
similarities 
\[
\Bar{w}_{1,t}\;\bydef\;\frac{\bar\mu_{u,t}^{\,2}}{\bar\mu_{u,t}^{\,2}+\bar\sigma_{u,t}^{\,2}},
\qquad
\Bar{w}_{2,t}\;\bydef\;\frac{\bar\mu_{v,t}^{\,2}}{\bar\mu_{v,t}^{\,2}+\bar\sigma_{v,t}^{\,2}}.
\]
Assume $\Bar{w}_{1,t}\to\overline w_1$ and $\Bar{w}_{2,t}\to\overline w_2$ as $t\to\infty$. Then the limiting overlaps satisfy the coupled fixed-point equations
\begin{equation}\label{eq:fp-rect}
\frac{\overline w_2}{1-\overline w_2}
=\theta^2\!\left[1-\mathrm{mmse}_{\mathsf U}(\overline w_1)\right],\qquad
\frac{\overline w_1}{1-\overline w_1}
=\frac{\theta^2}{\delta}\!\left[1-\mathrm{mmse}_{\mathsf V}(\overline w_2)\right].
\end{equation}
\end{lemma}

\begin{proof}
By \cite[Theorem.~3]{venkataramanan2022estimation}, the SE gives
\begin{align}
\bar\mu_{v,t+1}&=\theta\,\E\!\left[\mathsf U_*\,g_t\!\big(\bar\mu_{u,t}\mathsf U_*+\bar\sigma_{u,t}G\big)\right],
&
\bar\sigma_{v,t+1}^{\,2}&=\E\!\left[g_t\!\big(\bar\mu_{u,t}\mathsf U_*+\bar\sigma_{u,t}G\big)^2\right],\label{eq:SE-v}\\
\bar\mu_{u,t}&=\frac{\theta}{\delta}\,\E\!\left[\mathsf V_*\,f_t\!\big(\bar\mu_{v,t}\mathsf V_*+\bar\sigma_{v,t}G\big)\right],
&
\bar\sigma_{u,t}^{\,2}&=\frac{1}{\delta}\,\E\!\left[f_t\!\big(\bar\mu_{v,t}\mathsf V_*+\bar\sigma_{v,t}G\big)^2\right],\label{eq:SE-u}
\end{align}
where $G\sim\mathcal N(0,1)$ is independent of $(\mathsf U_*,\mathsf V_*)$, and the denoisers are the posterior means
\begin{align}
g_t(x)=\E\!\left[\mathsf U_*\,\big|\,\bar\mu_{u,t}\mathsf U_*+\bar\sigma_{u,t}G=x\right],\quad f_t(y)=\E\!\left[\mathsf V_*\,\big|\,\bar\mu_{v,t}\mathsf V_*+\bar\sigma_{v,t}G=y\right].\label{eq:ft-def}
\end{align}
By the tower property applied to \eqref{eq:ft-def},
\begin{align}
\E\!\Big[\mathsf U_*\,g_t\!\big(\bar\mu_{u,t}\mathsf U_*+\bar\sigma_{u,t}G\big)\Big]
=\E\!\Big[g_t\!\big(\bar\mu_{u,t}\mathsf U_*+\bar\sigma_{u,t}G\big)^2\Big].\label{eq:tower-U}
\end{align}
Using $\E[\mathsf U_*^2]=1$ from \assumpref{assump:main}, the scalar-channel identity gives
\begin{align}
\mathrm{mmse}_{\mathsf U}(\Bar{w}_{1,t})
&=\E\!\Big[\big(\mathsf U_*-\E[\mathsf U_*\,|\,\bar\mu_{u,t}\mathsf U_*+\bar\sigma_{u,t}G]\big)^2\Big]=1-\E\!\Big[g_t\!\big(\bar\mu_{u,t}\mathsf U_*+\bar\sigma_{u,t}G\big)^2\Big].\label{eq:mmse-U}
\end{align}
Combining \eqref{eq:SE-v} with \eqref{eq:tower-U}–\eqref{eq:mmse-U} yields
\begin{align}
\bar\mu_{v,t+1}&=\theta\Big(1-\mathrm{mmse}_{\mathsf U}(\Bar{w}_{1,t})\Big),\quad
\bar\sigma_{v,t+1}^{\,2}=1-\mathrm{mmse}_{\mathsf U}(\Bar{w}_{1,t}).\label{eq:sig-v-upd}
\end{align}
Hence, by definition of $\Bar{w}_{2,t}$,
\begin{align}
\frac{\Bar{w}_{2,t+1}}{1-\Bar{w}_{2,t+1}}
=\frac{\bar\mu_{v,t+1}^{\,2}}{\bar\sigma_{v,t+1}^{\,2}}
=\theta^2\Big(1-\mathrm{mmse}_{\mathsf U}(\Bar{w}_{1,t})\Big).\label{eq:w2-odds}
\end{align}
An analogous application of the tower property to \eqref{eq:ft-def} gives
\begin{align}
\frac{\Bar{w}_{1,t}}{1-\Bar{w}_{1,t}}
=\frac{\bar\mu_{u,t}^{\,2}}{\bar\sigma_{u,t}^{\,2}}
=\frac{\theta^2}{\delta}\Big(1-\mathrm{mmse}_{\mathsf V}(\Bar{w}_{2,t})\Big).\label{eq:w1-odds}
\end{align}
Finally, letting $t\to\infty$ in \eqref{eq:w2-odds}–\eqref{eq:w1-odds} under the assumed convergence
$\Bar{w}_{1,t}\to\overline w_1$ and $\Bar{w}_{2,t}\to\overline w_2$ gives exactly \eqref{eq:fp-rect}.
\end{proof}

\subsection{Proof of \propref{thm:OAMP_Wigner_FP}}\label{app:pf-thm:OAMP_Wigner_FP}

We now prove the main proposition regarding the equivalence of the OAMP and AMP state evolution equations.
The goal is to show that, by Lemma~\ref{lem:bayes-amp-fp}, the fixed-point equations of the OAMP algorithm \eqref{eq:OptimalRecursion} can be simplified into
\BS
\begin{align}
\frac{w_1}{1-w_1} & = \frac{\theta^2}{\delta} \left(1 -\mathrm{mmse}_{\mathsf{V}}(w_2)\right), \label{eq:Equivalence of OAMP and AMP in IID Gaussian_u_appendix}\\
\frac{w_2}{1-w_2}  & = \theta^2 \left(1 -\mathrm{mmse}_{{\mathsf{U}}}(w_1)\right). \label{eq:Equivalence of OAMP and AMP in IID Gaussian_v_appendix}
\end{align}
\ES
Our proof relies on the following lemma.
\begin{lemma}\label{lem:deformed_MP_system}
\BS
Define the auxiliary parameters
\begin{align}
&a(\rho_1) \bydef \frac{\delta}{\theta^2}(1+\rho_1) + 1, 
\quad 
b(\rho_2) \bydef \frac{\delta}{\theta^2}(1+\rho_2) + \delta,\\
&T(\lambda_*) = \rho_2\left(1+\frac{\delta}{\theta^2}\right) + \rho_1\left(\delta+\frac{\delta}{\theta^2}\right) + \rho_1\rho_2\frac{\delta}{\theta^2} + \lambda_*.
\end{align}
Let $\Smu(z)$ denote the Stieltjes transform of the MP law
\begin{equation}\label{eq: steiltjes transform of MP law}
\Smu(z) = \frac{1}{2\delta z}\left[z + \delta - 1 - \sqrt{(z - \delta - 1)^2 - 4\delta}\right].
\end{equation}
Then the resolvent $\Smu$ satisfies the coupled system for $T(\lambda_*)$ outside the bulk
\begin{align}
a(\rho_1)\bigl[\delta  \Smu(T(\lambda_*)) + \tfrac{1-\delta}{T(\lambda_*)}\bigr] &= 1 + \frac{\delta}{\theta^2}\left(1 - \frac{1 - \rho_1\, b(\rho_2)\, \Smu(T(\lambda_*))}{b(\rho_2)\, \Smu(T(\lambda_*))}\right), \label{eq: Wigner hatS}\\
b(\rho_2)\, \Smu(T(\lambda_*)) &= 1 + \frac{1}{\theta^2}\left(1 - \frac{1 - \rho_2\, a(\rho_1)\bigl[\delta  \Smu(T(\lambda_*)) + \tfrac{1-\delta}{T(\lambda_*)}\bigr]}{a(\rho_1)\bigl[\delta  \Smu(T(\lambda_*)) + \tfrac{1-\delta}{T(\lambda_*)}\bigr]}\right).\label{eq: Wigner S}
\end{align}
\ES
\end{lemma}

\begin{proof}
We show that \eqref{eq: Wigner hatS}–\eqref{eq: Wigner S} are equivalent to the MP self-consistent equation at $z=T(\lambda_*)$. Throughout, set
\[
\Smu \;\bydef\; \Smu\!\bigl(T(\lambda_*)\bigr), 
\qquad 
T \;\bydef\; T(\lambda_*),
\qquad 
\mathcal{S}_{\tilde{\mu}} \;\bydef\; \delta\,\Smu + \frac{1-\delta}{T}.
\]
Start from \eqref{eq: Wigner S} and isolate $a(\rho_1)\,\mathcal{S}_{\tilde{\mu}}$:
\begin{align}
b(\rho_2)\,\Smu 
&= 1 + \frac{1}{\theta^2}\!\left(1 - \frac{1 - \rho_2\,a(\rho_1)\,\mathcal{S}_{\tilde{\mu}}}{a(\rho_1)\,\mathcal{S}_{\tilde{\mu}}}\right)
= 1 + \frac{1}{\theta^2}\!\left(1+\rho_2 - \frac{1}{a(\rho_1)\,\mathcal{S}_{\tilde{\mu}}}\right),\\
\theta^2\bigl(b(\rho_2)\,\Smu - 1\bigr) &= 1+\rho_2 - \frac{1}{a(\rho_1)\,\mathcal{S}_{\tilde{\mu}}}
\quad\Longrightarrow\quad
a(\rho_1)\,\mathcal{S}_{\tilde{\mu}} \;=\; \frac{1}{\,1+\rho_2+\theta^2 - \theta^2 b(\rho_2)\,\Smu\,}.
\label{eq:aX-iso}
\end{align}
Substitute \eqref{eq:aX-iso} into \eqref{eq: Wigner hatS} and simplify the right-hand side:
\begin{align}
1 + \frac{\delta}{\theta^2}\!\left(1 - \frac{1 - \rho_1\, b(\rho_2)\,\Smu}{b(\rho_2)\,\Smu}\right)
= 1 + \frac{\delta}{\theta^2}\!\left(\frac{(1+\rho_1)\,b(\rho_2)\,\Smu - 1}{b(\rho_2)\,\Smu}\right)
= a(\rho_1) - \frac{\delta}{\theta^2\,b(\rho_2)\,\Smu}.
\label{eq:rhs-simpl}
\end{align}
Equating \eqref{eq:aX-iso} and \eqref{eq:rhs-simpl} gives a single identity in $\Smu$:
\begin{equation}
\frac{1}{\,1+\rho_2+\theta^2 - \theta^2 b(\rho_2)\,\Smu\,}
\;=\;
a(\rho_1) - \frac{\delta}{\theta^2\,b(\rho_2)\,\Smu}.
\label{eq:single}
\end{equation}
Clearing denominators and rearranging yields the quadratic
\begin{align}
\bigl(a(\rho_1)\,\theta^2\, b(\rho_2)\bigr)\,\Smu^2 
\;-\; \Bigl[a(\rho_1)\,(1+\rho_2+\theta^2) + \delta - 1\Bigr]\Smu 
\;+\; \frac{\delta(1+\rho_2+\theta^2)}{\theta^2\, b(\rho_2)} \;=\; 0.
\label{eq:int-quad}
\end{align}
Using $a(\rho_1)=1+\frac{\delta}{\theta^2}(1+\rho_1)$ and $b(\rho_2)=\delta+\frac{\delta}{\theta^2}(1+\rho_2)$, one directly computes
\begin{align}
\theta^2 a(\rho_1) b(\rho_2) \;=\; \delta\,T, 
\qquad
\frac{\delta(1+\rho_2+\theta^2)}{\theta^2\, b(\rho_2)} \;=\; 1,
\qquad
a(\rho_1)(1+\rho_2+\theta^2)+\delta-1 \;=\; T+\delta-1.
\end{align}
Substituting these into \eqref{eq:int-quad} gives
\begin{align}
\delta\,T\,\Smu^2 \;-\; (T+\delta-1)\Smu \;+\; 1 \;=\; 0,
\end{align}
which is precisely the MP self-consistent quadratic at $z=T(\lambda_*)$. This completes the proof.
\end{proof}

\begin{proof}[Proof of \propref{thm:OAMP_Wigner_FP}]
We prove the identity for the $\bm{u}$-channel; the proof for the symmetric identity is analogous. The derivation proceeds by simplifying the LHS and showing its equivalence to the RHS via the resolvent system in \lemref{lem:deformed_MP_system}.
\begin{align*}
\text{mmse}_{\mathsf{U}}(w_1)
    &\stackrel{\text{(a)}}{=} \frac{1}{\rho_1 + \frac{1}{1-w_1}}
      \stackrel{\text{(b)}}{=} b(\rho_2)\,\Smu(T(\lambda_*))
      \stackrel{\text{(c)}}{=} 1 + \frac{1}{\theta^2}\left(1 - \frac{1 - \rho_2\, a(\rho_1)\!\left[\delta \Smu(T(\lambda_*) ) + \frac{1-\delta}{T(\lambda_*)}\right]}{a(\rho_1)\!\left[\delta \Smu(T(\lambda_*) ) + \frac{1-\delta}{T(\lambda_*)}\right]}\right) \\
    &= 1 - \frac{1}{\theta^2}\left(\frac{1}{a(\rho_1)\!\left[\delta \Smu(T(\lambda_*)) + \frac{1-\delta}{T(\lambda_*)}\right]} - 1 - \rho_2\right)
      \stackrel{\text{(d)}}{=} 1 - \frac{1}{\theta^2} \frac{w_2}{1-w_2}.
\end{align*}
Here, step (a) follows from the definitions of the parameters. Step (b) is a simplification using the fixed-point equation for $w_1$ \eqref{eq:relationship:etau and wu,wv}. Step (c) applies the resolvent system identity \eqref{eq: Wigner S} from \lemref{lem:deformed_MP_system}. Step (d) uses the expression for effective SNR derived from the fixed-point equation for $w_2$ \eqref{eq:relationship: etav and wu,wv}.
\end{proof}

%% file: appendix/ap-spec_combo.tex
\section{Proof of The Optimal Spectral Estimators}

This appendix provides a detailed analysis comparing the performance of optimal
OAMP with that of optimal linear spectral methods. The core of the analysis is
to first establish the theoretical performance limit of any estimator based on
linear combinations of outlying singular vectors (cf.\
\propref{prop:outlier_characterization}). We then demonstrate how to construct
estimators that achieve this optimal performance in
\propref{prop:optimal_data_driven_estimators}. A crucial element of this
construction is a procedure for determining the relative signs of the different
outlying singular vector components, which is achievable under a non-Gaussian
signal assumption.

\subsection{Proof of \propref{prop:optimal_linear_estimator_cos}}
\label{app:pf-opt-linear-cos}

We prove the claim for $\bm u$; the argument for $\bm v$ is identical.
Using the empirical outlier index set $\mathcal I_M$ from \eqref{eq:emp-outlier-index}, write
\[
\bm u_{\mathrm{PCA}}(\bm c_u)
=\sqrt{M}\sum_{i\in\mathcal I_M}c_{u,i}\,\bm u_i(\bm Y),
\]
and fix any coefficient vector $\bm c_u\neq \bm 0$. By
Proposition~\ref{prop:outlier_characterization}(2)--(3), almost surely for all sufficiently large $M$,
each $i\in\mathcal I_M$ corresponds to a unique $\lambda_i\in\mathcal K_*$, and we use this identification below.
Then
\begin{align}
\limsup_{M\to\infty}\frac{\langle \bm u_{\mathrm{PCA}}(\bm c_u),\bm u_*\rangle^2}
{\|\bm u_{\mathrm{PCA}}(\bm c_u)\|^2\,\|\bm u_*\|^2}
&\stackrel{(a)}{=}\limsup_{M\to\infty}
\frac{\Big(\sum_{i\in\mathcal I_M} c_{u,i}\,\langle \bm u_i(\bm Y),\bm u_*\rangle\Big)^2}
{\Big\|\sum_{i\in\mathcal I_M} c_{u,i}\,\bm u_i(\bm Y)\Big\|^2\,\|\bm u_*\|^2}
\notag\\
&\stackrel{(b)}{=}\limsup_{M\to\infty}
\frac{\Big(\sum_{i\in\mathcal I_M} c_{u,i}\,\langle \bm u_i(\bm Y),\bm u_*/\sqrt M\rangle\Big)^2}
{\big(\|\bm u_*\|^2/M\big)\,\sum_{i\in\mathcal I_M} c_{u,i}^2}
\notag\\
&\stackrel{(c)}{\le}\limsup_{M\to\infty}
\frac{1}{\|\bm u_*\|^2/M}\;
\sum_{i\in\mathcal I_M}\big\langle \bm u_i(\bm Y),\bm u_*/\sqrt M\big\rangle^2
\notag\\
&\stackrel{(d)}{=}\sum_{\lambda\in\mathcal K^*}\nu_1(\{\lambda\}),\quad \text{a.s.}
\label{eq:cos-upper-final}
\end{align}
where (a) substitutes $\bm u_{\mathrm{PCA}}(\bm c_u)$ and cancels the common factor $M$; (b) uses the orthonormality
of $\{\bm u_i(\bm Y)\}_{i=1}^M$ and $\langle \bm u_i(\bm Y),\bm u_*\rangle=\sqrt M\langle \bm u_i(\bm Y),\bm u_*/\sqrt M\rangle$;
(c) follows from Cauchy--Schwarz in $\mathbb R^{|\mathcal I_M|}$; and (d) uses $M^{-1}\|\bm u_*\|^2\to 1$ and
Proposition~\ref{prop:outlier_characterization}(4) together with Claim~(2) and (3).
Moreover, (c) is tight if and only if $(c_{u,i})_{i\in\mathcal I_M}$ is proportional to
$\big(\langle \bm u_i(\bm Y),\bm u_*/\sqrt M\rangle\big)_{i\in\mathcal I_M}$, i.e., the oracle choice in
\eqref{eq:oracle_optimal_spectral_est.}.

\subsection{Proof of Proposition~\ref{prop:signal_plus_noise_sv}}
\label{app:pf-signal-plus-noise-sv}

To justify the empirical procedure for relative sign detection, we must
understand the \emph{joint} distribution of the singular vectors of
$\bm{Y}$ associated with all outlying singular values in
$\mathcal{K}^\ast$. This can be obtained from a self-consistent
representation of these singular vectors, in parallel with the OAMP state
evolution in Theorem~\ref{Thm: State Evolution}.

\begin{proof}[Proof of Proposition~\ref{prop:signal_plus_noise_sv}]
We focus on the left singular vectors; the proof for the right singular
vectors is entirely symmetric.

% \jjm{We first show that $0\notin \mathcal{K}$. Add proof.}

\textbf{Proof of Claim (1).} Fix the finite set of population outlying squared singular values
$\mathcal{K}^\ast=\{\lambda_1,\ldots,\lambda_K\}$, and write
$\sigma_k\bydef\sqrt{\lambda_k}$ for $1\le k\le K$.
For each $k$, let $\lambda_{k,M}$ denote the corresponding empirical
outlying squared singular value of $\bm Y$, with associated singular value
$\sigma_{k,M}\bydef\sqrt{\lambda_{k,M}}$, and let
$(\bm u_k,\bm v_k)$ be the associated left and right singular vectors.
By Proposition~\ref{prop:outlier_characterization}, for all sufficiently
large $M$, the indexing is well defined and
$\lambda_{k,M}\ac\lambda_k$
(equivalently, $\sigma_{k,M}\ac\sigma_k$);
in particular, $\sigma_{k,M}\neq 0$ almost surely for all large $M$ (cf. \lemref{lem:Gamma-analytic},~Claim~(3)).

Since $\sigma_{k,M}\neq 0$, for each $k$ we can write the singular vector into the following resolvent-based representation (cf.~Fact \ref{lem:specSE:resolvent-outliers}):
\BS
\BE
\bm u_k
=
\bigl(\lambda_{k,M}\bm I-\bm W\bm W^\UT\bigr)^{-1}
\left(
\frac{\theta\sigma_{k,M}}{\sqrt{MN}}
\langle\bm v_\ast,\bm v_k\rangle\,\bm u_\ast
+\frac{\theta}{\sqrt{MN}}
\langle\bm u_\ast,\bm u_k\rangle\,\bm W\bm v_\ast
\right),
\EE
\ES
where the inverse exists almost surely for all sufficiently large $M$, since
$\lambda_{k,M}$ lies outside the spectrum of $\bm W\bm W^\UT$. We are interested in the projections of $\bm u_\ast$ along the outlying left
singular vectors $\bm u_k$.
Multiplying both sides by $\langle\bm u_\ast,\bm u_k\rangle$ yields
\BS
\begin{eqnarray*}
\bm{u}_k^{\mathrm{OUT}}\bydef\langle\bm u_\ast,\bm u_k\rangle\,\bm u_k
& = &
\bigl(\lambda_{k,M}\bm I-\bm W\bm W^\UT\bigr)^{-1}
\left(
\frac{\theta\sigma_{k,M}}{\sqrt{MN}}
\langle\bm v_\ast,\bm v_k\rangle
\langle\bm u_\ast,\bm u_k\rangle\,\bm u_\ast
+
\frac{\theta}{\sqrt{MN}}
\langle\bm u_\ast,\bm u_k\rangle^2\,\bm W\bm v_\ast
\right)\\
&\explain{$M\to\infty$}{\simeq} &
\bigl(\lambda_k\bm I-\bm W\bm W^\UT\bigr)^{-1}
\left(
\frac{\theta\sigma_k}{\sqrt{MN}}
\langle\bm v_\ast,\bm v_k\rangle
\langle\bm u_\ast,\bm u_k\rangle\,\bm u_\ast
+
\frac{\theta}{\sqrt{MN}}
\langle\bm u_\ast,\bm u_k\rangle^2\,\bm W\bm v_\ast
\right)\\
&\bydef & \bm p_{u,k}.
\end{eqnarray*}
\ES
The approximation error above is due to the replacement of the empirical
eigenvalue $\lambda_{k,M}$ and singular value $\sigma_{k,M}$ by their
population limits $\lambda_k$ and $\sigma_k$.
To justify this replacement, note first that $\sigma_{k,M}\ac\sigma_k$, so the
difference between the corresponding scalar prefactors vanishes.
For the resolvent term, we use the resolvent identity
\BE\label{Eqn:resolvent_identity}
\bigl(\lambda_{k,M}\bm I-\bm W\bm W^\UT\bigr)^{-1}
-
\bigl(\lambda_k\bm I-\bm W\bm W^\UT\bigr)^{-1}
=
(\lambda_k-\lambda_{k,M})\,
\bigl(\lambda_{k,M}\bm I-\bm W\bm W^\UT\bigr)^{-1}
\bigl(\lambda_k\bm I-\bm W\bm W^\UT\bigr)^{-1}.
\EE
Since $\lambda_{k,M}$ and $\lambda_k$ remain at a strictly positive distance
from $\spp(\bm W\bm W^\UT)$ almost surely for all sufficiently large $M$, the
standard resolvent bound
$
\|(\lambda I-\bm W\bm W^\UT)^{-1}\|_{\mathrm{op}}
=
\mathrm{d}\bigl(\lambda,\spp(\bm W\bm W^\UT)\bigr)^{-1}
$
implies that both
$\|(\lambda_{k,M}\bm I-\bm W\bm W^\UT)^{-1}\|_{\mathrm{op}}$ and
$\|(\lambda_k\bm I-\bm W\bm W^\UT)^{-1}\|_{\mathrm{op}}$
are uniformly bounded.
Together with $\lambda_{k,M}\ac\lambda_k$, this shows that the error incurred
by replacing $(\lambda_{k,M},\sigma_{k,M})$ with $(\lambda_k,\sigma_k)$ is
asymptotically negligible under any of the convergence notions used below.

We next analyze the asymptotic distributions of $\bm p_{u,k}$:
\BE
\bm p_{u,k}
=
\bigl(\lambda_k\bm I-\bm W\bm W^\UT\bigr)^{-1}
\left(
\frac{\theta\sigma_k}{\sqrt{MN}}
\langle\bm v_\ast,\bm v_k\rangle
\langle\bm u_\ast,\bm u_k\rangle\,\bm u_\ast
+
\frac{\theta}{\sqrt{MN}}
\langle\bm u_\ast,\bm u_k\rangle^2\,\bm W\bm v_\ast
\right).
\EE
To isolate the deterministic signal-aligned contribution, we subtract and add
the averaged resolvent trace acting on $\bm u_\ast$.
Specifically, write
$$
\bigl(\lambda_k\bm I-\bm W\bm W^\UT\bigr)^{-1}\bm u_\ast
=
\frac{1}{M}\Tr\!\bigl(\lambda_k\bm I-\bm W\bm W^\UT\bigr)^{-1}\bm u_\ast
+
\Bigl[
\bigl(\lambda_k\bm I-\bm W\bm W^\UT\bigr)^{-1}
-
\frac{1}{M}\Tr\!\bigl(\lambda_k\bm I-\bm W\bm W^\UT\bigr)^{-1}\bm I_M
\Bigr]\bm u_\ast.
$$
Accordingly, we decompose $\bm p_{u,k}$ as
\BS
\BE
\bm p_{u,k}=\bm s_{u,k}+\bm n^{(1)}_{u,k}+\bm n^{(2)}_{u,k},
\EE
where the signal component $\bm s_{u,k}$ and the noise components
$\bm n^{(1)}_{u,k}$ and $\bm n^{(2)}_{u,k}$ are respectively given by
\begin{align}
\bm s_{u,k}
&\bydef
\left(
\frac{\theta\sigma_k}{\sqrt{MN}}
\langle\bm v_\ast,\bm v_k\rangle
\langle\bm u_\ast,\bm u_k\rangle
\right)
\left(
\frac{1}{M}\Tr\bigl(\lambda_k\bm I_M-\bm W\bm W^\UT\bigr)^{-1}
\right)\bm u_\ast,\\
\bm n^{(1)}_{u,k}
&\bydef
\left(
\frac{\theta\sigma_k}{\sqrt{MN}}
\langle\bm v_\ast,\bm v_k\rangle
\langle\bm u_\ast,\bm u_k\rangle
\right)
\Bigl[
\bigl(\lambda_k\bm I_M-\bm W\bm W^\UT\bigr)^{-1}
-
\frac{1}{M}\Tr\bigl(\lambda_k\bm I_M-\bm W\bm W^\UT\bigr)^{-1}\bm I_M
\Bigr]\bm u_\ast,\\
\bm n^{(2)}_{u,k}
&\bydef
\left(
\frac{\theta}{\sqrt{MN}}
\langle\bm u_\ast,\bm u_k\rangle^2
\right)
\bigl(\lambda_k\bm I_M-\bm W\bm W^\UT\bigr)^{-1}\bm W\bm v_\ast.
\end{align}
\ES

Collecting these vectors columnwise, define the matrices
$$
\bm P_u\bydef[\bm p_{u,1}\ \cdots\ \bm p_{u,K}],\qquad
\bm S_u\bydef[\bm s_{u,1}\ \cdots\ \bm s_{u,K}],\qquad
\bm N_u\bydef[\bm n_{u,1}\ \cdots\ \bm n_{u,K}],
$$
so that
$$
\bm P_u=\bm S_u+\bm N_u\in\mathbb R^{M\times K}.
$$
We now analyze the joint limit of $(\bm P_u,\bm S_u,\bm N_u)$, where each
$\bm n_{u,k}=\bm n^{(1)}_{u,k}+\bm n^{(2)}_{u,k}$ is the sum of the centered
(resolvent--$u_\ast$) fluctuation and the transverse ($\bm W\bm v_\ast$)
fluctuation.

\medskip\noindent{\it Signal component.}
Recall that
$$
\bm s_{u,k}
=
\left(
\frac{\theta\sigma_k}{\sqrt{MN}}
\langle\bm v_\ast,\bm v_k\rangle
\langle\bm u_\ast,\bm u_k\rangle
\right)
\left(
\frac{1}{M}\Tr(\lambda_k\bm I-\bm W\bm W^\UT)^{-1}
\right)\bm u_\ast.
$$
By Proposition~\ref{prop:outlier_characterization} (Claim~2) and the
definition of $\nu_3(\{\sigma_k\})$, the scalar prefactor converges almost
surely to
$$
\frac{\theta\sigma_k}{\sqrt{MN}}
\langle\bm v_\ast,\bm v_k\rangle
\langle\bm u_\ast,\bm u_k\rangle
\ac
2\theta\sigma_k\frac{1+\delta}{\sqrt\delta}\,\nu_3(\{\sigma_k\}).
$$
Moreover, by the standard trace convergence for resolvents:
$$
\frac{1}{M}\Tr(\lambda_k\bm I-\bm W\bm W^\UT)^{-1}
\ac
\Smu(\lambda_k).
$$
Combining these and substituting the explicit expression the point mass $\nu_3(\{\sigma_k\})$ from \lemref{lem:spectral_measures_properties} yields
\begin{align*}
\left( 2\theta\sigma_k \frac{1+\delta}{\sqrt{\delta}}\nu_3(\{\sigma_k\}) \right) \Smu(\lambda_k)
&= 2\theta\sigma_k \frac{1+\delta}{\sqrt{\delta}}
   \left( - \frac{\sqrt{\delta}}{1+\delta} \frac{1}{2\theta^3\sigma_k C'(\lambda_k)} \right) \Smu(\lambda_k) \\
&= \frac{\Smu(\lambda_k)}{-\theta^2 C'(\lambda_k)}
 = \nu_1(\{\lambda_k\}).
\end{align*}
Hence, in the sense of empirical row laws,
$$
\bm S_u \wc \bm{\mathsf S}_u
\bydef\bigl(\nu_1(\{\lambda_1\})\mathsf U_\ast,\ldots,
\nu_1(\{\lambda_K\})\mathsf U_\ast\bigr)^\UT,
$$
where $\mathsf U_\ast$ is the limiting signal coordinate distribution.

\medskip\noindent{\it Noise component: joint Gaussianity of }
$\bm n^{(1)}_{u,k}$ {\it and } $\bm n^{(2)}_{u,k}$.
Let $\bm W=\bm U_W\bm\Sigma_W\bm V_W^\UT$ be the singular value decomposition of
$\bm W$.  Write
$$
\tilde{\bm u}\bydef \bm U_W^\UT\bm u_\ast,\qquad
\tilde{\bm v}\bydef \bm V_W^\UT\bm v_\ast.
$$
Under Assumption~\ref{assump:main}, $\bm U_W$ and $\bm V_W$ are independent Haar
matrices and are independent of $(\bm u_\ast,\bm v_\ast)$, so
$\tilde{\bm u}$ and $\tilde{\bm v}$ have asymptotically standard normal
coordinates and are independent of $\bm\Sigma_W$ in the sense of empirical laws
(cf.\ \cite[Appendix~E--F]{fan2022approximate}).

Define the diagonal matrix
$$
\bm G_k\bydef (\lambda_k\bm I-\bm\Sigma_W\bm\Sigma_W^\UT)^{-1},
\qquad
\bar g_{k,M}\bydef \frac{1}{M}\Tr(\bm G_k).
$$
Then the centered resolvent term can be written as
$$
\bm n^{(1)}_{u,k}
=
\alpha_{k,M}\,
\bm U_W(\bm G_k-\bar g_{k,M}\bm I)\tilde{\bm u},
\qquad
\alpha_{k,M}\bydef
\frac{\theta\sigma_k}{\sqrt{MN}}
\langle\bm v_\ast,\bm v_k\rangle
\langle\bm u_\ast,\bm u_k\rangle,
$$
and the transverse term can be written as
$$
\bm n^{(2)}_{u,k}
=
\beta_{k,M}\,
\bm U_W\bm D_k\tilde{\bm v},
\qquad
\bm D_k\bydef (\lambda_k\bm I-\bm\Sigma_W\bm\Sigma_W^\UT)^{-1}\bm\Sigma_W,
\qquad
\beta_{k,M}\bydef \frac{\theta}{\sqrt{MN}}\langle\bm u_\ast,\bm u_k\rangle^2.
$$
Both $\alpha_{k,M}$ and $\beta_{k,M}$ converge almost surely to deterministic
limits (jointly over $k$) by Proposition~\ref{prop:outlier_characterization}.

For each $k$, set
$$
\bm q^{(1)}_{k}\bydef \alpha_{k,M}(\bm G_k-\bar g_{k,M}\bm I)\tilde{\bm u},
\qquad
\bm q^{(2)}_{k}\bydef \beta_{k,M}\bm D_k\tilde{\bm v},
\qquad
\bm q_k\bydef \bm q^{(1)}_{k}+\bm q^{(2)}_{k},
$$
so that
$$
\bm n_{u,k}=\bm n^{(1)}_{u,k}+\bm n^{(2)}_{u,k}=\bm U_W\bm q_k.
$$
Collecting $\bm q_k$ columnwise yields an $M\times K$ matrix
$$
\bm Q\bydef[\bm q_1\ \cdots\ \bm q_K],
\qquad\text{so that}\qquad
\bm N_u=\bm U_W\bm Q.
$$

\smallskip
{\it Step 1 (Gaussian limit for the coefficient rows).}
Each coordinate of $(\bm q_1,\ldots,\bm q_K)$ is an affine function of
$(\tilde u_i,\tilde v_i)$ with coefficients given by bounded continuous
functions of the singular values of $\bm W$ (through $\bm G_k$ and $\bm D_k$).
By \cite[Propositions~E.2 and~E.4]{fan2022approximate} and
\cite[Lemma~G.4]{zhong2021approximate}, for fixed $K$, the empirical joint law
of the rows of $\bm Q$ converges to a centered Gaussian vector
$\bm{\mathsf Q}\in\mathbb R^K$ with deterministic covariance matrix
$\bm\Sigma_Q$, i.e.,
$$
\bm Q \wc \bm{\mathsf Q},\qquad \bm{\mathsf Q}\sim\mathcal N(\bm 0,\bm\Sigma_Q).
$$

\smallskip
{\it Step 2 (Haar mixing).}
Since $\bm Q\indep \bm U_W$ under Assumption~\ref{assump:main}, rank-$K$ Haar
mixing (cf.\ \cite[Lemma~G.5]{zhong2021approximate}) implies that the empirical
row law of $\bm N_u=\bm U_W\bm Q$ converges to a centered Gaussian vector
$\bm{\mathsf N}_u\in\mathbb R^K$ with covariance $\bm\Sigma_Q$, independent of
$\mathsf U_\ast$, namely,
$$
(\bm S_u,\bm N_u)\xrightarrow[]{W}(\bm{\mathsf S}_u,\bm{\mathsf N}_u),
\qquad
\bm{\mathsf N}_u\sim\mathcal N(\bm 0,\bm\Sigma_Q),
\qquad
\bm{\mathsf N}_u\indep \mathsf U_\ast.
$$
Finally, since addition is Lipschitz on $\mathbb R^K$, applying similar arguments as in
\cite[Lemma~G.4]{zhong2021approximate} once more yields
$$
\bm P_u=\bm S_u+\bm N_u \xrightarrow[]{W} \bm{\mathsf S}_u+\bm{\mathsf N}_u
\bydef (\mathsf U^{\mathrm{OUT}}_1,\ldots,\mathsf U^{\mathrm{OUT}}_K)^\UT,
$$
where
$$
\mathsf U^{\mathrm{OUT}}_k=\nu_1(\{\lambda_k\})\,\mathsf U_\ast+\mathsf N_{u,k},
\qquad 1\le k\le K,
\qquad
(\mathsf N_{u,1},\ldots,\mathsf N_{u,K})\sim\mathcal N(\bm 0,\bm\Sigma_Q).
$$

\medskip\noindent
{\it Identification of $\bm\Sigma_Q$ via orthogonality.}
For $k\neq\ell$, the vectors $\bm p_{u,k}=\langle\bm u_\ast,\bm u_k\rangle\bm u_k$
are orthogonal for each finite $M$, hence $\langle\bm p_{u,k},\bm p_{u,\ell}\rangle=0$.
Passing to the row-law limit gives
$$
\mathbb E[\mathsf U^{\mathrm{OUT}}_k\mathsf U^{\mathrm{OUT}}_\ell]=0,
\qquad k\neq \ell.
$$
Since $\mathsf U^{\mathrm{OUT}}_k=\nu_1(\{\lambda_k\})\mathsf U_\ast+\mathsf N_{u,k}$
with $\bm{\mathsf N}_u\indep \mathsf U_\ast$ and $\mathbb E[\mathsf U_\ast^2]=1$,
this determines the off-diagonal entries of $\bm\Sigma_Q$ and yields the stated
formula for $\Sigma_u^{\mathrm{OUT}}$.
Moreover, $\mathbb{E}\left[\left(\mathsf U^{\mathrm{OUT}}_k\right)^2\right]=\nu_1(\{\lambda_k\})$ and therefore
$$
\mathrm{Var}(\mathsf N_{u,k})
=\mathbb{E}\left[\left(\mathsf U^{\mathrm{OUT}}_k\right)^2\right]-\nu_1(\{\lambda_k\})^2
=\nu_1(\{\lambda_k\})-\nu_1^2(\{\lambda_k\}).
$$
Thus we have
\begin{align}
 \mathsf{U}_k^{\mathrm{OUT}}
    &=\nu_1(\{\lambda_k\})\,\mathsf{U}_*
      +\sqrt{\nu_1(\{\lambda_k\})-\nu_1^2(\{\lambda_k\})}\,\mathsf{Z}_{u,k}, \qquad \mathsf{Z}_{u,k}\sim \mathcal{N}(0,1) \indep \mathsf{U}_*.   
\end{align}
The proof for the right singular vectors is identical, replacing
$\bm{W}\bm{W}^\UT$ by $\bm{W}^\UT\bm{W}$ and interchanging the roles
of $(\bm{u}_*,M)$ and $(\bm{v}_*,N)$ throughout, so we omit the
details.

\medskip\noindent \textbf{Proof of Claim (2).}
For $k\neq\ell$, the orthogonality of the empirical projected vectors
$\bm{p}_{u,k}=\langle\bm{u}_*,\bm{u}_k(\bm{Y})\rangle\bm{u}_k(\bm{Y})$ implies
$\langle\bm{p}_{u,k},\bm{p}_{u,\ell}\rangle=0$ for each finite $M$. Passing to
the row–law limit in~\eqref{eq:joint-distribution-projector-u} gives
\[
\mathbb{E}\bigl[\mathsf{U}_k^{\mathrm{OUT}}\mathsf{U}_\ell^{\mathrm{OUT}}\bigr]=0,
\qquad k\neq\ell.
\]
Using the independence 
$\bigl(\mathsf{Z}_{u,1},\ldots,\mathsf{Z}_{u,K}\bigr)\indep\mathsf{U}_*$ with $\E[\mathsf{U}_*^2]=1$,
\begin{align*}
0
&=\mathbb{E}\bigl[\mathsf{U}_k^{\mathrm{OUT}}\mathsf{U}_\ell^{\mathrm{OUT}}\bigr]=\nu_1(\{\lambda_k\})\nu_1(\{\lambda_\ell\})
+\sqrt{\nu_1(\{\lambda_k\})-\nu_1^2(\{\lambda_k\})}
 \sqrt{\nu_1(\{\lambda_\ell\})-\nu_1^2(\{\lambda_\ell\})}\,
 \mathbb{E}[\mathsf{Z}_{u,k}\mathsf{Z}_{u,\ell}],
\end{align*}
since by \lemref{lem:spectral_measures_properties} we have $\nu_1(\{\lambda_k\}) \in (0,1)$, so for $k\neq\ell$,
\[
\mathbb{E}[\mathsf{Z}_{u,k}\mathsf{Z}_{u,\ell}]
=-\frac{\nu_1(\{\lambda_k\})\,\nu_1(\{\lambda_\ell\})}
       {\sqrt{\nu_1(\{\lambda_k\})-\nu_1^2(\{\lambda_k\})}\,
        \sqrt{\nu_1(\{\lambda_\ell\})-\nu_1^2(\{\lambda_\ell\})}}.
\]
This yields the stated covariance entries for $\Sigma_u^{\mathrm{OUT}}$.
The same argument with $\nu_2$ and $(\mathsf{V}_k^{\mathrm{OUT}})$ gives the
expression for $\Sigma_v^{\mathrm{OUT}}$.
\medskip
It remains to justify that $\Sigma_u^{\mathrm{OUT}}$ is positive definite.
Since $\nu_1$ is a probability
measure and $\nu_1^{\parallel}(\mathbb{R}_+)>0$, the total atomic mass on the
outliers satisfies
\[
\sum_{k=1}^K \nu_1(\{\lambda_k\})
=1-\nu_1^{\parallel}(\mathbb{R}_+)\in(0,1).
\]
For any $\bm{x}=(x_1,\ldots,x_K)^\UT\in\mathbb{R}^K$, define
\[
y_k
\bydef
\frac{x_k}{\sqrt{\nu_1(\{\lambda_k\})-\nu_1^2(\{\lambda_k\})}},
\qquad k=1,\ldots,K.
\]
A direct computation using the explicit entries of $\Sigma_u^{\mathrm{OUT}}$
shows that
\begin{equation}\label{eq:xSigmax-rewrite}
\bm{x}^\UT\Sigma_u^{\mathrm{OUT}}\bm{x}
=\sum_{k=1}^K \nu_1(\{\lambda_k\})\,y_k^2
-\Bigl(\sum_{k=1}^K \nu_1(\{\lambda_k\})\,y_k\Bigr)^2.
\end{equation}
By the Cauchy--Schwarz inequality with weights
$\{\nu_1(\{\lambda_k\})\}_{k=1}^K$,
\[
\Bigl(\sum_{k=1}^K \nu_1(\{\lambda_k\})\,y_k\Bigr)^2
\le
\Bigl(\sum_{k=1}^K \nu_1(\{\lambda_k\})\Bigr)
\Bigl(\sum_{k=1}^K \nu_1(\{\lambda_k\})\,y_k^2\Bigr),
\]
and hence, combining with~\eqref{eq:xSigmax-rewrite},
\[
\bm{x}^\UT\Sigma_u^{\mathrm{OUT}}\bm{x}
\ge
\Bigl(1-\sum_{k=1}^K \nu_1(\{\lambda_k\})\Bigr)
\sum_{k=1}^K \nu_1(\{\lambda_k\})\,y_k^2.
\]
As observed above, $1-\sum_{k=1}^K \nu_1(\{\lambda_k\})>0$ and
$\nu_1(\{\lambda_k\})>0$ for all $k$, so the right-hand side is strictly
positive whenever $\bm{x}\neq\bm{0}$ (equivalently,
$(y_1,\ldots,y_K)\neq\bm{0}$). Thus $\Sigma_u^{\mathrm{OUT}}$ is positive
definite. The same reasoning, with $\nu_2$ in place of $\nu_1$, shows that
$\Sigma_v^{\mathrm{OUT}}$ is also positive definite.
\end{proof}

\subsection{Proof of \propref{prop:optimal_data_driven_estimators}}
\label{app:optimal_data_driven_estimators}

We treat the $\bm u$--channel; the $\bm v$--channel is identical. Recall from
\eqref{eq: optimal PCA estimators} that
\[
  \bm u^*_{\mathrm{PCA}}
  =
  \sum_{i\in\mathcal I_M} s_i^u \sqrt{\nu_1(\{\lambda_i\})}\,\bm u_i^\sharp,
  \qquad
  \bm u_i^\sharp=\sqrt{M}\,\xi_i\,\bm u_i(\bm Y).
\]
By Proposition~\ref{prop:outlier_characterization}(2)--(3), almost surely for all sufficiently large $M$,
each outlier window contains exactly one eigenvalue of $\bm Y\bm Y^\UT$. Hence, for each $i\in\mathcal I_M$
we may associate a unique population outlier $\lambda_i\in\mathcal K_*$.
Then
\begin{align}
\lim_{M\to\infty}\frac{\langle \bm{u}^*_{\mathrm{PCA}}, \bm{u}_* \rangle^2}
     {\| \bm{u}^*_{\mathrm{PCA}} \|^2 \|\bm{u}_*\|^2}
&\stackrel{(a)}{=}
\lim_{M\to\infty}\frac{\Big(\sum_{i\in\mathcal I_M} s_i^u\sqrt{\nu_1(\{\lambda_i\})}\,\langle \bm u_i^\sharp,\bm u_*\rangle\Big)^2}
     {\Big(M\sum_{i\in\mathcal I_M}\nu_1(\{\lambda_i\})\Big)\,\|\bm u_*\|^2}
\notag\\
&\stackrel{(b)}{=}\lim_{M\to\infty}
\frac{\Big(\sum_{i\in\mathcal I_M} s_i^u\sqrt{\nu_1(\{\lambda_i\})}\,
               \big(\tfrac1M\langle \bm u_i^\sharp,\bm u_*\rangle\big)\Big)^2}
     {\Big(\sum_{i\in\mathcal I_M}\nu_1(\{\lambda_i\})\Big)\,\big(\tfrac1M\|\bm u_*\|^2\big)}
\notag\\
&\stackrel{(c)}{\xrightarrow[\mathrm{a.s.}]{}}
\sum_{\lambda\in\mathcal K^*}\nu_1(\{\lambda\}).
\end{align}
Here (a) expands $\bm u^*_{\mathrm{PCA}}$ and uses
$\langle \bm u_i^\sharp,\bm u_j^\sharp\rangle=M\,\mathbf 1\{i=j\}$ to evaluate
$\|\bm u^*_{\mathrm{PCA}}\|^2=M\sum_{i\in\mathcal I_M}\nu_1(\{\lambda_i\})$.
Step (b) divides the numerator and denominator by $M^2$.
Step (c) uses $\tfrac1M\|\bm u_*\|^2\xrightarrow{\mathrm a.s.}1$ and, for each fixed $i\in\mathcal I_M$,
\[
\frac1M\langle \bm u_i^\sharp,\bm u_*\rangle
\xrightarrow{\mathrm a.s.}
[\bm s_{u,*}^{\mathrm R}]_i\,\sqrt{\nu_1(\{\lambda_i\})}
\qquad\text{and}\qquad
s_i^u\xrightarrow{\mathrm a.s.}[\bm s_{u,*}^{\mathrm R}]_i,
\]
which follow from Proposition~\ref{prop:signal_plus_noise_sv} together with
\eqref{eq:limit-matrix-U-rand}, and from Proposition~\ref{prop:outlier_characterization}(4) for the limiting overlaps.
Finally, by Proposition~\ref{prop:outlier_characterization}(2) the correspondence $i\in\mathcal I_M\leftrightarrow\lambda_i\in\mathcal K_*$ for all large $M$, so $\sum_{i\in\mathcal I_M}\nu_1(\{\lambda_i\})=\sum_{\lambda\in\mathcal K^*}\nu_1(\{\lambda\})$.

\subsection{Proof of Proposition~\ref{prop:MLE_relative_sign}}
\label{app:MLE_relative_sign}

We work throughout with the limit model \eqref{eq:limit-vector-U-true} and the notation introduced in Section~\ref{sec:relativesign}. The goal is to characterize when the global sign vector is identifiable from the rows of $\bm U^\sharp$, and to show that the MLE is asymptotically consistent whenever identifiability holds. We first analyze the Gaussian case, then the non-Gaussian case in the $\bm u$–channel, and finally use the inter-channel coupling to transfer the result to the $\bm v$–channel.

\subsubsection{Both Gaussian priors: impossibility of sign recovery}
In this part we show that, when both $\mathsf U_*$ and $\mathsf V_*$ are standard Gaussian, the distribution of the observed rows is independent of the sign vector. Consequently, the relative signs are not identifiable and no estimator can be consistent.

\begin{lemma}[Gaussian non-identifiability]\label{lem:gaussian_row_law}
Assume the setting of Proposition~\ref{prop:signal_plus_noise_sv}, and suppose $\mathsf U_*\sim\mathcal N(0,1)$.
Let $P_s$ be defined as in Proposition~\ref{prop:MLE_relative_sign}. Then $P_s=\mathcal N(\bm 0,\bm I_K)$; in particular, its law does not depend on $\bm s$.
\end{lemma}

\begin{proof}
By definition, $P_s$ is the joint density of
\[
\bigl([\bm s]_\ell \sqrt{\nu_1(\{\lambda_\ell\})}\,\mathsf U_*+\sqrt{1-\nu_1(\{\lambda_\ell\})}\,\mathsf Z_\ell\bigr)_{\ell\in\mathcal I_M},
\]
where $\mathsf U_*$ is independent of
\[
\bigl(\sqrt{1-\nu_1(\{\lambda_\ell\})}\,\mathsf Z_\ell\bigr)_{\ell\in\mathcal I_M}
\sim \mathcal N\!\bigl(\bm 0,\ \bm I_K-\bm\gamma(\bm s)\bm\gamma(\bm s)^\UT\bigr),
\qquad
\bm\gamma(\bm s)\bydef\bigl([\bm s]_\ell\sqrt{\nu_1(\{\lambda_\ell\})}\bigr)_{\ell\in\mathcal I_M}.
\]
If $\mathsf U_*\sim\mathcal N(0,1)$, then
\[
\bigl([\bm s]_\ell \sqrt{\nu_1(\{\lambda_\ell\})}\,\mathsf U_*\bigr)_{\ell\in\mathcal I_M}
\sim \mathcal N\!\bigl(\bm 0,\ \bm\gamma(\bm s)\bm\gamma(\bm s)^\UT\bigr),
\]
and hence, by independence,
\[
P_s=\mathcal N\!\bigl(\bm 0,\ \bm\gamma(\bm s)\bm\gamma(\bm s)^\UT+\bm I_K-\bm\gamma(\bm s)\bm\gamma(\bm s)^\UT\bigr)
=\mathcal N(\bm 0,\bm I_K),
\]
which does not depend on $\bm s$.
\end{proof}

\subsubsection{Non-Gaussian Priors: Identifiability and MLE Consistency}\label{sec:NonGaussian-MLE}

We now assume that $\mathsf U_*$ is non-Gaussian and focus on the
$\bm u$–channel. For each $\bm s\in\mathcal S_r$, let $P_{\bm s}$ be as defined in \propref{prop:MLE_relative_sign}. The proof mostly follows the classical MLE scheme \cite[Chapter~5]{vanderVaart1998asymptotic}, where the only caveat being that we obtain uniform convergence not from a direct law of large numbers but from the Wasserstein convergence of the empirical row measure to $P_{\bm{s}_{u,*}^{\mathrm{R}}}$ in \eqref{eq:limit-matrix-U-rand}. We proceed in the following steps: 

\paragraph{Uniform convergence of empirical log-likelihoods.}
Throughout this section we reindex the empirical outlier index set as
$\mathcal I_M=\{1,\ldots,K\}$, where $K=\card{\mathcal I_M}$.
For each $\bm s\in\mathcal S_r$, let $P_{\bm s}$ be as defined in \propref{prop:MLE_relative_sign}.
Denote the $M$ rows of the singular vectors matrix $\bm U^\sharp$ by $(\bm U^\sharp_{i,:})_{i\le M}$.
Define the sample and population log-likelihoods, respectively, by
\begin{align}
  L_M(\bm s;\bm U^\sharp)
  &\bydef \frac{1}{M}\sum_{i=1}^M \log P_s\big(\bm U^\sharp_{i,:}\big),
  \label{eq:LM-def-lem}\\
  L(\bm s;\bm{\mathsf U}^\sharp)
  &\bydef \E_{P_{\bm s_{u,*}^{\mathrm R}}}\!\big[\log P_s(\bm{\mathsf U}^\sharp)\big],
  \label{eq:L-def-lem}
\end{align}
i.e., $L(\bm s;\bm{\mathsf U}^\sharp)$ is computed under the true law $P_{\bm s_{u,*}^{\mathrm R}}$
(equivalently, conditioning on the ground truth sign vector $\bm s_{u,*}^{\mathrm R}$).
The following lemma establishes uniform convergence of the empirical criterion
$L_M(\bm s;\bm U^\sharp)$ to $L(\bm s;\bm{\mathsf U}^\sharp)$ over the finite set $\mathcal S_r$.

\begin{lemma}[Uniform convergence of empirical log-likelihoods]\label{lem:uniform-emp-log}
Let $L_M(\bm s;\bm U^\sharp)$ and $L(\bm s;\bm{\mathsf U}^\sharp)$ be defined in
\eqref{eq:LM-def-lem}--\eqref{eq:L-def-lem}. Then
\begin{align}
  \sup_{\bm s\in\mathcal S_r}
  \big|L_M(\bm s;\bm U^\sharp)-L(\bm s;\bm{\mathsf U}^\sharp)\big|
  \ac 0,
  \label{eq:LM-uniform-proof}
\end{align}
where it is understood that the same ground truth sign vector $\bm s_{u,*}^{\mathrm R}$ is shared by both
$\bm U^\sharp$ and $\bm{\mathsf U}^\sharp$.
\end{lemma}

\begin{proof}
By Proposition~\ref{prop:signal_plus_noise_sv}, the convergence
\begin{align}
  \bigl(\,\langle \bm u_*,\bm u_1(\bm Y)\rangle\,\bm u_1(\bm Y),\ \ldots,\
  \langle \bm u_*,\bm u_K(\bm Y)\rangle\,\bm u_K(\bm Y)\,\bigr)
  \xrightarrow[]{W}
  \bigl(\,\mathsf U_1^{\mathrm{OUT}},\ldots,\mathsf U_K^{\mathrm{OUT}}\,\bigr)^\UT
  \label{eq:Uout-W-proof}
\end{align}
holds in the sense of Wasserstein convergence of the empirical row measure.
Moreover, by Proposition~\ref{prop:outlier_characterization},
$|\langle \bm u_*,\bm u_k(\bm Y)\rangle|/\sqrt M \ac \sqrt{\nu_1(\{\lambda_k\})}$ for each $k\in[K]$.
Together with the sign randomization in \eqref{eq:limit-matrix-U-rand}, this implies that the relative sign vector
$\bm s^{\mathrm R}_{u,*}$ is independent of the remaining limit randomness, and hence the empirical row measure of
$\bigl(\bm u_1^\sharp,\ldots,\bm u_K^\sharp\bigr)$ converges in Wasserstein distance, under the conditional law given
$\bm s^{\mathrm R}_{u,*}$, i.e.,
\begin{align}\label{eq:limit-usharp-proof}
  \bigl(\bm u_1^\sharp,\ldots,\bm u_K^\sharp\bigr)
  \wc
  \bm{\mathsf U}^\sharp(\bm s^{\mathrm R}_{u,*})
  \in \R^K,
\end{align}
where the $\ell$-th marginal is
\begin{align}\label{eq:limit-vector-U-true}
  \mathsf U_\ell^\sharp\big([\bm s^{\mathrm R}_{u,*}]_\ell\big)
  =
  [\bm s^{\mathrm R}_{u,*}]_\ell \sqrt{\nu_1(\{\lambda_\ell\})}\,\mathsf U_*
  +\sqrt{1-\nu_1(\{\lambda_\ell\})}\,\mathsf Z_\ell,
  \qquad \ell\in[K].
\end{align}

We first show that for any fixed $\bm s\in\mathcal S_r$,
\begin{align}
  L_M(\bm s;\bm U^\sharp)\ac L(\bm s;\bm{\mathsf U}^\sharp).
  \label{eq:LM-fixeds-proof}
\end{align}
To invoke the Wasserstein convergence \eqref{eq:limit-usharp-proof}, it suffices to verify the quadratic growth condition
(cf.\ \cite[Definition~6.8]{Villani2009OptimalTransport})
\begin{align}
  \bigl|\log P_s(\bm x)\bigr|
  \le C(1+\|\bm x\|^2),
  \qquad \forall\,\bm x\in\R^K,\ \forall\,\bm s\in\mathcal S_r.
  \label{eq:logp-growth-proof}
\end{align}
Conditioning \eqref{eq:Usharps} on a fixed $\bm s\in\mathcal S_r$ gives
\begin{align}
  P_s(\bm x)
  &= c(\bm s)\,
     \E_{\mathsf U_*}\Big[
       \exp\Big(
         -\frac12
         (\bm x-\bm\gamma(\bm s)\mathsf U_*)^\UT
         \Sigma(\bm s)^{-1}
         (\bm x-\bm\gamma(\bm s)\mathsf U_*)
       \Big)
     \Big],
  \label{eq:density-mixture-proof}
\end{align}
where
\BS
\begin{align}
  \bm\gamma(\bm s)
  &\bydef \big([\bm s]_k\sqrt{\nu_1(\{\lambda_k\})}\big)_{k\in[K]}\in\R^K,
  \label{eq:gamma-def-proof}\\
  \Sigma(\bm s)
  &\bydef \bm I_K-\bm\gamma(\bm s)\bm\gamma(\bm s)^\UT,
  \label{eq:Ns-cov-proof}\\
  c(\bm s)
  &\bydef (2\pi)^{-K/2}\big(\det\Sigma(\bm s)\big)^{-1/2}.
  \label{eq:c-def-proof}
\end{align}
\ES
By Lemma~\ref{lem:spectral_measures_properties}, $\sum_{k=1}^K\nu_1(\{\lambda_k\})<1$, hence $\Sigma(\bm s)\succ0$ for all $\bm s\in\mathcal S_r$.
Since $\mathcal S_r$ is finite, the eigenvalues of $\Sigma(\bm s)$ are uniformly bounded away from $0$ and $\infty$, and there exist
$0<c_{\min}\le c_{\max}<\infty$ such that
\begin{align}
  c_{\min}\le c(\bm s)\le c_{\max},
  \qquad \bm s\in\mathcal S_r.
  \label{eq:c-bounds-proof}
\end{align}
Taking logarithms in \eqref{eq:density-mixture-proof} yields
\begin{align}
  \log P_s(\bm x)
  &= \log c(\bm s)
     + \log \E_{\mathsf U_*}\Big[
        \exp\Big(
          -\frac12
          (\bm x-\bm\gamma(\bm s)\mathsf U_*)^\UT
          \Sigma(\bm s)^{-1}
          (\bm x-\bm\gamma(\bm s)\mathsf U_*)
        \Big)
      \Big].
  \label{eq:logp-expand-proof}
\end{align}
For the upper bound, the quadratic form in \eqref{eq:logp-expand-proof} is nonnegative, hence
\begin{align}
  \log P_s(\bm x)
  &\le \log c(\bm s)
   \le \log c_{\max},
  \qquad \forall\,\bm x\in\R^K,\ \forall\,\bm s\in\mathcal S_r.
  \label{eq:logp-upper-proof}
\end{align}
For the lower bound we apply Jensen, expand the quadratic form, and then use uniform spectral bounds:
\begin{align}
  \log P_s(\bm x)
  &\stackrel{(a)}{\ge}
    \log c(\bm s)
    -\frac12\,
      \E\Big[
        (\bm x-\bm\gamma(\bm s)\mathsf U_*)^\UT
        \Sigma(\bm s)^{-1}
        (\bm x-\bm\gamma(\bm s)\mathsf U_*)
      \Big] \nonumber\\
  &\stackrel{(b)}{=}
    \log c(\bm s)
    -\frac12\Big(
       \bm x^\UT\Sigma(\bm s)^{-1}\bm x
       -2m\,\bm\gamma(\bm s)^\UT\Sigma(\bm s)^{-1}\bm x
       +\bm\gamma(\bm s)^\UT\Sigma(\bm s)^{-1}\bm\gamma(\bm s)
     \Big)\nonumber\\
  &\stackrel{(c)}{\ge}
    \log c_{\min}-\tfrac12(C_1\|\bm x\|^2+C_2)
  \stackrel{(d)}{\ge}
    -C(1+\|\bm x\|^2),
  \qquad \forall\,\bm x\in\R^K,\ \forall\,\bm s\in\mathcal S_r,
  \label{eq:logp-lower-proof}
\end{align}
where $m\bydef\E[\mathsf U_*]$ and $\E[\mathsf U_*^2]=1$.
Here (a) applies Jensen's inequality to \eqref{eq:logp-expand-proof}; (b) expands the expectation; (c) uses
\eqref{eq:c-bounds-proof}, uniform spectral bounds on $\Sigma(\bm s)^{-1}$, and boundedness of $\bm\gamma(\bm s)$ to obtain constants
$C_1,C_2>0$; and (d) absorbs constants into a single $C>0$.
Combining \eqref{eq:logp-upper-proof} and \eqref{eq:logp-lower-proof} yields \eqref{eq:logp-growth-proof}.

By the conditional Wasserstein convergence \eqref{eq:limit-usharp-proof} and the growth bound \eqref{eq:logp-growth-proof},
for each fixed $\bm s\in\mathcal S_r$ we obtain
\begin{align}
  L_M(\bm s;\bm U^\sharp)
  =\frac{1}{M}\sum_{i=1}^M \log P_s\big(\bm U^\sharp_{i,:}\big)
  \ac L(\bm s;\bm{\mathsf U}^\sharp).
  \label{eq:LM-pointwise-proof}
\end{align}
Since $\mathcal S_r$ is finite, pointwise almost-sure convergence implies the uniform convergence \eqref{eq:LM-uniform-proof}.
\end{proof}

%\label{eq:LM-uniform-proof}

\paragraph{Consistency of MLE.} With these ingredients in place, standard results of MLE consistency imply that any maximizer of the empirical log-likelihood over the finite set $\mathcal S_r$ converges almost surely to the unique maximizer of $L$, namely the true sign vector $\bm{s}_{u,*}^{\mathrm{R}}$.

\begin{lemma}[Consistency of the MLE under a non-Gaussian prior]\label{lem:MLE_sign_consistency}
Assume the setting of Proposition~\ref{prop:signal_plus_noise_sv} and Lemma~\ref{lem:uniform-emp-log}, and suppose
that the scalar signal $\mathsf U_*$ in \eqref{eq:limit-vector-U-true} is
not standard Gaussian with $\E[\mathsf U_*^2]=1$. Let $\mathcal I_M$ be the outlier
index set with $K=\card{\mathcal I_M}$, and let
\[
  \mathcal S_r
  \;\bydef\;
  \bigl\{\bm s\in\{\pm1\}^K : [\bm s]_r=+1\bigr\}
\]
and $\bm{s}_{u,*}^{\mathrm{R}}\in\mathcal S_r$ be the ground truth sign vector as in Section~\ref{sec:relativesign}.
Let $\hat{\bm s}_u^{\mathrm{MLE}}$ be an estimator defined in
\eqref{eq:MLE-est-global}. Then
\begin{align}\label{eq:MLE-sign-consistency}
  \hat{\bm s}_u^{\mathrm{MLE}}
  \ac
  \bm{s}_{u,*}^{\mathrm{R}} .
\end{align}
\end{lemma}

\begin{proof}
Throughout we reindex $\mathcal I_M=\{1,\dots,K\}$.
We work under the conditional law given the ground-truth sign vector
$\bm s_{u,*}^{\mathrm R}$. Let $L_M(\bm s;\bm U^\sharp)$ and
$L(\bm s;\bm{\mathsf U}^\sharp)$ be as in Lemma~\ref{lem:uniform-emp-log}.
To prove \eqref{eq:MLE-sign-consistency}, it suffices to show that
$\bm s_{u,*}^{\mathrm R}$ is the unique maximizer of $L(\cdot;\bm{\mathsf U}^\sharp)$ on $\mathcal S_r$ and to invoke the uniform convergence \eqref{eq:LM-uniform-proof}.

We first establish identifiability on $\mathcal S_r$, namely
\begin{equation}\label{eq:Ps-distinct-proof}
  \bm s\neq\bm t \quad\Longrightarrow\quad P_{\bm s}\neq P_{\bm t},
  \qquad \bm s,\bm t\in\mathcal S_r,
\end{equation}
where $P_{\bm s}$ denotes the law induced by \eqref{eq:Usharps}.
Let $\phi_{U_*}(\omega)\bydef\E[e^{i\omega\mathsf U_*}]$ and define
$\Psi(\omega)\bydef \phi_{U_*}(\omega)e^{\omega^2/2}$.
By \eqref{eq:Usharps}, for any $\bm w\in\R^K$, the characteristic function under $P_{\bm s}$ yields
\begin{equation}\label{eq:Phi-s-Psi-proof}
  \Phi_{\bm s}(\bm w)
  = \E\big[e^{i\langle \bm w,\bm{\mathsf U}^\sharp\rangle}\big]
  = e^{-\|\bm w\|^2/2}\,\Psi(\langle \bm w,\bm\gamma(\bm s)\rangle),
\end{equation}
with $\bm\gamma(\bm s)$ defined in \eqref{eq:gamma-def-proof}. To prove \eqref{eq:Ps-distinct-proof}, it suffices to show that
\begin{align}\label{eq:cf-separate-w0}
  \forall\,\bm s\neq\bm t\in\mathcal S_r,\ \exists\,\bm w_0\in\R^K
  \ \text{such that}\ 
  \Phi_{\bm s}(\bm w_0)\neq \Phi_{\bm t}(\bm w_0).
\end{align}
Since $\mathsf U_*$ is not standard Gaussian, $\Psi$ is non-constant (otherwise $\phi_{U_*}(\omega)=e^{-\omega^2/2}$).
Choose $\omega_1\neq\omega_2$ such that $\Psi(\omega_1)\neq\Psi(\omega_2)$.
Fix $\bm s\neq\bm t$ in $\mathcal S_r$. If $\bm\gamma(\bm t)=c\,\bm\gamma(\bm s)$ for some scalar $c$, then
\eqref{eq:gamma-def-proof} yields $c=[\bm t]_k/[\bm s]_k\in\{\pm1\}$ for every $k\in\mathcal I_M$.
Evaluating this identity at $k=r$ gives $c=1$ (since $[\bm s]_r=[\bm t]_r=+1$), whereas evaluating it at an index
$j$ with $[\bm s]_j\neq[\bm t]_j$ gives $c=-1$, a contradiction. Assume $K\ge2$ (the case $K=1$ is trivial since $\mathcal S_r$ is a singleton, no relative sign needed).
Hence the map $T(\bm w)\bydef(\langle\bm w,\bm\gamma(\bm s)\rangle,\langle\bm w,\bm\gamma(\bm t)\rangle)$ is surjective,
so there exists $\bm w_0$ such that $\langle\bm w_0,\bm\gamma(\bm s)\rangle=\omega_1$ and
$\langle\bm w_0,\bm\gamma(\bm t)\rangle=\omega_2$. Plugging into \eqref{eq:Phi-s-Psi-proof} yields \eqref{eq:cf-separate-w0},
and thus $P_{\bm s}\neq P_{\bm t}$, proving \eqref{eq:Ps-distinct-proof}.

Next, identifiability implies that $\bm s_{u,*}^{\mathrm R}$ uniquely maximizes the population log-likelihood. Since
$L(\bm s;\bm{\mathsf U}^\sharp)$ is computed under the true law $P_{\bm s_{u,*}^{\mathrm R}}$, the definition of KL divergence gives
\[
  L(\bm s_{u,*}^{\mathrm R};\bm{\mathsf U}^\sharp)-L(\bm s;\bm{\mathsf U}^\sharp)
  = \mathrm{KL}\!\left(P_{\bm s_{u,*}^{\mathrm R}}\;\Vert\;P_{\bm s}\right)\ge 0,
\]
with equality iff $\bm s=\bm s_{u,*}^{\mathrm R}$ by \eqref{eq:Ps-distinct-proof}. Since $\mathcal S_r$ is finite, the following gap is strictly positive:
\[
  \Delta
  \bydef
  \min_{\bm s\in\mathcal S_r,\ \bm s\neq\bm s_{u,*}^{\mathrm R}}
  \Big(L(\bm s_{u,*}^{\mathrm R};\bm{\mathsf U}^\sharp)-L(\bm s;\bm{\mathsf U}^\sharp)\Big)>0.
\]

Finally, Lemma~\ref{lem:uniform-emp-log} implies that, for any $\varepsilon>0$, almost surely for all sufficiently large $M$,
\begin{equation}\label{eq:uniform-eps}
  \sup_{\bm s\in\mathcal S_r}\big|L_M(\bm s;\bm U^\sharp)-L(\bm s;\bm{\mathsf U}^\sharp)\big|\le \varepsilon.
\end{equation}
Set $\varepsilon=\Delta/3$. Then for such $M$, the triangle inequality yields
\[
  L_M(\bm s_{u,*}^{\mathrm R};\bm U^\sharp)
  \ge L(\bm s_{u,*}^{\mathrm R};\bm{\mathsf U}^\sharp)-\varepsilon,
\]
and for any $\bm s\neq\bm s_{u,*}^{\mathrm R}$,
\[
  L_M(\bm s;\bm U^\sharp)
  \le L(\bm s;\bm{\mathsf U}^\sharp)+\varepsilon
  \le L(\bm s_{u,*}^{\mathrm R};\bm{\mathsf U}^\sharp)-\Delta+\varepsilon
  = L(\bm s_{u,*}^{\mathrm R};\bm{\mathsf U}^\sharp)-2\varepsilon.
\]
Hence $L_M(\bm s_{u,*}^{\mathrm R};\bm U^\sharp)>L_M(\bm s;\bm U^\sharp)$ for all $\bm s\neq\bm s_{u,*}^{\mathrm R}$, and therefore any maximizer
$\hat{\bm s}_u^{\mathrm{MLE}}\in\argmax_{\bm s\in\mathcal S_r}L_M(\bm s;\bm U^\sharp)$ satisfies
$\hat{\bm s}_u^{\mathrm{MLE}}=\bm s_{u,*}^{\mathrm R}$ for all sufficiently large $M$ almost surely, i.e.,
$\hat{\bm s}_u^{\mathrm{MLE}}\xrightarrow{\mathrm{a.s.}}\bm s_{u,*}^{\mathrm R}$.
\end{proof}

\subsubsection{Inter-channel Sign Coupling}
Finally, we couple the $\bm u$-- and $\bm v$--channel signs through the cross spectral measure
$\nu_3$ in \defref{def:spectral_measures}. The key point is that the outlier point mass
$\nu_3(\{\sigma_k\})$ is nonzero, which yields an asymptotically deterministic inter-channel sign relation.

\begin{lemma}[Inter-channel sign coupling]\label{lem:inter_channel_sign_coupling}
Under the assumptions of Lemma~\ref{lem:spectral_measures_properties} and
Proposition~\ref{prop:signal_plus_noise_sv}, for any outlier index $k\in\mathcal I_M$
with limiting eigenvalue $\lambda_k\in\mathcal K^*$ and singular value $\sigma_k=\sqrt{\lambda_k}$,
\begin{equation}\label{eq: sign coupling}
\sign\!\Big(
\langle \bm u_k(\bm Y),\bm u_*\rangle\,
\langle \bm v_k(\bm Y),\bm v_*\rangle
\Big)
\;\ac\;
\sign\big(\nu_3(\{\sigma_k\})\big),
\end{equation}
and $\nu_3(\{\sigma_k\})\neq 0$.
\end{lemma}

\begin{proof}
By the definition of $\nu_{L,3}$ in \defref{def:spectral_measures} and the overlap convergence in
\propref{prop:outlier_characterization},
\begin{equation}\label{eq:limit_to_mass-proof}
  \frac{1}{2(M+N)}\,
  \langle \bm u_k(\bm Y),\bm u_*\rangle\,
  \langle \bm v_k(\bm Y),\bm v_*\rangle
  \;\ac\;
  \nu_3(\{\sigma_k\}).
\end{equation}
Writing $\lambda_k=\sigma_k^2$, Lemma~\ref{lem:spectral_measures_properties} gives
\begin{equation}\label{eq:nu3-pointmass-proof}
  \nu_3(\{\sigma_k\})
  =
  \frac{\sqrt{\delta}}{1+\delta}\cdot
  \frac{\theta\,\CT(\lambda_k)}{2\sigma_k\,\Gamma'(\lambda_k)}
  =
  \frac{\sqrt{\delta}}{1+\delta}\cdot
  \frac{1}{2\theta\,\sigma_k\,\Gamma'(\lambda_k)},
\end{equation}
where the second equality uses the master equation $1-\theta^2\CT(\lambda_k)=0$.
Lemma~\ref{lem:Gamma-analytic} ensures $\Gamma'(\lambda_k)\neq 0$, hence $\nu_3(\{\sigma_k\})\neq 0$.
Since the limit in \eqref{eq:limit_to_mass-proof} is nonzero, taking signs in \eqref{eq:limit_to_mass-proof} yields
\eqref{eq: sign coupling}.
\end{proof}

\begin{proof}[Proof of Proposition~\ref{prop:MLE_relative_sign}]
For Claim~(1), assume without loss of generality that $\mathsf U_*$ is non-Gaussian. Then
Lemma~\ref{lem:MLE_sign_consistency} gives $\hat{\bm s}_u^{\mathrm{MLE}}\ac \bm s_{u,*}^{\mathrm R}$.
Lemma~\ref{lem:inter_channel_sign_coupling} yields an asymptotically deterministic relation between the
channel-wise relative sign vectors; in particular, with the reference index $r$ defining $\mathcal S_r$, set
\[
  \bigl[\hat{\bm s}_v^{\mathrm{MLE}}\bigr]_j
  \;\bydef\;
  \bigl[\hat{\bm s}_u^{\mathrm{MLE}}\bigr]_j\,
  \sign\!\big(\nu_3(\{\sigma_r\})\,\nu_3(\{\sigma_j\})\big),
  \qquad j\in\mathcal I_M.
\]
Then $\hat{\bm s}_v^{\mathrm{MLE}}\ac \bm s_{v,*}^{\mathrm R}$. By symmetry, the same conclusion holds when only
$\mathsf V_*$ is non-Gaussian.

For Claim~(2), when both $\mathsf U_*$ and $\mathsf V_*$ are standard Gaussian,
Lemma~\ref{lem:gaussian_row_law} (applied separately in each channel) shows that
the row law of the scaled outlier singular vectors is $\mathcal N(\bm 0,\bm I_K)$ and
does not depend on the sign vector.
Consequently, the likelihood is invariant over $\mathcal S_r$.
\end{proof}

\subsection{Proof of \propref{prop:NGMC-est}}\label{app:NGMC-est}

Before proving Proposition~\ref{prop:NGMC-est}, we explain heuristically why NGMC identifies the
relative sign $[\bm s_{u,*}^{\mathrm R}]_r[\bm s_{u,*}^{\mathrm R}]_j$. Fix distinct outlier indices $r,j\in\mathcal I_M$. Then the corresponding empirical singular-vector coordinates admit the limit representation by \eqref{eq:limit-vector-U-true}
\begin{align}\label{eq:NGMC-heuristic}
     \mathsf U_r^\sharp
  = [\bm s_{u,*}^{\mathrm R}]_r\sqrt{\nu_1(\{\lambda_r\})}\,\mathsf U_*
    + \sqrt{1-\nu_1(\{\lambda_r\})}\,\mathsf Z_r,
  \qquad
  \mathsf U_j^\sharp
  = [\bm s_{u,*}^{\mathrm R}]_j\sqrt{\nu_1(\{\lambda_j\})}\,\mathsf U_*
    + \sqrt{1-\nu_1(\{\lambda_j\})}\,\mathsf Z_j, 
\end{align}
where $\mathsf U_*$ is common and $(\mathsf Z_r,\mathsf Z_j)$ is a standard Gaussian pair independent of $\mathsf U_*$.
For a contrast $f:\R\to\R$, set
\[
  T_f(a,b)\bydef \E\!\left[f(\mathsf U_r^\sharp)\,\mathsf U_j^\sharp\right],
  \qquad a,b\in\{\pm1\},
\]
where $\mathsf U_r^\sharp,\mathsf U_j^\sharp$ are formed with $a=[\bm s_{u,*}^{\mathrm R}]_r$ and $b=[\bm s_{u,*}^{\mathrm R}]_j$.
If $f$ is odd, then by \eqref{eq:NGMC-heuristic} we have $T_f(-a,b)=-T_f(a,b)$ and $T_f(a,-b)=-T_f(a,b)$, hence $T_f$ depends on $(a,b)$ only through the product $ab$:
\[
  T_f(a,b)=ab\,C_f,
  \qquad C_f\bydef T_f(1,1).
\]
Therefore, whenever $C_f\neq 0$,
\[
  \sign T_f\big([\bm s_{u,*}^{\mathrm R}]_r,[\bm s_{u,*}^{\mathrm R}]_j\big)
  = \sign\!\big([\bm s_{u,*}^{\mathrm R}]_r[\bm s_{u,*}^{\mathrm R}]_j\big)\,\sign(C_f),
\]
so the sign of $T_f$ recovers the relative sign up to a fixed global orientation.
In NGMC we take $f(x)=x^{k+1}$ with even $k$ from Assumption~\ref{assump: Non-Gaussian Prior}; the ensuing calculation shows
$C_f\propto \E[\mathsf U_*^{k+2}]-(k+1)!!\neq 0$, ensuring non-degeneracy.

\begin{proof}[Proof of \propref{prop:NGMC-est}]
We treat the $\bm u$--channel; the $\bm v$--channel follows by the same argument together with the inter-channel sign
relation in \propref{prop:signal_plus_noise_sv}. Fix distinct outlier indices $r,j\in\mathcal I_M$. Define the deterministic alignments
\[
  \gamma_k \bydef \sqrt{\nu_1(\{\lambda_k\})}\in(0,1),\qquad
  \tilde\gamma_k \bydef \sqrt{1-\gamma_k^2},\qquad k\in\mathcal I_M,
\]
and the limiting pair
\begin{equation}\label{eq:Usharp-rj-NGMC}
  \mathsf U_r^\sharp
  = [\bm s_{u,*}^{\mathrm R}]_r\,\gamma_r\,\mathsf U_* + \tilde\gamma_r\,\mathsf Z_r,
  \qquad
  \mathsf U_j^\sharp
  = [\bm s_{u,*}^{\mathrm R}]_j\,\gamma_j\,\mathsf U_* + \tilde\gamma_j\,\mathsf Z_j,
\end{equation}
where $\E[\mathsf U_*^2]=1$ and $(\mathsf Z_r,\mathsf Z_j)$ is a standard Gaussian pair independent of $\mathsf U_*$.
By Proposition~\ref{prop:signal_plus_noise_sv} and \eqref{eq:limit-vector-U-true}, the joint law of empirical rows $(\bm u_r^\sharp,\bm u_j^\sharp)\wc(\mathsf U_r^\sharp,\mathsf U_j^\sharp)$ by \eqref{eq:Usharps}. Let $f(x)\bydef x^{k+1}$ and $g(x_1,x_2)\bydef x_1^{k+1}x_2$. Using the Wasserstein convergence above, the polynomial
growth of $g$, we have
\begin{equation}\label{eq:ngmc-stat-moment}
  \frac{1}{M}\,f(\bm u_r^\sharp)^\UT \bm u_j^\sharp
  = \frac{1}{M}\sum_{m=1}^M g([\bm u_r^\sharp]_m,[\bm u_j^\sharp]_m)
  \;\ac\;
  \E\big[(\mathsf U_r^\sharp)^{k+1}\mathsf U_j^\sharp\big].
\end{equation}
We will show that
\begin{equation}\label{eq:ngmc-moment-target}
  \E\big[(\mathsf U_r^\sharp)^{k+1}\mathsf U_j^\sharp\big]
  =
  [\bm s_{u,*}^{\mathrm R}]_r[\bm s_{u,*}^{\mathrm R}]_j\,
  \nu_1(\{\lambda_r\})^{(k+1)/2}\,
  \sqrt{\nu_1(\{\lambda_j\})}\,
  \big(\E[\mathsf U_*^{k+2}]-(k+1)!!\big),
\end{equation}
and therefore \eqref{eq:ngmc-stat-moment} implies
\begin{equation}\label{eq:ngmc-limit-claimed}
  \frac{1}{M}\,f(\bm u_r^\sharp)^\UT \bm u_j^\sharp
  \;\ac\;
  [\bm s_{u,*}^{\mathrm R}]_r[\bm s_{u,*}^{\mathrm R}]_j\,
  \nu_1(\{\lambda_r\})^{(k+1)/2}\,
  \sqrt{\nu_1(\{\lambda_j\})}\,
  \big(\E[\mathsf U_*^{k+2}]-(k+1)!!\big).
\end{equation}
Acknowledging \eqref{eq:ngmc-limit-claimed} and taking sign of both sides, the NGMC sign recovers the relative sign
$[\bm s_{u,*}^{\mathrm R}]_r[\bm s_{u,*}^{\mathrm R}]_j$ up to the deterministic orientation
$\sign(\E[\mathsf U_*^{k+2}]-(k+1)!!)\neq 0$ (Assumption~\ref{assump: Non-Gaussian Prior}); with the convention
$[\bm s_{u,*}^{\mathrm R}]_r=+1$, this yields $\hat s_{u,j}^{\mathrm{NGMC}}\ac[\bm s_{u,*}^{\mathrm R}]_j$.

Hence it remains to compute $\E[(\mathsf U_r^\sharp)^{k+1}\mathsf U_j^\sharp]$ and verify \eqref{eq:ngmc-moment-target}. We first determine $\Cov(\mathsf Z_r,\mathsf Z_j)$ from the limiting orthogonality
$\E[\mathsf U_r^\sharp\mathsf U_j^\sharp]=0$:
\begin{align}
  0=\E[\mathsf U_r^\sharp\mathsf U_j^\sharp]
  &\stackrel{(a)}{=}
    [\bm s_{u,*}^{\mathrm R}]_r[\bm s_{u,*}^{\mathrm R}]_j\,\gamma_r\gamma_j\,\E[\mathsf U_*^2]
    + \tilde\gamma_r\tilde\gamma_j\,\E[\mathsf Z_r\mathsf Z_j]\notag\\
  &\stackrel{(b)}{=}
    [\bm s_{u,*}^{\mathrm R}]_r[\bm s_{u,*}^{\mathrm R}]_j\,\gamma_r\gamma_j
    + \tilde\gamma_r\tilde\gamma_j\,\E[\mathsf Z_r\mathsf Z_j],
\end{align}
hence
\begin{equation}\label{eq:covZ-NGMC}
  \E[\mathsf Z_r\mathsf Z_j]
  = -\, [\bm s_{u,*}^{\mathrm R}]_r[\bm s_{u,*}^{\mathrm R}]_j\,
      \frac{\gamma_r\gamma_j}{\tilde\gamma_r\tilde\gamma_j}.
\end{equation}
Here (a) substitutes \eqref{eq:Usharp-rj-NGMC} and uses $(\mathsf Z_r,\mathsf Z_j)\indep \mathsf U_*$; (b) uses $\E[\mathsf U_*^2]=1$.

\smallskip
Next, expand
\begin{equation}\label{eq:split-NGMC}
  \E\big[(\mathsf U_r^\sharp)^{k+1}\mathsf U_j^\sharp\big]
  =
  [\bm s_{u,*}^{\mathrm R}]_j\gamma_j\,\E\big[(\mathsf U_r^\sharp)^{k+1}\mathsf U_*\big]
  + \tilde\gamma_j\,\E\big[(\mathsf U_r^\sharp)^{k+1}\mathsf Z_j\big]
 \bydef\text{\rm(I)}+\text{\rm(II)}.
\end{equation}

\smallskip
\noindent\emph{Term (I).} Using the binomial expansion and $\mathsf Z_r\indep \mathsf U_*$,
\begin{align}
\text{\rm(I)}
&= [\bm s_{u,*}^{\mathrm R}]_j\gamma_j\,
   \E\!\left[\big([\bm s_{u,*}^{\mathrm R}]_r\gamma_r\mathsf U_*+\tilde\gamma_r\mathsf Z_r\big)^{k+1}\mathsf U_*\right]\notag\\
&\stackrel{(a)}{=}
[\bm s_{u,*}^{\mathrm R}]_j\gamma_j
\sum_{m=0}^{k+1}\binom{k+1}{m}
\big([\bm s_{u,*}^{\mathrm R}]_r\gamma_r\big)^m\tilde\gamma_r^{\,k+1-m}\,
\E[\mathsf U_*^{m+1}]\,\E[\mathsf Z_r^{k+1-m}]\notag\\
&\stackrel{(b)}{=}
[\bm s_{u,*}^{\mathrm R}]_r[\bm s_{u,*}^{\mathrm R}]_j\,\gamma_j
\sum_{p=0}^{k/2}\binom{k+1}{2p}
\gamma_r^{\,k+1-2p}\tilde\gamma_r^{\,2p}\,(2p-1)!!\,\E[\mathsf U_*^{k+2-2p}],
\label{eq:I-NGMC}
\end{align}
where (b) retains only even Gaussian moments $\E[\mathsf Z_r^{2p}]=(2p-1)!!$ and uses
$[\bm s_{u,*}^{\mathrm R}]_r^{\,k+1-2p}=[\bm s_{u,*}^{\mathrm R}]_r$ for even $k$.

\smallskip
\noindent\emph{Term (II).} Let $h_u(z)\bydef\big([\bm s_{u,*}^{\mathrm R}]_r\gamma_r u+\tilde\gamma_r z\big)^{k+1}$ so that
$(\mathsf U_r^\sharp)^{k+1}=h_{\mathsf U_*}(\mathsf Z_r)$. Conditioning on $\mathsf U_*$ and applying Stein's identity for the
Gaussian pair $(\mathsf Z_r,\mathsf Z_j)$ yields
\begin{align}
\text{\rm(II)}
&= \tilde\gamma_j\,\E\big[h_{\mathsf U_*}(\mathsf Z_r)\,\mathsf Z_j\big]\notag\\
&\stackrel{(a)}{=}
\tilde\gamma_j\,\Cov(\mathsf Z_r,\mathsf Z_j)\,
\E\!\left[\frac{\partial}{\partial \mathsf Z_r}(\mathsf U_r^\sharp)^{k+1}\right]\notag\\
&\stackrel{(b)}{=}
\tilde\gamma_j\,\Cov(\mathsf Z_r,\mathsf Z_j)\,(k+1)\tilde\gamma_r\,
\E\big[(\mathsf U_r^\sharp)^k\big]\notag\\
&\stackrel{(c)}{=}
-\,[\bm s_{u,*}^{\mathrm R}]_r[\bm s_{u,*}^{\mathrm R}]_j\,\gamma_r\gamma_j\,(k+1)\,
\E\big[(\mathsf U_r^\sharp)^k\big]\notag\\
&\stackrel{(d)}{=}
-\,[\bm s_{u,*}^{\mathrm R}]_r[\bm s_{u,*}^{\mathrm R}]_j\,\gamma_r\gamma_j\,(k+1)
\sum_{p=0}^{k/2}\binom{k}{2p}\gamma_r^{\,k-2p}\tilde\gamma_r^{\,2p}\,(2p-1)!!\,\E[\mathsf U_*^{k-2p}],
\label{eq:II-NGMC}
\end{align}
where (a) is Stein, (b) differentiates with respect to $\mathsf Z_r$, (c) substitutes \eqref{eq:covZ-NGMC}, and (d) expands
$(\mathsf U_r^\sharp)^k$ and retains even Gaussian moments.

\smallskip
Finally, substituting \eqref{eq:I-NGMC}--\eqref{eq:II-NGMC} into \eqref{eq:split-NGMC} and using
$(k+1)\binom{k}{2p}=(k+1-2p)\binom{k+1}{2p}$ gives
\begin{align}\label{eq:moment-sum-NGMC}
  \E\big[(\mathsf U_r^\sharp)^{k+1}\mathsf U_j^\sharp\big]
  &=
  [\bm s_{u,*}^{\mathrm R}]_r[\bm s_{u,*}^{\mathrm R}]_j\,\gamma_j\gamma_r^{\,k+1}
  \sum_{p=0}^{k/2}\binom{k+1}{2p}(2p-1)!!\,\gamma_r^{-2p}\tilde\gamma_r^{\,2p}\,
  \Big(\E[\mathsf U_*^{k+2-2p}]-(k+1-2p)\E[\mathsf U_*^{k-2p}]\Big).
\end{align}
By the minimality of $k$ in \propref{prop:NGMC-est}, $\E[\mathsf U_*^{2\ell}]=(2\ell-1)!!$ for $2\ell\le k$ while
$\E[\mathsf U_*^{k+2}]\neq (k+1)!!$, which forces every summand in \eqref{eq:moment-sum-NGMC} with $p\ge 1$ to vanish; hence
\[
  \E\big[(\mathsf U_r^\sharp)^{k+1}\mathsf U_j^\sharp\big]
  =
  [\bm s_{u,*}^{\mathrm R}]_r[\bm s_{u,*}^{\mathrm R}]_j\,\gamma_r^{k+1}\gamma_j\,
  \big(\E[\mathsf U_*^{k+2}]-(k+1)!!\big).
\]
Using $\gamma_r^{k+1}\gamma_j=\nu_1(\{\lambda_r\})^{(k+1)/2}\sqrt{\nu_1(\{\lambda_j\})}$ yields \eqref{eq:ngmc-moment-target},
and hence \eqref{eq:ngmc-limit-claimed} by \eqref{eq:ngmc-stat-moment}. This completes the $\bm u$--channel.

For the $\bm v$--channel, define $\hat s_{v,j}^{\mathrm{NGMC}}\bydef \hat s_{u,j}^{\mathrm{NGMC}}\sign\!\big(\nu_3(\{\sigma_r\})\,\nu_3(\{\sigma_j\})\big)$;
the inter-channel sign coupling then yields $\hat s_{v,j}^{\mathrm{NGMC}}\ac[\bm s_{v,*}^{\mathrm R}]_j$.
\end{proof}

%% file: appendix/ap-Global_Sign2.tex
\section{Global Sign Detection}\label{app:pf-global_sign}

The spectral initializer \eqref{eq: spectral init} is correlated with the ground truth only up to an \emph{unknown global sign} (cf.\cite{montanari2021estimation,mondelli2021pca}). Under the auxiliary randomization in \eqref{eq:scaled-outlier-matrix-rand}, this sign is Rademacher distributed (cf.\cite[Remark~3.6]{feng2022unifying}). The purpose of this appendix is to estimate this global sign from the observed initializer (when it is identifiable) and to use the resulting estimate to select the signs $(s_1,s_2)$ in the signed DMMSE family \eqref{eq:signed-family} used by \eqref{eq: spec-OAMP}.

Concretely, the normalized spectral initializers in \eqref{eq: spectral init} admit the scalar Gaussian-channel limits
\begin{align}
\tilde{\bm u}_1 &\xrightarrow[]{W}\tilde{\mathsf U}_1(s_{u,*}^{\mathrm G})
\stackrel{d}{=} s_{u,*}^{\mathrm G}\sqrt{w_{1,1}}\,\mathsf U_* + \sqrt{1-w_{1,1}}\,\mathsf Z_u, \quad w_{1,1} = \sum_{\lambda_i \in \mathcal{K}^*} \nu_1(\{\lambda_i\}),
\label{eq:scalar-channel-u-init-app}\\
\tilde{\bm v}_1 &\xrightarrow[]{W}\tilde{\mathsf V}_1(s_{v,*}^{\mathrm G})
\stackrel{d}{=} s_{v,*}^{\mathrm G}\sqrt{w_{2,1}}\,\mathsf V_* + \sqrt{1-w_{2,1}}\,\mathsf Z_v,\quad w_{2,1} = \sum_{\lambda_i \in \mathcal{K}^*} \nu_2(\{\lambda_i\}).
\label{eq:scalar-channel-v-init-app}
\end{align}
where $s_{u,*}^{\mathrm G},s_{v,*}^{\mathrm G}\in\{\pm1\}$ are the realized global signs and all other variables are defined as in \eqref{sec: spec OAMP}. 

This representation determines when ${s}_{u,*}^{\mathrm{G}}$ is statistically
identifiable from $\tilde{\bm u}_1$. If $\mathsf U_*\stackrel{d}{=}-\mathsf U_*$, then the marginal law of $\tilde{\mathsf U}_1$ is invariant under
${s}_{u,*}^{\mathrm{G}}\mapsto -{s}_{u,*}^{\mathrm{G}}$, and the global sign
cannot be recovered from $\tilde{\bm u}_1$ alone. Conversely, under asymmetric
priors the two induced laws are distinct, and ${s}_{u,*}^{\mathrm{G}}$ can be
estimated consistently from the empirical distribution of $\tilde{\bm u}_1$.

Under the asymmetric regime, two constructions of consistent estimators of the true global signs are recorded below (GSMLE and odd-moment contrast), which manifest
global-sign counterparts of the relative-sign procedures in \propref{prop:NGMC-est}; their consistency
proofs repeat the similar arguments and are omitted.

\subsection{Global-Sign Maximum Likelihood Estimator (GSMLE) Scheme}
\label{app:pf-global_sign_gsmle}
We describe a two-hypothesis likelihood test induced by the scalar-channel limits
\eqref{eq:scalar-channel-u-init-app}--\eqref{eq:scalar-channel-v-init-app}.
For $s\in\{\pm1\}$, let $p_s^u$ denote the density of
\[
s\sqrt{w_{1,1}}\,\mathsf U_*+\sqrt{1-w_{1,1}}\,\mathsf Z_u,
\qquad \mathsf Z_u\sim\mathcal N(0,1),\ \mathsf Z_u\indep \mathsf U_*,
\]
and define $p_s^v$ analogously from \eqref{eq:scalar-channel-v-init-app}.
Given $\tilde{\bm u}_1\in\R^M$ and $\tilde{\bm v}_1\in\R^N$, set
\begin{equation}\label{eq:gsmle-LL-u}
\mathcal L_{u,M}(s)\;\bydef\;\frac{1}{M}\sum_{i=1}^M \log p_s^u\big([\tilde{\bm u}_1]_i\big),
\qquad
\mathcal L_{v,N}(s)\;\bydef\;\frac{1}{N}\sum_{i=1}^N \log p_s^v\big([\tilde{\bm v}_1]_i\big),
\end{equation}
and define the GSMLEs
\begin{equation}\label{eq:gsmle-def}
\hat s_u^{\mathrm{GSMLE}}\;\in\;\argmax_{s\in\{\pm1\}} \mathcal L_{u,M}(s),
\qquad
\hat s_v^{\mathrm{GSMLE}}\;\in\;\argmax_{s\in\{\pm1\}} \mathcal L_{v,N}(s).
\end{equation}
As $M,N\to\infty$, the following regimes hold:
\begin{enumerate}[label=\textup{(\roman*)},leftmargin=*]
\item \emph{Both priors asymmetric.}
\[
\hat s_u^{\mathrm{GSMLE}}\ac s_{u,*}^{\mathrm G},
\qquad
\hat s_v^{\mathrm{GSMLE}}\ac s_{v,*}^{\mathrm G}.
\]
\item \emph{Both priors symmetric.}
Then $p_{+1}^u=p_{-1}^u$ and $p_{+1}^v=p_{-1}^v$. Hence neither channel-wise global sign is
identifiable from its initializer. However, the relative global sign is identifiable: for any
baseline outlier $r\in\mathcal I_M$ with limit $\lambda_r=\sigma_r^2$,
\lemref{lem:inter_channel_sign_coupling} implies
\begin{equation}\label{eq:symmetric-prior-relative-global-sign}
s_{u,*}^{\mathrm G}\,s_{v,*}^{\mathrm G}\;\ac\;\sign\!\big(\nu_3(\{\sigma_r\})\big),
\end{equation}
where $\nu_3$ is the cross spectral measure in \lemref{lem:spectral_measures_properties}.
\item \emph{Exactly one prior symmetric.} Suppose $\mathsf U_*$ is asymmetric, then
\[
\hat s_u^{\mathrm{GSMLE}}\ac s_{u,*}^{\mathrm G},
\qquad
\hat s_v^{\mathrm{GSMLE}}\;\bydef\;\hat s_u^{\mathrm{GSMLE}}\cdot \sign\!\big(\nu_3(\{\sigma_r\})\big)
\;\ac\; s_{v,*}^{\mathrm G}.
\]
The converse case is analogous.
\end{enumerate}

\subsection{Global Sign Odd-moment Contrast Scheme (GOMC) Scheme} \label{app:pf-global_sign_GSOC}

In analogy with the NGMC rule for relative-sign alignment in \propref{prop:NGMC-est}, we record a simpler
method-of-moments estimator for the global sign in the $\bm u$--channel.
\begin{assumption}[Odd-moment asymmetry]\label{assump:global-odd-moment}
There exists an odd integer $j\ge1$ such that $\E[\mathsf U_*^{\,j}]\neq 0$.
\end{assumption}
Assume the setting of \propref{prop:optimal_data_driven_estimators} and
\assumpref{assump:global-odd-moment}, and let $j$ be the smallest odd integer such that
$\E[\mathsf U_*^{\,j}]\neq 0$. Define
\begin{equation}\label{eq:odd-moment-estimator}
\hat s_u^{\mathrm{GOMC}}
\;\bydef\;
\sign\!\Big(\frac{1}{M}\sum_{i=1}^M [\tilde{\bm u}_1]_i^{\,j}\Big)\,
\sign\!\big(\E[\mathsf U_*^{\,j}]\big).
\end{equation}
Under the scalar-channel convergence \eqref{eq:scalar-channel-u-init-app}, following similar steps in \eqref{eq:ngmc-stat-moment}-\eqref{eq:moment-sum-NGMC}, it can be shown that
\[
\frac{1}{M}\sum_{i=1}^M [\tilde{\bm u}_1]_i^{\,j}
\;\ac\;
s_{u,*}^{\mathrm G}\,w_{1,1}^{j/2}\,\E[\mathsf U_*^{\,j}],
\]
and hence $\hat s_u^{\mathrm{GOMC}}\ac s_{u,*}^{\mathrm G}$ as $M\to\infty$.

For the $\bm v$--channel, fix the same baseline outlier index $r\in\mathcal I_M$ as in
\eqref{eq:symmetric-prior-relative-global-sign}, and set
\begin{equation}\label{eq:pf-GOMC-v-def}
\hat s_v^{\mathrm{GOMC}}
\;\bydef\;
\hat s_u^{\mathrm{GOMC}}\cdot \sign\!\big(\nu_3(\{\sigma_r\})\big).
\end{equation}
By the inter-channel coupling in \eqref{eq:symmetric-prior-relative-global-sign},
$\hat s_v^{\mathrm{GOMC}}\ac s_{v,*}^{\mathrm G}$.

\subsection{Selection of the Signed DMMSE Estimators}
\label{app:pf-global_sign_dmmse}

We select the denoiser signs $(s_1,s_2)$ in \eqref{eq: spec-OAMP} by combining a global-sign
estimate from \eqref{eq:scalar-channel-u-init-app}--\eqref{eq:scalar-channel-v-init-app} with the
sign behavior of the signed DMMSE family \eqref{eq:signed-family}; see
Fact~\ref{fact:signed-dmmse-sign}.

\begin{itemize}[leftmargin=*]
\item \emph{At least one asymmetric prior.}
Without loss of generality, assume $\mathsf U_*$ is asymmetric.
Let $\hat s_u(\tilde{\bm u}_1)$ be any consistent estimator of $s_{u,*}^{\mathrm G}$, e.g.,
$\hat s_u^{\mathrm{GSMLE}}$ in \eqref{eq:gsmle-def} or $\hat s_u^{\mathrm{GOMC}}$ in
\eqref{eq:odd-moment-estimator}.
If $\mathsf V_*$ is asymmetric, define $\hat s_v(\tilde{\bm v}_1)$ analogously; otherwise, set
\begin{equation}\label{eq:global-coupling-choice}
\hat s_v \;\bydef\; \hat s_u\cdot \sign\!\big(\nu_3(\{\sigma_r\})\big),
\end{equation}
with any fixed baseline outlier $r\in\mathcal I_M$ (cf.\ \eqref{eq:symmetric-prior-relative-global-sign}).
We then choose
\begin{equation}\label{eq:sign-choice-asym}
s_1\;\bydef\;\hat s_u(\tilde{\bm u}_1),\qquad
s_2\;\bydef\;\hat s_v(\tilde{\bm v}_1).
\end{equation}
With this choice, the signed DMMSE update is aligned with the realized global signs, and hence
preserves a positive correlation with each signal; see Fact~\ref{fact:signed-dmmse-sign}. 

\item \emph{Both priors symmetric.}
Neither channel-wise global sign is identifiable from its initializer, but the relative global sign is
fixed by \eqref{eq:symmetric-prior-relative-global-sign}. We adopt the convention
\begin{equation}\label{eq:sign-choice-sym}
s_1 \bydef +1,
\qquad
s_2 \bydef \sign\!\big(\nu_3(\{\sigma_r\})\big),
\end{equation} (cf.\ \eqref{eq:symmetric-prior-relative-global-sign})
where $r\in\mathcal I_M$ is the reference outlier index (cf.\ \eqref{eq:symmetric-prior-relative-global-sign}). This convention is consistent with the asymptotic inter-channel sign coupling:
\begin{align}
\lim_{M\to \infty}\sign( \langle \tilde{\bm u}_1, \bm{u}_*\rangle \langle \tilde{\bm v}_1, \bm{v}_*\rangle) \ase \lim_{M\to \infty}\sign( \langle {\bm u}_r^\sharp, \bm{u}_*\rangle \langle {\bm v}_r^\sharp, \bm{v}_*\rangle) \stackrel{\text{Lemma}\ref{lem:inter_channel_sign_coupling}}{=} \sign\!\big(\nu_3(\{\sigma_r\})\big)
\end{align}
By Fact~\ref{fact:signed-dmmse-sign}, the signed DMMSE
is odd under symmetric priors, so any mismatch with the realized global signs induces only a
\emph{coherent} global flip of the iterates across both channels. Consequently, the SE recursion for
the squared overlaps (and hence the cosine similarities) is invariant under this flip.
\end{itemize}

%% file: appendix/ap-spec.init2.tex
\section{Proof of Theorem~\ref{thm: spec-SE}}
\label{app:pf-thm-spec-se}

\paragraph{Convention and asymptotic equivalence.}
Throughout this proof we work under Assumption~\ref{assump:main}. We use the asymptotic
vector equivalence notation introduced in Definition~\ref{def:eq}: for sequences
$\bm x\in\R^d$ and $\bm y\in\R^d$,
\[
\bm x\,\explain{$d\to\infty$}{\simeq}\,\bm y
\qquad\Longleftrightarrow\qquad
\frac{\|\bm x-\bm y\|^2}{d}\xrightarrow{\mathrm{a.s.}}0\quad\text{as }d\to\infty.
\]
When applying this notation below, $d$ is the ambient dimension of the corresponding vectors
(e.g., $d=M$ for $\R^M$-valued iterates and $d=N$ for $\R^N$-valued iterates).

\paragraph{Proof strategy.}
We recall the spiked model
\begin{equation}
\label{eq:specSE:model}
\bm Y=\frac{\theta}{\sqrt{MN}}\bm u_\ast\bm v_\ast^\UT+\bm W\in\R^{M\times N},
\end{equation}
with $\|\bm u_\ast\|^2/M\to 1$ and $\|\bm v_\ast\|^2/N\to 1$ almost surely.
The proof has two components.
First, we show that the (optimally) combined spectral initializer is asymptotically equivalent to a
single \emph{$\bm W$-driven} OAMP step in the sense of Definition~\ref{def: Noise OAMP}.
Second, we invoke the spiked-to-noise reduction developed in
Appendix~\ref{Sec:OAMP_SE_proof}---in particular
the construction of the $\bm W$-driven auxiliary recursion around \eqref{eq:aux-OAMP} and the induction establishing \eqref{eq:OAMP-to-auxOAMP-induction}---to transfer the state evolution
from the auxiliary recursion to the original $\bm Y$-driven recursion.

%============================================================
\paragraph{Spectral initializer as a $\bm W$-driven OAMP step.}

Let $(\sigma_{k,M},\bm u_k,\bm v_k)_{k\in\mathcal I_M}$ denote the outlier singular
triplets of $\bm Y$ (with $\|\bm u_k\|=\|\bm v_k\|=1$), where $\mathcal I_M$ is the (finite)
outlier index set and $\sigma_{k,M}^2\to\lambda_k\in\mathcal K_\ast$.
We adopt the randomized orientation convention
\begin{equation}
\label{eq:spec:init-sharp}
\bm u_k^\sharp \bydef \sqrt{M}\,\xi_k\,\bm u_k,
\qquad
\bm v_k^\sharp \bydef \sqrt{N}\,\xi_k\,\bm v_k,
\qquad k\in\mathcal I_M,
\end{equation}
where $(\xi_k)_{k\in\mathcal I_M}$ are i.i.d.\ Rademacher, independent of all other randomness.
Let $(s_k^u)_{k\in\mathcal I_M}$ and $(s_k^v)_{k\in\mathcal I_M}$ be consistent relative-sign estimators.
Define the \emph{normalized} combined spectral initializer by
\begin{align}
\label{eq:spec:init-u}
\tilde{\bm u}_1
&\bydef
\Big(\sum_{k\in\mathcal I_M}\nu_1(\{\lambda_k\})\Big)^{-1/2}
\sum_{k\in\mathcal I_M}s_k^u\sqrt{\nu_1(\{\lambda_k\})}\,\bm u_k^\sharp,
\\
\label{eq:spec:init-v}
\tilde{\bm v}_1
&\bydef
\Big(\sum_{k\in\mathcal I_M}\nu_2(\{\lambda_k\})\Big)^{-1/2}
\sum_{k\in\mathcal I_M}s_k^v\sqrt{\nu_2(\{\lambda_k\})}\,\bm v_k^\sharp.
\end{align}

\begin{lemma}[Optimal spectral estimate as a $\bm W$-driven OAMP step]
\label{lem:spec-init-as-OAMP}
Assume Assumption~\ref{assump:main} and that $\theta$ is super-critical so that
$\mathcal K_\ast=\{\lambda\in\R\setminus\supp(\mu):\Gamma(\lambda)=0\}$ is finite and nonempty.
Then there exist deterministic (dimension-independent) functions
$\Psi_1,\widetilde\Psi_1,\Phi_1,\widetilde\Phi_1:\R\to\R$, whose restrictions to
the (compact) support $\supp(\mu)$ are continuous, such that
\begin{align}
\label{eq:spec-init-as-OAMP-u}
\tilde{\bm u}_1
&\explain{$N\to\infty$}{\simeq}
\Psi_1(\bm W\bm W^\UT)\,\bm u_\ast
+\widetilde\Psi_1(\bm W\bm W^\UT)\,\bm W\bm v_\ast,
\\
\label{eq:spec-init-as-OAMP-v}
\tilde{\bm v}_1
&\explain{$N\to\infty$}{\simeq}
\Phi_1(\bm W^\UT\bm W)\,\bm v_\ast
+\widetilde\Phi_1(\bm W^\UT\bm W)\,\bm W^\UT\bm u_\ast.
\end{align}
In particular, $(\tilde{\bm u}_1,\tilde{\bm v}_1)$ is (up to $\explain{$N\to\infty$}{\simeq}$) a valid
(one-shot) $\bm W$-driven step in the sense of Definition~\ref{def: Noise OAMP}.
\end{lemma}

\begin{proof}
Fix $k\in\mathcal I_M$ and define the normalized overlaps
\[
\alpha_{k,M}\bydef \frac{1}{\sqrt{M}}\langle \bm u_k,\bm u_\ast\rangle,
\qquad
\beta_{k,M}\bydef \frac{1}{\sqrt{N}}\langle \bm v_k,\bm v_\ast\rangle.
\]
By Lemma~\ref{lem:specSE:resolvent-outliers}, whenever
$\sigma_{k,M}^2\notin\spp(\bm W\bm W^\UT)$ and $\sigma_{k,M}^2\notin\spp(\bm W^\UT\bm W)$,
\begin{align}
\bm u_k
&=
\bm R_M(\sigma_{k,M}^2)\Big(
\theta\,\frac{\alpha_{k,M}}{\sqrt{N}}\,\bm W\bm v_\ast
+\theta\,\frac{\sigma_{k,M}\beta_{k,M}}{\sqrt{M}}\,\bm u_\ast
\Big),
\label{eq:spec:init-resolvent-u}
\\
\bm v_k
&=
\widetilde{\bm R}_N(\sigma_{k,M}^2)\Big(
\theta\,\frac{\beta_{k,M}}{\sqrt{M}}\,\bm W^\UT\bm u_\ast
+\theta\,\frac{\sigma_{k,M}\alpha_{k,M}}{\sqrt{N}}\,\bm v_\ast
\Big),
\label{eq:spec:init-resolvent-v}
\end{align}
where $\bm R_M(x)\bydef(x\bm I_M-\bm W\bm W^\UT)^{-1}$ and
$\widetilde{\bm R}_N(x)\bydef(x\bm I_N-\bm W^\UT\bm W)^{-1}$.

\emph{Uniform resolvent control and replacement of $\sigma_{k,M}^2$ by $\lambda_k$.}
Since $\lambda_k\in\R\setminus\supp(\mu)$ and Assumption~\ref{assump:main}(d) enforces spectral
containment of $\bm W\bm W^\UT$ and $\bm W^\UT\bm W$, there is an a.s.\ event on which, for all large $M$,
\[
\mathrm{d}\big(\sigma_{k,M}^2,\spp(\bm W\bm W^\UT)\big)\wedge
\mathrm{d}\big(\sigma_{k,M}^2,\spp(\bm W^\UT\bm W)\big)\ge c_k>0,
\]
and hence $\|\bm R_M(\sigma_{k,M}^2)\|_{\op}\vee\|\widetilde{\bm R}_N(\sigma_{k,M}^2)\|_{\op}\le c_k^{-1}$.
Moreover, by the resolvent identity (cf.~\eqref{Eqn:resolvent_identity}),
\[
\big\|\bm R_M(\sigma_{k,M}^2)-\bm R_M(\lambda_k)\big\|_{\op}
\le \frac{|\sigma_{k,M}^2-\lambda_k|}{c_k^2}\xrightarrow{\mathrm{a.s.}}0,
\]
and similarly for $\widetilde{\bm R}_N(\sigma_{k,M}^2)$.
Since $\|\bm u_\ast\|^2/M$ and $\|\bm v_\ast\|^2/N$ stay bounded and $\|\bm W\|_{\op}=O(1)$,
replacing $\bm R_M(\sigma_{k,M}^2)$ by $\bm R_M(\lambda_k)$ (and similarly for $\tilde{\bm{R}}_N$)
incurs negligible normalized MSE errors.

\emph{Replacement of scalar coefficients by deterministic limits.}
Proposition~\ref{prop:outlier_characterization} yields
$\alpha_{k,M}^2\to\nu_1(\{\lambda_k\})$ and $\beta_{k,M}^2\to\nu_2(\{\lambda_k\})$ almost surely, and
$\sigma_{k,M}\to\sqrt{\lambda_k}$.
Because $|\mathcal I_M|<\infty$, we may replace the finitely many scalar coefficients
$\alpha_{k,M},\beta_{k,M},\sigma_{k,M}$ in \eqref{eq:spec:init-resolvent-u}--\eqref{eq:spec:init-resolvent-v}
by their almost sure\ limits at the cost of an overall $o_{\mathrm{a.s.}}(1)$ normalized MSE error in the resulting sums.

\emph{Handling the (relative) sign convention.}
By consistency of $(s_k^u)_{k\in\mathcal I_M}$ and $(s_k^v)_{k\in\mathcal I_M}$ and finiteness of
$\mathcal I_M$, there exists an almost sure\ event on which, for all large $M$, the effective orientations in
\eqref{eq:spec:init-u}--\eqref{eq:spec:init-v} agree across all outliers up to a single global sign.
Equivalently, the only remaining ambiguity is the global orientation of
$(\tilde{\bm u}_1,\tilde{\bm v}_1)$; this ambiguity is already accounted for in the signed denoiser
convention and the definition of the state evolution variables used in Theorem~\ref{thm: spec-SE} (see discussions in Section \ref{sec:spec-optimal-oamp}).

Putting these points together, substituting \eqref{eq:spec:init-resolvent-u}--\eqref{eq:spec:init-resolvent-v}
into \eqref{eq:spec:init-u}--\eqref{eq:spec:init-v} shows that $\tilde{\bm u}_1$ and $\tilde{\bm v}_1$ admit representations of the form
\eqref{eq:spec-init-as-OAMP-u}--\eqref{eq:spec-init-as-OAMP-v}, with $\Psi_1(x),\,\widetilde\Psi_1(x),\,\Phi_1(x),\,\widetilde\Phi_1(x)$ related to finite linear combinations of the resolvent kernels $ x\mapsto\frac{1}{\lambda_k-x}$
and hence continuous on $\supp(\mu)$ because $\text{dist}(\lambda_k,\spp(\mu))>0$ for each outlier $\lambda_k$.
This proves the claim.
\end{proof}

%============================================================
\paragraph{Reduction to the $\bm W$-driven auxiliary recursion.}

Let $(\tilde{\bm u}_t,\tilde{\bm v}_t)_{t\ge 1}$ denote the spectrally-initialized $\bm Y$-driven OAMP
iterates defined in Theorem~\ref{thm: spec-SE}, with initialization
$(\tilde{\bm u}_1,\tilde{\bm v}_1)$ given by \eqref{eq:spec:init-u}--\eqref{eq:spec:init-v}.
Let $(\hat{\bm u}_t,\hat{\bm v}_t)_{t\ge 1}$ denote the $\bm W$-driven {auxiliary} recursion
constructed as \eqref{eq:aux-OAMP} (in Appendix~\ref{Sec:OAMP_SE_proof}), using the same iterate denoisers as the original recursion, and initialized by
\[
(\hat{\bm u}_1,\hat{\bm v}_1)\bydef(\tilde{\bm u}_1,\tilde{\bm v}_1).
\]

\begin{lemma}[Spiked-to-noise reduction for spectrally-initialized OAMP]
\label{lem:specSE:reduction}
For every fixed $T<\infty$,
\begin{equation}
\label{eq:specSE:reduction}
(\tilde{\bm u}_t,\tilde{\bm v}_t)\,\explain{$d\to\infty$}{\simeq}\,(\hat{\bm u}_t,\hat{\bm v}_t),
\qquad \forall\,t\le T.
\end{equation}
\end{lemma}

\begin{proof}
The proof is an iteration-by-iteration comparison, and is identical in structure to the induction carried out
in Appendix~\ref{Sec:OAMP_SE_proof} leading to \eqref{eq:OAMP-to-auxOAMP-induction}. For completeness,
we isolate the only inputs used at each inductive step.

First, decompose the iterate denoiser output into its signal-aligned and orthogonal parts, as in
\eqref{eq:original-OAMP-decomp} of Appendix~\ref{Sec:OAMP_SE_proof}; the induction maintains the
orthogonality relations required to apply Lemma~\ref{lem:aux1} there (the ``transfer lemma'' controlling
the difference between applying polynomial spectral denoisers to $\bm Y\bm Y^\UT$ versus $\bm W\bm W^\UT$,
and similarly on the $N$-side). Concretely, Lemma~\ref{lem:aux1} is invoked exactly as in the derivation
around \eqref{eq:aux-OAMP}: it converts each matrix-denoiser action on $\bm Y\bm Y^\UT$ (resp.\ $\bm Y^\UT\bm Y$)
applied to an orthogonal residual into the corresponding $\bm W$-counterpart, plus explicit
$\bm u_\ast$ or $\bm v_\ast$ aligned forcing terms; the remainder is $o_{\mathrm{a.s.}}(1)$ in normalized MSE.
The error control uses only (i) the uniform operator-norm bounds on the involved matrices
(cf.\ \eqref{Eqn:YY_operator} in Appendix~\ref{Sec:OAMP_SE_proof}), (ii) the regularity of the matrix and iterate denoisers and (iii) the propagated orthogonality
relations (cf.\ \eqref{Eqn:lemma6_app_orthgonality2} there).

The base case $t=1$ is valid by construction (since $\hat{\bm u}_1=\tilde{\bm u}_1$ and
$\hat{\bm v}_1=\tilde{\bm v}_1$), and Lemma~\ref{lem:spec-init-as-OAMP} ensures that this common
initialization already has the $\bm W$-driven OAMP form (with initialization independent of $\bm{W}$) required by the transfer lemma and the auxiliary
construction. The inductive step then proceeds verbatim as in
Appendix~\ref{Sec:OAMP_SE_proof}, yielding \eqref{eq:specSE:reduction} for any fixed $T$.
\end{proof}

%============================================================
\paragraph{State evolution and conclusion.}

By construction, the auxiliary recursion \eqref{eq:aux-OAMP} is a $\bm W$-driven OAMP algorithm
in the sense of Definition~\ref{def: Noise OAMP}, with initialization satisfying the admissibility
requirements thanks to Lemma~\ref{lem:spec-init-as-OAMP} (and with the usual global-sign convention
handled exactly as in Theorem~\ref{thm: spec-SE}). Therefore, the general OAMP state evolution
theorem (Theorem~\ref{Th:OAMP_general_SE_app} and its specialization used in
Appendix~\ref{Sec:OAMP_SE_proof}) applies to $(\hat{\bm u}_t,\hat{\bm v}_t)$ and yields the state
evolution limits claimed in Theorem~\ref{thm: spec-SE} for the auxiliary iterates.

Finally, Lemma~\ref{lem:specSE:reduction} transfers these limits from the auxiliary recursion back to the
original $\bm Y$-driven recursion.
This completes the proof of Theorem~\ref{thm: spec-SE}.
\qed

%% file: appendix/misc.tex
\section{Some Miscellaneous Results}\label{sec:misc}

\begin{fact}[Non-tangential limit as Point Mass]\label{fact: Non-tangential limit}
The following result, a variant of \cite[Proposition~8]{mingo2017free}, shows that the point mass of a finite measure $\nu$ on $\mathbb{R}$ can be recovered from its Stieltjes transform $\mathcal{S}(z)$. For any $a \in \mathbb{R}$:
\[
\lim_{\substack{z \to a \\ \sphericalangle}} (z - a)\mathcal{S}(z) = \nu(\{a\}),
\]
where the limit $z \to a$ is non-tangential, meaning it is restricted to any cone of the form $\{ x + iy \in \mathbb{C}^+ \mid y > 0 \text{ and } |x - a| < \gamma y \}$ for any $\gamma > 0$.
\end{fact}

\begin{fact}[Location of Empirical Outliers]\label{fact:loc of emp out}
Any positive empirical eigenvalues of $\bm{YY}^\UT$ which are not eigenvalues of $\bm{WW}^\UT$ are solution to the following equation for $z \in \R$
\begin{align}\label{eq: empirical master equation}
\Gamma_M(z) = (1 -  \frac{\theta}{\sqrt{MN}}\bm{v_*}^\UT \bm{W}^\UT \bm{S}_1(z)\bm{u_*})(1 -  \frac{\theta}{\sqrt{MN}} \bm{u_*}^\UT \bm{W} \bm{S}_2(z)\bm{v_*}) - z( \frac{\theta}{M}  \bm{u_*}^\UT \bm{S}_1(z)\bm{u_*})( \frac{\theta}{N} \bm{v_*}^\UT \bm{S}_2(z)\bm{v_*}) =0,
\end{align}
where $\bm{S}_1(z)\bydef (z\bm{I}_M -\bm{WW}^\UT)^{-1}, \bm{S}_2(z)\bydef (z\bm{I}_N -\bm{W^\UT W})^{-1}$.
\end{fact}

\begin{proof}[Proof of \factref{fact:loc of emp out}]
A positive number $z$ is an eigenvalue of $\bm{YY}^\UT$ if and only if its square root, $\sigma = \sqrt{z}$, is a singular value of $\bm{Y}$. This is equivalent to the augmented matrix $\begin{pmatrix} \bm{0} & \bm{Y} \\ \bm{Y}^\UT & \bm{0} \end{pmatrix}$ having $\sigma$ as an eigenvalue. The characteristic equation for this condition is 
\begin{align}
    \det\begin{pmatrix} \sigma\bm{I}_M & -\bm{Y} \\ -\bm{Y}^\UT & \sigma\bm{I}_N \end{pmatrix} = 0.
\end{align}
Substituting the spiked model for $\bm{Y}$ from  \eqref{eq:rectangular spiked model}, we get
\begin{align}
    \det\left( \begin{pmatrix} \sigma\bm{I}_M & -\bm{W} \\ -\bm{W}^\UT & \sigma\bm{I}_N \end{pmatrix} - \begin{pmatrix} \bm{0} & \frac{\theta}{\sqrt{MN}}\bm{u_*}\bm{v_*}^\UT \\ \frac{\theta}{\sqrt{MN}}\bm{v_*}\bm{u_*}^\UT & \bm{0} \end{pmatrix} \right) = 0.
\end{align}
Let us define the unperturbed matrix $\widehat{\bm{W}} \bydef \begin{pmatrix} \sigma\bm{I}_M & -\bm{W} \\ -\bm{W}^\UT & \sigma\bm{I}_N \end{pmatrix}$. The perturbation is a rank-2 matrix which can be factored as $\widehat{\bm{U}}\widehat{\bm{\Theta}}\widehat{\bm{U}}^\UT$, where
\begin{align}
    \widehat{\bm{U}} \bydef \begin{pmatrix} \bm{u_*} & \bm{0} \\ \bm{0} & \bm{v_*} \end{pmatrix}, \quad \text{and} \quad
    \widehat{\bm{\Theta}} \bydef \frac{\theta}{\sqrt{MN}} \begin{pmatrix} 0 & 1 \\ 1 & 0 \end{pmatrix}.
\end{align}
The characteristic equation is now $\det(\widehat{\bm{W}} - \widehat{\bm{U}}\widehat{\bm{\Theta}}\widehat{\bm{U}}^\UT) = 0$. By the Weinstein-Aronszajn formula, this is equivalent to
\begin{align}
\det(\widehat{\bm{W}} - \widehat{\bm{U}}\widehat{\bm{\Theta}}\widehat{\bm{U}}^\UT) =    \det(\widehat{\bm{W}}) \det(\bm{I}_2 - \widehat{\bm{\Theta}}\widehat{\bm{U}}^\UT\widehat{\bm{W}}^{-1}\widehat{\bm{U}}) = 0.
\end{align}
From Assumption \ref{assump:main}(c), the Lemma's premise is that $z$ is not an eigenvalue of $\bm{WW}^\UT$. This ensures that $\widehat{\bm{W}}$ is invertible, so $\det(\widehat{\bm{W}}) \neq 0$. The condition thus simplifies to the singularity of the $2\times2$ matrix
\begin{align}
    \det(\bm{I}_2 - \widehat{\bm{\Theta}}\widehat{\bm{U}}^\UT\widehat{\bm{W}}^{-1}\widehat{\bm{U}}) = 0.
\end{align}
The inverse of $\widehat{\bm{W}}$ is given by the Shur complement formula
\begin{align}
    \widehat{\bm{W}}^{-1} = \begin{pmatrix} \sigma(z\bm{I}_M - \bm{WW}^\UT)^{-1} & \bm{W}(z\bm{I}_N - \bm{W}^\UT\bm{W})^{-1} \\ \bm{W}^\UT(z\bm{I}_M - \bm{WW}^\UT)^{-1} & \sigma(z\bm{I}_N - \bm{W}^\UT\bm{W})^{-1} \end{pmatrix} \bydef \begin{pmatrix} \sigma \bm{S}_1(z) & \bm{W}\bm{S}_2(z) \\ \bm{W}^\UT\bm{S}_1(z) & \sigma \bm{S}_2(z) \end{pmatrix},
\end{align}
where we have used $z=\sigma^2$ and the definitions of $\bm{S}_1(z)$ and $\bm{S}_2(z)$. The $2\times2$ core matrix is then
\BS
\begin{align}
\widehat{\bm{U}}^\UT\widehat{\bm{W}}^{-1}\widehat{\bm{U}} &= \begin{pmatrix} \bm{u_*}^\UT & \bm{0} \\ \bm{0} & \bm{v_*}^\UT \end{pmatrix} \begin{pmatrix} \sigma \bm{S}_1(z) & \bm{W}\bm{S}_2(z) \\ \bm{W}^\UT\bm{S}_1(z) & \sigma \bm{S}_2(z) \end{pmatrix} \begin{pmatrix} \bm{u_*} & \bm{0} \\ \bm{0} & \bm{v_*} \end{pmatrix} = \begin{pmatrix} \sigma \bm{u_*}^\UT \bm{S}_1(z) \bm{u_*} & \bm{u_*}^\UT \bm{W}\bm{S}_2(z) \bm{v_*} \\ \bm{v_*}^\UT \bm{W}^\UT\bm{S}_1(z) \bm{u_*} & \sigma \bm{v_*}^\UT \bm{S}_2(z) \bm{v_*} \end{pmatrix}.
\end{align}
\ES
Substituting this into the determinant condition gives
\begin{align}
    \det\left( \bm{I}_2 - \frac{\theta}{\sqrt{MN}} \begin{pmatrix} 0 & 1 \\ 1 & 0 \end{pmatrix} \begin{pmatrix} \sigma\bm{u_*}^\UT \bm{S}_1(z) \bm{u_*} & \bm{u_*}^\UT \bm{W}\bm{S}_2(z) \bm{v_*} \\ \bm{v_*}^\UT \bm{W}^\UT\bm{S}_1(z) \bm{u_*} & \sigma \bm{v_*}^\UT \bm{S}_2(z) \bm{v_*} \end{pmatrix} \right) = 0.
\end{align}
This expands to the determinant of the following $2\times2$ matrix
\begin{align}
    \det \begin{pmatrix} 1 - \frac{\theta}{\sqrt{MN}} \bm{v_*}^\UT \bm{W}^\UT\bm{S}_1(z) \bm{u_*} & - \frac{z\theta}{\sqrt{MN}} \bm{v_*}^\UT \bm{S}_2(z) \bm{v_*} \\ - \frac{z\theta}{\sqrt{MN}} \bm{u_*}^\UT \bm{S}_1(z) \bm{u_*} & 1 - \frac{\theta}{\sqrt{MN}} \bm{u_*}^\UT \bm{W}\bm{S}_2(z) \bm{v_*} \end{pmatrix} = 0.
\end{align}
Evaluating the determinant and recalling that $\sigma^2=z$ yields \eqref{eq: empirical master equation}.
\end{proof}

\begin{fact}[Special Form of the Sokhotski–Plemelj Formula]
\label{fact:special_SP_hilbert_form}
Let $\nu$ be a finite signed measure on $\mathbb{R}$, with its Stieltjes transform $\mathcal{S}_\nu$ and Hilbert transform $\mathcal{H}_\nu$ given by \eqref{eq:def of Stieltjes Transform} and \eqref{eq: Hilbert Transform Def}, respectively. For any non-zero point $x \in \mathbb{R} \setminus \{0\}$ where the measure $\nu$ possesses a density at $x^2$, denoted by $\nu(x^2)$, the following boundary limit for the Stieltjes transform holds
\begin{equation}
\lim_{\epsilon \to 0^+} \mathcal{S}_{\nu}((x - \mathrm{i}\epsilon)^2) = \pi \mathcal{H}_{\nu}(x^2) + \mathrm{i}\pi \, \mathrm{sign}(x) \nu(x^2).
\label{eq:special_inversion_hilbert}
\end{equation}
\end{fact}
\begin{proof}
Fix $x\in\mathbb{R}\setminus\{0\}$ such that the density $\frac{d\nu}{d\lambda}(x^2)$ exists, and denote this value by $\nu(x^2)$. Since $\nu$ is a finite signed measure, there exist finite positive measures $\nu^+$ and $\nu^-$ such that $\nu = \nu^+ - \nu^-$ by Jordan decomposition (cf.\ \cite[Theorem~4.1.5]{Cohn2013MeasureTheory}). By linearity of the Stieltjes transform, the Hilbert transform and the Radon--Nikodym derivative, this implies
\begin{align}\label{eq:linear-S-transform}
  \mathcal{S}_\nu = \mathcal{S}_{\nu^+} - \mathcal{S}_{\nu^-},\qquad
  \mathcal{H}_\nu = \mathcal{H}_{\nu^+} - \mathcal{H}_{\nu^-},\qquad
  \frac{d\nu}{d\lambda} = \frac{d\nu^+}{d\lambda} - \frac{d\nu^-}{d\lambda}.    
\end{align}
For a finite positive measure $\chi$ and for Lebesgue-almost every $t\in\mathbb{R}$ such that $\frac{d\chi}{d\lambda}(t)$ exists, the Sokhotski–Plemelj boundary value theorem (cf.\ \cite[Section~3.1]{belinschi2017outliers}) yields
\begin{align}
  \lim_{\epsilon\downarrow0}\mathcal{S}_\chi(t - \mathrm{i}\epsilon)
  &= \pi\mathcal{H}_\chi(t) + \mathrm{i}\pi\,\frac{d\chi}{d\lambda}(t),
  \label{eq:SP-pos-below}\\
  \lim_{\epsilon\downarrow0}\mathcal{S}_\chi(t + \mathrm{i}\epsilon)
  &= \pi\mathcal{H}_\chi(t) - \mathrm{i}\pi\,\frac{d\chi}{d\lambda}(t).
  \label{eq:SP-pos-above}
\end{align}
Applying \eqref{eq:SP-pos-below}–\eqref{eq:SP-pos-above} to $\nu^+$ and $\nu^-$ at $t=x^2$, and using the relations \eqref{eq:linear-S-transform}, we obtain
\begin{align}
  \lim_{\epsilon\downarrow0}\mathcal{S}_\nu(x^2 - \mathrm{i}\epsilon)
  &= \pi\mathcal{H}_\nu(x^2) + \mathrm{i}\pi\,\nu(x^2),
  \label{eq:SP-nu-below}\\
  \lim_{\epsilon\downarrow0}\mathcal{S}_\nu(x^2 + \mathrm{i}\epsilon)
  &= \pi\mathcal{H}_\nu(x^2) - \mathrm{i}\pi\,\nu(x^2).
  \label{eq:SP-nu-above}
\end{align}
Now consider the path
\begin{align}\label{eq:path-w}
  w(\epsilon) = (x - \mathrm{i}\epsilon)^2
  = (x^2 - \epsilon^2) - \mathrm{i}\,2x\epsilon,\qquad \epsilon>0.
\end{align}
Then $w(\epsilon)\to x^2$ as $\epsilon\downarrow0$. Moreover,
\[
\Im w(\epsilon) = -2x\epsilon,
\]
so for $x>0$ the points $w(\epsilon)$ lie in the lower half-plane, while for $x<0$ they lie in the upper half-plane.\\
\medskip
\noindent\textit{Case $x>0$.}
Here $\Im w(\epsilon)<0$ for all $\epsilon>0$, and $w(\epsilon)\to x^2$ from the lower half-plane. Since $\mathcal{S}_\nu$ is analytic on $\mathbb{C}\setminus\mathbb{R}$, the boundary value in \eqref{eq:SP-nu-below} is independent of the particular approach within the lower half-plane. Hence
\begin{align}\label{eq:limit-x-pos}
\lim_{\epsilon\downarrow0}\mathcal{S}_\nu\bigl(w(\epsilon)\bigr)
  = \lim_{\epsilon\downarrow0}\mathcal{S}_\nu(x^2 - \mathrm{i}\epsilon)
  = \pi\mathcal{H}_\nu(x^2) + \mathrm{i}\pi\,\nu(x^2).
\end{align}\\
\medskip
\noindent\textit{Case $x<0$.}
Now $\Im w(\epsilon)>0$ for all $\epsilon>0$, and $w(\epsilon)\to x^2$ from the upper half-plane. Using \eqref{eq:SP-nu-above} and the same reasoning,
\begin{align}\label{eq:limit-x-neg}
\lim_{\epsilon\downarrow0}\mathcal{S}_\nu\bigl(w(\epsilon)\bigr)
  = \lim_{\epsilon\downarrow0}\mathcal{S}_\nu(x^2 + \mathrm{i}\epsilon)
  = \pi\mathcal{H}_\nu(x^2) - \mathrm{i}\pi\,\nu(x^2).
\end{align}
\medskip
Combining \eqref{eq:limit-x-pos} and \eqref{eq:limit-x-neg}, and recalling $w(\epsilon)=(x-\mathrm{i}\epsilon)^2$, we obtain for all $x\neq0$ with density at $x^2$:
\[
\lim_{\epsilon\downarrow0}\mathcal{S}_\nu\bigl((x-\mathrm{i}\epsilon)^2\bigr)
= \pi\mathcal{H}_\nu(x^2) + \mathrm{i}\pi\,\mathrm{sign}(x)\,\nu(x^2),
\]
which is exactly \eqref{eq:special_inversion_hilbert}.
\end{proof}

\begin{fact}[Boundary values for the Stieltjes transform of a signed measure]
\label{fact:boundary-signed}
Let $\chi$ be a finite signed Borel measure on $\R$, and let
\[
\mathcal{S}_\chi(z)
\;\bydef\;
\int_{\R}\frac{1}{z-x}\,d\chi(x),
\qquad z\in\C\setminus\R,
\]
denote its Stieltjes transform. Let
\[
\chi = \chi_{\parallel} + \chi_{\perp}
\]
be the Lebesgue decomposition of $\chi$ with respect to Lebesgue measure
$\lambda$, where $\chi_{\parallel}\ll\lambda$ and $\chi_{\perp}\perp\lambda$.
Then the boundary values of $\mathcal{S}_\chi$ satisfy
\begin{align}
\lim_{\epsilon\downarrow0}
  \Im\mathcal{S}_\chi(x-\mathrm{i}\epsilon)
&=
\pi\,\frac{d\chi_{\parallel}}{d\lambda}(x),
&&\text{for Lebesgue-a.e.\ }x\in\R,
\label{eq:boundary-chi-ac} % (108)
\\[2mm]
\lim_{\epsilon\downarrow0}
  \bigl|\Im\mathcal{S}_\chi(x-\mathrm{i}\epsilon)\bigr|
&=\infty,
&&\text{for }|\chi_{\perp}|\text{-almost every }x\in\R.
\label{eq:boundary-chi-sing} % (109)
\end{align}
Here $|\chi_{\perp}|$ denotes the total variation measure of the signed
measure $\chi_{\perp}$. If $\chi_{\perp}=\chi_{\perp}^+-\chi_{\perp}^-$ is the Jordan decomposition
of $\chi_{\perp}$ into mutually singular finite positive measures, then\begin{equation}\label{eq:total-variation-jordan}
|\chi_{\perp}|=\chi_{\perp}^+ + \chi_{\perp}^-.
\end{equation}
In particular, \eqref{eq:boundary-chi-sing} is equivalent to the existence
of a Borel set $N\subset\R$ such that
\[
|\chi_{\perp}|(N)=0
\quad\text{and}\quad
\lim_{\epsilon\downarrow0}
  \bigl|\Im\mathcal{S}_\chi(x-\mathrm{i}\epsilon)\bigr|
= \infty
\quad\text{for all }x\in\R\setminus N.
\]
\end{fact}
\begin{proof}
Let $\chi=\chi^+-\chi^-$ be the Jordan decomposition, where $\chi^\pm$ are finite positive
Borel measures with $\chi^+\perp\chi^-$. For each $\chi^\pm$, let
\[
\chi^\pm=\chi^\pm_{\parallel}+\chi^\pm_{\perp}
\]
be the Lebesgue decomposition with respect to $\lambda$.
By uniqueness of the Lebesgue decomposition (see, e.g., \cite[Thm.~4.3.2]{Cohn2013MeasureTheory}),
\[
\chi_{\parallel}=\chi^+_{\parallel}-\chi^-_{\parallel},
\qquad
\chi_{\perp}=\chi^+_{\perp}-\chi^-_{\perp}.
\]

We first establish \eqref{eq:boundary-chi-ac}. For each finite positive measure $\chi^\pm$, the boundary-value formula for its Stieltjes transform
(see \cite[Section~3.1]{belinschi2017outliers}, with the convention $x-\mathrm{i}\epsilon$) yields
\begin{equation}\label{eq:boundary-chipm-ac}
\lim_{\epsilon\downarrow0}\Im\mathcal{S}_{\chi^\pm}(x-\mathrm{i}\epsilon)
=\pi\,\frac{d\chi^\pm_{\parallel}}{d\lambda}(x),
\qquad \text{for Lebesgue-a.e.\ }x\in\R.
\end{equation}
Since $\mathcal{S}_\chi=\mathcal{S}_{\chi^+}-\mathcal{S}_{\chi^-}$ and both the Stieltjes transform
and the Radon--Nikod\'ym derivative are linear in the measure, \eqref{eq:boundary-chipm-ac} implies,
for Lebesgue-a.e.\ $x\in\R$,
\[
\lim_{\epsilon\downarrow0}\Im\mathcal{S}_\chi(x-\mathrm{i}\epsilon)
=
\pi\!\left(\frac{d\chi^+_{\parallel}}{d\lambda}(x)-\frac{d\chi^-_{\parallel}}{d\lambda}(x)\right)
=\pi\,\frac{d\chi_{\parallel}}{d\lambda}(x),
\]
which is \eqref{eq:boundary-chi-ac}.

We now prove \eqref{eq:boundary-chi-sing}. Let the Jordan decomposition of the singular part be
\[
\chi_{\perp}=\chi_{\perp}^+-\chi_{\perp}^-,
\]
where $\chi_{\perp}^\pm$ are finite positive measures with $\chi_{\perp}^+\perp\chi_{\perp}^-$.
Then $|\chi_{\perp}|=\chi_{\perp}^+ + \chi_{\perp}^-$ by \eqref{eq:total-variation-jordan}.

For any finite signed Borel measure $\eta$ on $\R$ and $\epsilon>0$,
\[
\Im\mathcal{S}_\eta(x-\mathrm{i}\epsilon)
=
\epsilon\int_{\R}\frac{1}{(x-t)^2+\epsilon^2}\,d\eta(t).
\]
In particular, for $\eta\ge 0$ this is (up to the factor $\pi$) the Poisson integral of $\eta$.
Applying \cite[Section~3.1]{belinschi2017outliers} to the purely singular finite positive measures
$\chi_{\perp}^\pm$ yields
\[
\Im\mathcal{S}_{\chi_{\perp}^+}(x-\mathrm{i}\epsilon)\xrightarrow[\epsilon\downarrow0]{}+\infty
\quad\text{for }\chi_{\perp}^+\text{-a.e.\ }x\in\R,
\qquad
\Im\mathcal{S}_{\chi_{\perp}^-}(x-\mathrm{i}\epsilon)\xrightarrow[\epsilon\downarrow0]{}+\infty
\quad\text{for }\chi_{\perp}^-\text{-a.e.\ }x\in\R.
\]
Define
\[
u^+(x,\epsilon)\;\bydef\;\pi^{-1}\Im\mathcal{S}_{\chi_{\perp}^+}(x-\mathrm{i}\epsilon),
\qquad
u^-(x,\epsilon)\;\bydef\;\pi^{-1}\Im\mathcal{S}_{\chi_{\perp}^-}(x-\mathrm{i}\epsilon).
\]
By \cite[Lemma~1]{Watson1994Applications}, applied to the mutually singular finite positive measures
$\chi_{\perp}^+$ and $\chi_{\perp}^-$, as $\epsilon\downarrow0$ we have
\begin{align*}
u^-(x,\epsilon)&=o\bigl(u^+(x,\epsilon)\bigr),
&&\text{for }\chi_{\perp}^+\text{-a.e.\ }x\in\R,\\
u^+(x,\epsilon)&=o\bigl(u^-(x,\epsilon)\bigr),
&&\text{for }\chi_{\perp}^-\text{-a.e.\ }x\in\R.
\end{align*}

Since the Poisson kernel is nonnegative, for all $x\in\R$ and $\epsilon>0$ we have the domination
\begin{equation}\label{eq:ImS-TV-domination}
\bigl|\Im\mathcal{S}_{\chi_{\parallel}}(x-\mathrm{i}\epsilon)\bigr|
\le
\Im\mathcal{S}_{|\chi_{\parallel}|}(x-\mathrm{i}\epsilon).
\end{equation}
Set
\[
u^{\parallel}(x,\epsilon)\;\bydef\;\pi^{-1}\Im\mathcal{S}_{|\chi_{\parallel}|}(x-\mathrm{i}\epsilon)\;\ge\;0.
\]
Moreover, $|\chi_{\parallel}|\ll\lambda$ while $\chi_{\perp}^\pm\perp\lambda$, hence
$|\chi_{\parallel}|\perp\chi_{\perp}^\pm$. Another application of \cite[Lemma~1]{Watson1994Applications}
therefore yields, as $\epsilon\downarrow0$,
\begin{align*}
u^{\parallel}(x,\epsilon)&=o\bigl(u^+(x,\epsilon)\bigr),
&&\text{for }\chi_{\perp}^+\text{-a.e.\ }x\in\R,\\
u^{\parallel}(x,\epsilon)&=o\bigl(u^-(x,\epsilon)\bigr),
&&\text{for }\chi_{\perp}^-\text{-a.e.\ }x\in\R.
\end{align*}

Combining \eqref{eq:ImS-TV-domination} with the preceding little-$o$ relations, we obtain
\begin{align}
\Im \mathcal{S}_\chi(x-\mathrm{i}\epsilon)
&=
\Im\mathcal{S}_{\chi_{\parallel}}(x-\mathrm{i}\epsilon)
+\pi u^+(x,\epsilon)-\pi u^-(x,\epsilon)\notag\\
&\sim
\pi\,u^+(x,\epsilon)
\xrightarrow[\epsilon\downarrow0]{}+\infty,
&&\text{for }\chi_{\perp}^+\text{-a.e.\ }x\in\R,
\label{eq:boundary-chi-plus}\\
\Im\mathcal{S}_\chi(x-\mathrm{i}\epsilon)
&\sim
-\pi\,u^-(x,\epsilon)
\xrightarrow[\epsilon\downarrow0]{}-\infty,
&&\text{for }\chi_{\perp}^-\text{-a.e.\ }x\in\R.
\label{eq:boundary-chi-minus}
\end{align}
In particular,
\[
\lim_{\epsilon\downarrow0}\bigl|\Im\mathcal{S}_\chi(x-\mathrm{i}\epsilon)\bigr|
=\infty
\quad\text{for }\chi_{\perp}^+\text{-a.e.\ and }\chi_{\perp}^-\text{-a.e.\ }x\in\R.
\]
Since $|\chi_{\perp}|=\chi_{\perp}^+ + \chi_{\perp}^-$, the same property holds for
$|\chi_{\perp}|$-almost every $x\in\R$, which is \eqref{eq:boundary-chi-sing}.
\end{proof}
\medskip

%============================================
\begin{fact}[Resolvent representation of singular vectors]
\label{lem:specSE:resolvent-outliers}
Let $(\sigma,{\bm u},{\bm v})$ be a singular value--vector triplet of
\[
\bm Y \;=\; \frac{\theta}{\sqrt{MN}}\,{\bm u}_*{\bm v}_*^{\mathsf T} + \bm W
\;\in\; \mathbb{R}^{M\times N}.
\]
Assume that $\sigma^2 \notin \mathrm{spec}(\bm W\bm W^{\mathsf T}) \cup \{0\}$.
Then,
\begin{align}
{\bm u}
&= \bm R_M(\sigma^2)
\Bigg(
\theta\,\frac{\langle {\bm u},{\bm u}_*\rangle}{\sqrt{MN}}\,
\bm W{\bm v}_*
+ \theta\,\sigma\,\frac{\langle {\bm v},{\bm v}_*\rangle}{\sqrt{MN}}\,
{\bm u}_*
\Bigg),
\label{eq:specSE:u-resolvent-solved}\\
{\bm v}
&= \widetilde{\bm R}_N(\sigma^2)
\Bigg(
\theta\,\frac{\langle {\bm v},{\bm v}_*\rangle}{\sqrt{MN}}\,
\bm W^{\mathsf T}{\bm u}_*
+ \theta\,\sigma\,\frac{\langle {\bm u},{\bm u}_*\rangle}{\sqrt{MN}}\,
{\bm v}_*
\Bigg),
\label{eq:specSE:v-resolvent-solved}
\end{align}
where
\[
\bm R_M(x) \bydef (x\bm I_M - \bm W\bm W^{\mathsf T})^{-1},
\qquad
\widetilde{\bm R}_N(x) \bydef (x\bm I_N - \bm W^{\mathsf T}\bm W)^{-1}.
\]
\end{fact}
\begin{proof}[Proof of \factref{lem:specSE:resolvent-outliers}]
This follows directly from the singular equation and we omit the details.
\end{proof}

\medskip
%============================================

\begin{fact}[Properties of the signed DMMSE]\label{fact:signed-dmmse-sign}
Let
\begin{align}
\mathsf X^{\mathrm{sgn}} &\bydef s_*\sqrt{w}\,\mathsf X_*+\sqrt{1-w}\,\mathsf Z, \label{eq:ap-signed-channel}\\
\mathsf X^{\mathrm{std}} &\bydef s_*\mathsf X^{\mathrm{sgn}} \stackrel{d}{=} \sqrt{w}\,\mathsf X_*+\sqrt{1-w}\,\mathsf Z,\label{eq:ap-mirror-channel}
\end{align}
be scalar Gaussian channels with $w\in[0,1)$, $s_*\in\{\pm1\}$, $\E[\mathsf X_*^2]=1$, and
$\mathsf Z\sim\mathcal N(0,1)\indep \mathsf X_*$. Under \assumpref{assump:lip-MMSE}, let $\bar\phi(\cdot\,|\,w)$ be the DMMSE estimator
associated with \eqref{eq:ap-mirror-channel} defined in \eqref{eq:mmse-scalar}, and define the signed DMMSE
estimator associated with \eqref{eq:ap-signed-channel} by
\begin{equation}\label{eq:signed_DMMSE}
\bar\phi(x\,|\,w,s)\;\bydef\; \bar\phi(sx\,|\,w),\qquad s\in\{\pm1\}.
\end{equation}
Then:
\begin{enumerate}[label=\textup{(\roman*)},leftmargin=*]
\item (Matched Sign.) If $s=s_*$, then
\begin{equation}\label{eq:DMMSE-match}
\mathbb E\!\left[\mathsf X_*\,\bar\phi(\mathsf X^{\mathrm{sgn}}\,|\,w,s_*)\right]\ge 0,
\end{equation}
with strict positivity for non-degenerate $\mathsf X_*$ and $w\in(0,1)$.
\item (Possible Mismatched Sign.) If $\mathsf X_*\stackrel{d}{=}-\mathsf X_*$, then $\bar\phi(\cdot\,|\,w)$ is odd and, for any $s\in\{\pm1\}$,
\begin{equation}\label{eq:DMMSE-mismatch}
\mathbb E\!\left[\mathsf X_*\,\bar\phi(\mathsf X^{\mathrm{sgn}}\,|\,w,s)\right]
= ss_*\ \mathbb E\!\left[\bar\phi(\mathsf X^{\mathrm{std}}\,|\,w)^2\right].
\end{equation}
Hence the correlation flips sign with $ss_*$ while its magnitude is unchanged.
\end{enumerate}
\end{fact}

\begin{proof}[Proof of \factref{fact:signed-dmmse-sign}]
We use the DMMSE projection identity (see \cite[Appendix~A.2]{dudeja2024optimality}): for the standard channel
\eqref{eq:ap-mirror-channel},
\begin{equation}\label{eq:DMMSE-proj}
\mathbb E\!\left[\mathsf X_*\,\bar\phi(\mathsf X^{\mathrm{std}}\,|\,w)\right]
=\mathbb E\!\left[\bar\phi(\mathsf X^{\mathrm{std}}\,|\,w)^2\right].
\end{equation}

\smallskip
\noindent\textbf{Proof of (i).}
If $s=s_*$, then $\bar\phi(\mathsf X^{\mathrm{sgn}}\,|\,w,s_*)=\bar\phi(s_*\mathsf X^{\mathrm{sgn}}\,|\,w)=\bar\phi(\mathsf X^{\mathrm{std}}\,|\,w)$, and thus
\[
\mathbb E\!\left[\mathsf X_*\,\bar\phi(\mathsf X^{\mathrm{sgn}}\,|\,w,s_*)\right]
=\mathbb E\!\left[\bar\phi(\mathsf X^{\mathrm{std}}\,|\,w)^2\right]\ge 0,
\]
where the equality follows from \eqref{eq:DMMSE-proj}.

\smallskip
\noindent\textbf{Proof of (ii).}
If $\mathsf X_*\stackrel{d}{=}-\mathsf X_*$, then $\phi(\cdot\,|\,w)$ is odd by symmetry of
\eqref{eq:ap-mirror-channel}. Moreover, $\bar\phi(\cdot\,|\,w)$ is odd as well, since by \eqref{eq:mmse-scalar}
it is an affine combination of $\phi(\cdot\,|\,w)$ and $x$ with coefficients depending only on $w$.
By \eqref{eq:ap-mirror-channel} and \eqref{eq:signed_DMMSE}, and using the oddness of $\bar\phi(\cdot|w)$, we have
\[
\bar\phi(\mathsf X^{\mathrm{sgn}}\,|\,w,s)
=\bar\phi(ss_*\mathsf X^{\mathrm{std}}\,|\,w)
=ss_*\,\bar\phi(\mathsf X^{\mathrm{std}}\,|\,w).
\]
Therefore, applying \eqref{eq:DMMSE-proj},
\[
\mathbb E\!\left[\mathsf X_*\,\bar\phi(\mathsf X^{\mathrm{sgn}}\,|\,w,s)\right]
=ss_*\mathbb E\!\left[\mathsf X_*\,\bar\phi(\mathsf X^{\mathrm{std}}\,|\,w)\right]
=ss_*\mathbb E\!\left[\bar\phi(\mathsf X^{\mathrm{std}}\,|\,w)^2\right],
\]
which is \eqref{eq:DMMSE-mismatch}.
\end{proof}

\begin{remark}\label{rem:signed-dmmse-sign}
Fact~\ref{fact:signed-dmmse-sign} is used to select $s$ in the signed DMMSE \eqref{eq:signed-family} from a global-sign estimate (cf.\ Appendix~\ref{app:pf-global_sign}):
\begin{itemize}[leftmargin=*]
\item \emph{Asymmetric priors.} The ground truth $s_*$ is consistently estimable. Choosing $s$ accordingly yields a positively correlated signed DMMSE update, as in \eqref{eq:DMMSE-match}.
\item \emph{Symmetric priors.} The ground truth $s_*$ is not identifiable from a single channel. By \eqref{eq:DMMSE-mismatch}, a sign mismatch ($s=-s_*$) induces only a global flip of the output, while leaving its squared magnitude unchanged.
\end{itemize}
\end{remark}